\documentclass[11pt]{amsart}
\usepackage{enumitem}
\usepackage{amsthm}
\usepackage[paperheight=11in,paperwidth=8.5in,left=1in,top=1.25in,right=1in,bottom=1.25in,]{geometry}
 \usepackage{relsize}
\usepackage{rotating}
\usepackage[table]{xcolor}

\usepackage[super]{nth}
\usepackage[none]{hyphenat}
\usepackage{hyperref}
\hypersetup{
	colorlinks,
	citecolor=black,
	filecolor=black,
	linkcolor=black,
	urlcolor=black
}
\usepackage[capitalize, nameinlink, noabbrev, nosort]{cleveref}
\usepackage{graphicx, amssymb}
\usepackage{epstopdf} 
\usepackage{youngtab}
\usepackage[boxsize=.35 em]{ytableau}
\usepackage{comment}
\usepackage{setspace}
\usepackage{multicol}

\usepackage{tikz}   
\usetikzlibrary{arrows,calc,intersections,matrix,positioning,through}
\usepackage{tikz-cd}
\usepackage{comment}
\usepackage{diagbox}
\usepackage{blkarray}
\tikzset{commutative diagrams/.cd,every label/.append style = {font = \normalsize}}

\usepgflibrary{shapes.geometric}
\usetikzlibrary{shapes.geometric}
\usetikzlibrary{backgrounds}

\DeclareMathOperator{\Gr}{Gr}
\DeclareMathOperator{\Fl}{Fl}
\DeclareMathOperator{\GL}{GL}
\DeclareMathOperator{\SL}{SL}
\DeclareMathOperator{\Mat}{Mat}
\DeclareMathOperator{\sgn}{\text{sgn}}
\DeclareMathOperator{\rk}{\text{rank}}

\DeclareMathOperator{\Span}{\text{span}}

\DeclareMathOperator{\real}{\text{real}}

\newcommand{\bulR}{\triangleleft}
\newcommand{\ahead}{{\text{ahead}}}
\newcommand{\outside}{{\text{outside}}}
\newcommand{\woutside}{{\text{weak-outside}}}

\newcommand{\CC}{\mathbf{C}}

\newcommand{\abe}{{\delta_0}}

\newcommand{\ww}{{w}}
\newcommand{\pre}{{\text{pre}}}
\newcommand{\FL}{{\text{FL}}}

\DeclareMathOperator{\inc}{{\text{inc}}}
\DeclareMathOperator{\rem}{{\text{rem}}}
\DeclareMathOperator{\shift}{{\text{shift}}}
\DeclareMathOperator{\addL}{{\square}}

\newcommand{\Blue}{{\color{blue}\text{Blue}}}
\newcommand{\Black}{{\color{black}\text{Black}}}
\newcommand{\Red}{{\color{red}\text{Red}}}
\newcommand{\Purp}{{\color{purple}\text{Purple}}}
\newcommand{\tr}{{\mathrm{tr}}}

\newcommand{\CD}{\mathcal{CD}}

\newcommand{\bz}{\mathbf{z}}

\newcommand{\bw}{\mathbf{w}}

\newcommand{\ee}{\mathrm{e}}
\newcommand{\D}{{\mathcal{D}}}
\newcommand{\PosC}{\mathcal{P}^C}
\newcommand{\PosD}{\mathcal{P}^{C+D}}

\newcommand{\PosBC}{\mathcal{B}^C}
\newcommand{\PosBD}{\mathcal{B}^{C+D}}

\newcommand{\p}{\text{p}}
\newcommand{\match}{\text{match}}

\newcommand{\pr}{\text{pr}}

\newcommand{\BCFWV}{\text{BCFWVar}}
\newcommand{\BBCFWV}{\overline{\text{BCFWVar}}}

\newcommand{\Var}{{\text{Var}}}
\newcommand{\tVar}{\widetilde{\text{{Var}}}}

\newcommand{\halp}{{\hat\alpha}}
\newcommand{\hbet}{{\hat\beta}}
\newcommand{\hgam}{{\hat\gamma}}
\newcommand{\hdel}{{\hat\delta}}
\newcommand{\heps}{{\hat\varepsilon}}
\newcommand{\hzet}{{\hat\zeta}}
\newcommand{\hsigm}{{\hat\sigma}}

\newcommand{\balp}{{\bar\alpha}}
\newcommand{\bbet}{{\bar\beta}}
\newcommand{\bgam}{{\bar\gamma}}
\newcommand{\bdel}{{\bar\delta}}
\newcommand{\beps}{{\bar\varepsilon}}
\newcommand{\bzet}{{\bar\zeta}}
\newcommand{\bsigm}{{\bar\sigma}}

\newcommand{\Fdel}{{F_\delta}}
\newcommand{\Fdelr}{{F_{\delta_0}}}
\newcommand{\Fepsh}{{F_{\varepsilon_h}}}

\newcommand{\Ampl}{\mathcal{A}}
\newcommand{\Zu}{Z_{u_\star}}

\newcommand{\Grk}{\Gr_{k,n}^{\scriptscriptstyle \ge 0}}

\def\tZ{\tilde{Z}}
\newcommand{\bcfw}{{\bowtie}}
\newcommand{\mbcfw}{{{\iota_{\bowtie}}}}
\newcommand{\lr}[1]{\langle #1 \rangle}
\newcommand{\llrr}[1]{\langle\!\langle #1 \rangle\!\rangle}

\newcommand{\Pos}{\mathcal{P}}
\newcommand{\PosB}{\mathcal{B}}

\newcommand{\tzeta}{\tilde{\zeta}}
\newcommand{\talpha}{\tilde{\alpha}}
\newcommand{\tbeta}{\tilde{\beta}}
\newcommand{\tgamma}{\tilde{\gamma}}
\newcommand{\tdelta}{\tilde{\delta}}
\newcommand{\tepsilon}{\tilde{\varepsilon}}

\DeclareMathOperator{\after}{\text{after}}
\DeclareMathOperator{\Between}{\text{between}}
\DeclareMathOperator{\below}{\text{below}}

\newcommand{\scale}{\text{scale}}
\newcommand{\Loop}{\mathcal{L}}

\newcommand{\SA}{S^{\partial\Ampl}}
\newcommand{\Sred}{{\color{red}S^{\text{red}}}}
\newcommand{\Sblue}{{\color{blue}S^{\text{blue}}}}

\newcommand{\Srem}{S^{reg}}

\newcommand{\PC}{\Pi_C}
\newcommand{\PROJ}{\Pi}
\newcommand{\PD}{\Pi_{C+D}}
\newcommand{\PL}{\PROJ}

\newcommand{\C}{\mathbb{C}}

\newcommand{\B}{\mathcal{B}}
\newcommand{\TTT}{\mathcal{T}}

\def\R{\mathbb{R}}

\newtheorem*{theorem*}{Theorem}
\newtheorem{theorem}{Theorem}[section]
\newtheorem{thm}{Theorem}[section]
\newtheorem{lemma}[theorem]{Lemma}
\newtheorem{prop}[theorem]{Proposition}
\newtheorem{cor}[theorem]{Corollary}

\theoremstyle{definition}
\newtheorem{definition}[theorem]{Definition}
\newtheorem{ex}[theorem]{Example}
\newtheorem{defobs}[theorem]{Definition/Observation}

\newtheorem{obs}[theorem]{Observation}
\newtheorem{rmk}[theorem]{Remark}
\newtheorem{nn}[theorem]{Notation}
\newtheorem{cond}[theorem]{Condition}

\newtheorem{claim}[theorem]{Claim}

\setlist[itemize]{leftmargin=*}
\setlist[enumerate]{leftmargin=*}

\title{Notes on the one-loop amplituhedron and its BCFW tiling}
\author[R. Tessler]{Ran Tessler}
	\address{Department of Mathematics, Weizmann Institute of Science, Israel}
	
\begin{document}

\maketitle
\begin{abstract}
These notes are based on talks I gave in the seminar "Mathematical structures in scattering amplitudes in quantum field theories" I organized in Weizmann Institute on Fall 24'. They study amplituhedra, and extend the proof of \cite{even2021amplituhedron} of the BCFW conjecture for tree amplituhedra to one loop. 
\end{abstract}
This text consists of extended notes of a series of talks I gave as part of the seminar "Mathematical structures in scattering amplitudes in quantum field theories" I ran in Weizmann Institute on Fall 24', for math and physics graduate students. The perspective is purely mathematical.

The Amplituhedron was defined by Arkani-Hamed and Trnka \cite{arkani-hamed_trnka} in their study of scattering amplitudes for planar $\mathcal{N}=4$ SYM.
This space is the image of the nonnegative Grassmannian, or its loopy analogs under a certain \emph{Amplituhedron map} (see \cref{sec:back}).
Part of the motivation for defining this space was that it provided a geometric realization of the BCFW recursion \cite{BCFW}: it was conjectured implicitly in \cite{arkani-hamed_trnka} that images of certain subspaces of the loopy Grassmannian, constructed recursively via a process analogous to the BCFW recursion (see \cref{sec:bcfw_cells}), tile the amplituhedron.

These notes contain:
\begin{itemize}
\item A definition of the loopy Grassmannian and its (Pl\"ucker) nonnegative part, and a definition of loop amplituhedra, following \cite{arkani-hamed_trnka}. A definition of $B$-loop amplituhedra.
\item The ingredients of the BCFW recursion. In particular the \emph{forward limit} step that is known from physics literature, but has not appeared in the mathematical literature before.
\item A description of zero and one loop BCFW cells, and a detailed study of their geometry and combinatorics.
\item The algebraic counterpart of the BCFW recursion - promotions of functions on amplituhedra.
\item A proof that for zero and one loop BCFW cells tile the amplituhedron. For one loop this result is new.
\end{itemize}
This text follows the strategy of \cite{even2021amplituhedron} for studying BCFW cells and tilings. While this strategy is perfectly suited for studying the loopy case, a considerable complication comes from the absence of solid foundations for one loop objects, whose loopless analogues are well known in the rich theory of positroids. These difficulties are addressed by developing suitable tools.

We thank the Israel Science Foundation for support under grant ISF 1729/23.

\newpage
\setcounter{tocdepth}{1}
\tableofcontents
\section{Background: The loopy Grassmannian, its positive part, and the amplituhedron}\label{sec:back}
In this section we define loopy Grassmannians, which are the domains and targets of the amplituhedron maps, their nonnegative parts, amplituhedron maps and the images - the amplituhedra.
These objects were defined in \cite{arkani-hamed_trnka}. The special case of no loops involves the nonnegative Grassmannian, an elegant object that was defined in \cite{lusztig} and whose beautiful structure was discovered by Postnikov in \cite{postnikov}, and was extensively studied since.

The one loop Grassmannian is in fact a two-step flag variety. Lustig \cite{lusztig} also defined positivity structure for it. But it is another positivity structure, the \emph{Pl\"ucker positivity} which is relevant to the study of the amplituhedron. These two positivity structures are not equivalent, see \cite{bloch2023two} for an extended discussion.

We begin by making certain definitions and notations that will repeat throughout the text.
\begin{definition}\label{def:index_sets}
An \emph{index set} $N$ is a finite linearly ordered set. To make the wording less boring, we shall refer to elements of $N$ as \emph{indices, markers} and \emph{positions}. We write $\min N,~\max N$ for the minimal and maximal elements in $N,$ respectively. We usually denote the latter by $n.$ A repeating choice of an index set is $[n]=\{1,2,\ldots,n\}.$

For $i,j\in N$ we write $i<j,$ if $i$ comes before $j$ in the order. In this case we say that $i$ \emph{precedes} $j.$ We write $i\lessdot j$ if $i<j$ and they are consecutive in the order. In this case we say that $i,j$ are consecutive and that $j$ \emph{follows} $i.$ 
We say that $j$ \emph{cyclically follows $i$} if either $j$ follows $i$ or $j=\min N,~i=\max{N}.$ In case $j$ follows $i$ we sometimes write $j=i+1,~i=j-1$ or, if we want to stress that these relations are in $N,$ we may write $j=i+_N1,~i=j-_N1.$
 
\end{definition}
In what follows we work over the field $\R,$ unless specified otherwise.
\begin{nn}
We write $\ee_i$ for $i\in N$ for the $i$th standard basis element of $\R^N.$
\end{nn}
\begin{definition}\label{nn:Mat}
Fix index sets $N,~K$, and a finite set $\Loop$
A \emph{loopy matrix} is a tuple $(C\vdots \{D_p\}_{p\in\Loop})$ where $C\in\Mat_{K, N},~D_p\in\Mat_{2,N}$ for $p\in\Loop.$
$C$ is called the \emph{tree part} of the loopy matrix, $D_p$ are called \emph{loops}.

We denote the set of loopy matrices for given $K,N,\Loop$ by $\Mat_{K,N;\Loop}.$ If $N=[n]$, $K=[k]$ of $\Loop=[\ell]$ we write $n,k,\ell$ respectively instead of $[n],[k],[\ell].$ 
 We will often write $\Mat_{K,N;|\Loop|}$ for $\Mat_{K,N;\Loop}.$
The main case of interest in this text are $|\Loop|=0$ and $|\Loop|=1$. In the former cases we refer to a loopy matrix just as a matrix, and often erase $\Loop=\emptyset,|\Loop|=0$ from the subscript. In the latter case we denote the loopy matrix simply by $C\vdots D,$ and sometimes refer to it as a \emph{loopy pair}, or just as a \emph{pair}. We denote the $j$th column of a matrix $C$ or a loopy matrix $(C\vdots (D_i)_{i\in\Loop})$ by $C^j,~(C^j\vdots (D_i^j)_{i\in\Loop})$ respectively. We denote the $j$th row of a matrix $M$ by $M_j.$

For a subset $\Loop'=\{p_1,\ldots,p_j\}\subseteq \Loop$ we write \[C_{\Loop'}=\begin{pmatrix}D_{p_1}\\
\vdots\\
D_{p_j}\\C\end{pmatrix}\] for the matrix whose $k$ bottom rows are $C$, and the upper rows are the rows of $D_p,~p\in\Loop'$ where $\Loop'$ is ordered in an arbitrary order that will not affect what follows. We sometimes abuse notations and denote $C_{\Loop'}$ by $C+\sum_{p\in\Loop'}D_p$.

There are several important subspaces of $\Mat_{K,N;\Loop}.$ $\Mat_{K,N;\Loop}^\circ$ is the subset where
\begin{equation}\label{eq:ranks}
\forall \Loop'\subseteq \Loop,~~\rk( C_{\Loop'})=\min\{|N|,|K|+2|\Loop'|\}. \end{equation}
In this case the loopy matrix is said to be \emph{of full rank}. 
If all these matrices are in addition \emph{nonnegative}, meaning that all their maximal minors are nonnegative, then the loopy matrix is said to be \emph{nonnegative}. If they are all positive then the loopy matrix is said to be \emph{positive}. We write 
$\Mat_{K,N;\Loop}^{\geq},~\Mat_{K,N;\Loop}^>$ for the set of nonnegative and positive loopy matrices, respectively.

When $|K|+2|\Loop'|\leq |N|,$ for every subset $I\in\binom{N}{|K|+2|\Loop'|}$ we denote the maximal minor of the submatrix of $C_{\Loop'}$ whose columns are the columns indexed $I$ by the following notations
\[\lr{C_{\Loop'}}_I=\lr{C D_{p_1}\ldots D_{p_j}}_I=\lr{C +D_{p_1}+\ldots+ D_{p_j}}_I.\]
\end{definition}
\begin{definition}\label{nn:loopGrass}
Fix an index set $N,$ a nonnegative integer $k$, and a finite set $\Loop$.
A \emph{loopy vector space} is a tuple $(U\vdots \{V_p\}_{p\in\Loop})$ where $U\in\Gr_{K, N},~D_p\in\Gr_{2}(\R^N/U)$ for $p\in\Loop.$
If $N=[n]$ or $\Loop=[\ell]$ we write $n,\ell$ instead of $[n],[\ell]$ respectively. 
$U$ is called the \emph{tree part} of the loopy vector, $V_p$ are called \emph{loops}.

We denote the set of loopy vector spaces for given $k,N,\Loop$ by $\widetilde{\Gr}_{k,N;\Loop}.$ 

For a subset $\Loop'=\{p_1,\ldots,p_j\}\subseteq \Loop$ we write $U_{\Loop'}$ for the vector space $U+V_{p_1}+\ldots V_{p_j}.$ Note that since every $V_{p_i}$ is defined modulo $U,$ this summation makes sense and is well defined. 

There are several important subspaces of $\widetilde{\Gr}_{K,N;\Loop}.$ \[\Gr_{K,N;\Loop}=\{U\vdots\{V_p\}_{p\in\Loop}|~\forall \Loop'\subseteq\Loop,~~\dim(U_{\Loop'})=\min(|N|,k+2|\Loop'|)\}.\]Elements of this space are \emph{loopy vector spaces of full rank}. We will often write $\Gr_{k,N;|\Loop|}$ for $\Gr_{k,N;\Loop}.$
Again we will be mainly interested in $|\Loop|=0$ and $|\Loop|=1$. In the former cases we refer to a loopy vector space just as a vector space, and often erase $\Loop=\emptyset,|\Loop|=0$ from the subscript. In the latter case we denote the loopy vector space simply by $U\vdots V,$ and sometimes, as in the case of loopy matrices, refer to it as a \emph{loopy pair}, or just as a \emph{pair}. It will be clear from the context if when we speak about a loopy pair we mean a loopy matrix or a loopy vector space. 

There is a natural submersive surjection
\[\PROJ:\Mat^\circ_{K,N;\Loop}\to\Gr_{|K|,N;\Loop}\]
given by $(C\vdots\{D_p\}_{p\in\Loop})\mapsto(U\vdots\{V_p\}_{p\in\Loop})$
where $U$ is the row span of $C,$ and $V_p$ is the row span of $D_p,$ modulo $U.$
If $\PROJ(C\vdots\{D_p\}_{p\in\Loop})=(U\vdots\{V_p\}_{p\in\Loop})$ we say that $(C\vdots\{D_p\}_{p\in\Loop})$ is a \emph{representative} of $(U\vdots\{V_p\}_{p\in\Loop})$.
We will often not distinguish between a loopy vector space and its loopy matrix representative, and refer to a loopy matrix $(C\vdots\{D_p\}_{p\in\Loop})$ as a loopy vector space, or an element of $\Gr_{|K|,N;\Loop}$. We often write $C+D_{i_1}+\ldots D_{i_p}$ for the corresponding vector space, and may write $v\in C+D_{i_1}+\ldots D_{i_p}$ to indicate that the vector $v$ belongs to the corresponding vector space.  

We denote by $\Gr_{k,|N|;\Loop}^{\geq} ~(\Gr_{k,|N|;\Loop}^{>})$ the subspaces of $\Gr_{k,|N|;\Loop}$ which have a nonnegative (positive) representative. We refer to elements of this space as \emph{nonnegative (positive)} loopy vector spaces.
\end{definition}
\begin{definition}\label{def:group}
The fiber of $\PROJ$ is a torsor of the group \[\GL_{k;\Loop}=\GL_{k;|\Loop|}=\GL_{k,k+2}^{(1)}\times_{\GL_k}\GL_{k,k+2}^{(2)}\times_{\GL_k}\cdots \times_{\GL_k}\GL_{k,k+2}^{(|\Loop|)},\]where $\GL_{k,k+2}^{(i)}$ is a copy of the parabolic group which acts on pairs $(C\vdots D_i)$ by row operations involving $C$'s rows on $C$, and row operations involving both the rows of $C$ and of $D_i$ on $D,$ such that the total matrix representing the row operations on $C+D$ is invertible. The product is fibred on $\GL_k,$ in the sense that the restrictions of the row operations to $C$ are all the same.

We can write the elements of $\GL_{k;|\Loop|}$ as 
\[\begin{pmatrix}
    N_{|\Loop|} & 0_{2\times 2}&\ldots&0_{2\times 2}&K_{|\Loop|}\\
    \vdots&\ddots&\cdots&\vdots&\vdots\\
    0_{2\times 2} & 0_{2\times 2}&\ldots&N_1&K_{1}\\
    0_{k\times 2} & 0_{k\times 2}&\ldots&0_{k\times 2}&M
\end{pmatrix},\]
where $M$ is an invertible $k\times k$ matrix, each $N_i$ is an invertible $2\times 2$ matrix, $0_{i\times j}$ is an $i\times j$ matrix of $0$'s and each $K_i$ is an arbitrary $2\times k$ matrix. $\GL_{k;\Loop}$ has natural projections on $\GL_{k;\Loop'}$ for each $\Loop'\subseteq\Loop,$ obtained by restricting only to $(C\vdots(D_i)_{i\in\Loop'}).$ We also define $SL_{k;\Loop}$ as the subgroup of $\GL_{k;\Loop}$ where the representing matrix has determinant $1,$ that is
\[\det(M)\prod_{i\in\Loop'}\det N_i = 1.\]

In the special case of one loop the group $\GL_{k;1}$ is just the parabolic group preserving the two step flag variety \[\Fl_{k,k+2,n}=\{(U,W)|~U\in \Gr_{k,n},~W\in\Gr_{k+2,n},~U\subset W\}.\]
\end{definition}
The following lemma is easily verified.
\begin{lemma}\label{lem:submersion2Grass}
For $\Loop'\subseteq\Loop,$ write 
\[\PL_{\Loop'}:\Mat_{k,n;\Loop}\to\Mat_{k+2|\Loop'|,n},\qquad\qquad(C\vdots (D_i)_{i\in\Loop})\mapsto C+\sum_{i\in\Loop'}D_i.\]We abuse notations and write $\PL$ for the analogous map in the level of loopy vector spaces
\[\PL_{\Loop'}:\Gr_{k,n;\Loop}\to\Gr_{k+2|\Loop'|,n},\qquad\qquad(U\vdots (V_i)_{i\in\Loop})\mapsto U+\sum_{i\in\Loop'}V_i.\]
Note that the last map is defined only if $k+2|\Loop'|\leq n.$

Then these maps are continuous submersions, and in particular open. 
\end{lemma}
In the special case of $1-$loop we will use the notations:
\begin{align*}&\PC:\Mat_{k,n;1}\to \Mat_{k,n},~~(C\vdots D)\mapsto C,&\qquad\PD:\Mat_{k,n;1}\mapsto \Mat_{k+2,n},~~(C\vdots D)\to C+D,\\
&\PC:\Gr_{k,n;1}\to \Gr_{k,n},~~(U\vdots V)\to U,&\qquad
\PD:\Gr_{k,n;1}\to \Gr_{k+2,n},~~(U\vdots V)\to U+V.\end{align*}
\begin{definition}\label{def:pluckers_and_cone}
    Given a loopy matrix $(C\vdots(D_i)_{i\in\Loop}),~\Loop'\subseteq \Loop$ and $I\in\binom{[n]}{k+2|\Loop'|}$ we denote the $I$th \emph{maximal minor} of $(C\vdots(D_i)_{i\in\Loop'}),$ or of $C+\sum_{i\in\Loop'}D_i,$ made of the columns labeled $I$ by $\lr{C\vdots(D_i)_{i\in\Loop'}}_I=\lr{C+\sum_{i\in\Loop'}D_i}_I.$ We also use the slightly inaccurate term the $I$th \emph{Pl\"ucker coordinate} of $C+\sum_{i\in\Loop'}D_i.$ Note that if $g$ is an element of $\GL_{k;\Loop'}$ then
    \begin{equation}\label{eq:effect_of_grp_on_plckr}\lr{g\cdot (C\vdots(D_i)_{i\in\Loop})}_I=\det(g)\lr{(C\vdots(D_i)_{i\in\Loop})}.\end{equation}

    For a loopy vector space we define the Pl\"ucker coordinates
    \[\lr{U\vdots(V_i)_{i\in\Loop'}}_I=\lr{U+\sum_{i\in\Loop'}V_i}_I\]
    by first choosing a loopy matrix representative, and then taking its Pl\"ucker coordinates. While this definition depends on choices, for every given $\Loop'\subseteq\Loop$ with $k+2|\Loop'|\leq n$ the collection of Pl\"ucker coordinates $(\lr{U\vdots(V_i)_{i\in\Loop'}}_I)_{I\in\binom{[n]}{k+2|\Loop'|}}$ forms a \emph{projective vector}, by \eqref{eq:effect_of_grp_on_plckr}.

    More intrinsically, the Pl\"ucker coordinates $\lr{U\vdots(V_i)_{i\in\Loop'}}_I$ are \emph{global sections} of the line bundle
    \[\TTT_{\Loop'}:=\underline{\Mat_{k,n;\Loop'}^\circ}\times_{\SL_{k;\Loop'}}\underline{\mathbb{R}},\]where $\underline{\mathbb{R}}$ is the trivial real line bundle over $\Gr_{k,n;\Loop}$ with the trivial $\SL_{k;\Loop'}$ action, and $\underline{\Mat_{k,n;\Loop'}^\circ}$ is fiber bundle over $\Gr_{k,n;\Loop}$ with fiber over $(U\vdots (V_i)_{i\in\Loop})$ being matrix representatives to the sub loopy vector space $(U\vdots (V_i)_{i\in\Loop'}),$ obtained from $(U\vdots (V_i)_{i\in\Loop})$ by restriction, and it is a  
 $\GL_{k;\Loop'}-$principle bundle. $\SL_{k,\Loop'}$ acts by the natural action induced from the inclusion $\SL_{k;\Loop'}\hookrightarrow\GL_{k;\Loop'}$.
 We refer to the total space of $\TTT_\Loop$ as the \emph{cone} $\CC\Gr_{k,n;\Loop}$ over $\Gr_{k,n;\Loop},$ and it is endowed with the bundle map $\Pi:\CC\Gr_{k,n;\Loop}\to\Gr_{k,n;\Loop}.$ The Pl\"ucker coordinates are equivalently global functions over the cone. We use the same conventions as above to denote elements of the cone by loopy matrix representatives whose minors agree with the Pl\"ucker coordinates of the element.
\end{definition}
We think of the elements of the cones as matrix representatives up to $\SL_{k;\Loop}$ action.
In the case of one loop we simplify the notations:
\begin{nn}\label{nn:plucker}
For $I\in\binom{[n]}{k}$ or $\binom{[n]}{k+2}$ and a loopy matrix $C\vdots D,$ we write 
\[P_I(C)=P_I(C\vdots D)=\lr{C}_I,~\text{or }P_I(C+D)=P_I(C\vdots D)=\lr{C+D}_I=\lr{CD}_I,\]
respectively. We use similar notations for a loopy vector space $U\vdots V.$ We write $\TTT_0,\TTT_1$ for the corresponding line bundles.
\end{nn}

\subsection{Operations which preserve nonnegativity}
There are certain operations which preserve nonnegativity. 
\begin{definition}
\label{def:ops}
Let $N'$ be an index set containing an index $h,$ and $N=N'\setminus\{h\}.$
We define \[\pre_h:\Mat_{k,N;\Loop}\to\Mat_{k,N';\Loop}\]as the map which adds a zero column at the $h$th position of $C$ and every $D_i.$

Let $N'$ be an index set containing an index $h,$ and $N='N\setminus\{h\},$ $K'$ an index set containing an index $l,$ and write $K=K'\setminus\{l\}.$
The map $\inc_{h|l}$ is defined by \[\inc_{h|l}:\Mat_{K,N;\Loop}\to\Mat_{K',N';\Loop},\qquad (C\vdots(D_i)_{i\in\Loop})\mapsto (C'\vdots(D'_i)_{i\in\Loop})\]
where
\[C' = \begin{pmatrix}
C_{\min {K},\min {N}} & \cdots & C_{\min {K},h-1} & 0 & -C_{\min {K},h+1} & \cdots & -C_{\min {K},\max {N}}\\
\vdots  & \ddots & \vdots & \vdots & \vdots & \ddots & \vdots\\
C_{l-1,\min {N}} & \cdots & C_{l-1,h-1} & 0 & -C_{l-1,h-1} & \cdots & -A_{l-1,\max {N}}\\[0.25em]
0       & \cdots & 0         & 1 & 0       & \cdots & 0      \\
-C_{l+1,\min {N}} & \cdots & -C_{l+1,h-1} & 0 & C_{h+1,h+1} & \cdots & C_{h+1,\max {N}}\\
\vdots  & \ddots & \vdots & \vdots & \vdots & \ddots & \vdots\\
-C_{\max {K},\min {N}} & \cdots & -C_{\max {K},h-1} & 0 & C_{\max {K},h+1} & \cdots & C_{\max {K},\max {N}}
\end{pmatrix},
\qquad D'_i=\pre_h D_i.\]
For $i\in N,$ let $\rem_i:\Mat_{K,N;\Loop}\to\Mat_{K,N\setminus\{i\};\Loop}$ be the operation which removes the $i$th column of the loopy matrix.

For $i\in N,$ let $\scale_i(t)$ be the operation which scales the $i$th column of a loopy matrix by $t.$

For $h\in N,~t\in\mathbb{R}$ we define the map $x_h(t):\Mat_{k,N;\Loop}\to\Mat_{k,N;\Loop}$ by
\[[x_h(t)(C\vdots(D_i)_{i\in\Loop})]_j=\begin{cases}~(C\vdots(D_i)_{i\in\Loop})_j,&j\neq h+_N1\\
(C\vdots(D_i)_{i\in\Loop})_{h+1}+t(C\vdots(D_i)_{i\in\Loop})_h,&j= h+_N1,~h\neq\max{N},\\
(C\vdots(D_i)_{i\in\Loop})_{h+1}+(-1)^{k-1}t(C\vdots(D_i)_{i\in\Loop})_h,&j= h+_N1,~h=\max{N}
\end{cases}\]
We similarly define
$y_h(t):\Mat_{k,N;\Loop}\to\Mat_{k,N;\Loop}$ by
\[[y_h(t)(C\vdots(D_i)_{i\in\Loop})]_j=\begin{cases}~(C\vdots(D_i)_{i\in\Loop})_j,&j\neq h\\
(C\vdots(D_i)_{i\in\Loop})_{h}+t(C\vdots(D_i)_{i\in\Loop})_{h+1},&j= h,~h\neq\max{N},\\
(C\vdots(D_i)_{i\in\Loop})_{h}+(-1)^{k-1}t(C\vdots(D_i)_{i\in\Loop})_{h+1},&j= h,~i=\max{N}
\end{cases}.\]
These maps can be realized by right multiplication by the matrices we denote by the same notations (and there is an obvious modification for $h={\max{N}}$):
$$
\begin{matrix}
x_i(t) \;= &
\begin{pmatrix}
    1 & \cdots & 0 & 0 & \cdots & 0 \\
    \vdots & \ddots & \vdots & \vdots & \ddots &  \vdots \\
    0 & \cdots & 1 & t & \cdots & 0 \\
    0 & \cdots & 0 & 1 & \cdots & 0 \\
    \vdots & \ddots & \vdots & \vdots & \ddots & \vdots \\
    0 & \cdots & 0 & 0 & \cdots & 1
\end{pmatrix}
\end{matrix}
\;\;\;\;\;\;\;\;\;\;\;\;\;\;\;\;\;\;\;\;
\begin{matrix}
y_i(t) \;= &
\begin{pmatrix}
    1 & \cdots & 0 & 0 & \cdots & 0 \\
    \vdots & \ddots & \vdots & \vdots & \ddots & \vdots \\
    0 & \cdots & 1 & 0 & \cdots & 0 \\
    0 & \cdots & t & 1 & \cdots & 0 \\
    \vdots & \ddots & \vdots & \vdots & \ddots & \vdots \\
    0 & \cdots & 0 & 0 & \cdots & 1
\end{pmatrix}
\end{matrix}
$$

The maps $x_i,y_i,\inc_{i|j}$ and, if $k+2|\Loop|\leq |N|$ also $\pre_i$  descend to the loopy Grassmannian and to its nonnegative subspace (see \cref{lem:effect_ops} below). The same holds for $\rem_i$ as long as we restrict to the subspace where the $\rem_i$ operation does not harm the rank inequalities. We use the same notations for the Grassmannian-level maps $x_i,y_i,\pre_i,\rem_i$ and write $\inc_i$ for the induced $\inc_{i|j}$. If we perform the operation $\pre,\inc$ or $\rem$ with respect to the indices $I=\{i_1,\ldots,i_h\}$ we write $\pre_I,\inc_I,\rem_I$ respectively.

The last map we will now introduce adds loops.
Let $N$ be an index set and $A,B\in N$ cyclically consecutive indices.
For a $I=\{i_1,\ldots,i_r\}\subseteq N$ denote by
$\Mat^{(i_1\ldots i_r)}_{k,N;\ell},~\Gr^{(i_1\ldots i_r)}_{k,N;\ell} $ the subspaces of $\Mat_{k,N;\ell},~\Gr_{k,N;\ell}$ respectively, which consist of elements in which the columns labeled $I$ are linearly independent. 
We define \[\addL_{AB}:
\Mat^{(AB)}_{k,N;\ell}\to\Gr^{(AB)}_{k-2,N;\ell+1},\qquad(C\vdots (D_1,\ldots,D_\ell))\mapsto(U'\vdots (V'_1,\ldots,V'_\ell,V'_{\ell+1}))\]as follows. Perform row operations on $C$ to obtain $C'$ with $C'_{i,A},C'_{i,B}=0$ for all $i\geq 3.$ Then by the assumption on $C,$ the row span of the upper two rows is two dimensional. Denote it by $V'_{\ell+1}.$ $U'$ is the $k-2$ dimensional vector space spanned by the remaining rows of $C'.$ Note that $U'$ is independent of the choice of the representative, while $V'_{\ell+1}$ is independent of the choice of the representative as a two dimensional vector space \emph{in} $\mathbb{R}^{N}/U'.$ Finally, for $i=1,\ldots, \ell$ we add to the rows of $D'_i$ vectors from $V'_{\ell+1}$ to make its $A,B$ columns zero. $V'_i$ is the linear span of the resulting matrix, and it is easy to see it is well defined, independently of choices, as a subspace of $\R^N/U'.$ 
This map also descends to the loopy Grassmannian and by \cref{lem:effect_ops} below, also to the nonnegative loopy Grassmannian. We use the same notation for the maps in the loopy Grassmannians level.
\end{definition}
\begin{lemma}\label{lem:effect_ops}
The operations of \cref{def:ops} have the following effect on the Pl\"ucker coordinates
\begin{itemize}
\item $\pre_h:$
\[P_I(\pre_h W)=\begin{cases}P_I(W),&h\notin I\\
0,&h\in I
\end{cases}\]
\item $\inc_h:$
\[P_I(\inc_h W)=\begin{cases}P_{I\setminus\{h\}}(W),&h\in I\\
0,&h\notin I
\end{cases}\]
\item $\rem_h:$
\[P_I(\rem_h W)=P_I(W)\]
\item $x_h(t):$
\[P_I(x_h(t)W)=\begin{cases}P_{I}(W),& I\cap\{h,h+1\}\neq\{h+1\}\\P_{I}(W)+tP_{I\cup\{h\}\setminus\{h+1\}},& I\cap\{h,h+1\}=\{h+1\}
\end{cases}\]
\item $y_h(t):$
\[P_I(y_h(t)W)=\begin{cases}P_{I}(W),& I\cap\{h,h+1\}\neq\{h\},\\P_{I}(W)+tP_{I\cup\{h+1\}\setminus\{h\}},& I\cap\{h,h+1\}=\{h\}
\end{cases}\]
\item $\scale_h(t):$
\[P_I(\scale_h(t)W)=\begin{cases}tP_{I}(W),& h\in I,\\P_{I}(W),& h\notin I
\end{cases},\]
\end{itemize}
where $W$ is a loopy vector space and $I$ are sets in the appropriate size for which the left and right hand sides make sense.
We will be more specific regarding the effect of $\addL_{AB}.$ Assume $\addL_{AB}(U\vdots(V_1,\ldots,V_\ell))=(U'\vdots(V'_1,\ldots,V'_{\ell+1}).$ For $\Loop'\subseteq [\ell+1]$ with $k_{\Loop'}=k+2|\Loop'|\leq |N|,$ write $U'_{\Loop'}=U+\sum_{i\in\Loop'}V'_i.$ First, assume $\ell+1\notin\Loop.$ Then $P_I(U'_{\Loop'})=0$ whenever $I\cap\{A,B\}\neq\emptyset,$ and 
\[(P_I(U'_{\Loop'}))_{I\in\binom{N\setminus\{A,B\}}{k_{\Loop'}}}=(P_{I\cup\{A,B\}}(U+\sum_{i\in\Loop'}V_i))_{I\in\binom{N\setminus\{A,B\}}{k_{\Loop'}}}~~\text{as \emph{projective vectors}}.\]
If $\ell+1\in\Loop'$ then
\[(P_I(U'_{\Loop'}))_{I\in\binom{N}{k_{\Loop'}}}=(P_{I}(U+\sum_{i\in\Loop'\cap[\ell]}V_i))_{I\in\binom{N}{k_{\Loop'}}}~~\text{as \emph{projective vectors}}.\]
\end{lemma}
The proof is a direct computation.

\subsubsection{One step into one loop}
The Pl\"ucker coordinates satisfy certain relations. While we can write them for the most general loopy setting, the notations are lighter in the one loop case, which is the content of this text.
\begin{prop}[\cite{weyman2003cohomology}, Proposition 3.1.6]\label{prop:plucker_rels_flag}
For every $r,s\in\{k,k+2\},~I\in\binom{[n]}{r-1},~J=\{j_1,\ldots,j_{s+1}\}\in\binom{[n]}{s+1}$, the following relation holds for every loopy vector space $U\vdots V:$\[\sum_{h=1}^{s+1} (-1)^hP_{I\cup\{j_h\}}P_{J\setminus \{j_h\}}=0.\] 
    
\end{prop}

The next topological lemma will be useful in the study of the $1$-loop amplituhedron.
\begin{lemma}\label{lem:connected}
Fix nonnegative integers $n,k$ with $n\geq k+2.$ $\Gr_{k,n;1}^>$ is a connected submanifold of the real two step flag variety $\Gr_{k,n;1},$ of full dimension $k(n-k)+2(n-k-2).$  
$\Gr_{k,n;1}^\geq$ is its closure, it is compact and connected.\end{lemma}
\begin{proof}
Since the defining inequalities of $\Gr_{k,n;1}^>$ are strict, it is an open set.
For $0<x_1<\ldots<x_n$ write $v_i=(x_1^i,\ldots,x^i_n).$ Then $U\vdots V\in \Gr_{k,n;1}^>$, where $U=\Span(v_1,\ldots,v_k),$ and $V=\Span(v_{k+1},v_{k+2})$, by the well known properties of Vandermonde determinants. Thus, $\Gr_{k,n;1}^>\neq\emptyset.$

For connectedness we will use induction on $n.$ For $n=k+2,~\Gr_{k,n;1}^>\simeq\Gr_{k,k+2}^>$ which is known to be connected.
Note that for $n>k+2$ we have a forgetful map
\[\Gr_{k,n;1}^>\to\Gr_{k,n-1;1}^>\]obtained by forgetting the last column. We will show that the fiber of this map over every point of $\Gr_{k,n-1;1}^>$ is connected and non empty.
Take $U\vdots V\in \Gr_{k,n-1;1}^>,$ and fix a loopy matrix representative $C\vdots D$ for it. Extending the representative to a representative of a point in $\Gr^>_{k,n;1}$ amounts to picking an $n$th column $(a_1,\ldots,a_k\vdots a_{k+1}, a_{k+2}),$ so that all minors of the loopy matrix obtained from appending this column as the last column of the loopy matrix whose first $n-1$ columns are $C\vdots D,$ are positive. Each positivity constraint either does not involve the $n$th column, or is a linear inequality in its entries. Thus, the fiber is the intersection of open half spaces of $\mathbb{R}^{k+2}$ defined by these inequalities. Such an intersection, if non empty, is convex, hence connected.
We are left with showing the fiber is non empty.
Fix coefficients $\epsilon_1,\ldots,\epsilon_{n-k}$  to be specified later, and write \[C^n\vdots D^n = \sum_{i=1}^{n-k}(-1)^{i-1}\epsilon_i (C^{n-i}\vdots D^{n-i}),\]and let $\tilde{C}\vdots \tilde{D}$ be the matrix obtained from $C\vdots D$ by adding this column on the right. Then for every $I\in\binom{[n-1]}{k-1},$ 
\begin{equation}\label{eq:fiber_non_empty}\lr{\tilde{C}}_{I\cup\{n\}}=\sum_{j\in [n-1]\setminus (I\cup[k-1])}(-1)^{n-j-1}\epsilon_{n-j}\lr{C}_{I\cup \{j\}}.\end{equation}Note that for the maximal $j$ in the summation $(-1)^{n-j-1}\lr{C}_{I\cup \{j\}}>0,$ since $j\notin I$ but $j+1,\ldots,n-1\in I.$ We have similar expansion for $C+D$ minors. Thus, we can pick the coefficients $\epsilon_i$ in a way that the summand $\epsilon_i/\epsilon_{i+1}$ is large enough so that all expressions \eqref{eq:fiber_non_empty}, and the analogous expressions for the minors of $C+D$ are positive. This shows that the fiber is nonempty.

The topological closure $\overline{\Gr_{k,n;1}^>}$ of $\Gr_{k,n;1}^>$, which must be connected, since $\Gr_{k,n;1}^>$ is, and compact since the two-step flag variety $\Gr_{k,n;1}$ is, is clearly contained in $\Gr_{k,n;1}^\geq.$
Now, for an arbitrary point in $U\vdots V\Gr_{k,n;1}^\geq,$ note that applying a certain sequences of $x_i$ operators with arbitrarily small positive parameters on $U\vdots V\in$ results in $U'\vdots V'\in\Gr_{k,n;1}^>.$ Thus, $\Gr_{k,n;1}^\geq\subseteq\overline{\Gr_{k,n;1}^>},$ hence they are equal.
\end{proof}
\begin{rmk}
    Since $\Gr^>_{k,n}$ is contractible \cite{postnikov}, and the fibers from the above analysis are also contractible, simple induction shows that $\Gr^>_{k,n;1}$ itself is contractible.
\end{rmk}
\subsection{The loop amplituhedron}
\begin{definition}[Arkani-Hamed and Trnka \cite{arkani-hamed_trnka}]
Let $n,k,\ell$ be nonnegative integers satisfying $k+4\leq n.$
Let $Z$ be a positive $n\times (k+4)$ matrix.
The \emph{loop amplituhedron map} $\tZ:\Gr_{k,n;\ell}^\geq\to\widetilde\Gr_{k,k+4;\ell}$ is defined, in matrix representatives, by 
\[(C\vdots(D_1,\ldots,D_\ell))\mapsto (C\vdots(D_1,\ldots,D_\ell))Z=(CZ\vdots(D_1Z,\ldots,D_\ell Z)).\]The \emph{loop amplituhedron}
$\Ampl_{n,k,4}^\ell(Z)$ is the image of this map.

The $\ell=0-$amplituhedron is called the \emph{tree amplituhedron}. The $\ell=1-$amplituhedron is called the \emph{one loop amplituhedron}.    
\end{definition}
\begin{thm}\label{thm:ampli_well_def}For $\ell=0,1$ the amplituhedron is a compact connected subspace of $\Gr_{k,k+4;\ell}.$
\end{thm}
\begin{proof}
That the amplituhedron map is independent of choices is straightforward. 
In \cite[Section 4]{karp2017sign} Karp shows that the multiplying a matrix representative of $V\in\Gr_{k,n}^\geq$ by a positive $n\times(k+m)-$matrix $Z$ yields a $k$ dimensional subspace of $\R^{k+m},$ as long as $k+m\leq n.$

The case $\ell=0$ is well known, and follows from \cite[Section 4]{karp2017sign}, together with the compactness and connectedness of $\Gr_{k,n}^\geq$ \cite{postnikov}.

For $\ell=1$ we apply Karp's result both to $U$ and to $U+V,$ where $U\vdots V\in\Gr_{k,n;1}^\geq,$ to deduce that the image lies in $\Gr_{k,k+4}^1.$ The compactness and connectedness then follows from \cref{lem:connected}.
\end{proof}
\begin{rmk}
    While there are many works which study the tree amplituhedron in the mathematical literature, there are much less works which study the loopy case.
    The case $k=0,\ell=1$ is equivalent to the $m=2$ tree amplituhedron. This space has been studied a lot, see, e.g., \cite{BaoHe,PSW}. \cite{bai2016amplituhedron} studies the $k=\ell=1$ loop amplituhedron, and its cells.
\end{rmk}
\section{The BCFW cells}\label{sec:bcfw_cells}
The protagonists of this text are the recursively defined BCFW cells. These cells are interesting since they were implicitly conjectured in \cite{arkani-hamed_trnka} to \emph{tile} the amplituhedron, see \cref{def:tiling} below. See also \cite{Bai:2014cna}. This conjecture was a central motivation for introducing the amplituhedron in the study of scattering amplitudes of planar $\mathcal{N}=4$ SYM theory.
\subsection{The ingredients of the BCFW recursion}
\begin{nn}\label{nn:bcfwmap}
Fix elements $1 \leq a\lessdot b< c\lessdot d\lessdot n$ in $[n]$, and let $N_L= \{1, \dots, a, b, n\}$ and $N_R=\{b, \dots, c, d, n\}$. Fix $k \leq n$ and two nonnegative integers $k_L \leq |N_L|$ and $ k_R\leq |N_R|$ such that $k_L + k_R +1=k$. Finally, fix  $\ell_L,\ell_R,\ell\in\{0,1\}$ with $\ell=\ell_L+\ell_R.$
In what follows, whenever we refer to this notation, $\ell_L,\ell_R,\ell\in\{0,1\},$ unless explicitly specified differently. 
\end{nn}

\begin{definition}[BCFW Map]\label{def:bcfw-map} 
Using \cref{nn:bcfwmap},
the \emph{BCFW map} is the rational map
\[\mbcfw\;
:\;{\Mat}_{k_L, N_L;\ell_L}\times\; \R^5\;\times\;{\Mat}_{k_R,N_R,\ell_R}\dashrightarrow\Mat_{k,n;\ell},\]\[((C_L\vdots\{D_p\}_{p\in{\Loop}_L}\}), (\alpha, \beta, \gamma, \delta,\varepsilon), (C_R\vdots\{D_q\}_{q\in\Loop_R}))\mapsto(C\vdots\{D'_r\}_{r\in\Loop_L\cup\Loop_R})\]
 where
\begin{itemize}
\item \[C = C'_L+v+C'_R,\]
defined by 
\[C_L'=\pre_{N_R\setminus\{a,b,n\}}.y_a(\frac{\alpha}{\beta}).C_L,\qquad\qquad C_R'=\pre_{N_L\setminus\{a\}}.y_c(\frac{\gamma}{\delta}).y_d(\frac{\delta}{\varepsilon}).C_R,\]
and \[v=\alpha \ee_a+\beta \ee_b+(-1)^{k_R}\gamma \ee_c+(-1)^{k_R}\delta \ee_d +(-1)^{k_R}\varepsilon \ee_n\in \R^n.\]

\item For $r\in\Loop_L$ 
\[D'_r=\pre_{N_R\setminus\{a,b,n\}}.y_a(\frac{\alpha}{\beta}).D_r\]
\item For $r\in\Loop_R$ 
\[D'_r=\pre_{N_L\setminus\{a\}}.y_c(\frac{\gamma}{\delta}).y_d(\frac{\delta}{\varepsilon}).D_r\]
\end{itemize}
Schematically we obtain the loopy matrix
\begin{equation*}\label{eq:path-mtx} \begin{bmatrix}
    (D_{L})_1& \cdots & (D_{L})_{a-1} & (D'_{L})_a& (D_{L})_b&  0& \cdots& 0& 0&  0  & 	(-1)^{k_R+1} (D_{L})_n\\
    	0 & \cdots& \cdots&0 & (D_{R})_b & (D_{R})_{b+1}& \cdots&(D_{R})_{c-1}&(D'_{R})_c & (D'_{R})_d & (D_{R})_n\\
    \hline
	(C_L)_1& \cdots & (C_L)_{a-1} & (C'_L)_a& (C_L)_b&  0& \cdots& 0& 0&  0  & 	(-1)^{k_R+1} (C_L)_n\\
	0 & \cdots & 0 & \alpha & \beta & 0& \cdots &0 & (-1)^{k_R} \gamma& (-1)^{k_R}\delta& (-1)^{k_R}\varepsilon\\
	0 & \cdots& \cdots&0 & (C_R)_b & (C_R)_{b+1}& \cdots&(C_R)_{c-1}&(C'_R)_c & (C'_R)_d & (C_R)_n\\
	
\end{bmatrix}\end{equation*}
where $D_R,D_L$ are the union of rows of $D_i$ for $i\in\Loop_L,i\in\Loop_R$ respectively.

If the left or right factor is trivial, we omit it from the notations.
In particular, in the end case $a=1$ the left factor must trivial, and we refer to the map as the \emph{upper BCFW map}
\[\mbcfw\;
: \R^5\;\times\;{\Mat}_{k-1,\{2,\ldots,n\},\ell_R}\;\dashrightarrow\;\Mat_{k,n;\ell_R}.\]
Similarly, the \emph{lower BCFW map} is obtained when $b=n-3.$ In this case the right part must be trivial, and we write
\[\mbcfw\;
: {\Mat}_{k-1,\{1,\ldots,n-3,n\},\ell_L}\times\;\R^5\;\dashrightarrow\;\Mat_{k,n;\ell_L}\]
For later purposes we shall refer to the markers $a,b,c,d,n$ as the \emph{BCFW chord's support}.

\end{definition}
\begin{prop}\label{prop:BCFW_and_pos}
If $(\ell_L,\ell_R)\in\{(0,0),(0,1),(1,0)\},$ and $k_L\leq |N_L|-2,~k_R\leq |N_R|-2,$ then
the map $\mbcfw$ descends to a rational map (denoted using the same notations)
\[\mbcfw\;
:\;{\Gr}_{k_L, N_L;\ell_L}\times\; \Gr_{1,5}\;\times\;{\Gr}_{k_R,N_R,\ell_R}\;\dashrightarrow\;\Gr_{k,n;\ell}.\]
\[\mbcfw\left(\text{Domain}(\mbcfw)\cap({\Gr}^{\geq}_{k_L, N_L;\ell_L}\times\; \Gr_{1,5}^{>}\;\times\;{\Gr}^{\geq}_{k_R,N_R,\ell_R})\right)\subseteq\Gr_{k,n;\ell}^\geq.\]
The same holds for the upper and lower BCFW products.
\end{prop}
\begin{proof}
It is easy to see that without adding the positivity requirements that rank conditions \eqref{eq:ranks} are generically met.
We now move to the nonnegative case.
The tree case is known, it is also known that applying the BCFW product to two nonnegative matrices and a positive element from $\Gr_{1,5}^>$ results in a nonnegative matrix. If, in addition, the left and right factors are positive, then the result is of full rank (see, for example, \cite{even2023cluster} for these statements). This shows that $C$ is nonnegative. If we now apply these facts to the matrices $C_L+D,~C_R$, if $\ell_L=1$, or to $C_L,C_R+D$, if $\ell_R=1,$ the resulting $C+D'$ is non negative. Note that the assumptions on $k_L,k_R$ imply that $C_L+D$ or $C_R+D$ is nonnegative, which is needed for this derivation. Thus, the image is generically in $\Gr_{k,n;1}^{\geq}.$
The same argument applies in the end cases of upper and lower BCFW products.
\end{proof}
\begin{nn}\label{nn:butterfly}
With the above notations, if $S_L\subseteq \Gr_{k_L,N_L;\ell_L}^{\geq},S_R\subseteq \Gr_{k_R,N_R;\ell_R}^{\geq}$ we write $S_L\bcfw S_R$ for their \emph{BCFW product}, that is, the image of
$\mbcfw(S_L,\Gr_{1,5}^{>},S_R)$ whenever it is defined.
\end{nn}
\begin{definition}\label{def:forward_limit}
Let $N$ be the ordered index set $\{1,2,\ldots,n-1,A,B,n\}.$ We denote $c=n-2,d=n-1.$ The \emph{(one loop) forward limit} $\FL$ is the following rational map
\[\FL:\Mat^{(B)}_{k+1,N}\times\R^5\dashrightarrow\Gr_{k,[n];1}\] where $\Mat^{(B)}_{k+1,N}$ is the subspace of $\Mat_{k+1,N}$ made of matrices without a zero column at $B,$ defined as follows:
For $M\in\Mat^{(B)}_{k+1,N}$ and $(\gamma_\star,\delta_\star,\alpha_\star,\beta_\star,\varepsilon_\star)\in\R^5,$ with $\varepsilon_\star,\delta_\star,\alpha_\star\neq 0$ we 
\begin{itemize}
\item First apply $\inc_l$ on $M$, where $l$ is a new marker between $d$ and $A,$ and then $\scale_l(\varepsilon_\star).$ 
Call the resulting matrix $M_0.$
\item We then set
\[M_1=y_c(\frac{\gamma_\star}{\delta_\star}).y_d(\frac{\delta_\star}{\varepsilon_\star}).x_A(\frac{\beta_\star}{\alpha_\star}).x_l(\frac{\alpha_\star}{\varepsilon_\star}).M_0\]
Note that the rank of the columns $A,B$ of $M_1$ is $2.$
\item $M_2=\addL_{AB}(M_1).$ 
\item Finally, 
\[\FL(M)=\rem_l(x_l(1)
(\rem_{A,B}M_2)).\]
\end{itemize}    
\end{definition}
\begin{rmk}\label{rmk:last_y_c}
Observe that in the above construction of the forward limit we could have performed the last $y_c$ operation in the end, without changing the result. Denote by $\FL'$, the \emph{forward limit degenerated at $c$}, which is the operator defined like $\FL$ only without the $y_c$ operation, or equivalently, by substituting $\gamma_\star=0$ in $\FL$.
\end{rmk}
\begin{prop}\label{prop:fl_domain}
The map $\FL$ is generically defined, and is independent of choices.
Moreover, it descends to a rational map, still denoted by the same notation
\[\FL:\Gr_{k+1,n}^{\geq}\cap\Gr^{(B)}_{k+1,N}\times\Gr^>_{1,5}\dashrightarrow\Gr^{\geq}_{k,[n];1},\]
$\Gr^{(B)}_{k+1,N}$ is the subspace of $\Gr_{k+1,N}$ made of matrices without a zero column at $B.$
\end{prop}
\begin{proof}
Choose a representative $M$ for the element of $\Gr_{k+1,n}^{\geq}\cap\Gr^{(B)}_{k+1,N},$ and $(\gamma_\star,\ldots,\varepsilon_\star)$ for the element of $\Gr^{>}_{1,5}$. It is easy to see that generically $C,C+D$ have full rank, so that \eqref{eq:ranks} hold.

Since $\addL_{AB}$ descends to the Grassmannian, different choices for $M,$ which are the only choices in the process, do not affect the final answer. 

By \cref{lem:effect_ops} all the operations preserve nonnegativity. 
%
\end{proof}
\begin{nn}\label{nn:FL_set}
With the above notations, if $S\subseteq \Gr_{k,N}^{\geq}$ we write $\FL(S)$ for the image of
$\FL(S,\Gr_{1,5}^{>})$ whenever it is defined.
\end{nn}
\subsection{The tree and one loop BCFW cells}
We will now describe the BCFW cells, and some of their basic properties. 
\begin{definition}\label{def:1-loop-cells}\label{def:tree-bcfw-cells}
	\emph{$0$ and $1-$loop} BCFW cells are defined recursively as follows.
	\begin{enumerate}[align=left]
    \item[(Base case)]
			For $k=0$, the trivial cell $\Gr_{0,n}$ is a $0-$loop BCFW cell.
			\item[(Soft factor/adding a zero column)]
			If $S$ is a $\ell-$loop BCFW cell, for $\ell=0,1,$ then so is the cell obtained by inserting a zero column in the penultimate position.
			\item[(Factorization/Product)] In the notations of \cref{nn:bcfwmap}, if $S_L$ and $S_R$ are $\ell_L-$loop and $\ell_R-$loop BCFW cells on $N_L$ and $N_R$, for $\ell_L,\ell_R\in\{0,1\}$ with $\ell_L+\ell_R\leq 1,$ then their BCFW product $S_L \bcfw S_R$ is a standard BCFW cell. 

   \item[(Forward limit)]Let $S$ be a $0-$loop BCFW cell on the index set $N=\{1,2,\ldots,c,d,A,B,n\}$ (so that $A,B$ are located in this order just before $n$). Assume that $S$ has no zero column at $B,$ and that in the course of constructing $S$ there was no BCFW product step in which the BCFW chord was supported on $c,d,A,B,n.$ Then $\FL(S)$ is a $1-$loop BCFW cell.
	\end{enumerate}
    We will refer to $0-$loop BCFW cells as \emph{tree} BCFW cells. We will collectively refer to $0$ and $1$ loop BCFW cells by the name BCFW cells.
    We refer to the sequence of the above operations in the construction of a BCFW cell as the \emph{generation sequence} or \emph{recipe} for the BCFW cell, and each operation is a \emph{step} in this sequence. 
\end{definition}
Since the BCFW and forward limit maps are only rational maps, we will need to show that they are defined, meaning that the resulting loopy matrices are of full rank, throughout the above recursive process. This will be proven in \cref{prop:fixed_positroids} below. 

A few comments are in place.
\begin{rmk}\label{rmk:domain_cells_for_FL}
The reason tree BCFW cells with a zero column at $B$ are not included as input to the $\FL$ step is that the operation $\FL$ is not defined on this locus in the Grassmannian.

The other cells which are excluded from the collection of inputs to the forward limit operation are excluded for another reason. Applying $\FL$ to these cells results in a space of dimension less than $4(k+1)$, hence they will not be used to triangulate the $1-$loop amplituhedron. In fact, one can easily show using the techniques we develop in this text, that these cells actually map to the zero locus of a certain function $\llrr{cdAB},$ which will be defined below.
\end{rmk}
\begin{rmk}\label{rmk:cardinality}
The number of BCFW cells in a tiling of $\Grk$ has been shown in \cite{karp2020decompositions} to be the Narayana number $\frac{1}{k+1}\binom{n-4}{k}\binom{n-3}{k}$. It is known from the physics literature \cite{arkani-hamed_trnka} that the number of $1-$loop BCFW cells in a tiling is $\binom{n-2}{k}\binom{n-2}{k+2}.$
This fact can be easily deduced from this text, by counting recursively the number of different BCFW cells, and observing they are all different, a fact which follows directly from the analysis below.
\end{rmk}
\begin{rmk}\label{rmk:general_bcfw}
More general BCFW cells can be obtained from a similar recursive process if we allow in addition steps of cyclic rotations and reflections, see \cite[Section 6]{even2023cluster} for more general tree BCFW cells. If one proves the tiling conjecture for one collection, the proof for other collections follows directly via the strategy of \cite{even2023cluster}. In this text, mainly from length reasons, we will therefore restrict attention only to the BCFW cells described above. 
\end{rmk}
\subsubsection{Chord diagrams}\label{subsub:chord_diags}
\cite{even2021amplituhedron} had defined \emph{chord diagrams}, and had shown that they are in bijection with BCFW cells. In this subsection we extend these notions and correspondence to the $1$-loop case.
\begin{definition}[$\ell=0,1$ chord diagram]\label{def:chord_L=1} 
Let $k,n \in \mathbb{N}$. A~\emph{$\ell$-loop chord diagram} $\D$ for $\ell\in\{0,1\},$ is a set of $k+\ell$~quadruples named \emph{chords}, of non decreasing integers in the set $\{1,\dots,n-1\}$ named \emph{markers, indices} or \emph{positions}, of the following form:
	$$ \D \;=\; \{D_1,\ldots,\D_{k+\ell}\} 
    \;\;\text{ where }\;\;\D_i=(a_i,b_i,c_i,d_i)~\text{with }
	b_i=a_i+1 \text{ and }d_i=c_i+1.$$
    $(a_i,b_i)$ and $(c_i,d_i)$ are the \emph{start} and \emph{end} of $\D_i,$ respectively. 
    The chords satisfy the following assumptions:
 \begin{enumerate}
 \item No two chords $\D_i,\D_j \in D$ satisfy
$ a_i \;=\; a_j$ or $a_i \;<\; a_j \;<\; c_i \;<\; c_j.$ 
\item The chords are ordered so that for $i>j$ then $c_i\geq c_j,$ and, in case of equality then $b_i<b_j.$
\item If $\ell=0$ all chords are \emph{colored} black. If $\ell=1$ then every chord is colored black, red, blue or purple. We denote the sets of black, red, blue, and purple chords by $\Black=\Black(\D),\Blue=\Blue(\D),\Blue=\Blue(\D),\Purp=\Purp(\D),$ respectively, and $\Red\neq\emptyset$. The red chord $\D_i$ with minimal $a_i$ is termed the \emph{top red chord}, and we write $i=\tr=\tr(\D).$ For convenience we also denote the top red chord by $\D_0=(a_0,b_0,c_0,d_0).$
 \item All red, blue and purple chords end at $(c_0,d_0),$ every red chord $\D_i$ starts before every blue chord
$\D_j,$ that is $a_i<a_j.$ Every blue chord starts before every purple chord.
 Every other chord $\D_i$ which ends at $(c_0,d_0)$ has $a_i<a_0.$
 \item Every 
	chord $\D_i$ satisfies $b_i\leq c_i,$ and there is at most one chord with $c_i=b_i.$ If such a chord exists then it must be blue or red. 
 \end{enumerate}
 For $\ell=1$ we also denote $(c_0,d_0)$ by $(c_\star,d_\star)$, and refer to it the \emph{yellow chord}, which we denote by $\D_\star$. $\ell=1$ chord diagrams are also called \emph{one loop} chord diagrams. $\ell=0$ chord diagrams are also called \emph{tree} chord diagrams. We write $[k+\ell]_\D=[k+\ell]\setminus\{\tr(\D)\},$ if $\ell=1$ and $[k]$ if $\ell=0.$
\end{definition} 
\begin{figure}
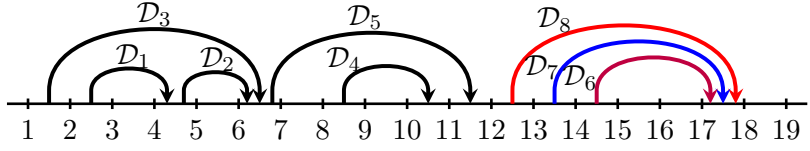

\begin{center}
\tikz[line width=1, scale=0.7]{
\draw (0.4,0) -- (0.8*19.5,0);
\foreach \i in {1,2,...,19}{
\def\x{\i*0.8}
\draw (\x,-0.1)--(\x,+0.1);
\node at (\x,-0.5) {\i};}
\foreach \i/\j in {1.5/6.5,2.5/4.3,4.7/6.2,6.8/11.5,8.5/10.5
}
\def\x{\i*0.8}
\def\y{\j*0.8}
\draw[line width=1.5,-stealth] (\x,0) -- (\x,0.25) to[in=90,out=90] (\y,0.25) -- (\y,0);

\draw[line width=1.5,red,-stealth] (12.5*0.8,0) -- (12.5*0.8,0.25) to[in=90,out=90] (17.8*0.8,0.25) -- (17.8*0.8,0);
\draw[line width=1.5,blue,-stealth] (13.5*0.8,0) -- (13.5*0.8,0.25) to[in=90,out=90] (17.5*0.8,0.25) -- (17.5*0.8,0);

\draw[line width=1.5,purple,-stealth] (14.5*0.8,0) -- (14.5*0.8,0.25) to[in=90,out=90] (17.2*0.8,0.25) -- (17.2*0.8,0);

\node at(4*0.8,1.7) {$\D_3$};
\node at(3.5*0.8,0.9) {$\D_1$};
\node at(5.5*0.8,0.85) {$\D_2$};
\node at(9*0.8,1.65) {$\D_5$};
\node at(8.5*0.8,0.85) {$\D_4$};
\node at(13.5*0.8,1.6) {$\D_8$};
\node at(13.2*0.8,0.7) {$\D_7$};
\node at(14.1*0.8,0.5) {$\D_6$};
}
\end{center}
\caption{
A one loop chord diagram $\D$ on $n=19$ markers.
} 
\label{fig:cd-example}
\end{figure}
See Figure~\ref{fig:cd-example}, where we visualize such a chord diagram $\D$ in the plane as a horizontal line with $n$ markers labeled $\{1,\dots,n\}$ from left to right, and $k+\ell$ nonintersecting chords above it, whose \emph{start} and \emph{end} lie inside the segments $(a_i,b_i)$ and $(c_i,d_i)$ respectively. The definition constrains the chords: they cannot start before $1$, end after $n-1$, or start or end on a marker; two chord cannot cross; two chords cannot start in the same segment $(h,h+1)$; and one chord cannot start and end in the same segment, and at most one chord, which must be red or blue can start and end in adjacent segments. 
\begin{definition}[Terminology and notations for chords]
\label{terminology:cd}
A~chord is a~\emph{top} chord 
if there is no chord above it, e.g. $\D_3,\D_5$ and $\D_8$ in Figure~\ref{fig:cd-example}.
 Otherwise it is a \emph{descendant} of the chords above it, called its \emph{ancestors}, and in particular a~\emph{child} of the chord immediately above it, which is called its~\emph{parent}. We denote by $\p(i)$ the parent of chord $i.$ For example, $\D_6$ has ancestors $\D_7$ and $\D_8$, and $\p(6)=7$. When $\D_i$ is a top chord we write $\p(i)=\emptyset,$ and it will be convenient to set $a_\emptyset=b_\emptyset=n,$ and sometimes think of $\D_\emptyset$ as the single element set $\{n\}.$
Two chords are \emph{siblings} if they are either top chords or children of a common parent, or both are top chords; for example, $\D_1$ and $\D_2$ are siblings, and $\D_3$ and $\D_5$ are siblings. We also consider all top chords as siblings.
Two chords are \emph{same-end} 
if their ends occur in a common segment $(e,e+1),$ they are \emph{head-to-tail} if the first ends in the segment where the second starts, and are \emph{sticky} if their 
starts lie in consecutive segments $(s,s+1)$ and~$(s+1,s+2)$.
For example, chords $\D_2$ and $\D_3$, and so are $\D_6,\D_7,\D_8$ are same-end, $\D_1$, 
	$\D_2$ and $\D_3,\D_5$ are head-to-tail, and chords $\D_1$ is a sticky child of $\D_3$. $\D_8,\D_7,\D_6$ form a \emph{sticky chain}, that is a sequence of chords, each a sticky child to the previous one. It is moreover a \emph{sticky-same end chain}, since these chords are also all same-end. A chord is \emph{short} if $c_i=b_i+1,$ and is \emph{very short} if $c_i=b_i.$ A very short chord is always red or blue. $\D_1,\D_2$ and $\D_4$ are short. Had we replaced $\D_6$ by a blue or red chord $\D'_6=(16,17,17,18)$ it would have been a very short chord. The yellow chord in \cref{fig:cd-example} is $(17,18).$

    If $\ell=1,$ every ancestor, descendent, parent, child or sibling to the top red chord is considered as having the same relation with the yellow chord. Otherwise we refer to the yellow chord also as a top chord. The yellow chord is considered same-ended with every other chord which ends at $(c_\star,d_\star),$ every chord which starts at $(c_\star,d_\star)$ is considered head-to-tail after the yellow chord.

    A chord $\D_i$ and a descendent $\D_j$ are said to be \emph{strictly same-ended} if they are same ended, and either they are of the same color, or $\D_i$ is black, or $\D_i$ is the yellow chord, and $\D_j$ is the top purple chord. In \cref{fig:cd-example} $\D_3,\D_1$ are strictly same end and also $\D_\star$ and $\D_6.$
\end{definition}
\begin{definition}\label{def:behind_above_etc}
Let $\D$ be a chord diagram. We define the following quantities per chord.
    \begin{itemize}
\item[--] $\ahead(\D_l)$ is the set of all chords which end after the end of $\D_l.$ 
We consider the end of the yellow chord and of the red chords to be $(c_\star,d_\star).$

\item[--] $\outside(\D_l)$ is the set of all chords which are neither $\D_l$ nor its descendants. $\woutside(\D_l)=\outside(\D_l)\cup\{\D_l\},$ if $l\neq0,\star,$ while if $\l\in\{0,\star\}$ it is $\outside(\D_l)\cup\{0,\star\}.$
\item[--] 
$\below(\D_l)$ is the number of descendants of $\D_l$. 
\item[--] 
$\after(\D_l),$ for $l\neq0,\star$, is the number of chords that start after the start of~$\D_l$. We \emph{exclude} from this count the top red and yellow chords. We define $\after(\D_0)=\after(\D_\star)$ as the number of chords with $a_i\geq c_\star$. 
\item[--] 
$\Between(\D_l),$ for $l\neq 0,\star,$ is the number of descendants of $\D_{\p(l)}$, that start before~$\D_l$ starts, excluding $0,\star$. If $\D_l$ is a top chord it is just the number of chords that start before $\D_l.$ $\Between(\D_0)=\Between(\D_\star)$ is one plus the number of descendants of $\D_{\p(0)},$ which either start before $\D_0$, or are descendants of $\D_0.$
\end{itemize}We will sometimes write $\Between(l)$ for $\Between(\D_l)$ etc.
\end{definition}

\begin{nn}\label{nn:k}
We extend the definition and notations of chord diagrams to general index sets $K,N$ in the natural way. For a $1$ loop chord diagram $\D$ we write $K_\D=K\setminus\{\tr(\D)\}.$ We write $\CD^\ell_{N,K}$ for $\ell\in\{0,1\}$ and $N,K$ index sets, for the collection of $\ell-$loops chord diagrams on the index set $N$ with chords labeled $K.$ $\CD_{N,K}:=\CD^0_{N,K}\cup\CD^1_{N,K}$. 
We will denote $\CD^\ell_{[n],[k+1]}$ by $\CD^\ell_{n,k}$.
\end{nn}
\begin{rmk}
In what follows, usually when we shall consider a single chord diagram we will take its index set to be either $[n]$, or $[n]\cup\{A,B\}$ where $A<B<n$ are consecutive in the order of $[n]\cup\{A,B\}$ (mainly when we perform a forward limit operation) or a subset of $[n]$ not containing $n-1$ (when we perform a $\pre$ operation). We will usually not be too careful about what the set $K$ is, except for its cardinality. Yet, in all cases the extension to more general index sets will be automatic.
\end{rmk}
\begin{definition}\label{def:L=1sub_diags}[Left, right, soft and FL subdiagrams]
\label{def:L=0_subdiags} 
Let $\D\in\CD^\ell_{N,K}$ be a chord diagram.
A \emph{subdiagram} is a chord diagram obtained by restricting to a subset of the chords, a 
subset of the markers which contains these chords and the marker~$n$. In addition, some of the chords may change color, where if a chord changes color then all chords must become black, and it may happen only if the rightmost top chord of the subdiagram is red, blue or purple in $\D.$
If the subdiagram contains a blue or a red chord, and their colors are unchanged, then it must also contain the top red chord. Note that in this case no color changes, and the resulting subdiagram is a $\ell=1$ diagram. Otherwise all chords become black and the result is a $\ell=0$ diagram.

If $p$ is the penultimate marker of $N,$ and no chord has $d_i=p$ then the \emph{soft limit subdiagram}, or the \emph{$\pre$ subdiagram}, $\D'$ of $\D$ is the subdiagram obtained by erasing this marker. We write $\D=\pre_p\D'.$

Let
$\D_{\max K} = (a,b,c,d)$ be the rightmost top chord of~$\D$, and assume $d$ is the penultimate marker of $N,$ and $\max K\neq\tr(\D).$ 
Define $\D_L$, the \emph{left subdiagram} of $\D$, on the markers $N_L=\{\min N,\dots,a,b,n\}$
and the \emph{right subdiagram} $\D_R$ on~$N_R=\{b,\dots,c,d,n\}$. The subdiagram $D_L$ contains all chords $D_i$ with $c_i\leq a,$ and $\D_R$ contains the descendants of~$\D_{\max K}$.

We write $\D=\D_L\bcfw\D_R,$ where $\D_L,\D_R$ are the left and right subdiagrams respectively. We do not rule out the case that one of $\D_L$ or $\D_R$ is an empty diagram.
Note that at most one will be an $\ell=1$ diagram.

For a chord $\D_i,$ where $i=\tr(K)$ is not excluded, we write $\D^{(i)}$ for the subdiagram obtained by restricting only to chords $\D_j$ that $(a_j,b_j)<(c_i,d_i)$ and, if $\D_i$ is not a top chord, are descendants of its parent. 

The index set of the $\FL$ subdiagram is obtained from the index set $N$ of $\D$ as follows. We erase all indices greater than $d_i$ other than $n,$ and all indices smaller than $b_{\p(i)}$ when $\p(i)\neq\emptyset.$ If $\D_i$ is not a top chord, we also erase from $N$ all indices smaller or equal $a_{\p(i)}.$  We erase from the index set. Note that $\D_i$ is the rightmost top chord for $\D^{(i)}.$

The subdiagram $\D^{(0)}=\D^{(\tr(\D))}$ is called the \emph{FL subdiagram}.

We also define the \emph{pre FL diagram} of $\D$ as the $\ell=0$ diagram obtained from the FL diagram of $\D$ by adding two markers $A\lessdot B$ between $d=d_{\tr(\D)}$ and $n,$ moving all red chords to end at $(A,B),$ moving all blue chords to end at $(d,A)$ and forgetting the colors.

In case $\D$ equals its $\FL$-subdiagram, and $\D'$ is its pre-$\FL$ diagram, then we write $\D=\FL(\D').$

A subdiagram $\D'$ of $\D$ is \emph{proper} if there is a sequence of subdiagrams $\D^{(0)},\ldots,\D^{(m)}$ with $\D^{(0)}=\D,~\D^{(m)}=\D'$ and $\D^{(i+1)}$ is either a pre, left, right or pre-FL subdiagram of $\D^{(i)}.$
\end{definition}
\begin{definition}[$\ell-$loop BCFW cell from a $\ell-$loop chord diagram]
\label{def:treeBCFWfromCD}
Let $\D\in \CD^\ell_{N,K}$ be a $\ell-$loop chord diagram. 
We associate $\D$ a $\ell-$loop BCFW cell $S_\D$ in ${\Gr}^{\geq}_{K, N;\ell}$ using the following recursive process:
\begin{enumerate}[align=left]
\item If $K=\emptyset$ then $S_\D$ is the trivial empty cell $\Gr_{\emptyset,N}^\geq.$
\item If $\D=\pre_p\D',$ where $\D'$ is the $\pre$-subdiagram of $\D,$ then $S_\D=\pre_p S_{\D'}.$
\item Otherwise, if the rightmost top chord is not red then let $S_L$ and $S_R$ be the standard BCFW cells on $N_L$ and $N_R$ associated to 
the left and right subdiagrams	$\D_L$ and $\D_R$ of $\D$.  Then, we set $S_\D := S_L \bcfw S_R$, the tree BCFW cell which is their BCFW product. Note that exactly at least one of $S_L,S_R$ is a $0-$loop cell and the other is an $\ell-$loop cell.
\item If the rightmost top chord is red, let $\D'$ be the pre-FL subdiagram of $\D.$ Then $S_\D=\FL(S_{\D'}).$ 
\end{enumerate}
\end{definition}
The following observation is straight forward
\begin{obs}\label{obs:1loopCDvs1loopBCFW}
There is a bijection between $\CD^\ell_{n,k}$ and recipes for constructing $\ell-$loop BCFW cells via the recursive process of \cref{def:1-loop-cells}.
\end{obs}

\begin{ex}
\label{right-left-diagrams}
The chord diagram $\D$ in \cref{fig:cd-example}, is its own $\FL$-subdiagram, and the rightmost top chord is $\D_8 = (12,13,17,18)$.  
The pre-$\FL$ subdiagram $\D'$ is the $\ell=0$ diagram on the index set $N'=\{1,\ldots,18,A,B,19\}$ with $\D'_1=\D_1,\ldots,\D'_6=\D_6,$ where $\D_6$ is black, $\D'_7=(13,14,18,A),~\D'_8=(12,13,A,B)$ and they are also black. 

The left subdiagram $\D'_L=\{\D_1,\ldots,\D_6\}$ of $\D'$ has index set $N_L = \{1,\dots,13,19\}$, while $N_R = \{13,\dots,18,A,B,19\}$ and $\D'_R = \{\D'_6,\D'_7\}$. 

The BCFW cell $S_\D$ equals $\FL(S_{\D'}),$ and $S_{\D'}=S_{\D_L} \bcfw S_{\D_R}.$
\end{ex}

In what follows we shall show that the resulting cells are all distinct. This will be an immediate consequence of \cref{thm:sep}, that will show that their images in the amplituhedron are disjoint.

\subsection{BCFW coordinates and form}\label{subsec:BCFW_form}
Let $\D$ be a chord diagram. 
We now explain how to write a canonical, up to row scaling, $\ell-$loop matrix representative for every element of $S_\D.$ The canonical form is called the \emph{BCFW form}, and is presented in terms of the BCFW coordinates. 
Before we describe the BCFW form, we would like to reduce a few unnecessary degrees of freedom. The next lemma does exactly this.
\begin{lemma}\label{lem:scale_inv}
Let $a,b,c,d$ be as in \cref{nn:bcfwmap}, and let $C_L\vdots D_L,C_R\vdots D_R,$ where one of $D_L,\D_R$ is trivial, be loopy matrices as in \cref{def:bcfw-map}. We write $A=c,B=d.$
Then, the loopy vector space
\[\FL\left(\mbcfw(C_L\vdots D_L,
(\alpha_0,\beta_0,\gamma_0,\delta_0,\varepsilon_0), C_R\vdots D_R),
(\gamma_\star,\delta_\star,\beta_\star,\alpha_\star,\varepsilon_\star)\right)\]
equals
\[\FL\left(\mbcfw(C_L\vdots D_L,
(\alpha_0,\beta_0,\gamma'_0,\delta\delta_0,\varepsilon_0), \scale_d(\alpha)\scale_c(\delta)[C_R\vdots D_R]),
(\gamma_\star,\delta_\star,\beta'_\star,\alpha\alpha_\star,\varepsilon_\star)\right),\]for every $\gamma_0,\gamma'_0,\beta_\star,\beta'_\star$ and non zero $\alpha_0,\beta_0,\delta_0,\varepsilon_0,\gamma_\star,\delta_\star,\alpha_\star,\varepsilon_\star,\alpha,\delta.$
In other words, the first expression is independent of $\gamma_0,\beta_\star$ and its dependence on the $\delta_0,\alpha_\star,(C_R\vdots D_R)^A,(C_R\vdots D_R)^B$ factors through $\frac{1}{\alpha_\star}(C_R\vdots D_R)^A,\frac{1}{\delta_0}(C_R\vdots D_R)^B.$ 
\end{lemma}The proof is a direct calculation.
This lemma motivates the following definition.

\begin{definition}\label{def:BCFW coords}
The set of \emph{BCFW coordinates} of a $\ell-$loop chord diagram $\D$ is the set $\{\halp_i,\ldots,\heps_i\}_{i\in[k]},$ if $\ell=0,$  
    or $\{\halp_i,\ldots,\heps_i\}_{i\in[k+1]_\D}\cup\{\halp_0,\hbet_0,\hdel_0,\heps_0,\halp_\star,\hgam_\star,\hdel_\star,\heps_\star\},$ if $\ell=1.$ These coordinates are subject to the following \emph{gauge equivalence} generated by
    \begin{itemize}
    \item Scaling all variables $\hzet_i$ by $\lambda\in\R_+,$ for a given $i.$
    \item Scaling $\alpha_\star,\{\hgam_i\}_{i\in\Red(\D)},\{\hdel_i\}_{i\in\Blue(\D)},$ by $\lambda\in\R_+.$
    \item Scaling $\hdel_0,\{\hdel_i\}_{i\in\Red(\D)},$ by $\lambda\in\R_+.$
    \end{itemize}
    We label the set of BCFW coordinates by     $\BCFWV=\BCFWV_\D.$ 

For later purposes we also define the set of \emph{reduced BCFW coordinates} $\BBCFWV$ as the set 
\[\{\bgam_\star,\bdel_\star,
\beps_\star\}\cup\{\balp_0,\bbet_0,
\beps_0\}\cup\bigcup_{i=1}^k\{\balp_i,\bbet_i,\bgam_i,\bdel_i,\beps_i\},\]
where $\bzet_i=\hzet_i$, except for $\balp_\star,\bdel_0$ which are not defined, $\bgam_i=\hgam_i/\halp_\star,~\bdel_i=\hdel_i/\hdel_0,~i\in\Red,$ and $\bzet_i=\hzet_i\halp_\star,~i\in\Blue,~\bzet\neq\bdel.$ The gauge equivalence for these variables is generated by scaling all the variables $\zeta_i,$ for a given $i,$ by $\lambda\in\R_+.$
\end{definition}

The description of the BCFW form follows closely \cite[Definition 12.11]{even2023cluster}.
We define a collection of evolving vectors $v^t_{\halp,i},\ldots,v^t_{\heps,i}$ where $t=0,\ldots,T$ index the steps $\pre,\bcfw,\FL$ in the generation sequence for $S_\D,$ so that $T$ is the last step, and $i\in [k],$ if $\ell=0$ or $i\in[k+1]\cup\{\star\}$. 
At $t=0$ all vectors are set to $0.$ 
As usual we denote $\tr(\D)$ by $0.$
We will construct the matrix in terms of the inputs to the BCFW and forward limit steps, which are $\halp_i,\ldots,\heps_i$ for $i\in [k],$ if $\ell=0$ or $i\in[k+1]\cup\{\star\}$. 
\begin{enumerate}
		\item If the $t$-th step is $\pre_p,$ where $p$ is the penultimate marker of $N_t$ then
        for every $i$ and every $\hzet\in\{\halp,\ldots,\heps\},$\[v_{\zeta,i}^t=\pre_pv_{\zeta,i}^{t-1}.\]
        \item If the $t$-th step is a BCFW step adding the chord $\D_l$,
        then we set
		\[v_{\halp,l}^t= \ee_{a_l}, \quad v_{\hbet,l}^t= \ee_{b_l}, \quad v_{\hgam,l}^t= (-1)^{k_R^t}\ee_{c_l}, \quad v_{\hdel,l}^t= (-1)^{k_R^t}\ee_{d_l}, \quad v_{\heps,l}^t= (-1)^{k_R^t}\ee_{n}, \quad \]where $k_R^t$ is the number of chords of the right subdiagram at time $t.$
	For other $i\neq l$  we set
		\begin{equation*}
			v_{\hzet,i}^t= \begin{cases}
				\scale_n((-1)^{k_R^t+1})y_{a_l}(\frac{\halp_l}{\hbet_l})v_{\hzet,i}^{t-1}   & \text{ if $\D_i$ belongs to the left subdiagram at time $t$ }\\
				 y_{c_l}(\frac{\hgam_l}{\hdel_l}) y_{d_l}(\frac{\hdel_l}{\heps_l})v_{\hzet,i}^{t-1} &\text{ if  $\D_i$ is in the right subcomponent at subdiagram $t$ }
	\\v_{\hzet,i}^{t-1}&\text{otherwise.}
    \end{cases}\end{equation*}If $l=\tr(\D)$ we use \cref{lem:scale_inv} to set $\hgam_l=0$ without affecting the final loopy vector space.
\item If the $t$-th step is a forward limit then we set 
\[v_{\hgam,\star}^t = \ee_{c_\star},~v_{\hdel,\star}^t = \ee_{d_\star},~v_{\halp,\star}^t = \ee_A,~v_{\heps,\star}^t = \ee_n.\]We do not need to define $v_{\hbet,\star}^t$ because of \cref{lem:scale_inv}.
Note that the $t-1$ step must be a BCFW step. The row which was added at the $t-1$th step is labeled $\tr(\D),$ and we also label it by $0.$ We would like to make this row the second top most in the matrix form, and to use it with the row $\star$ to perform the $\addL_{AB}$ operations.
Set \[\tilde{v}^{t-1}_{\hzet,i}=\begin{cases}(-1)^{k_R^{t-1}+1}\scale_{A}(-1)\scale_{B}(-1)\scale_{n}(-1){v}^{t-1}_{\hzet,i}&i=0\\
    \scale_{A}(-1)\scale_{B}(-1)\scale_{n}(-1){v}^{t-1}_{\hzet,i} &\text{otherwise}
\end{cases}.\]
The origin of these signs is that we apply $\inc_l$ to the matrix generated by the BCFW recursion up to step $t-1,$ where we order the rows in an \emph{increasing order of $a_i.$} $\inc_l$ adds the row labeled $\star$ as a topmost row, and adds signs as in \cref{def:ops}. We move row $0$ to be the second topmost, and fix the signs so that both the whole matrix and the matrix made of the $k_R^{t-1}+k_L^{t-1}$ bottom rows are nonnegative. 
Note that $N_t=N_{t-1}\setminus\{A,B\}.$ Put $v_{\hzet,0}^t=\tilde{v}_{\hzet,0}^{t-1},~\hzet\in\{\halp,\ldots,\heps\},$\[ \tilde{v}_\star^t=\hgam_\star v_{\hgam,\star}^t+\hdel_\star v_{\hdel,\star}^t+\halp_\star v^t_{\halp,\star}+\heps_\star v_{\heps,\star}^t,\qquad\qquad \tilde{v}_0^t=\halp_0 v_{\halp,0}^t+\hbet_0 v_{\hbet,0}^t+\hdel_0 v^t_{\hdel,0}+\heps_0 v_{\heps,0}^t\]and
\[ v_\star^t=\rem_{AB}\tilde{v}^t_\star=\hgam_\star v_{\hgam,\star}^t+\hdel_\star v_{\hdel,\star}^t+\heps_\star  v_{\heps,\star}^t,\qquad\qquad v_0^t=\rem_{AB}\tilde{v}^t_0=\halp_0 v_{\halp,0}^t+\hbet_0 v_{\hbet,0}^t+\heps_0 v_{\heps,0}^t.\]
For any other chord $\D_i$ we use $\tilde{v}_0^t,\tilde{v}_\star^t$ to cancel its $A,B$ entries and perform the $\FL$ operation.
\begin{align*}v_{\hzet,i}^t&=
\rem_{AB}\left(y_{c_\star}(\frac{\gamma_\star}{\delta_\star})\tilde{v}_{\hzet,i}^{t-1}-\frac{\pr_{A}\tilde{v}_{\hzet,i}^{t-1}}{\halp_\star}\tilde{v}_\star^t-(-1)^{k_L^{t-1}}\frac{\pr_{B}\tilde{v}_{\hzet,i}^{t-1}}{\hdel_0}\tilde{v}_0^t\right)\\&=
\rem_{AB}\left(y_{c_\star}(\frac{\gamma_\star}{\delta_\star})\tilde{v}_{\hzet,i}^{t-1}\right)-\frac{\pr_{A}\tilde{v}_{\hzet,i}^{t-1}}{\halp_\star}{v}_\star^t-(-1)^{k_L^{t-1}}\frac{\pr_{B}\tilde{v}_{\hzet,i}^{t-1}}{\hdel_0}{v}_0^t
\end{align*}where we refer to the projection into a one dimensional space, oriented according to $\ee_A$ or $\ee_B$, as a scalar in $\R$.
	\end{enumerate}
    \begin{rmk}\label{rmk:bcfw_vecs_fl}
    We can also obtain the same vectors $v^t_{\hzet,i}$ if we do not set $\hbet_\star=\hgam_0=0.$
    In this case, in the notations of the forward limit step, we can write for $i\in[k+1],$ \[\check{v}^{t-1}_{\hzet,i}= y_{A}(\frac{\hgam_0}{\hdel_0})\tilde{v}^{t-1}_{\hzet,i},\]which are (the proper scalings) of the vectors one would have obtained after the $t-1$ step had $\hgam_0\neq0,$ and \[\check{v}^t_\star=x_{A}(\frac{\hbet_\star}{\halp_\star})\tilde{v}^t_\star, \qquad\qquad\check{v}^{t}_0= x_{A}(\frac{\hbet_\star}{\halp_\star})y_A(\frac{\hgam_0}{\hdel_0})\tilde{v}^{t}_{0}.\]
    Then $v_{\hzet,i}^t=
\rem_{AB}(\check{v}_{\hzet,i}^t),$ where
    \begin{align}\label{eq:v_in_terms_check_v}
    \check{v}_{\hzet,i}^t=
y_{c_\star}(\frac{\gamma_\star}{\delta_\star})x_{A}(\frac{\hbet_\star}{\halp_\star})\check{v}_{\hzet,i}^{t-1}-s\check{v}_\star^t-r\check{v}_0^t.
\end{align}
where $r=r(\hzet,i)=(-1)^{k_L^{t-1}}\frac{\pr_{B}\tilde{v}_{\hzet,i}^{t-1}}{\hdel_0},$ and $s=s(\hzet,i)=\frac{\pr_{A}\tilde{v}_{\hzet,i}^{t-1}-\hgam_0r}{\halp_\star}$ are chosen to make the $A,B$ entries of the vector in the parenthesis $0.$ Direct calculation shows that the result is independent of $\hbet_\star,\hgam_0.$

\end{rmk}

It will also we useful to define the vectors\begin{equation}\label{eq:passage_from_v_hat_to_v_bar}v^t_{\bzet,i}=\begin{cases}
\hdel_0v^t_{\hdel,i},&\bzet=\bdel,\D_i\in\Red_\D\\
   \halp_\star v^t_{\hzet,i},&\bzet=\bgam,~\D_i\in\Red_\D~\text{or }\bzet=\bdel,~\D_i\in\Blue_\D\\v^t_{\hzet,i},&\text{otherwise.}
   \end{cases}\end{equation}
Then, for every $t\in [T]$, the row indexed $i$ of the BCFW pair at time $t$ equals
\begin{equation}\label{eq:recipevecs}
v_{i}^t=\begin{cases}\frac{1}{c_i}\sum_{\hzet\in\{\halp,\ldots,\heps\}} \hzet_i v_{\hzet,i}^t=\sum_{\bzet\in\{\balp,\ldots,\beps\}} \bzet_i v_{\bzet,i}^t,,&i\neq0,\star\\
\sum_{\hzet\in\{\halp,\hbet,\heps\}} \hzet_i v_{\hzet,i}^t
=\sum_{\bzet\in\{\balp,\bbet,\beps\}} \bzet_i v_{\bzet,i}^t
,&i=0,\\
\sum_{\hzet\in\{\hgam,\hdel,\heps\}} \hzet_i v_{\hzet,i}^t=\sum_{\bzet\in\{\bgam,\bdel,\beps\}} \bzet_i v_{\bzet,i}^t,&i=\star,\end{cases}\end{equation}
where $c_i=\halp_\star$ for $\D_i\in\Blue_\D$ and is $1$ otherwise.
We denote $v^T_{\hzet,i}$ by $v^{\D_i}_\hzet,$ $v^T_{\bzet,i}$ by $v^{\D_i}_\bzet,$ and $v^T_i$ by $v^{\D_i}.$

\begin{obs}\label{obs:bcfw_form}
With the above notations, let $C\vdots D$ be the matrices defined as follows.
If $\ell=0$ then $C$ is the matrix whose rows are the vectors $v^{\D_i},~i\in[k],$ ordered in an increasing order of $a_i,$ and $D$ is the empty matrix. If $\ell=1$ then $C$ is the matrix whose rows are the vectors $v^{\D_i},~i\in[k+1]_\D,$ ordered in an increasing order of $a_i,$ and $D$ is the matrix whose first row is $v^{\D_\star}$ and the second is $v^{\D_0}.$
    Then $C\vdots D$ is the $\ell-$loop vector space $U\vdots V\in S_\D$ which is the result of the BCFW recursion with the parameters $\halp_i,\ldots,\heps_i,$ for $i\in[k]$ if $\ell=0$ or $[k+\ell]\cup\star,$ if $\ell=1$.

    The vectors $v^{\D_i}_\hzet$ are Laurent polynomials in the BCFW coordinates of the chords $\D_h$ for $h\in\ahead(i).$ The vectors $v^{\D_i},v^{\D_i}_\bzet$ are Laurent polynomials in the reduced BCFW coordinates of the same chords set of chords. 
\end{obs}The proof is immediate from unveiling the effect of $\pre,\bcfw,$ and $\FL$ steps in the generation sequence of the BCFW cell. 
\begin{obs}\label{obs:parameterizing_S_check_D}
Let $t_\FL$ be the time of the forward limit step. Write $S_{\check\D}$ for the space of loopy vector spaces which have a matrix representation  whose rows are the non zero vectors $\check{v}^{t_\FL}_i=\sum_{\hzet\in\{\halp,\ldots,\heps\}}\hzet_i\check{v}^{t_\FL}_{\hzet,i},$ where each $\hzet_i>0,$ such that the tree part is the span of the rows other than $\tr(\D),\star.$
Then the cell after the forward limit is obtained by taking elements of $S_{\check\D}$, and chopping columns $A,B.$ 

We slightly abuse notations and refer to the variables $\hzet_i$, for $i$ such that the corresponding row is non zero, as the BCFW variables for $S_{\check\D}.$ 
\end{obs}
\begin{rmk}\label{rmk:v_becomes_explicit}
The vectors $v^{\D_i}_\bzet$ are relatively simple for $\bzet=\balp,\bbet,$ and equal $\ee_{a_i},\ee_{b_i}$ respectively. For $\bzet=\bgam$ 
\[v^{\D_i}_\bgam = \begin{cases}-v^{\D_\star},&\D_i\in\Red\\\frac{u_\star}{\hgam_\star},&\D_i\in\Blue\\\ee_{c_i},&\text{otherwise}\end{cases},\]where $u_\star=\pr_{c_\star d_\star}v^{\D_\star}.$ The vector $v^{\D_i}_\bdel$ might be more complicated, but always has the form
\[v^{\D_i}_\bdel = \begin{cases}-v^{\D_0}+c_{\D,i}v^{\D_i}_\bgam,&\D_i\in\Red\\-v^{\D_\star}+c_{\D,i}v^{\D_i}_\bgam,&\D_i\in\Blue\\\ee_{d_i}+c_{\D,i}v^{\D_i}_\bgam,&\text{otherwise}\end{cases},\]
\end{rmk}
\begin{definition}\label{def:bcfw_matrix_form}
   The loopy matrix $C\vdots D$ defined in the above observation is called the \emph{BCFW representative}, \emph{BCFW form} and the \emph{BCFW matrix} of $U\vdots V.$ We will transfer terms from chord diagrams to the rows of the matrix, e.g. a red row, or row $i$ is an ancestor of row $j,$ when these terms hold for the corresponding chords.  
\end{definition}
\begin{rmk}
   From the construction the matrix $C$ is always upper triangular when $\ell=0$ in the sense of rectangular matrices, meaning that the first non zero entry of the $i$th row comes before the first non zero entry of the $(i+1)-$th row. When $\ell=1$, this property breaks because of the red rows. However, by subtracting appropriate multiples of $D_0$ to the red rows of $C$, $C$ becomes an upper triangular matrix, and moreover one can add the rows $D_0,\D_\star$ in appropriate places to obtain an upper triangular representative of $U+V$ in a canonical (up to row scalings) way. The entries of $C\vdots D$ are Laurent polynomials in the variables of the BCFW recursion. The entries of the row corresponding to $\D_i$ are linear in the variables $\hzet_i.$
   Note that the construction of the BCFW matrix is essentially following the steps of the BCFW recursion in the level of matrix representatives. The only time a choice is being made is in the forward limit step, in which we chose to use the $\tr(\D)$ row, in addition to the yellow row, in order to cancel the $A,B$ columns of the other rows.
\end{rmk}
\begin{cor}\label{cor:BCFW_up_2_gauge}
For $\D\in\CD_{n,k},$ every point in $S_\D$ is the result of applying the BCFW recursion step with respect to some choice of positive BCFW coordinates in $\R_+^{\BCFWV}$. Different choices which are equivalent under the gauge equivalence yield the same element of $S_\D,$ hence every point in $S_\D$ is also the result of applying the BCFW recursion step with respect to some choice of positive reduced BCFW coordinates in $\R_+^{\BBCFWV}$\end{cor}
\subsection{BCFW cells and positroids}\label{subsec:positroids}
We recall that a \emph{positroid} \cite{postnikov} for $\Gr_{k,N}^\geq,$ where $N$ is an ordered index set is the subset of $\mathcal{B}\subseteq\binom{N}{k}$ for which there exists an element of $\Gr_{k,N}^\geq $ for which precisely the Pl\"ucker coordinates labeled $\mathcal{B}$ are non zero. The elements of $\mathcal{B}$ are called \emph{bases} and the subspace of $\Gr_{k,N}^\geq$ with these non zero Pl\"ucker coordinates is called the \emph{positroid cell}. Positroid and positroid cells are well studied objects, with many beautiful properties, and can be labeled by different combinatorial object. See \cite{postnikov} for a beautiful introduction.
\begin{definition}\label{def:1-loop positroid}
A \emph{$1-$loop positroid cell} $\Pos=(\PosC,\PosD)$ for $\Gr_{k,N;1}^{\geq}$ is
a pair of positroid cells of $\Gr_{k,N}^{\geq},\Gr_{k+2,N}^{\geq}$ respectively,
such that there exists $U\vdots V\in\Gr_{k,N;1}^{\geq}$ 
with $U\in\PosC,~U+V\in\PosD.$ For such $(U\vdots V)$ we refer to $(\PosC,\PosD)$ as the $1-$loop positroid cell of $U\vdots V$. 

The \emph{$1-$loop} positroid cell of $U\vdots V$ is denoted $\Pos(U\vdots V)=(\PosC(U\vdots V),\PosD(U\vdots V)).$
If a set $S\subseteq\Gr_{k,N;1}^{\geq}$ is contained in a single $1-$loop positroid cell, we refer to the cell as the $1-$loop positroid cell of $S,$ and denote it by $\Pos(S)=(\PosC(S),\PosD(S)).$

The \emph{$1-$loop positroid} of a $1-$loop positroid cell $(\PosC,\PosD)$ is the pair $\PosB=(\PosBC,\PosBD)$ where $\PosBC,~\PosBD$ are the collection of bases for $\PosC,~\PosD$ respectively.
We will sometimes refer to a basis which \emph{belongs} to $\PosBC~(\PosBD)$ as a basis for $\PosC~(\PosD).$ We will sometimes refer to $\PosC$ as the \emph{tree part} of $(\PosC,\PosD),$ and to $\PosD$ as the \emph{loop part}. By slight abuse of notations we extend these notions to standard, that is, non loopy ($\ell=0$), positroids and positroid cells: the \emph{tree part} of a standard positroid (cell) is the positroid (cell) itself, while the loop part is trivial.
In cases that it is clear from the context that we discuss a $1-$loop positroid, or in case we state a claim for both $1-$loop positroids and standard positroids we may omit the prefix '$1-$loop'. In cases we would like to stress that a positroid is not a $1-$loop positroid we may refer to it as a \emph{$0-$loop} positroid. 

The intersection of a $1-$loop positroid cell $(\PosC,\PosD)$ with $\Gr_{k,N;\ell}^{\geq}$ for $\ell\in\{0,1\}$ is called \emph{the realizable part of $(\PosC,\PosD)$} and is denoted by $(\PosC,\PosD)^{\real}.$

It follows from \cref{lem:effect_ops} that the operations $\inc_i,\pre_i,\rem_i,x_i,y_i,\addL_{AB}$ also act on positroids and positroid cells. We use the same notations to describe their actions on these objects.
\end{definition}
Every point in $\Gr_{k,n;\ell}^\geq$ for $\ell=0,1$ belongs, by definition, to some $1-$loop positroid cell. We will see below that every BCFW cell is contained in a $1-$loop positroid cell. While tree level BCFW cells are positroid cells, the same does not hold in the loop level, as the next example shows.
\begin{ex}
Let $k=1,n=5,\ell=1.$ The dimension of $\Gr_{1,5;1}$ is $8.$ There is only one $1-$loop positroid cell of this dimension, where $\PosC,\PosD$ are the big cells $\Gr_{1,5}^>,~\Gr_{3,5}^>$ of $\Gr_{1,5}^\geq,\Gr_{3,5}^{\geq}$ respectively. But there are three $1-$loop BCFW cells, which must therefore correspond to it.
\end{ex}

In this section we will study the effect of the operations $\pre,\bcfw,\FL$ on $1-$loop positroids.

The proof of the following proposition is straight forward.
\begin{prop}\label{prop:matroid_under_pre}
Let $N$ be an index set, $p\in N$ the penultimate marker.
If $S\subseteq\Gr_{k,N\setminus\{p\};1}$ is contained in a $1-$loop positroid cell $(\PosC,\PosD)$ then $\pre_p S$ is contained in the $1-$loop positroid cell $\pre_p \Pos=(\pre_p\PosC,\pre_p\PosD)$ whose bases are the same bases as $(\PosC,\PosD),$ thought of as subsets of $N$ rather than $N\setminus \{p\}$. A similar statement holds in tree level.
\end{prop}
\begin{definition}\label{def:4biden}
Let $V \in \Gr_{k,N}$. A subset $J \subseteq N$ is \emph{coindependent for $V$} if $V$ has a nonzero Pl\"ucker coordinate $\lr{V}_I$, such that $I \cap J = \emptyset$. By convention, if $k=0$ then all subsets are coindependent for $V$. If $S \subset  \Gr_{k,n}$, then $J$ is \emph{coindependent for} $S$ if $J$ is coindependent for all~$V \in S$. If $S$ is a positroid cell with a positroid $\mathcal{B}$ we will sometimes say that $J$ is coindependent for $\mathcal{B}.$

	For $r\in\mathbb{N},$ $P\subset\Gr_{k,N}^{\geq}$ is \emph{$r$-coindependent} if every $r$-element subset of $N$ is co-independent for $P$. 

 Let $V \in \Gr_{k,N}$. $j\in N$ is \emph{independent for $V$} if $V$ has a nonzero Pl\"ucker coordinate $\lr{V}_I$, with $j\in I.$
 If $S \subset  \Gr_{k,n}$, then $j$ is \emph{independent for} $S$ if $j$ is independent for every~$V \in S$. If $S$ is a positroid cell with a positroid $\mathcal{B}$ we will sometimes say that $j$ is independent for $\mathcal{B}.$ In these cases we will also say that $V,~S$ or $\mathcal{B}$ are $j-$independent. 
 \begin{obs}\label{obs:ind_coind}
 
$\{i_1,\ldots,i_r\}\subseteq[n]$ is coindependent for the vector space $V$ if and only if $V$ does not contain a vector supported on the entries labeled $\{i_1,\ldots,i_r\}.$ 

 $j$ is independent for $V$ if and only if the $j$th column of $V$ is non zero.
 \end{obs}
 \begin{proof}
If $V$ contains such a vector $v$, let $C$ be a matrix representation including $v$ as a row. Then every minor of $C$ avoiding $\{i_1,\ldots,i_r\}$ is $0$ as the determinant of a matrix with a zero column. For the opposite direction, if $C$ is a matrix representation with rows $v_1,\ldots,v_k$, then the rows of $C^{[n]\setminus\{i_1,\ldots,i_r\}}$ are not linearly independent. Therefore there is a non trivial linear combination $\sum_{i=1}^ka_iv_i^{[n]\setminus\{i_1,\ldots,i_r\}}$ which vanishes. Thus, $\sum_{i=1}^ka_iv_i$, which cannot be $0$ since $C$ has rank $k$, is supported on $\{i_1,\ldots,i_r\}.$

For the second part, if $V^j$ is a zero column $j$ is not independent. Otherwise, starting from a matrix representative $C$ and performing row operations we can get to a matrix representation $C'$ in which the $j$th column has only the first entry non zero. The remaining rows have rank $k-1.$ If $i_1,\ldots,i_{k-1}$ are linearly independent columns for the remaining rows, the minor $\{i_1,\ldots,i_{k-1},j\}$ is non zero.
 \end{proof}


\end{definition}
\begin{cond}\label{condition:BCFW}
Let $k_L,N_L,\ell_L,k_R,N_R,\ell_R$ be as in \cref{nn:bcfwmap}. We say that two realizable positroid cells $\Pos_L\subseteq\Gr_{k_L,N_L;\ell_L}^{\geq},~\Pos_R\subseteq\Gr_{k_R,N_R;\ell_R}^{\geq},$ or the corresponding positroids, satisfy the \emph{BCFW product condition} if $a,b,n$ are coindependent for the tree part of $\Pos_L,$ and $b,c,d,n$ are coindependent for the tree part of $\Pos_R.$
\end{cond}
The following proposition modifies \cite[Lemma 10.1]{even2023cluster} to the one loop level.
\begin{prop}
\label{prop:matroid_under_bcfw} Fix $k_L,\ldots,\ell_R$ as in \cref{nn:bcfwmap}.
	Let $S_L \subset {\Gr}^{\geq}_{k_L, N_L;\ell_L}$ and $S_R \subset {\Gr}^{\geq}_{k_R, N_R;\ell_R}$, be two spaces such that $S_L\subseteq\Pos_L,~S_R\subseteq\Pos_R,$ for positroid cells $\Pos_L,\Pos_R$ satisfying \cref{condition:BCFW}. Let $\PosBC_L,\PosBC_R$ be the tree parts of the associated positroids.
	
Then the BCFW product $S_L\bcfw S_R$ is defined, and is contained in a single $\ell-$loop positroid cell. The bases of the tree part of that positroid cell are exactly the sets $I_L \sqcup \{f\} \sqcup I_R$ where $I_L, f, I_R$ are disjoint and satisfy one of the following:
	\begin{center}
	\begin{tabular}
 {|c | c|c|c|}
		\hline
		&$I_L$ & $f$ & $I_R$\\
		\hline
		(1)&$I_L \in \PosBC_L, ~ b \notin I_L$& $a$ &$I_R \in \widehat{\PosBC_R}$\\
		\hline
		(2)&$I_L \in \PosBC_L$& $b$ &$I_R \in \widehat{\PosBC_R}$\\
				\hline
		(3)&$I_L \in \widehat{\PosBC_L}$& $c$ &$I_R \in \PosBC_R,~d \notin I_R$\\
		\hline
		(4)&$I_L \in \widehat{\PosBC_L}$& $d$ &$I_R \in \PosBC_R$\\
		\hline
		(5)&$I_L \in \widehat{\PosBC_L}$& $n$ &$I_R \in \PosBC_R$\\
		\hline
		(6)&$I_L \in \widehat{\PosBC_L}$& $n$ &$(I_R \setminus \{c\}) \cup \{d\} \in \PosBC_R$\\
		\hline
	\end{tabular}
\end{center}
where $\widehat{\PosBC_L}\supseteq \PosBC_L$ is the positroid obtained from $\PosBC_L$ by applying $y_a,$ and $\widehat{\PosBC_R}\supseteq\PosBC_R$ is the positroid obtained from $\PosBC_R$ by first applying $y_d$ and then $y_c.$ The same holds for bases of $\PosD$, where if the loopy part is on the left $\PosBC_L,\widehat{\PosBC_L}$ are replaced by $\PosBD_L,\widehat{\PosBD_L},$ respectively, and if the loopy part is on the right, then $\PosBC_R,\widehat{\PosBC_R}$ are replaced by $\PosBD_R,\widehat{\PosBD_R},$ respectively, where $\widehat{\PosBD_L},~\widehat{\PosBD_R}$ are defined similarly to $\widehat{\PosBC_L},~\widehat{\PosBC_R}$.   
\end{prop}
\begin{proof}
The $\ell_L=\ell_R=0$ case is \cite[Lemma 10.1]{even2023cluster}. The description of the BCFW operation \cref{def:bcfw-map} and \cref{lem:effect_ops} show that, given loopy matrix representatives for an element of $S_L$ and an element of $S_R$, the minors of the resulting $\ell-$loop matrix are multilinear functions with positive coefficients in the minors of the left and right matrices, the BCFW variables and their inverses $\alpha^\pm,\ldots,\varepsilon^\pm.$ Thus, also in the $1-$loop case $S_L\bcfw S_R$ is contained in a single loopy positroid cell. 
The claim about its tree level part is again \cite[Lemma 10.1]{even2023cluster} applied to the tree level parts of $S_L,S_R.$ For the loopy part we use the following trick. We add two columns, in two consecutive positions $A,B$ in the side which contains the loop. These columns contain zeroes in all rows but those of the loopy part $D_L$ or $D_R,$ where they contain an invertible $2\times 2$ matrix. By the assumption on the tree level the resulting $k_L+2$ or $k_R+2$ dimensional vector space is $a,b,n$ or $b,c,d,n$ coindependent, respectively. We apply \cite[Lemma 10.1]{even2023cluster} to it, and ignore bases which include $A$ or $B.$ This implies the result.
\end{proof}
The following claim is a direct consequence of the matroid description in \cref{prop:matroid_under_bcfw}.
\begin{cor}\label{cor:matroid_after_bcfw}
Using the notations of \cref{prop:matroid_under_bcfw}, write $\Pos$ for the positroid cell containing $S_L\bcfw S_R,$ and $\PosBC$ the tree part of its positroid. 
Suppose that $\PosBC_L$ is $\{a,b,n\}-$coindependent, and that $\PosBC_R$ is $\{b,c,d,n\}-$coindependent. Then the following holds for $\PosBC.$
\begin{enumerate}
\item\label{it:matroid_after_bcfw_bcdn}
It is $\{1,c,d,n\}-$coindependent.
\item\label{it:matroid_after_bcfw_chord}
It is $J$-coindependent for every $J\in\binom{\{a,b,c,d,n\}}{4}.$  
\item\label{it:matroid_after_bcfw_promoting_1_side}
For every basis $I$ in $\PosBC_L$ ($\PosBC_R$) there is a basis in $\PosBC$ containing it. Moreover, there exists a basis for $\PosBC$ whose intersection with $N_L$ ($N_R$) is $I.$ With the exception of $I\in\PosBC_L$ with $b\in I,$ one can also find, for every  $t\in\{a,b,c,d,n\}\setminus I,$ a basis of $\PosBC$ whose intersection with $N_L\cup\{a,b,c,d,n\}~(N_R\cup\{a,b,c,d,n\})$ is $I\cup\{t\}$. 
\end{enumerate}
\end{cor}
The following was proven in \cite{even2023cluster}.
\begin{lemma}\label{lem:4coind_tree}\cite[Corollary 10.7]{even2023cluster}
Tree BCFW cells are $4-$coindependent.
\end{lemma}

We can apply \cref{prop:matroid_under_bcfw},\cref{cor:matroid_after_bcfw} and \cref{lem:4coind_tree} to deduce the following results on the matroids of tree level BCFW cells.
\begin{lemma}\label{lem:bcfw_before_fl_tree}
Let $\D\in\CD_{n,k}^0$ be a chord diagram with at least one chord ending at $(n-2,n-1)$ and no chord of the form $(n-4,n-3,n-2,n-1).$
Then the following holds.
\begin{enumerate}
\item\label{it:bcfw_before_fl_yellow}$S_\D$ is $n-1-$independent and $\{n-4,n-3,n-2,n-1,n\}-$coindependent.
\item\label{it:bcfw_before_fl_bcdn}
$S_\D$ has a basis $I\cap\{1,n-4,n-3,n-2,n-1,n\}=\{n-1\}.$
\item\label{it:bcfw_before_fl_chord}
For every chord $\D_i=(a_i,b_i,c_i,d_i)$ with $c_i,d_i\neq n-2$ there is a basis $I$ for $S_\D$ with  $I\cap\{a_i,b_i,c_i,d_i,n-2,n-1\}=\{n-1\}.$
\item\label{it:bcfw_before_fl_reds}For every chord $\D_i$ with $d_i\in\{n-2,n-1\}$ there exists a basis $I$ with $I\cap\{a_i,b_i,n-3,n-2,n-1,n\}=\{n\}.$ If in addition $b_i=n-4,$ then there exists a basis $I$ with $I\cap\{n-4,\ldots,n\}=\{n-4\}.$
\item\label{it:bcfw_before_fl_blues_reds}If $d_i\in\{n-2,n-1\}$ there exists a basis $I$ for $S_\D$ which does not intersect $\{a_i,b_i,n-3,n-2,n\}$ and includes $n-1.$
\end{enumerate}
\end{lemma}
\begin{proof}
For \cref{it:bcfw_before_fl_yellow} we first show that if $\D$ has no chord supported on $\{n-4,\ldots,n\}$ then it is $\{n-4,\ldots,n\}-$independent. We induct on the number of chords in the diagram. The case of no chords is true in an empty sense.  If there is one chord then by the assumption on $\D$ it is straightforward that $S_\D$ is $\{n-4,\ldots,n\}-$independent. In the general case $\D=\D_L\bcfw\D_R.$ $\D_R$ has no chord $(n-4,n-3,n-2,n-1)$, it either has a chord ending at $(n-2,n-1)$ or $\D_R=\pre_{n-1}\D'_R.$ In the former case $\D_R$ is $\{n-4,\ldots,n\}-$independent by induction. In the latter $\D'_R$ is $\{n-4,n-3,n-2,n\}$-coindependent by \cref{lem:4coind_tree}, hence by \cref{prop:matroid_under_pre} $\D_R$ is $\{n-4,\ldots,n\}$ coindependent. \cref{cor:matroid_after_bcfw},~\cref{it:matroid_after_bcfw_promoting_1_side} shows the same for $\D.$ 

Elements of $S_\D$ have a matrix representative whose $k$th row which corresponds to $\D_k$ has support $\{a_k,b_k,c_k,d_k,n\}.$ Thus, by \cref{obs:ind_coind}, every non zero maximal minor must include one of $a_k,b_k,c_k,d_k,n.$ By \cref{cor:matroid_after_bcfw},~\cref{it:matroid_after_bcfw_chord} there is a minor not containing $\{a_k,b_k,c_k,n\}.$ It thus must contain $d_k=n-1.$

For \cref{it:bcfw_before_fl_bcdn} we induct on the number of chords. For $k=0,1$ this is straightforward, under the assumptions. Assume the claim holds for chord diagrams with less than $k$ chords which satisfy the assumptions. Consider $\D_R.$ If it also satisfies the assumptions, then we can pick a basis $I_R $ for $S_{\D_R}$ satisfying the assumptions, meaning that it intersects $\{b,n-4,\ldots,n\}$ only at $n-1.$ We use \cref{lem:4coind_tree} to find a basis $I_L$ for $S_{\D_L}$ which avoids $1,a,b,n$ and pick $I=I_L\cup\{b\}\cup I_R,$ which satisfies the requirements. The other possibility is that $\D_R=\pre\D'_R.$ In this case we use \cref{lem:4coind_tree} to pick a basis $I_R$ for the positroid of $S_{\D'_R}$ which avoids $n-4,n-3,n-2,n,$ and $I_L$ as before. Now $I= I_L\cup\{n-1\}\cup I_R$ is our required basis.

For \cref{it:bcfw_before_fl_chord}, let $\D'$ be the subdiagram of $\D$ which is the result of the $i$th BCFW product step. By \cref{cor:matroid_after_bcfw},~\cref{it:matroid_after_bcfw_chord} $S_{\D'}$ has a basis which avoids $a_i,\ldots,d_i.$ It also avoids $n-2,n-1$ which do not belong to the index set of $\D'.$ We use \cref{prop:matroid_under_pre} and \cref{cor:matroid_after_bcfw},~\cref{it:matroid_after_bcfw_promoting_1_side} to lift this basis to a basis with the same property for $S_{\D''}$ where $\D''$ is the subdiagram of $\D$ which descends from the maximal top chord. Then, using \cref{cor:matroid_after_bcfw},~\cref{it:matroid_after_bcfw_promoting_1_side} again we find a basis for $S_\D$ whose intersection with $\{a_i,b_i,c_i,d_i,n-2,n-2\}$ is $\{n-1\}.$

For \cref{it:bcfw_before_fl_reds}, by applying \cref{cor:matroid_after_bcfw},~\cref{it:matroid_after_bcfw_promoting_1_side} repeatedly we may assume $i=k.$ Similarly to the proof of the first two items we will induct of the number of chords to show that $S_{\D_R}$ has a basis which avoids $a_k,b_k,n-3,n-2,n-1,n.$ 
For at most one chord this is straight forward. $\pre$ steps do not affect the basis, and for $\bcfw$ steps with respect to indices $a,b,c,d,n$ we will use \cref{cor:matroid_after_bcfw},~\cref{it:matroid_after_bcfw_promoting_1_side} with $t=b,$ using $b\neq n-3$, which follows from the assumptions. Using \cref{cor:matroid_after_bcfw},~\cref{it:matroid_after_bcfw_promoting_1_side} once more we may find a basis for $S_\D$ whose intersection with $\{a_i,b_i,n-3,n-2,n-1,n\}$ is $\{n\}.$ The second part of the item is proven similarly.

The proof of \cref{it:bcfw_before_fl_blues_reds} when $d_i=n-1$ is similar to the previous ones. We use induction and \cref{cor:matroid_after_bcfw},~\cref{it:matroid_after_bcfw_promoting_1_side} to find a basis for $S_{\D}$ which avoids $a_i,b_i,n-3,n$ and includes exactly one of $n-2,n-1.$ We omit the details. When $d_i=n-2$ we use similar considerations to show that for the subdiagram obtained at the end of the $i$th BCFW step there is a basis $I'_R$ which avoids $a_i,b_i,c_i,d_i$ but uses $n.$ In the $\p(i)$th BCFW step that will add a parent to $\D_i$ we use case $(a)$ in the table of \cref{prop:matroid_under_bcfw}, and \cref{lem:4coind_tree}, to find a basis $I$ for the resulting positroid which satisfies $I=I_L\cup\{a_{\p(i)}\}\cup I_R$ with
$I_L$ a basis for the left component which avoids $a_{\p(i)},b_{\p(i)},n,$ and $I_R = I'_R\cup\{n-1\}\setminus\{d\}.$ This is a basis to the subdiagram obtained after the $p(i)$th BCFW step in the generation sequence for $\D,$ which satisfies our requirements. We now use \cref{cor:matroid_after_bcfw},~\cref{it:matroid_after_bcfw_promoting_1_side}, as in the proofs of the previous items, to lift this base to a base of our positroid of interest.
\end{proof}
\begin{cond}\label{condition:FL}
 Let $N=[n]\cup\{A,B\}$ be an index set. Let $S\subseteq \Gr_{k+1,[n]\cup\{A,B\}}^{\geq}$ be a positroid cell. We say that $S$ satisfies the \emph{forward limit condition} if $B$ is independent for $S,$ and $\{c,d,A,B,n\}$ are coindependent for $S,$ where $c=n-2,d=n-1$ and $c\lessdot d\lessdot A\lessdot B\lessdot n.$   
\end{cond}
\begin{prop}\label{prop:matroid_under_FL}
Let $S\subseteq \Gr_{k+1,[n]\cup\{A,B\}}^{\geq}$ be a positroid cell with positroid $\PosB$ for which $B$ is independent. Then if $\FL(S)\neq\emptyset$ it is contained in a single $1-$loop positroid cell $(\PosC,\PosD)$,with bases $(\PosBC,\PosBD)$ which we now describe.
$I\in\binom{[n]}{k}$ belongs to $\PosBC$ in the following cases. Let $J=I\cap\{c,d,n\}.$
	\begin{center}
	\begin{tabular}
 {|c | c|c|}
		\hline
		&$J$ & $I$\\
		\hline
		(1)&$\{c,d,n\}$& $I\cup\{A,B\}\setminus\{n\}\in\PosB$ or $I\cup\{A,B\}\setminus\{d\}\in\PosB$ or $I\cup\{B\}\in\PosB$\\
		\hline
		(2)&$\{d,n\}$& $I\cup\{A,B\}\setminus\{n\}\in\PosB$ or $I\cup\{A,B\}\setminus\{d\}\in\PosB$ or $I\cup\{B\}\in\PosB$\\
				\hline
		(3)&$\{n\}$& $I\cup\{A,B\}\setminus\{n\}\in\PosB$ or $I\cup\{B\}\in\PosB$\\
		\hline
		(4)&$\{c,d\}$&  $I\cup\{A,B\}\setminus\{d\}\in\PosB$ or $I\cup\{B\}\in\PosB$\\
		\hline
		(5)&$\{d\}$& $I\cup\{A,B\}\setminus\{d\}\in\PosB$ or $I\cup\{B\}\in\PosB$\\
		\hline
  	(6)&$\{c,n\}$& $I\cup\{A,B\}\setminus\{n\}\in\PosB$ or $I\cup\{A,B\}\setminus\{c\}\in\PosB$ or $I\cup\{A,B,d\}\setminus\{c,n\}\in\PosB$ or $I\cup\{B\}\in\PosB$\\
		\hline
		
  (7)&$\{c\}$& $I\cup\{A,B\}\setminus\{c\}\in\PosB$ or $I\cup\{B\}\in\PosB$\\
		\hline
  (8)&$\emptyset$& $I\cup\{B\}\in\PosB$\\
		\hline

	\end{tabular}
\end{center}

$I\in\binom{[n]}{k+2}$ belongs to $\PosBD$ in the following cases. Let $J=I\cap\{c,d,n\}.$ Then $J\neq\emptyset$ and
	\begin{center}
	\begin{tabular}
 {|c | c|c|}
		\hline
		&$J$ & $I$\\
		\hline
		(1)&$\{c,d,n\}$& $I\setminus\{n\}\in\PosB$ or $I\setminus\{d\}\in\PosB$\\
		\hline
		(2)&$\{d,n\}$& $I\setminus\{n\}\in\PosB$ or $I\setminus\{d\}\in\PosB$ \\
				\hline
		(3)&$\{n\}$& $I\setminus\{n\}\in\PosB$\\
		\hline
		(4)&$\{c,d\}$& $I\setminus\{d\}\in\PosB$\\
		\hline
		(5)&$\{d\}$& $I\setminus\{d\}\in\PosB$\\
		\hline
  	(6)&$\{c,n\}$& $I\setminus\{n\}\in\PosB$ or $I\setminus\{c\}\in\PosB$ or $I\cup\{d\}\setminus\{c,n\}\in\PosB$\\
		\hline
		
  (7)&$\{c\}$& $I\setminus\{c\}\in\PosB$\\
		\hline
	\end{tabular}
\end{center}

\end{prop}
\begin{proof}
We follow the process in which we defined the forward limit. We first apply $\inc_l,$ and this operations takes $\mathcal{B}$ to 
\[\mathcal{B}_0:=\{J\in \binom{[n]\cup\{A,B\}}{k+2}|~J=\{l\}\cup I,I\in\mathcal{B}\}.\]
The applications of $x_*,y_*$ operations take $\mathcal{B}_0$ to $\mathcal{B}_1,$ which is calculated from $\mathcal{B}_0$ by applying \cref{lem:effect_ops} to $x_l,x_A,y_d,y_c$ and $\mathcal{B}_0.$
In order to define $\PosD$ we first erase columns $A,B$ which corresponds to passing to $\mathcal{B}'=\{I\in\mathcal{B}_1|~A,B\notin I\}.$ 
A case-by-case examination shows that $I\in\PosB'$ is described via $J'=I\cap\{c,d,l\},$ which must not be empty, as follows:

	\begin{center}
	\begin{tabular}
 {|c | c|c|}
		\hline
		&$J'$ & $I$\\
		\hline
		(1)&$\{c,d,l\}$& $I\setminus\{l\}\in\PosB$\\
		\hline
		(2)&$\{d,l\}$& $I\setminus\{n\}\in\PosB$ \\
				\hline
		(3)&$\{l\}$& $I\setminus\{n\}\in\PosB$\\
		\hline
		(4)&$\{c,d\}$& $I\setminus\{d\}\in\PosB$\\
		\hline
		(5)&$\{d\}$& $I\setminus\{d\}\in\PosB$\\
		\hline
  	(6)&$\{c,l\}$& $I\setminus\{l\}\in\PosB$ or $I\cup\{d\}\setminus\{c,l\}\in\PosB$\\
		\hline
		
  (7)&$\{c\}$& $I\setminus\{c\}\in\PosB$\\
		\hline
	\end{tabular}
\end{center}

Finally, adding the $l$th column to the $n$ th column, and erasing $l$ takes us to
\[\PosBD=\{I\in\binom{[n]}{k+2}|~I\in\PosB'~\text{or }I\cup\{l\}\setminus\{n\}\in\PosB'\}.\]
The result of this operation is as stated.
Note that had we not erased the $A,B$ columns the operation of adding the $l$th column to the $n$th column could have resulted in losing the positivity. 

A similar case-by-case analysis shows the following rule for a basis $I$ which contains $A,B$, to belong to $\PosB_1$, in terms of $J'=I\cap\{c,d,l\}:$
	\begin{center}
	\begin{tabular}
 {|c | c|c|}
		\hline
		&$J'$ & $I$\\
		\hline
		(1)&$\{c,d,l\}$& $I\cup\{A,B\}\setminus\{l\}\in\PosB$\\
		\hline
		(2)&$\{d,l\}$& $I\cup\{A,B\}\setminus\{l\}\in\PosB$\\
				\hline
		(3)&$\{l\}$& $I\cup\{A,B\}\setminus\{l\}\in\PosB$\\
		\hline
		(4)&$\{c,d\}$&  $I\cup\{A,B\}\setminus\{d\}\in\PosB$ or $I\cup\{B\}\in\PosB$\\
		\hline
		(5)&$\{d\}$& $I\cup\{A,B\}\setminus\{d\}\in\PosB$ or $I\cup\{B\}\in\PosB$\\
		\hline
  	(6)&$\{c,l\}$& $I\cup\{A,B\}\setminus\{l\}\in\PosB$ or $I\cup\{A,B,d\}\setminus\{c,l\}\in\PosB$ or $I\cup\{B\}\in\PosB$\\
		\hline
		
  (7)&$\{c\}$& $I\cup\{A,B\}\setminus\{c\}\in\PosB$ or $I\cup\{B\}\in\PosB$\\
		\hline
  (8)&$\emptyset$& $I\cup\{B\}\in\PosB$\\
		\hline

	\end{tabular}
\end{center}
And finally, after the last step involving the columns $l,n$ we see that\begin{equation}\label{eq:PoSC}\PosBC=\{J\in\binom{[n]}{k}|~J\cup\{A,B\}\in\mathcal{B}_1~\text{or }J\cup\{A,B,l\}\setminus\{n\}\in\mathcal{B}_1\}.
\end{equation}
satisfies the description in the statement of the claim.

The above analysis writes $\PosBC,\PosBD$ purely in terms of $\PosB.$ Thus, whenever they are non empty, the resulting space is contained in a single $1-$loop positroid cell.
\end{proof}
\begin{rmk}\label{rmk:degenerate_pos}
For later purposes we also define the positroid cells $\widetilde{\PosC},~\widetilde{\PosD}$ obtained from $S$ by $\FL',$ the forward limit degenerate at $c$ of \cref{rmk:last_y_c}. The same argument shows that these two positroid cells are well defined. The rules for the associated positroids are this time in terms of $J'=I\cap\{d,n\},$ and are given by the rows of the tables of \cref{prop:matroid_under_FL} in which $J$ does not include $c.$ The proof is the same.
\end{rmk}
\cref{prop:matroid_under_FL} and \cref{rmk:degenerate_pos} have the following consequences.
\begin{cor}\label{cor:matroid_after_fl}
With the notations of \cref{prop:matroid_under_FL}, let $\Pos=(\PosC,\PosD)$ be the $1-$loop positroid cell containing $\FL(S),$ and $(\PosBC,\PosBD)$ the associated positroid. Then:
\begin{enumerate}
\item\label{it:matroid_after_fl_B_case}
If $S\subseteq\Gr_{k+1,N}$ has a basis $I$ with $I\cap\{A,B\}=\{B\}$ then $\PosBC$ has a basis $I\setminus\{B\}$. In particular, if it has a basis $I$ with $I\cap\{1,c,d,A,B,n\}=\{B\}$ then $\PosBC$ is $\{1,c,d,n\}$-independent. If it has a basis $I$ with $I\cap\{c,d,A,B,n\}=\{B\}$ then $\PosBC$ is $\{c,d,n\}$-independent.
\item\label{it:matroid_after_FL_yellow}If $\{c,d,A,B,n\}$ is coindependent for $S$ then every $\{i,j\}\in\binom{c,d,n}{2}$ is coindependent for $\PosBD.$ Moreover, for every basis $I$ for $\PosB$ not intersecting $\{c,\ldots, n\}$ and every $i\in\{c,d,n\},$ $I\cup\{i\}$ is a basis for $\PosBC$
\item\label{it:matroid_after_fl_reds}If there is $I\in\PosB$ satisfying $I\cap\{d,A,B,n\}=\{n\}$ or $I\cap\{c,d,A,B,n\}=\{c\}$ then $I\cup\{d\}$ is a basis for $\PosBD.$
\end{enumerate}
\end{cor}
\begin{prop}\label{prop:fixed_positroids}
Let $\D\in \CD_{n,k}$ be a chord diagram. Then every BCFW or forward limit step in the generation sequence of the cell is defined, and there exists a $1-$loop positroid cell $\Pos_D=(\PosC_\D,\PosD_\D)$ of $\Gr_{k,n;1}^\geq$ which contains $S_\D.$
\end{prop}
\begin{proof}
We will use induction on the lexicographically ordered $(\ell,k,n)$ to show a stronger claim, that moreover the tree part of the $1-$loop positroid cell at each step is coindependent with respect to the set made of the minimal and three maximal elements of the index set of that step. 
If $\ell=0$ then $\D$ is a tree chord diagram. In this case it is well known it is contained in a positroid cell, see, for example, \cite[Section 3]{even2021amplituhedron}, which also describes the decorated permutation of the positroid cell and algorithms for constructing it. \cref{lem:4coind_tree} shows the coindependent result.

Assume $\ell=1.$ If $\D=\pre_{p}\D'$ then the result follows from the induction and \cref{prop:matroid_under_pre}.

If $\D=\FL(\D')$ then $\D'$ is a tree diagram, which by assumption has no chord $(c,d,A,B)$ but has some chord ending at $(A,B).$ By \cref{lem:bcfw_before_fl_tree},~\cref{it:bcfw_before_fl_yellow},~\cref{it:bcfw_before_fl_bcdn} the matroid $\PosB$ of $S_{\D'}$ is $\{c,d,A,B,n\}$ coindependent, and has a basis $I$ with $I\cap\{1,c,d,A,B,n\}=\{B\}.$ \cref{prop:matroid_under_FL} then shows that the $1-$loop positroid of $S_\D$ is nonempty, hence the forward limit step is defined, and moreover its tree part is $\{1,c,d,n\}-$coindependent. In particular $S_\D$ is nonempty and is contained in a single $1-$loop positroid cell.

If $\D=\D_L\bcfw \D_R,$ where both are assumed to satisfy the inductive assumptions, then by \cref{prop:matroid_under_bcfw} and \cref{cor:matroid_after_bcfw},~\cref{it:matroid_after_bcfw_bcdn}, also $S_\D$ satisfies the inductive assumptions.
\end{proof}
The following lemma summarizes some useful properties of the BCFW cells' positroids.
\begin{lemma}\label{lem:bcfw_after_fl}Let $\D\in\CD_{n,k}^1$ be a chord diagram. Let $\Pos=(\PosC,\PosD)$ be the $1-$loop positroid cell containing it. 
\begin{enumerate}
\item\label{it:bcfw_after_fl_bcdn}$\{1,n-2,n-1,n\}$ is coindependent for $\PosC$.
\item\label{it:bcfw_after_fl_chords}For every black or purple chord $\D_i,$ $\{a_i,b_i,c_i,d_i\}$ is coindependent for $\PosC.$ If $\D=\D_L\bcfw\D_R$ for $\D_L,\D_R$ where $\D_L,\D_R$ have index sets as in \cref{nn:bcfwmap}, then $S_\D$ is also $J-$coindependent for every $J\in\binom{\{a,b,c,d,n\}}{4}.$
\item\label{it:bcfw_after_fl_for_Sblue_argument} $\{c_\star,d_\star,n\}$, if $\D_0$ is a top chord, or $\{a_{\p(0)},b_{\p(0)},c_\star,d_\star\},$ otherwise, is coindependent for $\PosC.$
\item\label{it:bcfw_after_fl_yellow}
$\{c_\star,d_\star\}$ is coindependent for $\PosD.$ If $\D=\FL(\D')$ then every $J\in\binom{\{c_\star,d_\star,n\}}{2}$ is coindependent for
$\PosD.$ 
\item\label{it:bcfw_after_fl_reds}For every red or blue chord $\D_i$ $\{a_i,b_i\}$ is coindependent for $\PosD.$ If $\D=\FL(\D')$ then  
every $J\in\binom{\{a_i,b_i,n\}}{2}$ is coindependent for
$\PosD.$
\end{enumerate}
\end{lemma}
\begin{proof}
\cref{it:bcfw_after_fl_bcdn}  was proven in \cref{prop:fixed_positroids}.

In case $\D=\FL(\D')$ the first part of \cref{it:bcfw_after_fl_chords} follows from \cref{cor:matroid_after_fl}~\cref{it:matroid_after_fl_B_case} and \cref{lem:bcfw_before_fl_tree},~\cref{it:bcfw_before_fl_chord}.

Recall $\D$ can be constructed from its $\FL-$subdiagram by repeated applications of $\pre$ and BCFW products with tree level BCFW cells. We thus apply \cref{prop:matroid_under_pre}, \cref{cor:matroid_after_bcfw} and \cref{lem:4coind_tree} to deduce both parts of \cref{it:bcfw_after_fl_chords} from the case $\D=\FL(\D')$ described above.

For \cref{it:bcfw_after_fl_for_Sblue_argument}, we will show by induction on the number of steps in the recipe of the BCFW cell in terms of BCFW and $\pre$ steps performed after the forward limit step, that for every $1-$loop chord diagram $\D$ on the index set $N,$ either $\{\min N,c_\star,d_\star, n\}$ is coindependent for the corresponding matroid $\PosC,$ when $\D_0$ is a top chord, or $\{a_{\p(0)},a_{\p(0)},c_\star,d_\star\}$ is coindependent for $\PosC.$ We will refer to the above quadruple as the \emph{forbidden quadruple for $\PosC$}.

If $\D=\FL(\D')$ then \cref{it:bcfw_after_fl_for_Sblue_argument} follows from the first item, and moreover $\{\min N,c_\star,d_\star,n\}$ is coindependent for $\PosC$. 

If $\D=\pre_{p}\D'$ and the claim holds for $\D'$ then it clearly holds for $\D$ by induction and  \cref{prop:matroid_under_pre}. If $\D=\D_L\bcfw\D_R,$ then if $\D_L$ is the $1-$loop component, let $I_L$ be a basis for the positroid $\PosC_L$ of $S_{\D_L}$ avoiding the forbidden quadruple for $\PosC_L$, which exists by induction. Choose a basis $I_R$ for the positroid cell of $S_{\D_R},$ avoiding $\{b,c,d,n\}$, which is possible thanks to \cref{lem:4coind_tree}, and pick $f=d.$ Then $I_L\cup\{f\}\cup I_R$ is a basis for $\PosC,$ by \cref{prop:matroid_under_bcfw} which avoids the forbidden quadruple for $\PosC_L,$ which in this case equals the forbidden quadruple for $\PosC$.

If $\D_R$ is the $1-$loop component there are two cases. Either $\D_0$ is a top chord of $\D_R$, or not. If it is, pick $I_R$ which avoids the forbidden quadruple for $\PosC_R,$ which must be $\{b,c_\star,d_\star,n\}.$ For this we use induction. We pick a basis $I_L$ for the left part which avoids $a,b,n.$ For this we use \cref{lem:4coind_tree}. We pick $f=n.$ $I=I_L\cup\{f\}\cup I_R$ avoids $a,b,c_\star,d_\star.$
If $\D_0$ is not a top chord of $\D_R$ then the forbidden quadruple for $\PosC_R$ is $\{a_{\p(0)},b_{\p(0)},c_\star,d_\star\}$. We pick a basis $I_R$ for $\PosC_R$ avoiding it, using the induction, a basis $I_L$ for the left part avoiding $a,b,n$, using \cref{lem:4coind_tree} and $f=a.$ $I=I_L\cup\{f\}\cup I_R$ avoids $\{a_{\p(0)},b_{\p(0)},c_\star,d_\star\}$, which is the forbidden quadruple for $\PosC.$
In both cases, by \cref{prop:matroid_under_bcfw} $I_L\cup\{f\}\cup I_R.$

\cref{it:bcfw_after_fl_yellow}, in case $\D=\FL(\D')$ follows from \cref{cor:matroid_after_fl},~\cref{it:matroid_after_FL_yellow} and \cref{lem:bcfw_before_fl_tree},~\cref{it:bcfw_before_fl_yellow}. 
For general $\D\in\CD_{n,k}^1$ let $\D'$ be its $\FL-$subdaigram. Then $\{c_\star,d_\star\}$ is coindependent for its loopy part, meaning we can find a basis for the loopy part not involving $c_\star,d_\star.$ $S_\D$ is constructed from $S_{\D'}$ by a series of $\pre$ and $\bcfw$ steps, which is a subsequence of the generation sequence. Let $\D_i,~i=1,\ldots, T$ be the sequence of intermediate terms with $\D_1=\D', \D_T=\D$ and every $\D_i$ is either $\pre_p\D_{i-1},$ $\D_{i-1}\bcfw\D''$ or $\D''\bcfw\D_{i-1}.$ Let $I_{i-1}$ be a basis for the loopy part of $\D_{i-1}$ not involving $\{c_\star,d_\star\},$ whose existence we assume by induction. We will constructs a basis $I_i$ with the same property for the loop part of $\D_{i}.$ If $\D_i=\pre_p\D_{i-1}$ then $I_i=I_{i-1}.$
If $\D_{i-1}\bcfw\D''$ or $\D''\bcfw\D_{i-1},$ where the new chord is $(a,b,c,d)$,  let $j\in\{a,b,c,d,n\}$ be an index not included in the index set of $\D_{i-1},$ defined as follows. If $\D_{i-1}$ is the left component of the BCFW product, we pick $j=d.$ If it is the right part we pick $j=a.$ In both case $j\neq c_\star,d_\star.$
We pick a basis $J$ for $\D'',$ which is a tree level BCFW cell, not including $\{a,b,c,d,n\}.$ For this we use the $4-$coindependence of $S_{\D''},$ \cref{lem:4coind_tree}. Then by \cref{prop:matroid_under_bcfw} $I_i=J\cup \{j\}\cup I_{i-1}$ is a basis satisfying our requirements.

\cref{it:bcfw_after_fl_reds}, in case $\D=\FL(\D')$ follows from \cref{cor:matroid_after_fl},~\cref{it:matroid_after_fl_reds},~\cref{it:matroid_after_FL_yellow} and \cref{lem:bcfw_before_fl_tree},~\cref{it:bcfw_before_fl_reds}. The coindependence of $\{a_i,b_i\}$ for $\PosD$ in the general case is proven similarly to the previous item.
\end{proof}

\subsection{The domino form}
Karp, Williams, Zhang and Thomas~\cite[Appendix~A]{karp2020decompositions} suggested representing the points in the BCFW cells by special matrices, whose rows are called \emph{domino bases}. \cite{even2021amplituhedron} used this form extensively in their proof of the $\ell=0$ BCFW conjecture. We now recall the definition, and generalize it to $1-$loop.
\begin{definition}
\label{def:domino_entries}
The \emph{($\ell-$loop) domino matrix} of a chord diagram $\D \in \mathcal{CD}^1_{n,k}$ is a pair $(C\vdots D)_\D=(C\vdots D)_\D(\{\alpha_0,\beta_0,\varepsilon_0\}\cup\{\gamma_\star,\delta_\star,\varepsilon_\star\}\cup\{\alpha_i,\ldots,\varepsilon_i\}_{i\in[k+\ell]_\D}),$
where $C_\D$ is a $k \times n$ matrix, and $D_\D$ is a $2\times n$ matrix. If $\ell=0$ there are no $\{\alpha_0,\beta_0,\varepsilon_0\}\cup\{\gamma_\star,\delta_\star,\varepsilon_\star\},$ but we may keep them in the notations to avoid splitting into cases.
The rows of $C_\D$ are indexed by $[k+\ell]_\D=[k+\ell]\setminus\{\tr(\D)\}$ where the $l$th row corresponds to the chord $\D_l.$ They \emph{appear in the matrix}, however, according to the increasing order of the starting points $a_i$ of the chords. We say that a row is black, red, blue, or purple if the corresponding chord is.
$D_\D$ consists of two rows, the top row is the \emph{yellow row} $D_\star,$ and the bottom row is the \emph{top red row} $D_0.$ 
We slightly abuse terminology by transferring terms from chords to describe relations between corresponding row, for example we say that $D_p$ is a parent of $D_l$ if $\D_p$ is a parent of $\D_l.$
The matrices $C_\D,D_\D$ are defined as follows.
We start with all rows being $0.$
\begin{itemize}
\item[(a)]
Write $(\alpha_l,\beta_l)$ at the respective positions $(a_l,b_l)$ of $C_l,$ if $l\in[k+\ell]_\D$, and of $D_0,$ if $l=0.$ 
\item[(b)]
If the $l$th chord is a top chord, write  $\varepsilon_l$ at the $n$th position of $C_l$ (if $l\in[k+\ell]_\D$) or $D_l$ (if $l\in\{0,\star\}$). Otherwise, add to the $l$th row $\varepsilon_l\alpha_{\p(l)},\varepsilon_l\beta_{\p(l)}$ at positions $a_{\p(l)},b_{\p(l)}$  respectively. 
\item[(c)]
If $\D_l$ is black or purple or yellow,
write $(\gamma_l, \delta_l)$ at the respective positions $(c_l,d_l)$. 
\item[(d)]
If $\D_l$ is red, add to it $\delta_lD_0+\gamma_lD_\star.$
\item[(e)] If $D_l$ is blue, add to it $\delta_lD_\star+\gamma_l\pr_{(c_\star,d_\star)}D_\star.$ 
\end{itemize}
We will refer to $\alpha_i \ee_{a_i}+\beta_i\ee_{b_i},~i\in[k+\ell]$ and to $\gamma_i\ee_{c_i}+\delta_i\ee_{d_i},~i\in[k+l]\cup\{\star\}\setminus(\Red(\D)\cup\Blue(\D)),$ as the \emph{starting} and \emph{ending domino} of $\D_i,$ respectively.
By slight abuse of nations we will sometimes also refer to the pairs
$(\alpha_i,\beta_i),~i\in[k+\ell]$ and to $(\gamma_i,\delta_i),~i\in[k+l]\cup\{\star\}\setminus(\Red(\D)\cup\Blue(\D))$ as the starting and ending dominoes.
It will be convenient to refer to refer to $\ee_n$ as the \emph{starting domino of $\D_\phi$}. This will allow us to call the vector which multiplies $\varepsilon_i$ in the row corresponding to $\D_i$ the \emph{starting domino inherited from the parent of $\D_i$} or just the starting domino of the parent.
We write $u_\star = \pr_{c_\star d_\star}D_\star.$ 
As usual we use $\tr=\tr(\D)$ and $0$ interchangeably. 
\end{definition}
\begin{ex}\label{ex:domino_matrices}
\cite{even2021amplituhedron} provides several examples of $\ell=0$ domino matrices, e.g. Examples 2.6 and 2.31.
We now provide $\ell=1$ examples. There are three chord diagrams for $(\ell,n,k)=(1,5,1).$ Both have two cords, one $(1,2,3,4)$ and the other is a very short $(2,3,4).$
The first is when the two chords are red:
\[\begin{pmatrix}
0 &0 &\gamma_\star &\delta_\star &\varepsilon_\star \\
\alpha_0 &\beta_0 &0 &0&\varepsilon_0\\\hline
(\delta_1+\varepsilon_1)\alpha_0 &\alpha_1+(\delta_1+\varepsilon_1)\beta_0 &\beta_1+\gamma_1\gamma_\star &\gamma_1\delta_\star &\gamma_1\varepsilon_\star +\delta_1\varepsilon_0
\end{pmatrix}\]
The second case is that $(1,2,3,4)$ is red and $(2,3,4)$ is blue:
\[\begin{pmatrix}
0 &0 &\gamma_\star &\delta_\star &\varepsilon_\star \\
\alpha_0 &\beta_0 &0 &0&\varepsilon_0\\\hline\varepsilon_1\alpha_0 &\alpha_1+\varepsilon_1\beta_0 &\beta_1+(\gamma_1+\delta_1)\gamma_\star &(\gamma_1+\delta_1)\delta_\star &\delta_1\varepsilon_\star 
\end{pmatrix}\]
The last one is when $(1,2,3,4)$ is black and $(2,3,4)$ is the top red chord.
\[\begin{pmatrix}
\varepsilon_\star\alpha_2 &\varepsilon_\star\beta_2 &\gamma_\star &\delta_\star &0 \\
\varepsilon_0\alpha_2&\alpha_0+\varepsilon_0\beta_2 &\beta_0 &0&0\\\hline\alpha_2&\beta_2 &\gamma_2 &\delta_2&\varepsilon_2
\end{pmatrix}\]
When we add the sign constraints of \cref{def:L=1domino_signs} we obtain the domino parameterizations of the three BCFW cells for these values of $(\ell,n,k).$ Note that the three cells are contained in the same $1$-loop positroid cell, which consists of the two big cells $(\Gr^>_{1,5},\Gr^>_{3,5}).$
\end{ex}
\begin{definition}\label{def:L=1domino_signs}
The \emph{sign rules} of the domino matrix $(C\vdots D)_\D$ 
\begin{enumerate}
\item 
For every chord $l\neq 0,\star$: \; $\alpha_l,\beta_l > 0,$ and,
\[(-1)^{1+\below(\D_0)}\alpha_0,~(-1)^{1+\below(\D_0)}\beta_0>0 .\]
\item 
$\gamma_\star,\delta_\star>0,$ and for $l\neq 0,\star,$
\[(-1)^{\below(\D_l)}\gamma_l,~(-1)^{\below(\D_l)}\delta_l>0.\]
\item $\varepsilon_l$ satisfies
\[\begin{cases}(-1)^{\after(\D_l)}\varepsilon_l>0,&\D_l~\text{is a top chord, including }l=0,\star\\(-1)^{1+\Between(\D_l)+\below(\D_0)}\varepsilon_l>0,&\p(l)=0\\(-1)^{\Between(\D_l)}\varepsilon_l>0,&\text{otherwise}\end{cases}\]
\item\label{it:for_eta}
Suppose $\D_l$ is a strict same end child of $\D_m.~l\in[k+\ell]_\D\cup\{\star\},~m\in[k+\ell]_\D,$ or $m=\star$ and $\D_l$ is the top purple chord. We recall that this condition means that either $\D_l$ is the same-end child of $\D_m,$ and either both are of the same color or $\D_m$ is black, or that $\D_m=\D_\star$ and $\D_l$ is the top purple chord. Then \; $\delta_l/\gamma_l < \delta_m/\gamma_m.$
\item\label{it:for_theta}
If $\D_l,\D_m$ are siblings such that $\D_l$ is head-to-tail after $\D_m,$ for $l\in[k+\ell],~m\in[k+\ell]_\D\cup\{\star\}$, where $\D_m$ is either black, yellow or purple, then: \; $\beta_l/\alpha_l > \delta_m/\gamma_m,$ unless $l=0$ in the opposite inequality holds.
\end{enumerate}
For later purposes we add a few more notations.
If $\D_m$ is a strict same-end ancestor of $\D_l$ we write 
\begin{equation}\label{eq:eta}
\eta_{ml}=\det\begin{pmatrix}
   \gamma_m&\delta_m\\
   \gamma_l&\delta_l,
\end{pmatrix}
\end{equation}where in the special case of \cref{it:for_eta} we write $\eta_l:=\eta_{ml}.$
Similarly, when $\D_l$ starts where $\D_m$ ends, and $m\notin\Blue_\D\cup\Red_\D$ we write
\begin{equation}\label{eq:theta}
\theta_{ml}=\det\begin{pmatrix}
   \alpha_l&\beta_l\\
   \gamma_m&\delta_m
\end{pmatrix}
\end{equation}
and in the special case of \cref{it:for_theta} we write $\theta_m:=\theta_{ml}.$

We write $\tVar=\tVar_\D$ for the set of domino entries of $\D,$ and every $\eta_{ml},\theta_{ml}$ as above. Its subset $\Var=\Var_\D$ is obtained by taking every $\eta_i,\theta_i$ and every domino entry, except $\gamma_i$ whenever $\eta_i$ or $\theta_i$ is defined.
\end{definition}
\begin{obs}\label{obs:gauge_domino}
The domino representation is also subject to gauge.
The group generated by the following scalings acts on domino forms for $U\vdots V\in S_\D.$
\begin{itemize}
\item[--]Scaling the row $D_\star$ by $\lambda,$ every $\gamma_i,~\D_i\in\Red_\D\cup\Blue_\D$ and $\delta_i,~\D_i\in\Blue_\D$ by $\lambda^{-1}.$
\item[--]Scaling the row $D_0$ by $\lambda,$ every $\delta_i,~\D_i\in\Red_\D$ and $\varepsilon_i,~\p(i)=0,$ by $\lambda^{-1}.$
\item[--]Scaling the row $C_h$ by $\lambda,$ and every and $\varepsilon_i,~\p(i)=h,$ by $\lambda^{-1}.$
\end{itemize}
\end{obs}
The domino sign rule immediately imply
\begin{obs}\label{obs:iterated_etas} 
\[(-1)^{1+\below(\D_l)+\below(\D_m)}\eta_{ml}>0,\]if $l,m\neq\star.$ If $l$ or $m$ equals $\star$ we subtract $\below(\D_\star)$ from the power of $-1.$
Similarly,
\[(-1)^{1+\below(\D_m)}\theta_{ml}>0,\]if $l\neq 0,m\neq \star.$ If $l=0$ we add $1+\below(\D_0)$ to the power of $-1,$ while if $m=\star$ we subtract $\below(\D_\star)$ from the power. 
In particular these signs hold for $\eta_l,\theta_m.$
\end{obs}
The following simple observation describes domino matrices purely in terms of equalities and inequalities instead of parameterizations, and will be useful when we study limit of domino matrices.
\begin{obs}\label{obs:poly_ineqs}
Fix $\D\in\CD_{n,k}^\ell.$ The collection of domino matrices is defined as $(k+2\ell)$ rows indexed $[k+1]\cup\{\star\},$ such that certain $\D-$dependent entries and $2\times 2$ minors vanish, and other $\D-$dependent entries and $2\times 2$ minors have prescribed signs. 
\end{obs}
\subsubsection{From BCFW representation to domino - and back}
\begin{prop}\label{prop:BCFWndDominoEquiv}
Let $S$ be a $\ell-$loop BCFW cell, and $S_\D$ be the set of points in $\Gr_{k,n;\ell}^{\geq}$ which have a $\D-$domino representative, where $\D$ is the chord diagram which corresponds to $S$ by the bijection of \cref{obs:1loopCDvs1loopBCFW}. Then $S=S_\D.$
\end{prop}
The proposition follows immediately from the next lemma.
\begin{lemma}\label{lem:L=1_BCFW2Domino}
The domino variables for $\D\in\CD_{n,k}^\ell$ can be written as
\[\alpha_i=\alpha_i^\D(\BCFWV_\D),\ldots,\varepsilon_i=\varepsilon_i^\D(\BCFWV_\D),\]
where $\alpha_i^\D,\ldots,\varepsilon_i^\D$ are Laurent polynomials in $\BCFWV$ defined using the following recursive process:
\begin{enumerate}
\item If $\D=\pre_p(\D'),$ where $p$ is the penultimate marker of $N$ then 
\[\alpha_i^\D=\alpha_i^{\D'},\ldots,\varepsilon_i^\D=\varepsilon_i^{\D'}.\]
\item If $\D=\D_L\bcfw\D_R.$ Write $k_L,k_R$ for the number of chords of $\D_L,\D_R$ respectively.
In this case, \begin{enumerate}
\item\label{it:BCFW2DOM_BCFW_chord} $\alpha_k^\D=\halp_k,~\beta_k^\D=\hbet_k,~\gamma_k^\D=(-1)^{k_R}\hgam_k,~\delta_k^\D=(-1)^{k_R}\hdel_k,~\varepsilon_k^\D=(-1)^{k_R}\heps_k.$
\item\label{it:BCFW2DOM_BCFW_left}  If $\D_i$ for belongs to the left subdiagram $\D_L$ of $\D$ (we do not exclude the case $i=0,\star$) then
\[\alpha_i^\D=\alpha_i^{\D_L},~\beta_i^\D=\beta_i^{\D_L},~\delta_i^\D=\delta_i^{\D_L},~\varepsilon_i^\D=(-1)^{k_R+1}\varepsilon_i^{\D_L}.\]
If the $i$th chord is not head-to-tail with the $k$th chord then also $\gamma_i^\D=\gamma_i^{\D_L}.$ Otherwise
\[\gamma_i^\D = \gamma_i^{\D_L}+\frac{\halp_k}{\hbet_k}\delta_i^{\D_L}.\]
In case $i=0$ the same rules apply, only that there are no $\gamma_i,\delta_i.$ In case $i=\star$ the same rules apply only that there are no $\alpha_i,\beta_i.$
\item\label{it:BCFW2DOM_BCFW_right} If $\D_i$ belongs to the right subdiagram $\D_R$ of $\D$ then
\[\alpha_i^\D=\alpha_i^{\D_R},~\beta_i^\D=\beta_i^{\D_R},~\delta_i^\D=\delta_i^{\D_R}.\]
If $\D_i$ is not a child of $\D_k$ then also $\varepsilon_i^\D=\varepsilon_i^{\D_R}.$ Otherwise 
\[\varepsilon_i^\D=-\varepsilon_i^{\D_R}/\varepsilon^\D_k=(-1)^{k_R}\varepsilon_i^{\D_R}/\heps_k.\]
If $\D_i$ is a same-end descendant of $\D_k$ which is not red or blue \emph{at the time of this BCFW step}
\[\gamma_i^\D=\gamma_i^{\D_R}+\frac{\hgam_k}{\hdel_k}\delta_i^{\D_R}.\] Otherwise $\gamma_i^\D=\gamma_i^{\D_R}.$
\end{enumerate}
\item If $\D=\FL(\D'),$ then $\D_\star,\D_0$ are top chords. In this case, even though $\D_0$ was labelled $\tr(\D)$ in $\D',$ we refer to it as $\D'_0$ to reduce confusion.
\begin{enumerate}
\item\label{it:BCFW2DOM_FL_yellow} \[\gamma_\star^\D=\hgam_\star,\qquad\delta_\star^\D=\hdel_\star,\qquad\varepsilon_\star^\D=\heps_\star.\]
\item\label{it:BCFW2DOM_FL_topred}
\[\alpha_0^\D=(-1)^{k'_R+1}\alpha_0^{\D'}=(-1)^{k'_R+1}\halp_0^{\D'},\qquad\beta_0^\D=(-1)^{k'_R+1}\beta_0^{\D'}=(-1)^{k'_R+1}\hbet_0^{\D'},\qquad\varepsilon_0^\D=(-1)^{k'_R}\varepsilon_0^{\D'}=\heps_0^{\D'},\]where $k'_R$ is the number of chords in the right subdiagram of $\D'$
\item\label{it:BCFW2DOM_FL_general}For every $i\in [k+1]_\D$ \[\alpha_i^\D=\alpha_i^{\D'},~\beta_i^\D=\beta_i^{\D'}.\]
\[\varepsilon_i^\D=\begin{cases}
    -\varepsilon_i^{\D'},&\D'_i~\text{is a top chord of }\D'\\
    (-1)^{1+k'_R}\varepsilon_i^{\D'},&\p(i)=0\\
    \varepsilon_i^{\D'},&\D'_i~\text{otherwise.} 
\end{cases}\]
\item For a black chord $\D_i,$ $\gamma_i^\D=\gamma_i^{\D'},~\delta_i^\D=\delta_i^{\D'}$
\item\label{it:BCFW2DOM_FL_purple} For a purple chord $\D_i$  
\[\gamma_i^\D = \gamma_i^{\D'}+\frac{\gamma_\star}{\delta_\star}\delta_i^{\D'},~~\delta_i^\D=\delta_i^{\D'}.\]
\item\label{it:BCFW2DOM_FL_blue} For a blue chord $\D_i$ 
\[\gamma_i^\D=\frac{\gamma_i^{\D'}}{\hdel_\star},~~\delta_i^\D=\frac{\delta_i^{\D'}}{\halp_\star}.\]
\item\label{it:BCFW2DOM_FL_red} For a red chord $\D_i,~i\neq 0$
\[\gamma_i^\D = \frac{\gamma_i^{D''}}{\halp_\star},~~\delta_i^\D = \frac{\delta_i^{D'}}{\hdel_0} 
,\]
where ${D''}$ is the right subdiagram of $\D',$ obtained by restricting to $\D_0$'s descendants.
\end{enumerate}
\end{enumerate}
The positivity of the BCFW coordinates implies the sign rules for the domino coordinates.

The reduced BCFW coordinates are uniquely determined, up to the gauge freedom, from the domino variables, where the sign rules for the domino variables imply the positivity of the BCFW variables.
\end{lemma}
\begin{proof}
The proof for the first item is immediate. The second item also follows easily from the description of the domino form and the recursive description of the BCFW cell: we apply the BCFW step to the domino forms, and correct what is needed in order to obtain a domino form again.
If we work in the notations of \cref{nn:bcfwmap}, the only domino variables which are affected by a BCFW step are $\varepsilon_i,$ for a top chord of the left or right subdiagram; $\gamma_i$ for a chord $\D_i$ which either ends at $(a,b)$ or at $(c,d);$ and $\delta_i$ if $\D_i$ ends at $(c,d).$
$\varepsilon_i$ for a top chord of the left subdiagram gains only the prescribed sign.
If $\D_i$ ends at $(a,b)$ it is easy to see that $\gamma_i$ is affected as in \cref{it:BCFW2DOM_BCFW_left}. If $\D_i$ ends at $(c,d)$ but is not a top chord of the right subdiagram then $\gamma_i$ is affected as in \cref{it:BCFW2DOM_BCFW_right}. If $\D_i$ a top chord of the right subdiagram then its $(c,d,n)$ entries transform under the BCFW step to
\[(x, y,\varepsilon_i^{\D_R})\to
(x+\frac{\hgam_k}{\hdel_k}y+\frac{\hgam_k}{\heps_k}\varepsilon_i^{\D_R}, y+\frac{\hdel_k}{\heps_k}\varepsilon_i^{\D_R},\varepsilon_i^{\D_R}).
\]
In order to pass to the domino form we must set the $n$th entry of this row to $0$ by subtracting a multiple of the $k$th row. This multiple is $\frac{\varepsilon_i^{\D_R}}{(-1)^{k_R}\heps_k},$ and after subtracting it the $(c,d,n)$ component becomes
\[(x+\frac{\hgam_k}{\hdel_k}y,y,0).\]
This agrees with our formula for $\gamma_i,\delta_i,\varepsilon_i$ of \cref{it:BCFW2DOM_BCFW_right}, where we have exclude the red and blue case in \cref{it:BCFW2DOM_BCFW_right} since in the definition of the domino form in this case $\gamma_i,\delta_i$ are coefficients of vectors that can be determined from the yellow and top red row, and we consider the column operations as performed on these rows rather than on $\gamma_i,\delta_i$.

Moving to the third item, the forward limit. Note that the last operation before the forward limit has to be a BCFW step, with respect to the chord $k+1=\tr(\D')$, which for convenience we index $0.$ Its BCFW parameters are $(\halp_0,\ldots,\heps_0)\in\R_+^5.$ The $\FL$ parameters are $(\hgam_\star,\ldots,\heps_\star)\in\R_+^5.$ We use \cref{lem:scale_inv} to set $\hgam_0=\hbet_\star=0.$
In order to perform the forward limit we make the $\star,0$ rows being topmost, and fix signs so that both the resulting matrix, and the submatrix made of the bottom $k$ rows, will be nonnegative. The change of sign of the row labeled $0$ induces a change of signs in the definition of $\varepsilon_i$ whenever $\p(i)=0.$ We then use the rows $\star,0$ to cancel the $A,B$ entries of every row in which at least one of them is non zero.


Now, a chord $\D'_i$ which ends at $(d,A)$ can be corrected only using the row $\star.$ We correct such a chord by erasing $\delta_i^{\D'}/\halp_\star$ times the row $\star.$ Such chords will become blue in $\D.$

The red chords in $\D$ are those whose counterpart in $\D'$ ends at $(A,B).$
Consider the $i$th such row. Its $(A,B)$ entries are $(\gamma_i^{\D'},\delta_i^{\D'}).$ Note that since $\gamma_0^{\D'}=0,$ $\gamma_i^{\D'}=\gamma_i^{\D''}.$
We wish so subtract a linear combination of the $0$th and $\star$-th rows to cancel the $(A,B)$ part.
The $(A,B)$ part of the $0$th row is $(\gamma_0^{\D'},\delta_0^{\D'}).$ 
The $(A,B)$ part of row $\star$ is $(\halp_\star,0).$ 
Note that, from \cref{it:BCFW2DOM_BCFW_right},
\[(\gamma_i^{\D'},\delta_i^{\D'})=(\gamma_i^{\D''}+\frac{\gamma_0^{\D'}}{\delta_0^{\D'}}\delta_i^{\D''},\delta_i^{\D''}).\]
From here it is easy to see that 
\[(\gamma_i^{\D'},\delta_i^{\D'})=\frac{\gamma_i^{\D''}}{\halp_\star}(\halp_\star,0)+\frac{\delta_i^{D'}}{\hdel_0}(\gamma_0^{\D'},\delta_0^{\D'}).\]


We now turn to discuss signs. The tree case was proven in \cite[Corollary 7.1]{even2021amplituhedron}. We will follow the same logic of studying how do signs evolve under the three different types of steps that appear in the BCFW cell construction.

We will prove by induction on the lexicographically ordered $(k,n)$ that the signs rules hold. 
We split into cases.
\begin{enumerate}
\item$S_\D=\pre_p S_{\D'}:$ The sign rules for $S_\D$ and $S_{\D'}$ are the same, and $\pre_p$ does not affect the signs.

\item$\D=\D_L\bcfw\D_R:$ By \cref{def:bcfw-map},
the sign constraints on the domino entries of $\D_k,$ the rightmost top chord, hold by definition, as they agree with the BCFW coordinates up to exactly the required signs.

The signs of entries in the left component do not change, unless they are the $n$th entries of rows, and these change by $(-1)^{k_R+1}.$ Since also $\after(\D_i)=\after((\D_L)_i)+k_R+1$ the sign rule for these entries persists to hold.

If $\D_k$ is head-to-tail after $\D_i,$ then from the above we know that
\[\gamma_i^{\D}=\gamma_i^{\D_L}+\frac{\alpha_k}{\beta_k}\delta_i^{\D_L}.\]From the induction $\gamma_i^{\D_L},\delta_i^{\D_L}$ have the same sign.
Thus, 
\[\gamma_i^\D/\delta_i^\D=\gamma_i^{\D_L}/\delta_i^{\D_L}+\alpha_k/\beta_k>\alpha_k/\beta_k.\]

The signs of entries in the right side are not affected, except than those of $\varepsilon_i$ for children of $\D_k,$ whose signs change, as explained in the above analysis of canceling their $n$th entry, by $(-1)\sgn(\varepsilon_i^\D)=(-1)^{k_R}.$ It follows from the definition of $\Between,$ that the resulting sign is just $\Between(\D_i).$ For same end children $\D_i$ of $\D_k$ which are either yellow, top red, black or purple, we have
\[\gamma_i^\D=\gamma_i^{\D_R}+\frac{\gamma_k}{\delta_k}\delta_i^{\D_L}.\]Since $\gamma_i^{\D_R},\delta_i^{\D_R}$ have the same sign, from induction, we have
\[\gamma_i^\D/\delta_i^\D=\gamma_i^{\D_R}/\delta_i^{\D_R}+\gamma_k/\delta_k>\gamma_k/\delta_k.\]
\item$\D=\FL(\D'):$
The discussion of the BCFW form \cref{subsec:BCFW_form} explains the signs of the top red and yellow rows \cref{it:BCFW2DOM_FL_topred},~\cref{it:BCFW2DOM_FL_yellow}.
The signs of $\alpha_0,\beta_0,\varepsilon_0,\gamma_\star,\delta_\star,\varepsilon_\star$ are automatic from the constructions and \cref{it:BCFW2DOM_FL_topred},~\cref{it:BCFW2DOM_FL_yellow} above. The treatment in $\gamma_l/\delta_l$ versus $\gamma_l/\delta_l$ for $\D_l,\D_m$ being purple, yellow, or both red but neither top red, follows from \cref{it:BCFW2DOM_FL_purple},~\cref{it:BCFW2DOM_FL_red},~\cref{it:BCFW2DOM_FL_yellow} and is similar to the treatment of same end chords above. For blue chords the same analysis shows a reverse inequality, because of \cref{it:BCFW2DOM_FL_blue}.
The remaining signs are automatic.
\end{enumerate}
For the last statement, note that the above process of writing the domino variables in terms of the BCFW coordinates is easily reversed, exactly up to the gauge freedom in the BCFW representation, by following the same recursive steps, and the domino sign constraints are precisely what allows for this reversal of operations to preserve the positivity, and to finally impose the positivity of the BCFW coordinates.
\end{proof}
The recursive description of \cref{lem:L=1_BCFW2Domino} is easily made non recursive, and in terms of the reduced BCFW variables:
\begin{cor}\label{cor:BCFW2dominoExplicit}
Let $\D\in\CD_{n,k}^\ell$ be a chord diagram. 
Then, the domino variables of $\D_i,$ up to scaling of the row, are given by
\[\alpha_i=\begin{cases}\balp_i,&i\in[k+1]_\D\\
(-1)^{1+\below(\D_0)}\balp_0,&i=0\end{cases},\qquad\qquad\beta_i=\begin{cases}\bbet_i,&i\in[k+1]_\D\\
(-1)^{1+\below(\D_0)}\bbet_0,&i=0\end{cases}\]
\[\delta_i=\begin{cases}(-1)^{\below(\D_i)}\bdel_i,&i\neq 0,\star,
\\
\bdel_\star,&i=\star
\end{cases}.\]
\[\varepsilon_i=\begin{cases}(-1)^{\after(\D_i)}\beps_i>0,&\D_i~\text{is a top chord, including }i=0,\star\\
(-1)^{\Between(\D_i)+\below(\D_0)+1}\frac{\beps_i}{\beps_{\p(i)}},&\p(i)=0\\
(-1)^{\Between(\D_i)}\frac{\beps_i}{\beps_{\p(i)}},&\text{otherwise }.\end{cases}\]
The rule for $\gamma_i$ is as follows. For $i\in\Black(\D)\cup\Purp(\D)$
\[\gamma_i=(-1)^{\below(\D_i)}\left\{\bgam_i+\bdel_i[\frac{\balp_j}{\bbet_j}]+\bdel_i\sum_{\substack{\D_l~\text{is a strict}\\\text{same end ancestor of }\D_i}}\frac{\bgam_l}{\bdel_l}\right\},
\]
where the $[\cdot]$ term appears only $\D_i$ has a sibling $\D_j$ which is head to tail after it. Note that if $\D_i$ is purple then the summation is only over purple, black and yellow chords. The rule for blue chords is the same, only that we need to further divide the right hand side by $\bdel_\star,$ and in this case the summation is only over blue chords, and there is never a sibling term. For the yellow chord the same formula holds, only without the sign in the beginning. For a red chord which is not the top red chord
\[\gamma_i=(-1)^{\below(\D_i)}
[\bgam_i+\bdel_i(\sum_{\substack{l\neq 0~\text{and }\D_l~\text{is a strict}\\\text{same end ancestor of }\D_i}}\frac{\bgam_l}{\bdel_l})].\] 
Writing the reduced BCFW coordinates in terms of the domino variables, up to the gauge again, is immediate from the above, for all variables except $\heps_i,\hgam_i.$
For $\heps_i$ we have
\[\beps_i=
    |\varepsilon_i\prod_{\D_j~\text{is an ancestor of $\D_i$}}\varepsilon_j|.\]
By direct calculations we obtain
that in cases $\eta_i$ is defined it holds that
\[\bgam_i=\begin{cases}|\frac{\eta_i}{\delta_j}|,&i\notin\Blue_\D\\
|\frac{\delta_\star\eta_i}{\delta_j}|&i\in\Blue_\D\end{cases}.\]
Similarly, when $\theta_i$ is defined we have\[
\bgam_i=|\frac{\theta_i}{\beta_j}|.
\]
Otherwise $\bgam_i=\gamma_i.$
\end{cor}
\begin{prop}\label{prop:pos_poly_plucker_rep}
Let $\D\in\CD_{n,k}^\ell,~\ell\in\{0,1\}$ be a chord diagram. Every maximal minor of a domino pair $C\vdots D$ can be written as a polynomial with non negative coefficients in the absolute values of the elements of $\tVar.$
\end{prop}
The proof is almost identical to the proof of the analogous $\ell=0$ claim \cite[Proposition 7.3]{even2021amplituhedron}.
\begin{rmk}\label{rmk:inclusion_of_Var}
If $\D'$ is a subdiagram of $\D$ then $\BBCFWV^{\D'}\setminus\{\hgam_{\tr(\D)},\hdel_{\tr(\D)}\}\hookrightarrow\BCFWV^\D,$ naturally via $\bzet_i^{\D'}\to\bzet_i^\D.$ There is a similar natural injection of domino variables.
\end{rmk}

\subsection{Uniqueness of the domino and BCFW forms}\label{sec:domino_param}
In this section we prove that, up to the global scaling freedom, every point in a $\ell=1$ BCFW cell has a unique domino representative, and that also the BCFW coordinates are determined up to scaling.
\begin{definition}\label{def:weak_sub_mfld}
Let $X$ be a manifold. A subset $Y$ is a \emph{weak (smooth) submanifold (with boundary)} if $Y$ is the injective image of a continuous (smooth) map from a (smooth) manifold to $X.$
$Y$ is a \emph{topological submanifold (with boundary)} of $X$ if $Y$ is the homeomorphic image of a manifold (with boundary) to $X.$
\end{definition}
\begin{thm}\label{thm:dominoAndBCFWparams}
Let $\D\in\CD^\ell_{N,K}$ for $\ell\in\{0,1\},$ be a chord diagram. Then every element in $S_\D$ has unique domino and BCFW forms, up to gauge. 
In particular, $S_\D$ is a $4(k+\ell)$ dimensional manifold. $S_\D$ is moreover a submanifold of $\Gr_{K,N;\ell}$ both in the weak smooth sense, and in the topological sense. 
\end{thm}
One can actually show that $S_\D$ is a smooth submanifold, but we do not need this result, and its proof requires additional calculations.

We will rely on the following lemma
\begin{lemma}\label{lem:4biden}
\begin{enumerate}
\item\label{it:4bidenC}
Let $\D$ be a chord diagram whose rightmost top chord $\D_k=(a,b,c,d)$ is not red. Let $(C\vdots D)$ be a BCFW representative of an element in $(U\vdots V)\in S_\D.$ Let $(U'_L\vdots V'_L),~(U'_R\vdots V'_R)$ be the (maybe trivial) linear spans of the rows corresponding to the left and right subdiagrams respectively.
Then $U'_L=U\cap\mathbb{R}^{1,2,\ldots,a,b,n},~U'_R=U\cap\mathbb{R}^{b,b+1,\ldots,c,d,n}$ and neither $U'_L$ nor $U'_R$ contain a vector supported on positions $\{a,b,c,d,n\}.$ In particular $U\cap \R^{\{a,b,c,d,n\}}$ is one dimensional spanned by $C_k,$ and there is non zero vector in $V$ supported on a proper subset of $\{a,b,c,d,n\}.$ 

In addition, if $V'_L\neq 0$ then it has a representative made of two vectors contained in $\mathbb{R}^{1,2,\ldots,a,b,n},$ unique up to $\GL_2$ action and addition of elements from $U'_L.$ If $V'_R\neq 0$ then it has a representative made of two vectors contained in $\mathbb{R}^{b,b+1,\ldots,c,d,n},$ unique up to $\GL_2$ action and addition of elements from $U'_R.$
\item\label{it:4bidenD} Let $\D=\FL(\D'),$ where $\D'$ has a rightmost top chord supported on positions $(a,b,A,B),$ and no chord supported on $c,d,A,B,$ where $c,d,A,B$ are consecutive in the marker set of $\D'$.
Let $(C\vdots D)\in S_\D$ be a BCFW representative. Then $C+D$ contains a unique vector, up to scalar multiplication, supported on positions $a,b,n$ and a unique vector, up to scalar multiplication, supported on the positions $c\lessdot d\lessdot n.$
\end{enumerate}
\end{lemma}
\begin{proof}
For the first item, note first that the form of the BCFW map shows that $U$ contains a vector $v$ supported on the entries $\{a,b,c,d,n\},$ and not on a proper subset of them. The BCFW form also shows that $U=U'_L+U'_R+v.$
Let $N_L=\{\min N,\ldots,a,b,n\},~N_R=\{b,\ldots,c,d,n\}$ be their respective index sets. 
Observe that $U'_L\subseteq\R^{N_L},~U'_R\subset\R^{N_R}$ and that $N_R\cap N_L = \{b,n\}.$ 
Thus, if $v$ were not unique up to scaling, then there must have been a linear combination of vectors in $U'_L+U'_R$ with support contained in $\{a,b,c,d,n\}.$ But this would have implied the existence of vectors $v_L\in U'_L\cap\Span\{\ee_a,\ee_b,\ee_n\},~v_R\in U'_R\cap\Span\{\ee_b,\ee_c,\ee_d,\ee_n\},$ at least one of which is nonzero. Thus, the 'In particular' part indeed follows from the first part of the statement. 

To rule out the existence of $v_L,v_R$, note that the existence of a non zero $v_L$ would imply that $\{a,b,n\}$ is not coindependent for $U'_L$, and the existence of a non zero $v_R$ would imply that $\{b,c,d,n\}$ is not coindependent for $U'_R$. Now, $(U'_L\vdots V'_L)$ and $(U'_R\vdots V'_R)$ are obtained from the inputs $(U_L\vdots V_L)$ and $(U_R\vdots V_R)$ to the BCFW map, by applying certain $\pre$ and $y_i(t)$ operations for $t>0$, see \cref{def:bcfw-map}. Thus, by \cref{lem:effect_ops} the existence of a non zero $v_L$ or $v_R$ would imply that either $\{a,b,n\}$ is not coindependent for $U_L$, or $\{b,c,d,n\}$ is not coindependent for $U_R.$ But $U_L,U_R$ are the tree parts of elements in the BCFW cells $S_{\D_L},S_{\D_R}$ respectively. Thus, this is impossible by \cref{lem:bcfw_after_fl},~\cref{it:bcfw_after_fl_bcdn}. This also shows $U'_L=U\cap\mathbb{R}^{1,2,\ldots,a,b,n},~U'_R=U\cap\mathbb{R}^{b,b+1,\ldots,c,d,n}.$ Indeed, in both cases the left hand side is contained in the right hand side. If $U'_L\subsetneq U\cap\mathbb{R}^{1,2,\ldots,a,b,n}$ then there must have been a linear combination of $v$ and a vector $v_R\in U'_R$ supported on positions $a,b,n$, which is impossible by the above. The case of $U'_R$ is treated similarly.

For the 'In addition' part, the BCFW representation yields one such matrix representation. If $V'_L\neq 0$ then if the representation were not unique up to $\GL_2$ and addition an elements of $U'_L,$ it would imply the existence of two vectors $u,u'\in U+V ,$ supported on $\R^{1,\ldots,a,b,n}$ with $u-u'\in U\setminus U'_L.$ Thus $u-u' = v_L+v_R+\lambda v$ with $v'_L\in V_L,v'_R\in V_R,\lambda\in\R$ where at least one of $v_R,\lambda\neq 0.$ Since the support of $u-u'-v_L$ is contained in $\{1,\ldots,a,b,n\}$ and the support of $v_R+\lambda v$ is contained in $\{a,b,\ldots,c,d,n\}$ the non zero vector $u-u'-v_L=v_R+\lambda v$ must be supported on $a,b,n.$ This means that $v_R$ cannot be $0,$ and must be supported on $b,c,d,n$, but the existence of such a vector was ruled out above. The same argument works for the right component.

For the second item, the existence of the vectors follows from the definition of the BCFW form, and the uniqueness follows from \cref{lem:bcfw_after_fl},~\cref{it:bcfw_after_fl_reds},~\cref{it:bcfw_after_fl_yellow}. 
\end{proof}
\begin{proof}[Proof of Theorem \ref{thm:dominoAndBCFWparams}]
We will prove the uniqueness part of the theorem by induction on $|K|+|N|.$
We note first that for $\ell=0$ the uniqueness of the domino form is a consequence of \cite[Proposition 3.15,~Lemma 3.24]{even2021amplituhedron}, while the uniqueness of the BCFW form is \cite[Proposition 6.22]{even2023cluster}. Assume now $\ell=1.$ There are three possibilities either $\D=\pre_p \D',$ where $p$ is the penultimate marker of $N,$ or $\D=\D_L\bcfw\D_R$ or $\D=\FL(\D').$
\begin{enumerate}
\item$\D=\pre_p \D':$ From the induction the elements of $S_{\D'}$ have unique BCFW and domino forms. Since every $(U\vdots V)\in S_\D$ equals $\pre_p (U'\vdots V')$ for $(U'\vdots V')\in S_{\D'}$, it is easy to see that also $(U\vdots V)$ has unique BCFW and domino forms.

\item $\D=\D_L\bcfw\D_R:$
Let $\D_k=(a,b,c,d)$ be the rightmost top chord. Every element in $S_\D$ has a form 
\[(U\vdots V) = \mbcfw((U_L\vdots V_L),[\alpha_k:\beta_k:\gamma_k:\delta_k:\varepsilon_k],(U_R\vdots V_R)),\]where exactly one of $V_L,V_R$ is a trivial vector space, and 
$(U_L\vdots V_L)\in S_{\D_L},~(U_R\vdots V_R)\in S_{\D_R}$.
The proof of this case is relatively similar to the proof of \cite[Lemma 3.24]{even2021amplituhedron}.

By \cref{lem:4biden},~\cref{it:4bidenC} the line $v_k,$ corresponding to $\D_k$ can be uniquely recovered as $U\cap \mathbb{R}^{a,b,c,d,n},$ and $U'_L,U'_R$ which are the rows spans of the rows corresponding to the chords of the left and right components are also uniquely determined. Knowing $v_k$ we can reverse the operations used in \cref{def:bcfw-map} to retrieve $U_L,U_R$ from $U'_L,U'_R$ respectively. Moreover, if we take the representation of the non trivial $V'_L$ or $V'_R$ as given in the 'Moreover' part of \cref{lem:4biden},~\cref{it:4bidenC}, we may recover the non trivial $V_L$ or $V_R$ by again reversing the BCFW map operations. 
From the induction the BCFW parameters of $U_L\vdots V_L,~U_R\vdots V_R$ are uniquely determined, hence also those of $U\vdots V.$ Using \cref{prop:BCFWndDominoEquiv} the same can be said about the domino representation.
\item $\D=\FL(\D'):$
Let $(C\vdots D)$ be a domino representation of an element of $(U\vdots V)\in S_\D.$ We would like to show it is unique. We will do that by first identifying uniquely, up to scaling, the rows corresponding to $\D_0,\D_\star,$ then inverting the forward limit operation, which we will also be able to do uniquely up to the irrelevant degrees of freedom. Then the result will follow by induction the on $\D'.$

For convenience we assume $N=[n],$ and write $(a,b)$ for $(a_0,b_0)$ The rows $D_\star,D_0$ are contained respectively in $(U+V)\cap\R^{n-2,n-1,n},~(U+V)\cap\R^{a,b,n}.$ They are uniquely specified by this requirement, as \cref{lem:4biden},~\cref{it:4bidenD} shows. 
Denote the non zero entries of $D_\star$ in positions $n-2,n-1,n$ respectively by $\gamma_\star,\delta_\star,\varepsilon_\star.$ 
Denote the non zero entries of $D_0$ in positions $a,b,n$ respectively by $\alpha_0,\beta_0,\varepsilon_0.$ 

Let $(\hat{C}\vdots \hat{D})$ be obtained from $(C\vdots D)$ by adding two new columns, $A,B$ between $n-1$ and $n.$ In $\hat{C}$ we fill these rows by $0$s, while in $\hat{D}$ we embed the $2\times 2$ identity matrix.
Let ${V'}$ be the vector space spanned by the vectors $C'_i,~i\in[k+1]$ defined as follows:
\begin{itemize}
\item Choose an arbitrary positive $\gamma_{k+1}.$
\item $C'_0$ is $y_{n-2}(\gamma_{k+1})(\hat{D}_0).$
\item For $i\in [k],$ if $\D_i$ is black or purple,  $C'_i = y_{n-2}(-\frac{\alpha_\star}{\beta_\star})\hat{C}_i.$
\item For $i\in [k]$ with $\D_i$ blue, $C_i$ is obtained from $\hat{C}_i$ by first subtracting the multiple of $\hat{D}_\star$ required to make the $n$ entry $0,$ and then applying $y_{n-2}(-\frac{\alpha_\star}{\beta_\star}).$
\item For $i\in[k]$ with $\D_i$ red, $C'_i$ is obtained by first subtracting from $\hat{C}_i$ the multiple of $\hat{D}_\star$ required to make the $n-1$th coordinate $0,$ and then subtracting the multiple of $C'_0$ required to make the $n$th coordinate zero as well.
\end{itemize}
These operations precisely invert the forward limit operation, and they determine uniquely, up to the choice of $\gamma_{k+1}$ a space $V'\in S_{\D'}$ which gave rise to $(V\vdots U)$ under the forward limit operation. Different choices of $\gamma_{k+1}$ do not affect the domino form of $U\vdots V$ or the BCFW form.
\end{enumerate}

$V'$ has a unique domino and BCFW form, by induction, hence also $U\vdots V$ have unique domino and BCFW form.

Regarding the dimension, we may fix the scaling of each row of a domino pair by setting $\varepsilon_i=1$ for $i\in\{0,\star,1,\ldots,k\},$ and obtain $4k+4\ell$ parameters. These parameters map injectively, by the above analysis, and smoothly, by construction, to $\Gr_{k,k+4;\ell}$, hence parameterizing a manifold of dimension $4k+4\ell.$

The parameterization via BCFW or domino coordinates is clearly smooth, hence $S_\D$ is a smooth weak submanifold  of $\Gr_{K,N;\ell}.$ Moreover, the process of finding the BCFW or domino coordinates from the data of $U\vdots V\in S_\D$ is clearly reversible, and in a continuous manner, thus $S_\D$ is also a topological submanifold of $\Gr_{K,N;\ell}.$\end{proof}

\section{The boundaries of the $1-$loop BCFW cells}
It turns out that boundaries of BCFW cells can be obtained by weakening the sign constraints of \cref{def:L=1domino_signs}. This weakening may result in losing rank or losing the uniqueness of the domino representation.
This section studies the boundaries of BCFW cells. 
\cref{subsec:SA} defines an important subspace of the nonnegative loop Grassmannian that will turn out to be the preimage of the amplituhedron's boundary.
\cref{subsec:weak_domino} defines weak domino forms, which are a way to parameterize boundary strata. \cref{subsec:mathc_bdries} shows that certain codimension $1$ boundaries of BCFW cells are also boundaries of other BCFW cells. The way to detect which cells share a codimension $1$ boundary involves a combinatorial operation on chord diagrams, which we call a \emph{shift}. Two special boundaries of BCFW cells, the red and blue boundary, are studied in \cref{subsec:redblue_bdries}. Finally \cref{subse:mfld_str} studies the manifold structure on boundaries.
These detailed analyses will be used later to show that the amplituhedron is tiled by the BCFW cells.
\subsection{The space $\SA$}\label{subsec:SA}
\begin{definition}\label{def:strata}
Let $\D$ be a chord diagram. A maximal subspace of $\overline{S}_\D,$ the Hausdorff topological closure of $S_\D$, in which all points have the same $1-$loop positroid is called a \emph{positroid stratum}.
\end{definition}
\begin{definition}
\label{def:SA}
For nonnegative integers $k,n$ with $n\geq k+4,$ write
\begin{align*}
\SA_C 
&= \left\{(C\vdots D) \in \Gr_{k,n;\ell}^{\geq} \;:\; \text{There exist } i \text{ and } j\text{ such that }\langle C^J\rangle =0~\text{for every }J\subseteq[n]\setminus \{i,i',j,j'\}\right\}\\
&= \left\{C\vdots D \in \Gr_{k,n;\ell}^{\geq} \;:\; C \cap \Span\left\{\ee_i,\ee_{i'},\ee_j,\ee_{j'}\right\} \neq \{0\} \text{ for some } i \text{ and } j \right\},
\end{align*}
where $i'=i+1\mod n,~j'=j+1\mod n,$ and we consider only $i,j$ such that $\{i,i',j,j'\}$ are four different elements. If $\ell=0$ then as usual by $C\vdots D$ we just mean $C.$
Similarly, write
\begin{align*}
\SA_D 
&= \left\{(C\vdots D) \in \Gr_{k,n;1}^{\geq} \;:\; \text{There exists } i\text{ such that }\langle (C+D)^J\rangle =0~\text{for every }J\subseteq[n]\setminus \{i,i'\}\right\}\\
&= \left\{(C\vdots D) \in \Gr_{k,n;1}^{\geq} \;:\; (C+D) \cap \Span\left\{\ee_i,\ee_{i'}\right\} \neq \{0\} \text{ for some } i \right\}, 
\end{align*}and let $\SA=\SA_C\cup\SA_D.$
\end{definition}
The following observation is an immediate consequence of the definition. 
\begin{obs}\label{obs:SA_strata}
For every $\ell-$loop chord diagram $\D,$ $\overline{S}_\D\cap\SA$ is a union of positroid strata of  $\overline{S}_\D.$
\end{obs}We will show later that $\SA$ is the preimage of the boundary of the amplituhedron $\Ampl_{n,k,4}^\ell(Z)$.
\subsection{Weak domino forms}\label{subsec:weak_domino}
Recall \cref{obs:poly_ineqs}.
\begin{definition}\label{def:almost_and_weak_domino}
Let $\D$ be a chord diagram. 
An \emph{almost $\D-$domino pair} $C\vdots D$ is a pair of matrices $C,D$ without zero rows, whose non zero entries are contained in the non-zero entries dictated by the $\D-$domino forms, and every $2\times 2$ minor which vanishes in a domino form also vanishes for $C\vdots D.$ Note that we do not require that the remaining entries or $2\times 2$ minors have the prescribed signs of a $\D-$domino form. 
An almost $\D-$domino pair $C\vdots D$ is said to be an \emph{almost domino pair} for $U\vdots V$ if the row spans of $C,C+D$ are contained in $U,V$ respectively. 
We index the rows of $C,D$ by the same indexing conventions we use for $\D-$domino forms, e.g $C_i,D_0,D_\star.$

If $C\vdots D$ is an almost domino pair,  
then, up to the gauge group action of \cref{obs:gauge_domino}, the following expressions are well defined functions of $C,D$ just like in the usual $\D-$domino case: 
\begin{itemize}
\item $\alpha_h(C\vdots D),\beta_h(C\vdots D),\delta_h(C\vdots D).$ 
\item \begin{itemize}
\item[--]$\gamma_h(C\vdots D),$ if $\D_h\notin\Blue_\D$; \item[--]$\gamma_h(C\vdots D)\pr_{c_\star d_\star}D_\star,$ if $\D_h$ is blue. 
\item[--]Moreover, if $\pr_{c_\star d_\star}D_\star=(\gamma_\star,\delta_\star)\neq (0,0)$ then also $\gamma_h$ is defined. 
\item[--]If $(\gamma_\star,\delta_\star)=(0,0)$ but $\gamma_h(C\vdots D)\pr_{c_\star d_\star}D_\star\neq (0,0),$ we set $\gamma_h=\infty.$ Otherwise $\gamma_h$ is ill defined.
\end{itemize}
\item \begin{itemize}
    \item[--]$\varepsilon_h(C\vdots D)$ if $\D_h$ is a top chord; 
\item[--]$\varepsilon_h\pr_{a_{\p(h)}b_{\p(h)}}C_{\p(h)}$ otherwise, where if $\p(h)=0$ then by $C_{\p(h)}$ we mean $D_0.$ \item[--]Moreover, if $\pr_{a_{\p(h)}b_{\p(h)}}C_{\p(h)}=(\alpha_{\p(h)},\beta_{\p(h)})\neq (0,0)$ then also $\varepsilon_h$ is defined. 
\item[--]If $(\alpha_{\p(h)},\beta_{\p(h)})=(0,0)$ but $\varepsilon_h\pr_{a_{\p(h)}b_{\p(h)}}C_{\p(h)}\neq(0,0)$ then $\varepsilon_h=\infty.$ Otherwise $\varepsilon_h$ is ill defined.
\end{itemize}
\item $\eta_h$ unless $\gamma_h$ or $\gamma_{\p(h)}$ is infinite or ill defined, and $\theta_h$ unless $\gamma_h$ is infinite or ill defined.
\end{itemize}
If for $C\vdots D$ all domino variables are finite or ill defined, and the inequalities implied by \cref{obs:poly_ineqs} are satisfied in the weak sense (replacing a definite sign by that sign or $0$),
we say that $C\vdots D$ has a \emph{weak domino form}.
If in addition both $C,C+D$ are of full rank we say that $C\vdots D$ is a \emph{weak $\D-$domino representative} of $U\vdots V$.

We can similarly define for a subset of rows of $C\vdots D$ the notions of almost and weak domino forms, by requiring the above properties only for that subset of rows. In particular,
recall \cref{def:behind_above_etc}. For a chord $\D_h$, we say that $U\vdots V$ 
has a \emph{weak (almost) domino form weakly outside $\D_h$} if there is a set of vectors $v_f\in U,~\D_f\in\woutside(\D_h)$ and possibly $v_0,v_\star\in U+V$ if $0\in\woutside(h)$ such that the collection of $\{v_f\},$ satisfy the above conditions for defining a weak (almost) domino form, when they are restricted only to the set $\woutside(\D_h)$ of rows. If $h\in\{0,\star\}$ . $U\vdots V$ has a weak (almost) domino form \emph{strictly outside} $\D_h$ if the same holds only for the rows in $\outside(\D_h).$ 
If there is a unique up to scalings collection of vectors realizing the weak (almost) domino form (weakly/strictly) outside $\D_h$, we say that $U\vdots V$ has a \emph{unique} weak (almost) $\D-$domino form (weakly/strictly) outside $\D_h.$ 

Recall the set $\tVar_\D$ defined in \cref{def:domino_entries}. For $A\subseteq \tVar_D,$ denote by $\partial_AS_\D\subseteq\overline{S}_\D,$ where $\overline{S}_\D$ is the Hausdorff topological closure of $S_\D,$ the subspace of $\overline{S}_\D$ consists of points which have a weak $\D-$domino representative in which the elements of $A$ vanish, and the other elements of $\tVar_\D\setminus A$ do not. 
\end{definition}
\begin{rmk}\label{rmk:inf_gamma_eps}
Note that $\gamma_h=\infty$ or ill defined only if $\pr_{c_\star d_\star}D_\star=0.$ Similarly $\varepsilon_h=\infty$ or ill defined only if $\pr_{a_{\p(h)}b_{\p(h)}}D_{\p(h)}=0.$ 
\end{rmk}
\begin{definition}\label{def:domino_limit}
Let $(U\vdots V)\in \overline{S}_{\D}$ an element in the Hausdorff topological closure of $S_\D.$

A \emph{domino limit pair} for $U\vdots V$ is any pair $C\vdots D,$ where neither $C$ nor $D$ has a zero row, that can be written as the entry-wise limit
\[C\vdots D=\lim_{i\to\infty}C^i\vdots D^i,\] where $C^i\vdots D^i$ is a $\D-$domino representative for $U^i\vdots V^i,$ where $(U^i\vdots V^i)\in S_\D,~i=1,2,\ldots$ is a sequence converging to $U\vdots V.$
\end{definition}
\begin{obs}\label{obs:domino_limits equal almost domino}
Every $\D-$domino limit is an almost $\D-$domino pair. 
We can thus define the domino expressions of \cref{def:almost_and_weak_domino} also for a domino limit pair. 
\end{obs}
\begin{obs}\label{obs:domino_limits}
Every $(U\vdots V)\in \overline{S}_{\D}$ has a domino limit pair $C\vdots D$. Every row of $C,D$ corresponds to a chord, and is a vector in $U,~U+V$ respectively. If $C$ is of full rank then $V$ is the row span of $C,$ if $C+D$ is of full rank then its row span is $U+V.$ 
Every entry or a $2\times 2$ minor of the limit, which is identically zero for a $\D-$domino pair is identically zero for the limit. Every entry or a $2\times 2$ minor of the limit, which has a definite sign for a $\D$-domino pair, has that sign, or is $0$, in the limit.
In particular, every $\zeta_h(C\vdots D),$ for $\zeta\in\{\alpha,\beta,\gamma,\delta,\varepsilon,\theta,\eta\},$ if not ill defined, satisfies the sign rules of \cref{def:L=1domino_signs} in the weak sense.
\end{obs}
The proofs are  immediate from the rows' scaling freedom and continuity.

\begin{rmk}\label{rmk:weak_vs_almost}
The difference between weak $\D$ domino, a domino limit and an almost domino form, is that the last two allow for the starting domino of a row, or the ending domino of the yellow row, to vanish, while other rows whose domino forms should have included a multiple of these dominos still contain such a domino. This may happen when the domino $(\zeta_h(C^i\vdots D^i),\sigma_h(C^i\vdots D^i))$ of the parent or yellow chord goes to $0$ as $i\to\infty$ but its contribution to the child or blue row, $\xi_j(C^i\vdots D^i)(\zeta_h(C^i\vdots D^i),\sigma_h(C^i\vdots D^i))$ does not tend to zero, which implies that $\lim_{i\to\infty}\xi_j(C^i\vdots D^i)\to\infty,$ justifying our definition.

An almost domino form differs from a domino limit in that the entries and $2\times2$ minors of the former that are not identically zero in a domino form, see \cref{obs:poly_ineqs}, have no sign constraints. In a domino limit case they are constrained to have prescribed signs, at least in the weak sense.
\end{rmk}

\begin{definition}\label{def:Sred_blue}
    Let $\D$ be a chord diagram. 
    The subset $\Sred_\D\subset \overline{S}_{\D}$ is defined as the set $(U\vdots V)\in\overline{S}_{\D}$ such that $U$ contains a vector which is a linear combination of $D_0,D_\star$ for some almost $\D-$domino pair for $U\vdots V.$  

    The subset $\Sblue_\D\subset \overline{S}_{\D}$ is defined as the set $(U\vdots V)\in\overline{S}_{\D}$ such that for some almost domino pair $C\vdots D$ for $U\vdots V$, and some blue chord $\D_h,$ $U$ contains a vector in $\Span(\ee_{a_h},\ee_{b_h},D_\star,\pr_{c_\star d_\star}D_\star).$

Set \[\Srem_\D = 
\overline{S}_\D\setminus(\SA\cup\Sred_\D\cup\Sblue_\D).\]
\end{definition}
\begin{nn}
In light of the previous definition,  
by slight abuse of notations, if $C\vdots D$ is an almost $\D-$domino pair which witnesses that $U\vdots V$ belongs to $\Sblue_\D,~\Sred_\D,~\Sblue,~\Sred,$ we sometimes write that $C\vdots D$ belongs to the corresponding set.  
\end{nn}
For the treatment in $\Sred_\D,\Sblue_\D$ below we shall need certain distinguished $\ww_h\in U+V$, for every red or blue chord $\D_h$. 
\begin{lemma}\label{lem:ww}
For every domino limit $C\vdots D$ for $U\vdots V\in\overline{S}_\D$, and each red or blue chord $\D_h,$ there exists a non zero vector \[\ww_h\in (U+V)\cap\Span(\ee_{a_h},\ee_{b_h},\ee_{c_\star},\ee_{d_\star})\]whose projection on $(a_h,b_h)$ is proportional to the starting domino of $h$'th row of $C\vdots D$, and to the dominoes $C_h$'s children inherit from it, and its projection on $(c_\star,d_\star)$ is proportional to $u_\star$.
\end{lemma}
\begin{proof}
We can find domino pairs $C^i\vdots D^i\in S_\D$ which tend to $C\vdots D.$ For $h=0$ there is a unique linear combination of $D^i_0,D^i_\star$ which cancels the $n$th entry or the starting domino of $\p(0).$ Call it $\ww^i_0.$ 
Suppose we have defined $\ww^i_{\p(h)}.$ If $C_h$ is red then the vector
\[C_h+\gamma_h D_\star+\delta_h D_0
\in (U+V)\cap\Span(\ee_{a_h},\ee_{b_h},\pr_{a_{\p(h)b_{\p(h)}}}\ww^i_{\p(h)}).\]
By subtracting from it $\varepsilon^i_h\ww^i_{\p(h)}$ we obtain $\ww^i_h.$

Similarly, if $\D_h$ is blue then 
\[C_h+\delta_h D_\star\in (U+V)\cap\Span(\ee_{a_h},\ee_{b_h},\pr_{a_{\p(h)b_{\p(h)}}}C^i_{\p(h)},u_\star),\]where if $\p(h)=0$ then by $C_{\p(h)}$ we refer to $D_0.$
Again, subtracting $\varepsilon_h^i\ww^i_{\p(h)}$ from this vector, we obtain $\ww^i_h.$

We may take a subsequence limit of the properly scaled  $\ww^i_h$ to obtain the required $\ww_h.$

The last property of the vectors $\ww_h$ holds for $C^i\vdots D^i$ before taking the limits, hence also holds after the limit, even though some of the projections might vanish. 
\end{proof}
The recursive process which describes the vectors $\ww_h$ for $U\vdots V\in S_\D$ can be solved to give
\begin{obs}\label{obs:ww}For $U\vdots V\in S_\D$ it holds that
\[\ww_h=\alpha_h\ee_{a_h}+\beta_h \ee_{b_h}+\Gamma_h\pr_{c_\star d_\star}D_\star,\]
where \[\Gamma_h = \begin{cases}\frac{(-\varepsilon_h)(-\varepsilon_{\p(h)})\cdots(-\varepsilon_0)}{\varepsilon_\star} &\D_h\in\Red_\D\\
\gamma_h-\gamma_{\p(h)}\varepsilon_h+\cdots\pm\gamma_{\p^m(h)}\varepsilon_h\cdots\varepsilon_{\p^{m-1}(h)}\mp\frac{(-\varepsilon_h)(-\varepsilon_{\p(h)})(-\varepsilon_{\p^2(h)})\cdots(-\varepsilon_0)}{\varepsilon_\star}&\D_h\in\Blue_\D\\
\end{cases}\]
where in the blue case $m$ is such that $\D_{\p^m(h)}\in\Blue_\D,\D_{\p^{m+1}(h)}\in\Red_\D.$
  In particular all summands in $\Gamma_h$ have the same sign.
\end{obs}
\begin{proof}
For $h=0$, $\Gamma_0=-\frac{\varepsilon_0}{\varepsilon_\star}<0.$
For other $\D_h\in\Red_\D,$ $\ww_h$ is obtained from starting from $C_h$ by subtracting $\gamma_hD_\star+\delta_hD_0+\varepsilon_h\ww_{\p(h)},$ hence $\Gamma_h=-\varepsilon_h\Gamma_{p(h)},$ and this recursion is solved to yield the required $\Gamma_h$, whose sign is $-1$ to the power
\[(1+\Between(\D_h))+(1+\Between(\D_{\p(h)})+\ldots+(1+\Between(\D_{\p^r(h)})+\below(\D_0)+1)+1,\]where $r$ is the number of red ancestors of $\D_h$ which are not the  top red, and we have used the sign rules \cref{def:L=1domino_signs}. This sum equals $\below(\D_h),$ from the definitions of $\below,\Between.$

For $\D_h\in\Blue_\D,$ \begin{equation}\label{eq:ww_C_lin_comb}\ww_h=C_h-\delta_hD_\star-\varepsilon_h\ww_{\p(h)}\Rightarrow \Gamma_h=\gamma_h-\varepsilon_h\Gamma_{\p(h)}.\end{equation}Again this recursion yields the required expression, and the sign is seen to be $(-1)^{\below(\D_h)},$ from the same reasoning as above.
\end{proof}
\begin{obs}\label{obs:cell_not_red_blue}For a chord diagram $\D,$\[S_{\D}\cap(\Sred_\D\cup \Sblue_\D)=\emptyset.\]\end{obs}
Note that it is not yet clear that $\Srem_\D$ is non empty and contains $S_\D,$ but it will be shown below.
\begin{proof}
Consider $U\vdots V\in S_\D.$ The proof of \cref{thm:dominoAndBCFWparams}, which calculates the domino or BCFW variables row-by-row, easily extends to showing that $U\vdots V$ has a unique almost domino pair, hence, it is a domino pair. Denote it by $C\vdots D.$ Since it is of full rank, the span of $D_0,D_\star,$ must intersect the row span of $C$ trivially, implying $S_\D\cap\Sred_\D=\emptyset.$

Suppose there is a non trivial linear combination $v$ of $\ee_{a_h},\ee_{b_h},u_\star,D_\star$ in $U,$ where $\D_h\in\Blue_\D,$ then there are three possibilities. Either its projection on the entries $(a_h,b_h)$ is $0,$ or it is nonzero but proportional to the starting domino of $C_h,$ or it is not proportional to the starting domino of $C_h$. In the first case, since $D_\star$ is not in the row span of $C$ then we can find a non zero linear combination of $D_\star,v$ supported on $c_\star,d_\star,$ contradicting \cref{lem:bcfw_after_fl},~\cref{it:bcfw_after_fl_yellow}.

In the second case, adding a small multiple of $v$ to $C_h$ does not change any domino sign, and preserves the rank and the domino form, contradicting the uniqueness of the domino form \cref{thm:dominoAndBCFWparams}.

In the last case, there is a non zero linear combination $v'$ of $v$ and $D_\star$ which lies in the linear span of $\ee_{a_h},\ee_{b_h}$ and $u_\star.$ It must be different from $\ww_h$ by \cref{lem:ww}. Thus, we can find a linear combination of $v',\ww_h$ supported only on $a_h,b_h$, contradicting \cref{lem:bcfw_after_fl},~\cref{it:bcfw_after_fl_reds}. 
\end{proof}

The following key lemma will be useful for our boundary analysis below. 
\begin{lemma}\label{lem:uniqueness_or_SA}
\begin{enumerate}
    \item Every $U\vdots V\in\Srem_\D$ contains a unique almost $\D-$domino pair $C\vdots D$, up to the gauge equivalence of \cref{obs:gauge_domino}. $C\vdots D$ is a weak $\D-$domino representative.

\item If $U\vdots V\in\overline{S}_\D\setminus\SA$ then it has a unique almost domino form weakly outside $\D_0.$ This form is a weak domino form weakly outside $\D_0.$ We recall the reader that 'weakly' means including $D_0,D_\star.$ 
In this case also the vectors $\ww_h$ defined in \cref{lem:ww} are unique up to scaling. Moreover, the matrix obtained from any domino limit for $U\vdots V,$ by replacing the rows corresponding to $\D_h$ for $h\in\Red_\D\cup\Blue_\D$ by $\ww_h,$ is unique up to row scalings, and is of full rank.
\item\label{rmk:Sred_without_red_chords}
If $C\vdots D$ is a domino limit for $U\vdots V\in\overline{S}_\D\setminus\SA$, then if $U\cap\Span(D_0,D_\star)\neq0$ then there is a non top red row $C_h\in\Span(D_0,D_\star)$. In particular, when $|\Red_\D|=1$ $\Sred_\D\subseteq\SA.$\end{enumerate}
\end{lemma}
\begin{nn}\label{nn:ww}
In what follows, for a domino limit $C\vdots D$ of $U\vdots V\in\overline{S}_\D\setminus\SA$ we will write $\ww_h=\ww_h(C\vdots D)$ for the vectors $\ww_h$ defined in \cref{lem:ww}.  They are well defined up to scaling by Item (2) above. if $U\vdots V\notin \SA$, and then we denote them by $\ww_h=\ww_h(U\vdots V)$.
\end{nn}
An important consequence of the lemma is that in order to determine whether a point which does not belong to $\SA$ is in $\Sred$ or $\Sblue$, it is enough to take one domino limit and check if the point satisfies the conditions for belonging to $\Sred, \Sblue$ with respect to it.

As intermediate steps we will need the following two claims. The first is a simple criterion for belonging to $\SA,$ and the second is a list of cases in which almost domino pairs fall into $\SA.$
\begin{obs}\label{lem:2SA^D_cases}
Let $\D\in\CD_{n,k}^1$ be a chord diagram, $C\vdots D$ be a $\D-$ domino limit, not necessarily of full rank contained in $U\vdots V$. Assume that the row span of $C$ contains a vector $v=\gamma D_\star +\delta D_0 +u$ where $u\neq0$ is supported on two cyclically consecutive indices. Then $U\vdots V\in\SA_\D.$
\end{obs}
\begin{proof}
$u=v-\gamma D_\star -\delta D_0\in \Span{C+D}$ is non zero, and is thus a witness for belonging to $\SA_D.$
\end{proof}
\begin{lemma}\label{lem:Srem_cases}
Let $C\vdots D$ be an almost $\D$ domino pair contained in $U\vdots V\in\overline{S}_\D\setminus\SA.$ 
\begin{itemize}
    \item Assume that $U\vdots V\in \Srem_\D.$ Then the following hold:
    \begin{enumerate}
        \item For every chord neither the starting domino $(\alpha_h,\beta_h)$,nor the ending domino $(\gamma_h,\delta_h),$ nor $\varepsilon_h$ vanish. In particular every $\varepsilon_h,\gamma_h$ are defined and finite.
        \item $\gamma_h\neq 0$ for a top chord, which may be the yellow chord, ending at $n-1.$
        \item $\gamma_h\neq 0$ for $\D_h$ being the red sticky child of $\D_0,$ if $\D_0$ starts at $(1,2)$ or is itself a sticky child.
        \item $\delta_h\neq 0$ for a short black or purple $C_h$ which is a sticky child or starts at $(1,2).$
        \item $\alpha_h\neq 0$ for a chord with a sticky child.
        \item $\beta_h\neq 0$ for a chord with a sticky parent or which starts at $(1,2).$
        \item $\alpha_h\neq 0$ for a short black chord which ends at $(n-2,n-1),$ or a very short chord.
        \item $\eta_i\neq0$ if $\D_i$ is a sticky same end child of $\D_j$ for which $\eta_i$ is defined, and $\D_j$ is either a sticky child or starts at $(1,2).$
    \end{enumerate}
    Moreover, the almost domino pair is unique up to row scalings.
    The same conclusions and uniqueness hold for every $U\vdots V\in\overline{S}_\D\setminus\SA$ for rows corresponding to $\D_h\in\woutside(\D_0).$
    \item Assume $C\vdots D$ is a domino limit and no red row of $C$ lies in $\Span(D_0,D_\star).$ Then all the above conclusions hold, except that possibly $\varepsilon_h$ may vanish, for $\D_h\in\Blue_\D,$ if $\D_h$ is not very short ending at $(n-2,n-1).$
    If $\varepsilon_h$ vanishes for a blue chord, then for this chord $\delta_h\neq 0.$
    \item If $C\vdots D$ is a domino limit, and no further assumptions are made, then the conclusions of the previous item hold with the following exceptions.
    \begin{itemize}
    \item[--] Both $(\alpha_h,\beta_h)$ and $\varepsilon_h$ may vanish (or $(\alpha_h,\beta_h)=(0,0)$ and $\varepsilon_h$ ill defined), but not only one of them, for $\D_h\in\Red_\D\setminus\{0\}$. 
    \item[--] $\alpha_h$ may vanish if $C_h\in\Red_\D\setminus\{0\}$ has a sticky child or is very short.
    \item[--] $\beta_h$ may vanish if $C_h\in \in\Red_\D\setminus\{0\}$ is a sticky child.
        \item[--] $\eta_h$ may vanish for $C_h\in\Red_\D\setminus\{0\}$, if $C_h$ is a sticky-same end child of another red chord.
    \end{itemize}
    In all these cases it holds that $C_h\in\Span(D_0,D_\star),$ hence $U\vdots V\in\Sred_\D.$
    Moreover, if $\ww_h,~h\in\Red_\D\cup\Blue_\D\cup\{\star\}$ are as in \cref{lem:ww} then
     $\pr_{a_hb_h}\ww_h\neq 0$, and in addition $\pr_{a_h}\ww_h\neq 0$ if $\D_h$ has a sticky child, and $\pr_{b_h}\ww_h\neq0$ if $\D_h$ is a sticky child. 
\end{itemize}
In the last two items the black, purple and yellow rows, and the vectors $\ww_h$ are unique up to scaling.
\end{lemma}
\begin{rmk}This lemma and \cref{lem:uniqueness_or_SA} generalize corresponding results of \cite{even2021amplituhedron} for the tree level. The analysis of the tree level is much simpler as $\Sred=\Sblue=\emptyset$ in this case, and these spaces are responsible for most of the complication. 
\end{rmk}
\begin{proof}
We will go over the different types of vanishings, and see what is implied from them. In all the cases we will find a vector whose support guarantees that $U\vdots V\in\SA.$ It will sometimes be straightforward, and sometimes involve some work. We will study the chords in parent-to-child order, and induct, in this order, on the various uniqueness claims.
\begin{itemize}
\item\textbf{The (non) vanishing of the starting, ending, or inherited domino for a chord $\D_h$, for $\ww_h$ and uniqueness:} If the chord is black or purple the support of the resulting row is $(a_{\p(0)},b_{\p(0)},c_h,d_h),$ $(a_{\p(h)},b_{\p(h)},a_h,b_h)$ or $(a_h,b_h,c_h,d_h)$ respectively, where we set $a_{\p(h)}=b_{\p(h)}=n$ if $\p(h)=\emptyset.$ In case $\gamma_h=\delta_h=0$ we obtain the same support also for $\D_h\in\Red_\D\cup\Blue_\D\setminus\{0\}.$ In these cases $U\vdots V\in \SA_C.$ Similarly, for the top red or yellow chord we get supports of the form $(a_0,b_0)$, $(c_\star,d_\star)$ or $(a_{\p(h)},b_{\p(h)})$ which imply $U\vdots V\in\SA_D.$ 

Consider $\D_h\in\Red_\D\setminus\{0\}.$ Assume $\alpha_h=\beta_h=0$ but $\varepsilon_h$ is defined and non zero, or $(\alpha_h,\beta_h)\neq(0,0)$ but $\varepsilon_h=0$ or ill defined. We may apply \cref{lem:2SA^D_cases} 
to $C_h,$ with \[u = \begin{cases}\alpha_{\p(h)}\ee_{a_{\p(h)}}+\beta_{\p(h)}\ee_{b_{\p(h)}}&~\text{if $\varepsilon_h\neq 0$ and $\D_h$ is not a top chord:}\\
\alpha_{h}\ee_{a_{h}}+\beta_{h}\ee_{b_{h}}&~\text{if $\varepsilon_h= 0$}\end{cases}\]
$u\neq 0$ and hence $(U\vdots V)\in \SA_D.$ If $\alpha_h=\beta_h=\varepsilon_h=0$ then $C_h\in\Span(D_0,D_\star).$

Consider $\D_h\in \Blue_\D.$ If $\alpha_h=\beta_h=\varepsilon_h=0$ then $C_h$ has the same support as $D_\star$ which shows $U\vdots V\in\SA_C.$

If $\varepsilon_h=0$ then by definition $U\vdots V\in\Sblue_\D.$ If $\D_h$ is in addition very short and ends at $(n-2,n-1)=(c_\star,d_\star)$ then the support of $C_h$ becomes $(c_\star,d_\star,a_{\p(0)},b_{\p(0)}),$ which shows that in this case $U\vdots V\in\SA_C.$ 

Before we treat the case $\alpha_h=\beta_h=0$ for a non top blue chord, we turn to uniqueness.
Assuming we know, by induction, the uniqueness up to scaling of the parent of $\D_h$, which is trivial if $\p(h)=\emptyset.$ If there is another weak domino form whose $h$th row $C'_h$ ($D'_h$) is not proportional to $C_h$ ($D_h$) then a linear combination of them that cancels the domino inherited from the parent would witness that $U\vdots V\in\SA,$ if $\D_h$ is black, purple, top red or yellow. If $\D_h\in \Red_\D\setminus\{0\}$ it shows, as above, $U\vdots V\in\SA\cup\Sred_\D.$ If $\D_h\in\Blue_\D$ it shows $U\vdots V\in\Sblue_\D,$ 
This proves the uniqueness in $\Srem_\D,$ and the uniqueness for $\D_h\in\woutside(\D_0).$ It shows, in particular, that all domino limits are the same in $\Srem_\D,$ and their restrictions to chords in $\woutside(\D_0)$ are the same on $\overline{S}_\D\setminus\SA.$

Returning to $\alpha_h=\beta_h=0$ for $\D_h\in\Blue_\D.$
We will show that the starting domino of $\D_h$ is never zero on a $\D$-domino limit for $U\vdots V\notin\SA:$
If also $\gamma_h=0$ the $C_h$ is a linear combination of its inherited domino and $D_\star.$ Using \cref{lem:2SA^D_cases} this implies $U\vdots V\in\SA.$
Assume now that $\gamma_h\neq 0.$  
Now, if $C^i\vdots D^i$ is a sequence of properly scaled domino pairs for $U^i\vdots V^i$ which tend to $U\vdots V,$ such that $\lim_{i\to\infty}C^i\vdots D^i=C\vdots D$ entry-wise as matrices, consider the vectors $\ww^i_h$ defined in \cref{nn:ww}, and a limit $\ww_h$ of a properly scaled subsequence. Note that by our assumption \[\lim_{i\to\infty}\alpha^i_h/\gamma_h^i=\lim_{i\to\infty}\beta^i_h/\gamma_h^i=0.\]  $\varepsilon_\star\neq0$ since $U\vdots V\notin\SA.$ Using \cref{obs:ww} and \cref{lem:ww} we see that $|\Gamma_h|\geq|\gamma_h|,$ hence\[\pr_{a_hb_h}\ww_h=0\Rightarrow U\vdots V\in\SA_D.\]
Note that the above shows that in $\Srem_\D$ all $\varepsilon_h$ are defined and finite, since otherwise the starting domino of $\D_{\p(h)}$ vanishes. On $\overline{S}_\D\setminus\SA$ all $\varepsilon_h\neq 0$, except possibly for children of red chords with $C_{\p(h)}\in\Span(D_0,D_\star).$ Blue $\gamma_h$ are well defined and finite, since otherwise the domino $u_\star$ of the yellow chord vanishes.

The above argument shows that also the projection of $\ww_h$ on $(a_h,b_h)$ is non zero. They must be unique, since by the same inductive argument, as above, if they were not unique there would have been a linear combination of them supported on a consecutive pair of indices, hence $U\vdots V\in\SA_D.$ Moreover, by \cref{lem:ww} we can deduce the uniqueness of all black and purple rows, even those that descend from the top red chord, by the same inductive argument we used above for uniqueness, but using $\ww_h$ to keep track of the starting domino of $C_h,$ in case the latter vanishes. 
\item\textbf{$\gamma_h\neq 0$ for a top chord, including the yellow chord, ending at $n-1.$} Otherwise, if the chord is black then it is supported on $(a_h,b_h,n-1,n).$ If it is yellow then on $(n-1,n).$ In both cases $U\vdots V\in\SA.$
\item\textbf{If $\D_h$ is the red sticky child of $\D_0,$ and $\D_0$ starts at $(1,2)$ or is itself sticky then $\gamma_h\neq 0$.}
Otherwise the support of $C_h$ is $(1,2,3,n)$ or $(a_{h}-2,a_h-1,a_h,b_h).$ Both show $U\vdots V\in\SA.$
\item\textbf{$\delta_h\neq 0$ for a short black or purple $C_h$ which is a sticky child or starts at $(1,2).$}
Otherwise $C_h$ is supported on four cyclically consecutive indices, hence $U\vdots V\in\SA_C$.
\item\textbf{$\alpha_h$ for a chord with a sticky child. $\pr_{a_h}\ww_h\neq 0.$}
Assume $\alpha_h=0.$ If the child $\D_l$ is black, purple, top red or yellow and $\varepsilon_l$ is not infinite then the support of $\D_l$ is either its two dominoes, in the first two cases, or the single domino, in the last two. In both cases $U\vdots V\in\SA.$ $\varepsilon_l$ may be infinite only if $C_h\in\Span(D_0,D_\star),$ which implies in particular that $U\vdots V\in\Sred_\D.$

We now rule out also the cases when the child is blue or red, and when the parent is red. The case that $\D_h\in\Red_\D\setminus\{0\}$ and also $\beta_0=0$ was analyzed above. Otherwise, by \cref{lem:ww} it is enough to show that $\pr_{a_h}\ww_h\neq0$ for $h\in\Red_\D\cup\Blue_\D\setminus\{0\}.$ Assume this is not the case. We know that the projection of $\ww_h$ on $(a_h,b_h)$ is non zero, and similarly its projection on the $(c_\star,d_\star)$ entries is non zero or else $U\vdots V\in\SA_D.$
If $\pr_{a_h}\ww_h=0$, by \cref{lem:ww}, $\ww_h$ has the form $x\ee_{b_h}+z\pr_{c_\star d_\star}D_\star.$ If the sticky child $\D_l$ of $\D_h$ is purple or black, the proof is as above. Otherwise $\D_l$ is blue or red. In these cases, consider the vector
\[\ww_l = x'\ee_{b_h}+y'\ee_{b_l}+z'\pr_{c_\star d_\star}D_\star.\]Again if either $(x',y')=0$ or $z'=0$ then $U\vdots V\in\SA_D.$ If $\ww_h,\ww_l$ non proportional 
\[z\ww_l-z'\ww_h\in\Span(\ee_{b_h},\ee_{b_l})\Rightarrow U\vdots V\in\SA.\]
We are left with ruling out the possibility that $\ww_h,\ww_l$ are proportional. Let $U^i\vdots V^i,~i=1,2,\ldots$ be a sequence of $1$-loop vector spaces tending to $U\vdots V.$ Then the uniquely defined $\ww_h^i,\ww_l^i$ tend to $\ww_h,\ww_l$ respectively.
Let $W_{h,l}^i$ be the two plane spanned by $\ww_l,\ww_i.$ Its limit $W_{h,l}\subset U+V$ contains $\ww_h.$ Note that from the domino sign rules and \cref{obs:ww} it holds that
\begin{align*}&\lr{W_{h,l}^i}_{a_hb_h}>0,~\lr{W_{h,l}^i}_{a_hb_l}>0~\lr{W_{h,l}^i}_{b_hb_{l}}>0,~\text{and}\\&\sgn(\lr{W_{h,l}^i}_{a_hc_\star})=\sgn(\lr{W_{h,l}^i}_{a_hd_\star})=\sgn(\lr{W_{h,l}^i}_{b_lc_\star})=\sgn(\lr{W_{h,l}^i}_{b_ld_\star}).
\end{align*}
We may find a vector $v\in W_{h,l}$ completing $\ww_h$ to a basis, while keeping the above inequalities weakly satisfied. $v\in\Span(\ee_{a_h},\ee_{b_h},\ee_{b_l},\pr_{c_\star d_\star} D_\star).$ If its support, or $\ww_h$'s support, are contained in a consecutive pair of entries, $U\vdots V\in\SA.$ Otherwise, by subtracting an appropriate multiple of $\ww_h$ we may assume $v$'s $b_h$'s entry is $0,$ that $v=x''\ee_{a_h}+b''\ee_{b_l}+z''\pr_{c_\star d_\star}D_\star.$ If $x'',y''\neq 0$ then the limiting inequalities of the upper row above imply $\sgn(x'')=-\sgn(y''),$ but the bottom row's limiting inequalities imply $\sgn(x'')=-\sgn(y'').$ Thus, at least one of $x'',y''$ is $0.$ This implies that a nonzero linear combination of $\ww_h,v$ is supported on $c_\star,d_\star,$ hence $U\vdots V\in\SA.$
\item\textbf{$\beta_h\neq 0$ for a chord with a sticky parent or which starts at $(1,2).$ $\pr_{b_h}\ww_h\neq 0.$}
If the chord is black, purple, top red or yellow the support of the chord shows immediately that $\beta_h=0$ implies $U\vdots V\in\SA.$ If it is non top red, or blue, the proof is similar to the previous item, using combinations of $\ww_h,\ww_{\p(h)}.$ When $\D_h\in\Red_\D\setminus\{0\}$ if $C_h\notin\Span(D_0,D_\star)$ then we may apply \cref{lem:2SA^D_cases} as in the first item to show $U\vdots V\in\SA.$
\item\textbf{$\alpha_h\neq 0$ for a short black chord which ends at $(n-2,n-1),$ or a very short chord.
} 
In case $\D_h$ is black and $\alpha_h=0$, if it is a top chord then the corresponding row has support $(n-3,n-2,n-1,n).$ Otherwise, by performing row operations in a child-to-parent order, with the ancestors of $\D_h$ to cancel the inherited dominoes, as in the passage from the domino form to the BCFW form, which are possible since all $\varepsilon_l$ variables of its ancestors are finite and well defined, we obtain a row of such support.

If the chord is very short and blue then by the above we may assume that $\beta_h\neq 0.$ Thus, the corresponding $\ww_h$ is supported on $(c_\star,d_\star)$ and $U\vdots V\in\SA$. If $\D_h$ is very short red then if also $\beta_h=0$ then either $C_h\in\Span(D_0,D_\star)$ or, by the above, $U\vdots V\in \SA.$ If $\beta_h\neq 0$ then again $\ww_h$ has support $(c_\star,d_\star)$ showing $U\vdots V\in\SA.$
\item\textbf{$\eta_i\neq0$ if $\D_i$ is a sticky same end child of $\D_j$ for which $\eta_i$ is defined, and $\D_j$ is either a sticky child or starts at $(1,2).$}
Otherwise a linear combination of $C_i,C_j$ is supported on $(1,2,3,n)$ or $(a_{\p(\p(i))},a_{\p(i)},a_i,b_i).$
If this vector is non zero then it witnesses $U\vdots V\in\SA.$ If it is zero then $\beta_i$ must be $0,$ which implies, from previous items, that if $\D_i$ is black or purple or blue, $U\vdots V\in\SA.$ If $\D_i$ is red then it also implies $U\vdots V\in\SA$ unless $C_i\in\Span(D_0,D_\star).$ 
\item\textbf{$\varepsilon_h\neq 0$ for a very short chord ending at $(n-2,n-1).$}
If $\p(0)=\emptyset$ the support of $C_h$ is $(n-3,n-2,n-1,n)$ showing that $U\vdots V\in \SA_C.$ Otherwise we can add to $C_h$ a combination of $C_{\p(0)},C_{\p(\p(0))},\ldots$ to obtain such a support.
\item\textbf{$(\delta_h,\varepsilon_h)\neq(0,0)$ for blue chords. }Otherwise we obtain a vector of support $(a_h,b_h,c_\star,d_\star)$ which shows $U\vdots V\in\SA.$
\end{itemize}
\end{proof}
\begin{proof}
[Proof of \cref{lem:uniqueness_or_SA}]
Let $C\vdots D$ be an 
$\D$-domino limit for $U\vdots V.$ There is always at least one such limit by \cref{obs:domino_limits}. It follows from \cref{lem:Srem_cases} that if it is not in a weak $\D-$domino form then $U\vdots V\in \SA\cup\Sred\cup\Sblue.$ \cref{lem:Srem_cases} also implies the uniqueness claims of this lemma.

Turning to ranks, we now show that if $U\vdots V\notin \SA\cup\Sred\cup\Sblue$ has a weak domino form $C\vdots D$ then $C+D$ has full rank, hence also $C$, meaning that this is a weak domino representative, which by the above will also be unique. By the assumption on $U\vdots V$ and \cref{lem:Srem_cases} $\varepsilon_h$ is defined and finite for every $h.$ Consider the \emph{weak BCFW form} of $C+D$, which is the matrix $M$ obtained from $C+D$ by 
subtracting, in a parent-to-child order, $\varepsilon_h$ times the row labeled $\p(h)$ from the row labeled $h.$ Since $U\vdots V\in\Srem,$ $\varepsilon_h$ is defined and finite, hence also these operations are defined. Observe that the row $M_h$ corresponding to $\D_h$ has zeroes in all indices smaller than $a_h$ or $c_\star$ if $h=\star$. 

Let $f_h,~h\in[k+1]\cup\{\star\}$ be the index of the first non zero entry of $M_h.$ Then by \cref{lem:Srem_cases} and the assumption $U\vdots V\in\Srem$ the submatrix $(M)_{h\in[k+1]}^{\{f_h|~h\in[k+1]\}},$ if its rows are reordered in an increasing order of $a_h$, is upper triangular with non zero diagonal entries. Thus, the rank of $M$ is at least $k+1.$ 

To finish we need to show that $M_\star$ is not in the linear span of the rows of $M_{[k+1]}.$ By the triangularity above, if $M_\star$ were a linear combination of these rows, then only rows with $f_i\geq f_\star$ could have appeared with non zero coefficients. Now, \[\{\D_i|~f_i\geq f_\star\}\subseteq\outside(\D_\star).\]
In particular $\D_0$ is not in this set.
The passage from $C_i$ to $M_i$ involves only the row $C_i$ and its ancestors. The same holds for $D_\star,M_\star.$ Thus, 
\[\Span(M_i|\D_i\in\outside(\D_\star))\subseteq\Span(C_i|\D_i\in\outside(\D_\star)).\]
Hence, if $M_\star\in\Span(M_i|\D_i\in\outside(\D_\star))$ then 
\[D_\star\in\Span(C_i|\D_i\in\outside(\D_\star))\subseteq\Span(C)\subseteq U.\]
But this implies a vector of support $(c_\star,d_\star,a_{\p(\star)},b_{\p(\star)})$ in $U,$ hence $U\vdots V\in\SA.$ A contradiction.

The proof that the matrix $M'$ made of the rows of $C+D$ for $h\notin\Red\cup\Blue,$ and $\{\ww_h\}_{h\in\Red\cup\Blue}$ has full rank is identical to the above, only that we use $\ww_h$ to cancel the inherited dominoes of children of $C_h,~h\in\Red_\D\cup\Blue_\D,D_0$.

For \cref{rmk:Sred_without_red_chords}, assume towards contradiction that $U\vdots V\notin\SA,$ there is no non top red row in the well defined $\Span(D_0,D_\star)$ but $U$ contains a linear combination $v$ of $D_0,D_\star.$ If this combination is $D_0$ or $D_\star$, or if one of the dominoes or inherited dominos of $D_0,D_\star$ vanishes then $U\vdots V\in\SA.$ Otherwise, let  $C\vdots D$ be a domino limit for $U\vdots V.$ By \cref{lem:Srem_cases}, in this setting, the starting domino of every chord is non zero, hence every $\varepsilon_h$ is defined and finite. Let $C'$ be the matrix obtained from $C$ by adding the row $v$. We perform row operations as above to cancel the inherited domino of each chord, where for a child of the top red chord we subtract a multiple of $v,$ for other red rows we cancel both the inherited domino using the parent, and the starting domino of $D_0$ using $v$, and for $v$ we subtract a multiple of $C_{\p(0)}$, if exists. The above argument which showed that $C+D_0$ has rank $k+1$, shows the same for the linear span of $v$ and the rows of $C$. Thus, $v$ is linearly independent of $C,$ which is a contradiction.
\end{proof}
\subsection{Matching boundaries}\label{subsec:mathc_bdries}
We now define several boundary strata that will be later understood to be precisely the codimension $1$ boundaries of $S_\D$ which map to codimension $1$ boundaries in the image of $S_\D$ under the amplituhedron map, but avoid the external boundary of the amplituhedron. We will also show that each such boundary is also a boundary of another BCFW cell.
\begin{definition}\label{def:shift_ops}
Let $\D\in\CD_{n,k}^\ell,$ and $\D_i=(a_i,b_i,c_i,d_i)$ be a chord. We define certain other diagrams by \emph{shifting} the chord $\D_i$ in various ways. In all cases chords we have not described their change - do not change, and the labels of chords may change according to the labeling convention of the shifted diagram.
\\\textbf{$\shift_{\alpha_i}\D$ - the shift  of $\D_i$ at $\alpha_i$ 
$i\neq\star$:} 
\\Assume $\D_i$ has no sticky child, is not very short, and if it is short then $d_i\neq n-1.$ Then 
\begin{itemize}
\item if $c_i>b_{i}+1$ $\shift_{\alpha_i}\D$ is the same as $\D$ only $\D_i$ is replaced by $(b_i,b_i+1,c_i,d_i).$ 
\item If $\D_i$ is short and purple, it becomes very short and blue; if it is short and blue it becomes very short and red; if it is short and red it becomes very short and red. 
\item If it is short and black then every ancestor of $\D_i$ which ends at $(c_i,d_i)$ is replaced by a chord with the same start and color, but which ends at $(b_i,c_i).$ If no chord starts at $(c_i,d_i)$ then $\D_i$ is replaced by $(b_i,c_i,d_i,d_i+1).$ Otherwise, if there exists such a chord, which ends at $(e_i,f_i)$ then $\D_i$ is replaced by $(b_i,c_i,e_i,f_i).$
\end{itemize}
\textbf{$\shift_{\beta_i}\D$ 
- the shift of $\D_i$ at $\beta_i,~i\neq\star$:} 
\\Assume $\D_i$ has no sticky parent, and does not start at $(1,2)$ or at the ending domino of another chord, then $\shift_{\beta_i}\D$ is the chord diagram having the same chords as $\D,$ only that $\D_i$ is replaced by $(a_i-1,a_i,c_i,d_i).$
\\\textbf{$\shift_{\gamma_i}\D$ - 
the shift of $\D_i$ at $\gamma_i,~i\neq 0$:} 
\begin{itemize}
\item\underline{$\gamma_i\in\Var_\D,~\D_i\notin\Red_\D$ and $d_i\neq n-1$:}\\If $i\neq\star$, the shifted diagram $\shift_{\gamma_i}\D$ is the same as $\D,$ only that the ending domino of $\D_i$ is shifted to $(d_i,d_i+1)$ if $\D_i$ is black or yellow. Otherwise it is unchanged, but if 
$\D_i$ is (the top) blue the new chord is colored red. If $i=\star$ then the yellow, red and blue chords are shifted to end at $(d_\star,d_\star+1),$ keeping their colors. The purple chords are declared black.
\item\underline{$\theta_i\in\Var_\D:$} Let $\D_h$ be the sibling of $\D_i$ which starts where $\D_i$ ends.
If $i\neq\star$ then $\D_i$ is replaced by $(a_i,b_i,c_h,d_h)$ and gets the color of $\D_h,$ unless $\D_h$ is the top red chord, and then $\D_i$ remains black. 
If $i=\star$ then $\D_\star$ is replaced by $(c_h,d_h)$, all red and blue chords of $\D$ change their end to $(c_h,d_h)$ without changing their start, and keep their colors. $D_h$ is colored blue.
\item\underline{$\eta_i\in\Var_\D,~\D_i$ is the top purple chord:}\\The shifted diagram $\shift_{\gamma_i}\D$ is the same as $\D,$ only that $\D_i$ becomes blue.
\item\underline{$\D_i$ is the second topmost red chord or $\eta_i\in\Var_\D,$ for $\D_i$ which is not the top purple chord:}
\\Write $\D_h$ for the same end parent of $\D_i.$ In order for this shift to be defined we require that if $\D_i$ is a sticky child, then $\D_h$ is not a sticky child and does not start at $(1,2).$ 
\\In the shifted diagram $\D_h$ is replaced by the chord $(a_h,b_h,a_i,b_i)$ unless $b_h=a_i,$ which means $\D_i$ is a sticky child. In this case if $a_h\neq c_\star,$ $\D_h$ is replaced by a black chord $(a_h-1,a_h,a_i,b_i).$ 
In both cases, if $h=0$ then $\D_i$ is declared the top red chord.
If $a_h=c_\star$ then $\D_h$ is replaced by a very short blue $(a_h,b_h,b_h+1).$

If there is a chain of same end chords ending at $(a_h,b_h),$ and we include the case $a_h=c_\star,$ in which this chain contains the yellow chord, then all non purple chords ending at $(a_h,b_h)$ are replaced by chords with the same starting point, but which end at $(b_h,b_h+1),$ and they keep their colors. If $a_h=c_\star$ then the purple chords are declared black.
\end{itemize}
\textbf{$\shift_{\delta_i}\D$ - 
the shift of $\D_i$ at $\delta_i,~i\neq 0$:} 
\\This shift is defined when $\delta_i\in\Var_\D,$ that is, $\D_i$ has no strict same end child, and, if $\D_i$ is short black, purple, or very short, we also assume it is not a sticky child and does not start at $(1,2)$. Note that for $i=\star$ the absence of strict same end children means that there are no purple chords.
\begin{itemize}
    \item If $\D_i\in\Black_\D\cup\Purp_\D$ is not short, then it is replaced by a black chord $(a_i,b_i,c_i-1,c_i).$
    \item If $\D_i$ is red and not very short, it just changes color to blue; if $\D_i$ is blue and not very short, then it changes color to purple.
    \item If $i=\star$, and there are no purple chords, then in the shifted diagram the end of every yellow, red and blue chord is shifted one to the left, that is, becomes $(c_\star-1,c_\star),$ and their colors are unchanged. Chords of $\D$ which end at $(c_\star-1,c_\star),$ are declared purple in the new diagram.
    \item If $\D_i$ is short black or purple and $a_i\neq c_\star$, then in the shifted diagram it is replaced by a chord $(a_i-1,a_i,a_i+1,a_i+2)$ and the end of every chord in $\D$ which ends at $(a_i,a_i+1)$ is shifted to $(a_i+1,a_i+2).$ All colors are the same, except for $\D_i$ which becomes black. 
    
    If $a_i=c_\star$ then $\D_i$ is replaced by a very short blue $(a_i,a_i+1,a_i+2).$ The yellow, red and blue chords of $\D$ change their end to $(a_i+1,a_i+2),$ but keep their colors. Purple chords are declared black.
     \item If $\D_i$ is very short, then it is replaced by $(c_\star-2,c_\star-1,c_\star,d_\star)$, and the end of every chord of $\D$ which ends at $(c_\star-1,c_\star)$ is shifted to end at $(c_\star,d_\star)$ in the shifted diagram. Other chords are untouched. If $\D_i$ is blue then it becomes purple, as well as all the other chords which have changed. If $\D_i$ is red, then it and the other chords which have changed become blue.
\end{itemize}
\textbf{$\shift_{\varepsilon_h}\D$ - 
the blue shift, or the shift of $\D_h$ at $\varepsilon_h,$ where $\D_h$ is the lowest blue chord:}
\\This diagram is obtained from $\D$ by declaring $\D_h$ and all the purple chords as black. The red and blue chords other than $\D_h$, and if $\D_h$ is very short then also the black chords ending at $(c_h,d_h)=(c_\star,d_\star)$, are replaced by chords, of the same color and starting domino, which end at $(a_h,b_h).$ Their descendants which end at $(a_h,b_h)$ are declared purple and $\D_\star=(a_h,b_h)$. 
If $\D_h$ is not very short it is untouched. Otherwise $b_h=c_\star.$ If there is a chord $\D_l$ starting at $(c_\star,d_\star)$ then $\D_h$ is replaced by $(a_h,b_h,c_l,d_l).$ If no such $\D_l$ exists, $\D_h$ is replaced by $(a_h,b_h,b_{h}+1,b_h+2).$
\\\textbf{$\shift_{\abe}\D$ - the red shift: }
\\This shift is defined only when $|\Red_\D|>1.$ The resulting diagram is obtained from $\D$ by changing the color of the top red chord to be black, and declaring its red child $\D_h$ to be the new top red chord. Other chords are untouched.

If $\D_1,\D_2$ differ by a shift, we refer to $\D_1,\D_2$ as \emph{neighboring chord diagrams}, and to $S_{\D_1},S_{\D_2}$ as \emph{neighboring BCFW} cells.
\end{definition}
\begin{obs}\label{obs:bdry_matching_comb}
    If $\shift_{\zeta_i}\D=\D'$ then there exists a unique $\sigma_j\in\tVar_{\D'}\cup\{\abe\}$ with $\shift_{\sigma_j}\D'=\D.$ In this case we write $\match(\D,\zeta_i)=(\D',\sigma_j),$ and say that $(\D,\zeta_i)$ is \emph{matched} to $(\D',\sigma_i).$ In particular, $\shift_{\zeta_i}\D\neq\shift_{\sigma_j}\D$ whenever the left and right side are defined and $\zeta_i\neq\sigma_j.$
\end{obs}
\begin{proof}
We list the matched pairs, and in all cases we assume that the shift is defined. 
\begin{itemize}
\item[--]$\match(\D,\alpha_i)=(\shift_{\alpha_i}\D,\zeta_j),$ where $\zeta_j$ is $\beta_i$ if $c_i>b_i+1.$ If $\D_i$ is short  $\zeta_j=\delta_i$ if no chord starts where it ends, or $\gamma_j$ is there is a chord which starts where $\D_i$ ends, and $j$ is the label of that chord \emph{after} the shift. 
\item[--]$\match(\D,\beta_i)=(\shift_{\beta_i}\D,\alpha_i).$
\item[--]$\match(\D,\gamma_i)=(\shift_{\gamma_i}\D,\zeta_j),$ where $\zeta_j$ is defined as follows. 
First, if $i=\star$ then if $\D_{\p(\star)}$ does not end at $(c_\star,d_\star)$ and no chord starts at $(c_\star,d_\star)$ then $\zeta_j=\delta_\star.$ If $\D_{\p(\star)}$ ends at $(c_\star,d_\star)$ then $\zeta_\star=\delta_0,$ and if a sibling of $\D_\star$ starts at $(c_\star,d_\star)$ then $\zeta_j=\varepsilon_j,$ where $j$ is the lowest blue chord.

If $i\neq\star,$ then $\zeta_j=\delta_i$ if no chord starts and no ancestor of $\D_i$ of the same color ends at $(c_i,d_i).$ If $\D_i$ has a same end parent whose label \emph{after} the shift is $j$ then if this parent is sticky $\zeta_j=\alpha_j,$ and otherwise $\zeta_j=\gamma_j.$ 
If $\D_i$ has sibling which starts at $(c_i,d_i)$, whose label after the shift is $j$ then $\zeta_j=\gamma_j.$ 
\item[--]$\match(\D,\delta_i)=(\shift_{\delta_i}\D,\zeta_j),$ where $\zeta_j=\gamma_i$ if $i=\star,$ or if it is black or purple and not short, or if it is red or blue and not very short.

$\zeta_j=\alpha_j$ if $\D_i$ is short black or purple and does not start at $(c_\star,d_\star),$ or if it is very short. In these cases $j$ is the label of the shifted chord. If $\D_i$ is short and starts at $(c_\star,d_\star)$ then $\zeta_j=\varepsilon_j,$ where $j$ is the index of the shifted chord, which is the lowest blue chord in the shifted diagram.  
\item[--]$\match(\D,\varepsilon_h)=(\shift_{\varepsilon_h}\D,\zeta_j),$ where $\zeta_j=\gamma_\star$ if the lowest blue chord $\D_h$ is not very short. If it is very short and no chord starts at $(c_\star,d_\star)$ then $\zeta_j=\delta_j,$ where $j$ is the index of the shifted blue chord which now ends at $(d_\star,d_{\star}+1).$ If there is a chord which starts at $(c_\star,d_\star)$ and its label after the shift is $j$ then $\zeta_j=\gamma_j.$
 \item[--]$\match(\D,\delta_0)=(\shift_{\abe}\D,\gamma_\star).$
\end{itemize}
The verification is straight forward, using the diagrams below.
\end{proof}
\begin{rmk}\label{rmk:diag_shifts}
The diagrams below include all possible shifts.
In every diagram the top and bottom rows represent the shifts which are matched to each other, their types are written beneath.

In diagrams (1)-(3), (6)-(9) and the top rows of (4),(5),(10),(11) we illustrate $\D_i$ as the chord in black, but it does not necessarily mean $\D_i\in\Black_\D$: we allow all different colors which are allowed in \cref{def:shift_ops}. 

The diagrams are self explanatory, but we make the following comments:
In diagrams (4),(5) if the top row is blue the bottom is red, and if it is purple the bottom is blue. In diagram (13), bottom row, the very short chord indicates that $(i+2,i+3)=(c_\star,d_\star)$ for that row, and $\D_i=\D_\star$.
Similarly, in diagrams (10),(11) top row $(i+1,i+2)=(c_\star,d_\star),$ while in the bottom rows $(i+2,i+3)=(c_\star,d_\star).$ The $\eta$ cases in diagrams (6),(7),(9)
 also include the possibility of second top most red chord $\D_i.$

The cases which do not involve colored chords, that is, (1)-(3),(6)-(9), are identical to the different possibilities for shifts in \cite[Section 7]{even2021amplituhedron}, though our notations slightly differ. \begin{center}
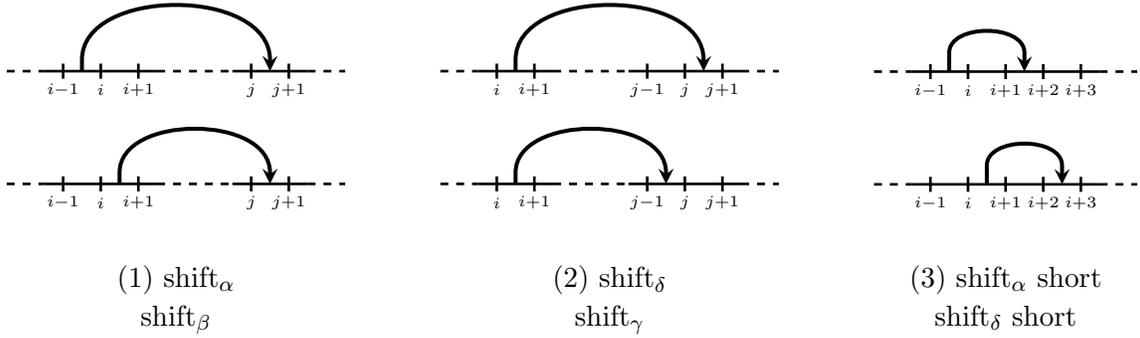
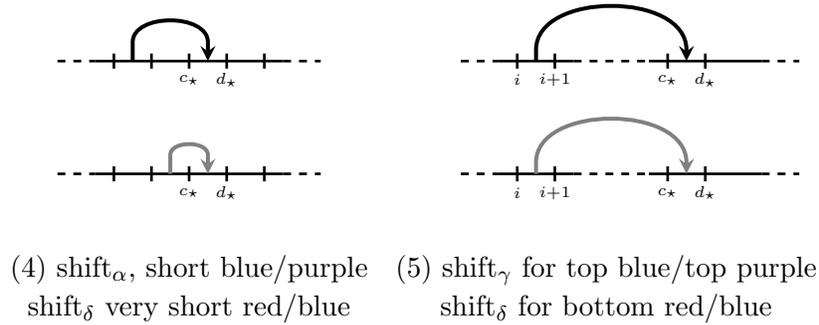
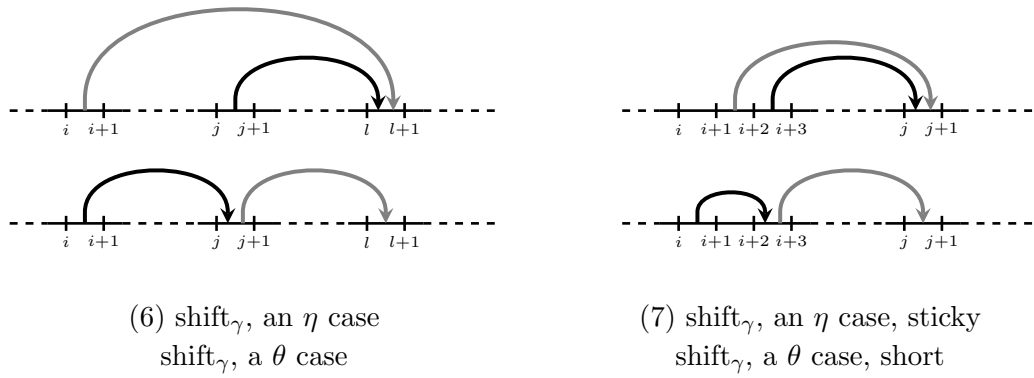
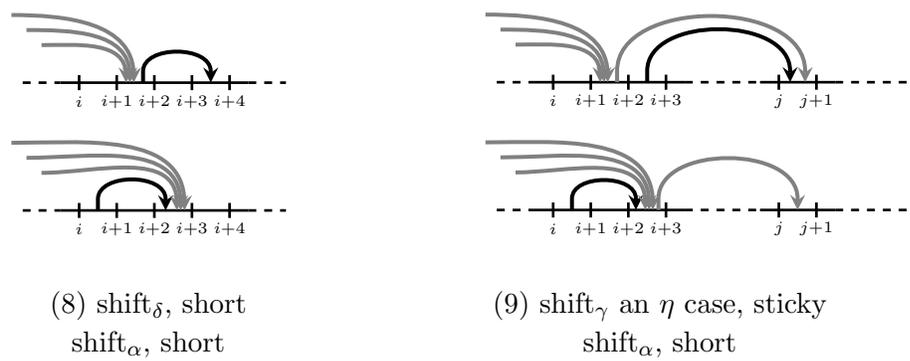
\begin{tabular}{ccc}
\tikz[line width=1]{
\def\dh{1.5}
\foreach \h in {0,\dh}{
    \draw[dashed] (0.25,\h) -- (0.75,\h);
    \draw (0.75,\h) -- (2.25,\h);
    \draw[dashed] (2.25,\h) -- (3.25,\h);
    \draw (3.25,\h) -- (4.25,\h);
    \draw[dashed] (4.25,\h) -- (4.75,\h);
    \foreach \i/\j in {2/i{-}1,3/i,4/i{+}1,7/j,8/j{+}1}{
        \def\x{\i/2}
        \draw (\x,\h-0.1)--(\x,\h+0.1);
        \node at (\x,\h-0.25) {$\scriptscriptstyle\j$};
    }
}
\foreach \i/\j/\h in {3/7/0,2/7/\dh}{
    \def\x{\i/2+0.25}
    \def\y{\j/2+0.25}
    \draw[line width=1.5,-stealth] (\x,\h) -- (\x,\h+0.15) to[in=90,out=90] (\y,\h+0.15) -- (\y,\h);
}
}
& \hspace{0.75cm}
\tikz[line width=1]{
\def\dh{1.5}
\foreach \h in {0,\dh}{
\draw[dashed] (0.25,\h) -- (0.75,\h);
\draw (0.75,\h) -- (1.75,\h);
\draw[dashed] (1.75,\h) -- (2.75,\h);
\draw (2.75,\h) -- (4.25,\h);
\draw[dashed] (4.25,\h) -- (4.75,\h);
\foreach \i/\j in {2/i,3/i{+}1,6/j{-}1,7/j,8/j{+}1}{
\def\x{\i/2}
\draw (\x,\h-0.1)--(\x,\h+0.1);
\node at (\x,\h-0.25) {$\scriptscriptstyle\j$};}}
\foreach \i/\j/\h in {2/6/0,2/7/\dh}{
\def\x{\i/2+0.25}
\def\y{\j/2+0.25}
\draw[line width=1.5,-stealth] (\x,\h) -- (\x,\h+0.15) to[in=90,out=90] (\y,\h+0.15) -- (\y,\h);}
}
\hspace{0.75cm} & 
\tikz[line width=1]{
\def\dh{1.5}
\foreach \h in {0,\dh}{
\draw[dashed] (0.25,\h) -- (0.75,\h);
\draw (0.75,\h) -- (3.25,\h);
\draw[dashed] (3.25,\h) -- (3.75,\h);
\foreach \i/\j in {2/i{-}1,3/i,4/i{+}1,5/i{+}2,6/i{+}3}{
\def\x{\i/2}
\draw (\x,\h-0.1)--(\x,\h+0.1);
\node at (\x,\h-0.25) {$\scriptscriptstyle\j$};}}
\foreach \i/\j/\h in {3/5/0,2/4/\dh}{
\def\x{\i/2+0.25}
\def\y{\j/2+0.25}
\draw[line width=1.5,-stealth] (\x,\h) -- (\x,\h+0.25) to[in=90,out=90] (\y,\h+0.25) -- (\y,\h);}
}
\\[1em]
(1) $\shift_\alpha$ & 
(2) $\shift_\delta$ & 
(3) $\shift_\alpha$ short \\
$\shift_\beta$ &
$\shift_\gamma$ &
$\shift_\delta$ short\\[1em]
\end{tabular}
\end{center}

\begin{center}
\begin{tabular}{ccc}
\hspace{0.5cm}
\tikz[line width=1]{
\def\dh{1.5}
\foreach \h in {0,\dh}{
\draw[dashed] (0.25,\h) -- (0.75,\h);
\draw (0.75,\h) -- (3.25,\h);
\draw[dashed] (3.25,\h) -- (3.75,\h);
\foreach \i/\j in {2/
,3/
,4/c_\star,5/d_\star,6/
}{
\def\x{\i/2}
\draw (\x,\h-0.1)--(\x,\h+0.1);
\node at (\x,\h-0.25) {$\scriptscriptstyle\j$};}}
\foreach \i/\j/\h/\c in {3/4/0/gray,2/4/\dh/black}{
\def\x{\i/2+0.25}
\def\y{\j/2+0.25}
\draw[line width=1.5,-stealth,\c] (\x,\h) -- (\x,\h+0.25) to[in=90,out=90] (\y,\h+0.25) -- (\y,\h);}
}
\hspace{0.5cm} & \hspace{0.5cm}
\tikz[line width=1]{
\def\dh{1.5}
\foreach \h in {0,\dh}{
\draw[dashed] (0.25,\h) -- (0.75,\h);
\draw (0.75,\h) -- (1.75,\h);
\draw[dashed] (1.75,\h) -- (2.75,\h);
\draw (2.75,\h) -- (4.25,\h);
\draw[dashed] (4.25,\h) -- (4.75,\h);
\foreach \i/\j in {2/i,3/i{+}1,6/c_\star,7/d_\star}{
\def\x{\i/2}
\draw (\x,\h-0.1)--(\x,\h+0.1);
\node at (\x,\h-0.25) {$\scriptscriptstyle\j$};}}
\foreach \i/\j/\h/\c in {2/6/0/gray,2/6/\dh/black}{
\def\x{\i/2+0.25}
\def\y{\j/2+0.25}
\draw[line width=1.5,-stealth,\c] (\x,\h) -- (\x,\h+0.15) to[in=90,out=90] (\y,\h+0.15) -- (\y,\h);}
}\hspace{0.5cm}
\\[1em]
(4) $\shift_\alpha$, short blue/purple & 
(5) 
$\shift_\gamma$ for top blue/top purple  \\
$\shift_\delta$ very short red/blue&
 $\shift_\delta$ for bottom red/blue\\[1em]
\end{tabular}
\end{center}

\begin{center}
\begin{tabular}{ccc}
\hspace{0.5cm}
\tikz[line width=1]{
\def\dh{1.5}
\foreach \h in {0,\dh}{
\draw[dashed] (0.25,\h) -- (0.75,\h);
\draw (0.75,\h) -- (1.75,\h);
\draw[dashed] (1.75,\h) -- (2.75,\h);
\draw (2.75,\h) -- (3.75,\h);
\draw[dashed] (3.75,\h) -- (4.75,\h);
\draw (4.75,\h) -- (5.75,\h);
\draw[dashed] (5.75,\h) -- (6.75,\h);
\foreach \i/\j in {2/i,3/i{+}1,6/j,7/j{+}1,10/l,11/l{+}1}{
\def\x{\i/2}
\draw (\x,\h-0.1)--(\x,\h+0.1);
\node at (\x,\h-0.25) {$\scriptscriptstyle\j$};}}
\foreach \i/\j/\h/\c in {2/5.8/0/black,6.2/10/0/gray,6/9.8/\dh/black,2/10.2/\dh/gray}{
\def\x{\i/2+0.25}
\def\y{\j/2+0.25}
\draw[line width=1.5,-stealth,\c] (\x,\h) -- (\x,\h+0.15) to[in=90,out=90] (\y,\h+0.15) -- (\y,\h);}
}
\hspace{0.5cm} & \hspace{0.5cm}
\tikz[line width=1]{
\def\dh{1.5}
\foreach \h in {0,\dh}{
\draw[dashed] (1.25,\h) -- (1.75,\h);
\draw (1.75,\h) -- (3.75,\h);
\draw[dashed] (3.75,\h) -- (4.75,\h);
\draw (4.75,\h) -- (5.75,\h);
\draw[dashed] (5.75,\h) -- (6.75,\h);
\foreach \i/\j in {4/i,5/i{+}1,6/i{+}2,7/i{+}3,10/j,11/j{+}1}{
\def\x{\i/2}
\draw (\x,\h-0.1)--(\x,\h+0.1);
\node at (\x,\h-0.25) {$\scriptscriptstyle\j$};}}
\foreach \i/\j/\h/\c in {4/5.8/0/black,6.2/10/0/gray,6/9.8/\dh/black,5/10.2/\dh/gray}{
\def\x{\i/2+0.25}
\def\y{\j/2+0.25}
\draw[line width=1.5,-stealth,\c] (\x,\h) -- (\x,\h+0.15) to[in=90,out=90] (\y,\h+0.15) -- (\y,\h);}
} \hspace{0.5cm}
\\[1em]
(6) $\shift_\gamma$, an $\eta$ case & 
(7) $\shift_\gamma,$ an $\eta$ case, sticky \\
$\shift_\gamma$, a $\theta$ case&
$\shift_\gamma,$ a $\theta$ case, short \\[1em]
\end{tabular}
\end{center}

\begin{center}
\begin{tabular}{ccc}
\hspace{1cm}
\tikz[line width=1]{
\def\dh{1.7}
\foreach \h in {0,\dh}{
\draw[dashed] (1.25,\h) -- (1.75,\h);
\draw (1.75,\h) -- (4.25,\h);
\draw[dashed] (4.25,\h) -- (4.75,\h);
\foreach \i/\j in {4/i,5/i{+}1,6/i{+}2,7/i{+}3,8/i{+}4}{
\def\x{\i/2}
\draw (\x,\h-0.1)--(\x,\h+0.1);
\node at (\x,\h-0.25) {$\scriptscriptstyle\j$};}}
\foreach \i/\j/\h/\c in {4/5.8/0/black,5.2/7/\dh/black}{
\def\x{\i/2+0.25}
\def\y{\j/2+0.25}
\draw[line width=1.5,-stealth,\c] (\x,\h) -- (\x,\h+0.15) to[in=90,out=90] (\y,\h+0.15) -- (\y,\h);}
\foreach \y in {0.5,0.7,0.9}{
\draw[line width=1.5,-stealth,gray] (2-\y,\y) to[in=90,out=0] (6.35/2+\y/4,0.1) -- (6.35/2+\y/4,0);
\draw[line width=1.5,-stealth,gray] (2-\y,\y+\dh) to[in=90,out=0] (5/2+\y/4,\dh+0.1) -- (5/2+\y/4,\dh);}
}
\hspace{1cm} & \hspace{1cm}
\tikz[line width=1]{
\def\dh{1.7}
\foreach \h in {0,\dh}{
\draw[dashed] (1.25,\h) -- (1.75,\h);
\draw (1.75,\h) -- (3.75,\h);
\draw[dashed] (3.75,\h) -- (4.75,\h);
\draw (4.75,\h) -- (5.75,\h);
\draw[dashed] (5.75,\h) -- (6.75,\h);
\foreach \i/\j in {4/i,5/i{+}1,6/i{+}2,7/i{+}3,10/j,11/j{+}1}{
\def\x{\i/2}
\draw (\x,\h-0.1)--(\x,\h+0.1);
\node at (\x,\h-0.25) {$\scriptscriptstyle\j$};}}
\foreach \i/\j/\h/\c in {4/5.7/0/black,6.3/10/0/gray,6/9.8/\dh/black,5.2/10.2/\dh/gray}{
\def\x{\i/2+0.25}
\def\y{\j/2+0.25}
\draw[line width=1.5,-stealth,\c] (\x,\h) -- (\x,\h+0.15) to[in=90,out=90] (\y,\h+0.15) -- (\y,\h);}
\foreach \y in {0.5,0.7,0.9}{
\draw[line width=1.5,-stealth,gray] (2-\y,\y) to[in=90,out=0] (6.2/2+\y/4,0.1) -- (6.2/2+\y/4,0);
\draw[line width=1.5,-stealth,gray] (2-\y,\y+\dh) to[in=90,out=0] (5/2+\y/4,\dh+0.1) -- (5/2+\y/4,\dh);}
}
\hspace{1cm}
\\[1em]
(8) $\shift_\delta$, short & 
(9) $\shift_\gamma$ an $\eta$ case, sticky \\
$\shift_\alpha$, short &
$\shift_\alpha$, short \\[1em]
\end{tabular}
\end{center}
\begin{center}
\begin{tabular}{ccc}
\hspace{1cm}
\tikz[line width=1]{
\def\dh{1.7}
\foreach \h in {0,\dh}{
\draw[dashed] (1.25,\h) -- (1.75,\h);
\draw (1.75,\h) -- (4.25,\h);
\draw[dashed] (4.25,\h) -- (4.75,\h);
\foreach \i/\j in {4/i,5/i{+}1,6/i{+}2,7/i{+}3,8/i{+}4}{
\def\x{\i/2}
\draw (\x,\h-0.1)--(\x,\h+0.1);
\node at (\x,\h-0.25) {$\scriptscriptstyle\j$};}}
\foreach \i/\j/\h/\c in {4.8/5.8/0/blue,5.2/7/\dh/black}{
\def\x{\i/2+0.25}
\def\y{\j/2+0.25}
\draw[line width=1.5,-stealth,\c] (\x,\h) -- (\x,\h+0.15) to[in=90,out=90] (\y,\h+0.15) -- (\y,\h);}
\foreach \y in {0.5,0.7,0.9}{
\draw[line width=1.5,-stealth,gray] (2-\y,\y) to[in=90,out=0] (6.35/2+\y/4,0.1) -- (6.35/2+\y/4,0);
\draw[line width=1.5,-stealth,gray] (2-\y,\y+\dh) to[in=90,out=0] (5/2+\y/4,\dh+0.1) -- (5/2+\y/4,\dh);}
}
\hspace{1cm} & \hspace{1cm}
\tikz[line width=1]{
\def\dh{1.7}
\foreach \h in {0,\dh}{
\draw[dashed] (1.25,\h) -- (1.75,\h);
\draw (1.75,\h) -- (3.75,\h);
\draw[dashed] (3.75,\h) -- (4.75,\h);
\draw (4.75,\h) -- (5.75,\h);
\draw[dashed] (5.75,\h) -- (6.75,\h);
\foreach \i/\j in {4/i,5/i{+}1,6/i{+}2,7/i{+}3,10/j,11/j{+}1}{
\def\x{\i/2}
\draw (\x,\h-0.1)--(\x,\h+0.1);
\node at (\x,\h-0.25) {$\scriptscriptstyle\j$};}}
\foreach \i/\j/\h/\c in {4.8/5.7/0/blue,6.3/10/0/gray,6/9.8/\dh/black,5.2/10.2/\dh/gray}{
\def\x{\i/2+0.25}
\def\y{\j/2+0.25}
\draw[line width=1.5,-stealth,\c] (\x,\h) -- (\x,\h+0.15) to[in=90,out=90] (\y,\h+0.15) -- (\y,\h);}
\foreach \y in {0.5,0.7,0.9}{
\draw[line width=1.5,-stealth,gray] (2-\y,\y) to[in=90,out=0] (6.2/2+\y/4,0.1) -- (6.2/2+\y/4,0);
\draw[line width=1.5,-stealth,gray] (2-\y,\y+\dh) to[in=90,out=0] (5/2+\y/4,\dh+0.1) -- (5/2+\y/4,\dh);}
}
\hspace{1cm}
\\[1em]
(10) $\shift_\delta$, short & 
(11) $\shift_\gamma$, an $\eta$ case, sticky \\
$\shift_\varepsilon$, very short &
$\shift_\varepsilon$, very short \\[1em]
\end{tabular}
\end{center}

\begin{center}
\begin{tabular}{ccc}
\hspace{0.5cm}
\tikz[line width=1]{
\def\dh{1.5}
\foreach \h in {0,\dh}{
\draw[dashed] (0.25,\h) -- (0.75,\h);
\draw (0.75,\h) -- (1.75,\h);
\draw[dashed] (1.75,\h) -- (2.75,\h);
\draw (2.75,\h) -- (3.75,\h);
\draw[dashed] (3.75,\h) -- (4.75,\h);
\draw (4.75,\h) -- (5.75,\h);
\draw[dashed] (5.75,\h) -- (6.75,\h);
\foreach \i/\j in {2/a_0,3/b_0,6/j,7/j{+}1,10/l,11/l{+}1}{
\def\x{\i/2}
\draw (\x,\h-0.1)--(\x,\h+0.1);
\node at (\x,\h-0.25) {$\scriptscriptstyle\j$};}}
\foreach \i/\j/\h/\c in {2/5.8/0/red,6.2/10/0/black,6/9.8/\dh/blue,2/10.2/\dh/red}{
\def\x{\i/2+0.25}
\def\y{\j/2+0.25}
\draw[line width=1.5,-stealth,\c] (\x,\h) -- (\x,\h+0.15) to[in=90,out=90] (\y,\h+0.15) -- (\y,\h);}
}
\hspace{0.5cm} & \hspace{0.5cm}
\tikz[line width=1]{
\def\dh{1.5}
\foreach \h in {0,\dh}{
\draw[dashed] (1.25,\h) -- (1.75,\h);
\draw (1.75,\h) -- (3.75,\h);
\draw[dashed] (3.75,\h) -- (4.75,\h);
\draw (4.75,\h) -- (5.75,\h);
\draw[dashed] (5.75,\h) -- (6.75,\h);
\foreach \i/\j in {4/i,5/i{+}1,6/i{+}2,7/i{+}3,10/j,11/j{+}1}{
\def\x{\i/2}
\draw (\x,\h-0.1)--(\x,\h+0.1);
\node at (\x,\h-0.25) {$\scriptscriptstyle\j$};}}
\foreach \i/\j/\h/\c in {5/5.8/0/black,6.2/10/0/blue,6/9.8/\dh/blue,5/10.2/\dh/black}{
\def\x{\i/2+0.25}
\def\y{\j/2+0.25}
\draw[line width=1.5,-stealth,\c] (\x,\h) -- (\x,\h+0.15) to[in=90,out=90] (\y,\h+0.15) -- (\y,\h);}
} \hspace{0.5cm}
\\[1em]
(12) $\shift_\varepsilon$, nonsticky & 
(13) $\shift_{\varepsilon}$, sticky \\
$\shift_{\gamma_\star}$, a $\theta$ case&
$\shift_{\gamma_\star},$ a $\theta$ case, very short\\[1em]
\end{tabular}
\end{center}

\begin{center}
\begin{tabular}{cc}
\hspace{1cm}
\tikz[line width=1]{
\def\dh{1.7}
\foreach \h in {0,\dh}{
\draw[dashed] (1.25,\h) -- (1.75,\h);
\draw (1.75,\h) -- (4.25,\h);
\draw[dashed] (4.25,\h) -- (4.75,\h);
\foreach \i/\j in {4/i,5/i{+}1,6/i{+}2,7/i{+}3,8/i{+}4}{
\def\x{\i/2}
\draw (\x,\h-0.1)--(\x,\h+0.1);
\node at (\x,\h-0.25) {$\scriptscriptstyle\j$};}}
\foreach \y in {0.5}{
\draw[line width=1.5,-stealth,blue] (2-\y,\y) to[in=90,out=0] (6.35/2+\y/4,0.1) -- (6.35/2+\y/4,0);
\draw[line width=1.5,-stealth,blue] (2-\y,\y+\dh) to[in=90,out=0] (5/2+\y/4,\dh+0.1) -- (5/2+\y/4,\dh);}
\foreach \y in {0.7,0.9}{
\draw[line width=1.5,-stealth,red] (2-\y,\y) to[in=90,out=0] (6.35/2+\y/4,0.1) -- (6.35/2+\y/4,0);
\draw[line width=1.5,-stealth,red] (2-\y,\y+\dh) to[in=90,out=0] (5/2+\y/4,\dh+0.1) -- (5/2+\y/4,\dh);}
}
\hspace{1cm} & \hspace{1cm}
\tikz[line width=1]{
\def\dh{1.7}
\foreach \h in {0,\dh}{
\draw[dashed] (1.25,\h) -- (1.75,\h);
\draw (1.75,\h) -- (3.75,\h);
\draw[dashed] (3.75,\h) -- (4.25,\h);
\foreach \i/\j in {4/i,5/i{+}1,6/i{+}2,7/i{+}3}{
\def\x{\i/2}
\draw (\x,\h-0.1)--(\x,\h+0.1);
}
\foreach \y in {0.5}{
\draw[line width=1.5,-stealth,blue] (2-\y,\y) to[in=90,out=0] (5/2+\y/4,0.1) -- (5/2+\y/4,0);
\draw[line width=1.5,-stealth,blue] (2-\y,\y+\dh) to[in=90,out=0] (5/2+\y/4,\dh+0.1) -- (5/2+\y/4,\dh);}
\foreach \y in {0.7}{
\draw[line width=1.5,-stealth,red] (2-\y,\y) to[in=90,out=0] (5/2+\y/4,0.1) -- (5/2+\y/4,0);
\draw[line width=1.5,-stealth,red] (2-\y,\y+\dh) to[in=90,out=0] (5/2+\y/4,\dh+0.1) -- (5/2+\y/4,\dh);}
\foreach \y in {0.9}{
\draw[line width=1.5,-stealth,red] (2-\y,\y) to[in=90,out=0] (5/2+\y/4,0.1) -- (5/2+\y/4,0);
\draw[line width=1.5,-stealth,black] (2-\y,\y+\dh) to[in=90,out=0] (5/2+\y/4,\dh+0.1) -- (5/2+\y/4,\dh);}
}}
\hspace{1cm} 
\\[1em]
(14) $\shift_{\gamma_\star}$& 
(15) $\shift_{\gamma_\star}$ an $\eta$ case \\
$\shift_{\delta_\star}$&
$\shift_\abe$ \\[1em]
\end{tabular}
\end{center}
\end{rmk}
\begin{definition}\label{def:internal_bdries_domino}
Let $\D$ be a chord diagram, and $\bullet\in\{\alpha_i,\ldots,\varepsilon_i\}_{i\in[k+1]_\D}\cup\{\gamma_\star,\delta_\star,\alpha_0,\beta_0\}$ is such that $\shift_\bullet\D$ is defined. Then, if $\bullet\neq\delta_0,\gamma_i$ we define $S_{\partial_\bullet\D}$ to be the subspace of $\Gr_{k,n;\ell}$ whose element have a weak domino representative for which all inequalities implied in \cref{obs:poly_ineqs} hold, except that $\bullet=0.$ If $\bullet = \gamma_i$ then exactly one of $\gamma_i,\eta_i,\theta_i\in\Var_\D$. Denote it by $\zeta_i.$ Then if the shift at $\gamma_i$ is defined, $S_{\partial_{\gamma_i}\D}$ is defined to be the subspace of $\Gr_{k,n;\ell}$ whose element have a weak domino representative for which all inequalities of \cref{obs:poly_ineqs} hold, except that $\zeta_i=0.$
If $\bullet=\delta_0,$ $S_{\partial_\bullet\D}$ is defined to be $S_{\partial_{\gamma_\star}\shift_{\abe}\D}.$   
We also use the notation $\partial_\bullet S_\D=S_{\partial_\bullet\D}.$

For reasons that will be clarified later, we refer to the above strata as the \emph{matchable} boundary strata of $S_\D.$ $S_{\partial_{\varepsilon_h}}\D,$ where $\D_h$ is the lowest blue chord is called the \emph{blue boundary of $S_\D.$} $S_{\partial_{\delta_0}}\D,$ is called the \emph{red boundary of $S_\D.$} The other matchable boundaries are called \emph{uncolored}.
\end{definition}
\begin{rmk}
In \cite{even2021amplituhedron} both the notations for shifts and for boundaries were slightly different. In particular, the notation there for shifts have indicated which part of chord is being shifted. The notations we use here are more compact, and allow stating claims more uniformly, but are slightly less suggestive of the chord diagrammatic meaning. 
\end{rmk}
\begin{nn}\label{nn:reg_match}Matchable boundaries which are not contained in $\SA$ 
are called the \emph{internal matchable boundaries}. The uncolored internal matchable boundaries are called \emph{regular matchable boundaries}. 
\end{nn}
\begin{rmk}\label{rmk:all_matchable_reg}
It can be shown, using the techniques of this text, that all matchable boundaries of a BCFW cell $S_\D,$ except the red and blue boundaries, are regular, but we will not rely on it.
\end{rmk}
\subsubsection{A BCFW-like description of $S_{\partial\bullet}$}
It is evident from the definition of the matchable boundaries, the description of the domino form given in Section \cref{subsec:BCFW_form} and the passage between domino and BCFW forms \cref{cor:BCFW2dominoExplicit} that an alternative description of the space $S_{\partial\zeta_i}$ is as the subspace of $\Gr_{k,n;1}^\geq$ consisting of points $C\vdots D$ having a BCFW form, with the exception that the parameter $\bzet_i$ is set to $0.$
It will also be instrumental to present a construction of this space which is similar to that of BCFW cells. To this end we define degenerate versions of the BCFW step and the forward limit.
\begin{definition}\label{def:degenerate_ops}
Using \cref{nn:bcfwmap} and the notations of \cref{def:bcfw-map}, we define the \emph{BCFW map degenerated at $\zeta,$ for $\zeta\in\{\alpha,\ldots,\varepsilon\}$} as follows. 
\begin{itemize}
\item \underline{$\zeta=\gamma:$} We define
\[\mbcfw^{\gamma} :\;{\Mat}_{k_L, N_L;\ell_L}\times\; \R^4\;\times\;{\Mat}_{k_R,N_R,\ell_R}\rightarrow\;\Mat_{k,n;\ell}\]
where $((C_L\vdots\{D_p\}_{p\in{\Loop}_L}\}), (\alpha, \beta, \delta,\varepsilon), (C_R\vdots\{D_q\}_{q\in\Loop_R}))$ is mapped to
$(C\vdots\{D'_r\}_{r\in\Loop_L\cup\Loop_R}),$
defined as in \cref{def:bcfw-map}, only that we require \[C'_R=\pre_{N_L\setminus\{a\}}.y_d(\frac{\delta}{\varepsilon}).C_R,~D'_r=\pre_{N_L\setminus\{a\}}.y_d(\frac{\delta}{\varepsilon}).D_r,~\text{for }r\in\Loop_R,\]and the $c$th entry of $v$ is put to $0.$
\item \underline{$\zeta=\alpha$:} We define 
    \[\mbcfw^{\alpha} :\;{\Mat}_{k_L, N_L;\ell_L}\times\; \R^4\;\times\;\pre_{b}{\Mat}_{k_R,N_R\setminus\{b\},\ell_R}\rightarrow\;\Mat_{k,n;\ell}\]
where $((C_L\vdots\{D_p\}_{p\in{\Loop}_L}\}), ( \beta,\gamma ,\delta,\varepsilon), (C_R\vdots\{D_q\}_{q\in\Loop_R}))$ is mapped to
$(C\vdots\{D'_r\}_{r\in\Loop_L\cup\Loop_R}),$
defined as in \cref{def:bcfw-map} only that \[C'_L=\pre_{N_R\setminus\{a,b,n\}}.C_L,~D'_r=\pre_{N_R\setminus\{a,b,n\}}.D_r,~\text{for }r\in\Loop_L,\]
and the $a$th entry of $v$ is put to $0.$
\item\underline{$\zeta=\beta:$}
We define 
    \[\mbcfw^{\beta} :\;\pre_{b}{\Mat}_{k_L, N_L\setminus\{b\};\ell_L}\times\; \R^4\;\times\;{\Mat}_{k_R,N_R,\ell_R}\rightarrow\;\Mat_{k,n;\ell}\]
where $((C_L\vdots\{D_p\}_{p\in{\Loop}_L}\}), (\alpha,\gamma, \delta,\varepsilon), (C_R\vdots\{D_q\}_{q\in\Loop_R}))$ is mapped to
$(C\vdots\{D'_r\}_{r\in\Loop_L\cup\Loop_R}),$
defined as in \cref{def:bcfw-map} only that 
\[C'_L=\pre_{N_R\setminus\{a,b,n\}}.C,~D'_r=\pre_{N_R\setminus\{a,b,n\}}.D_r,~\text{for }r\in\Loop_L,\]and the $b$th entry of $v$ is put to $0.$
\item\underline{$\zeta = \delta:$}
We define 
    \[\mbcfw^{\delta} :\;{\Mat}_{k_L, N_L;\ell_L}\times\; \R^4\;\times\;\pre_d{\Mat}_{k_R,N_R\setminus\{d\},\ell_R}\rightarrow\;\Mat_{k,n;\ell}\]
where $((C_L\vdots\{D_p\}_{p\in{\Loop}_L}\}), (\alpha,\beta,\gamma,\varepsilon), (C_R\vdots\{D_q\}_{q\in\Loop_R}))$ is mapped to
$(C\vdots\{D'_r\}_{r\in\Loop_L\cup\Loop_R}),$
defined as in \cref{def:bcfw-map} only that 
\[C'_R=\pre_{N_L\cup\{d\}\setminus\{a\}}.y_c(\frac{\gamma}{\varepsilon}).\rem_d(C_R),~D'_r=\pre_{N_L\setminus\{a\}}.y_d(\frac{\gamma}{\varepsilon}).\rem_d(D_r),~\text{for }r\in\Loop_R,\] and the $d$th entry of $v$ is put to $0.$
\item\underline{$\zeta=\varepsilon:$}
We define 
    \[\mbcfw^{\varepsilon} :\;{\Mat}_{k_L, N_L;\ell_L}\times\; \R^4\;\times\;\pre_d{\Mat}_{k_R,N_R\setminus\{d\},\ell_R}\rightarrow\;\Mat_{k,n;\ell}\]
where $((C_L\vdots\{D_p\}_{p\in{\Loop}_L}\}), (\alpha,\beta,\gamma,\delta), (C_R\vdots\{D_q\}_{q\in\Loop_R}))$ is mapped to
$(C\vdots\{D'_r\}_{r\in\Loop_L\cup\Loop_R}),$
defined as in \cref{def:bcfw-map} with the following difference: 
\[C'_R = \pre_ny_{c}(\frac{\gamma}{\delta})(\rem_n(y_d(\delta).C_R)),~~D'_r=\pre_ny_{c}(\frac{\gamma}{\delta})(\rem_n(y_d(\delta).C_R)),~r\in\Loop_R.\]
The $n$th entry of $v$ is put to $0.$
    \end{itemize}
For $\zeta\in\{\gamma,\delta\}$ we define the \emph{forward limit degenerated at $\zeta$}, and denote it by $\FL^\zeta,$ as follows. For $\zeta=\gamma$ it is just the map $\FL'$ defined in \cref{rmk:last_y_c}.
For $\zeta=\delta$ we use the notations of \cref{def:forward_limit}, and define
\[\FL^{\delta}:\Gr'_{k+1,N}\times\Gr_{1,4}\dashrightarrow\Gr_{k,[n];1}\]
as follows.
Write $M_0'=\rem_d(y_{c}(\gamma)M_0)$,
    and set $M'_1$ to be    \[M'_1=\pre_d.y_c(\frac{\gamma_\star}{\varepsilon_\star}).x_A(\frac{\beta_\star}{\alpha_\star}).x_l(\frac{\alpha_\star}{\varepsilon_\star}).M'_0.\] The construction continues as in \cref{def:forward_limit}:\[\FL^\delta(M)=\rem_l(x_l(1)
(\rem_{A,B}(\addL_{AB}M'_1)))).\]
\end{definition}
In analogy to \cref{prop:BCFW_and_pos},~\cref{prop:fl_domain} and with similar proofs, we have 
\begin{prop}\label{prop:degen_ops_domain}
If $(\ell_L,\ell_R)\in\{(0,0),(0,1),(1,0)\},$ and $k_L\leq |N_L|-2,~k_R\leq |N_R|-2,$ then
the maps $\mbcfw^\zeta,~\zeta\in\{\alpha,\ldots,\varepsilon\}$ descend to rational maps (denoted using the same notations) between the corresponding Grassmannians and nonnegative Grassmannians. They are called the \emph{BCFW product degenerated at $\zeta$}.
If if $S_L\subseteq \Gr_{k_L,N_L;\ell_L}^{\geq},S_R\subseteq \Gr_{k_R,N_R;\ell_R}^{\geq}$ we write $S_L\bcfw^\zeta S_R$ for the image of
$\mbcfw(S_L,\Gr_{1,5}^{>},S_R)$ whenever it is defined.

The map $\FL^\zeta$ is generically defined, and is independent of choices.
Moreover, it restricts to a rational map, and preserves non negativity.
\end{prop}
\begin{prop}\label{prop:boundaries_domain}
Let $\CD_{n,k}^\ell$ be a chord diagram. Every matchable boundary $S_{\partial_{\zeta_i}\D}$, but the red one, if exists, can be constructed via the following algorithm. We construct $\partial_{\zeta_i}S_\D$ by the same recipe of $\pre,\bcfw,\FL$ as $S_\D,$ only that at the $i$th BCFW step corresponding to $\D_i,$ if $i\neq\star,$ or in the $\FL$ step, if $i=\star,$ we perform the BCFW map degenerated at $\hzet_i,$ or the forward limit map degenerated at $\hzet_\star,$ respectively. These spaces are contained in the closure $\overline{S}_\D$ in the Hausdorff topology.
\end{prop}
\begin{proof}
In all these cases the weak domino form differs from the $\D$-domino form by setting $\tilde\zeta_i=0,$ where $\tilde\zeta_i=\zeta_i,$ unless $\zeta_i=\gamma_i$ and $\eta_i$ or $\theta_i$ exist - in which case $\tilde\zeta_i$ is this $\eta_i$ or $\theta_i.$ 
In all regular cases, except the case $\zeta_i=\delta_\star,$ from \cref{cor:BCFW2dominoExplicit} and the discussion of \cref{subsec:BCFW_form}, that setting $\tilde\zeta_i=0$ is equivalent to putting $\bzet_i=\hzet_i=0$ in the BCFW form, which is equivalent to constructing the cell using the usual BCFW recipe, except that the operation corresponding to $\D_i$ is degenerated at $\zeta_i.$ In the remaining two cases $\zeta_i=\delta_\star$ and $\zeta_i=\varepsilon_h,$ where $\D_h$ is the lowest blue chord, it is easily checked by hand that the two descriptions for the spaces agree, by the same argument used in \cref{subsec:BCFW_form}.

Denote by $(U\vdots V)_{\partial_{\zeta_i}\D}(\{\hgam_\star,\hdel_\star,
\heps_\star\}\cup\{\halp_0,\hbet_0,
\heps_0\}\cup\bigcup_{j=1}^k\{\halp_j,\hbet_j,\hgam_j,\hdel_j,\heps_j\}\setminus\{\hzet_i\})$ the loopy vector space in $S_{\partial_{\zeta_i}\D}$ obtained by performing the suggested algorithm for the space $\partial_{\zeta_i}S_\D$ with the written inputs. 
Note that the set of inputs for which this vector loopy space is of full rank is open.
Denote by $(U\vdots V)_{\D}(\{\hgam_\star,\hdel_\star,
\heps_\star\}\cup\{\halp_0,\hbet_0,
\heps_0\}\cup\bigcup_{j=1}^k\{\halp_j,\hbet_j,\hgam_j,\hdel_j,\heps_j\})$ the loopy vector space in $S_\D$ obtained by the standard BCFW process, with the same inputs and $\hzet_i.$ Then in the regular cases other than $\zeta_i=\delta_\star$ we can take the limit directly in coordinates and get
\begin{align*}(U\vdots V)_{\partial_{\zeta_i}\D}&(\{\hgam_\star,\hdel_\star,
\heps_\star\}\cup\{\halp_0,\hbet_0,
\heps_0\}\cup\bigcup_{j=1}^k\{\halp_j,\hbet_j,\hgam_j,\hdel_j,\heps_j\}\setminus\{\hzet_i\})
\\&=
\lim_{\hzet_i\to 0}(C\vdots D)_{\D}(\{\hgam_\star,\hdel_\star,
\heps_\star\}\cup\{\halp_0,\hbet_0,
\heps_0\}\cup\bigcup_{j=1}^k\{\halp_j,\hbet_j,\hgam_j,\hdel_j,\heps_j\}),\end{align*}as long as the left hand side is of full rank.
In the remaining two cases we should take the limit more carefully. Still, for $\hdel_\star$ with $\Purp_\D=\emptyset$, a simple calculation shows
\begin{align*}
(U\vdots V)_{\partial_{\delta_\star}\D}&(\{\hgam_\star,
\heps_\star\}\cup\{\halp_0,\hbet_0,
\heps_0\}\cup\bigcup_{j=1}^k\{\halp_j,\hbet_j,\hgam_j,\hdel_j,\heps_j\})
\\&=
\lim_{\substack{\hdel_\star\to 0\\\tilde\gamma_j/\delta_\star\to\hgam_j,~\text{if $\D_j$ is blue}\\\tilde\gamma_j=\hgam_j~\text{otherwise}}}(U\vdots V)_{\D}(\{\hgam_\star,\hdel_\star,
\heps_\star\}\cup\{\halp_0,\hbet_0,
\heps_0\}\cup
\bigcup_{j=1}^k
\{\halp_j,\hbet_j,\hgam_j,\hdel_j,
\heps_j
\}).
\end{align*}
When $\zeta_i=\varepsilon_h$ where $\D_h$ is the lowest blue chord, a direct calculation reveals that
\begin{align*}
(U\vdots V)_{\partial_{\varepsilon_h}\D}&(\{\hgam_\star,\hdel_\star,
\heps_\star\}\cup\{\halp_0,\hbet_0,
\heps_0\}\cup\bigcup_{j=1}^k\{\halp_j,\hbet_j,\hgam_j,\hdel_j,\heps_j\}\setminus\{\heps_h\})
\\&=
\lim_{\substack{\tilde{\varepsilon}_i\to 0\\\tilde\varepsilon_j/\tilde{\varepsilon}_i\to\heps_j,~\text{if $\D_j$ is a descendent of $\D_h$}\\\tilde\varepsilon_j=\heps_j~\text{otherwise}}}
(U\vdots V)_{\D}(\{\hgam_\star,\hdel_\star,
\heps_\star\}\cup\{\halp_0,\hbet_0,
\heps_0\}\cup
\bigcup_{j=1}^k
\{\halp_j,\hbet_j,\tilde{\gamma}_j,\hdel_j,
\heps_j
\}).
\end{align*}
Again the two equalities above hold as long as the left hand sides are of full ranks.
Thus, these spaces are indeed contained in the closure of the BCFW cell.
\end{proof}
\begin{rmk}
    It requires some work, but one can actually show that whenever the parameters we use are positive the outputs of the degenerate BCFW algorithms always have full rank.
\end{rmk}
\cref{prop:boundaries_domain} and the definition of neighboring boundaries have the following consequence.
\begin{obs}\label{obs:minimal_neighboring}
Let $\D_1,\D_2$ be two neighboring chord diagrams. Then there exist two proper subdiagrams $\D'_1,\D'_2$ of $\D_1,\D_2$ respectively, which are also neighboring, no proper subdiagrams of $\D'_1$ and $\D'_2$ are neighboring, and the generation sequence for passing from $\D'_1$ to $\D_1$ involves exactly the same steps and the generation sequence for passing from $\D'_2$ to $\D_2$. The last property is equivalent to requiring that diagrams obtained from $D_i$ by erasing the chords of $\D'_i$ are the same (including colors). 
Moreover, the rightmost top chord of $\D'_1.\D'_2$ is different.
\end{obs}
\begin{definition}\label{def:minimal_neighboring}
    In the scenario of \cref{obs:minimal_neighboring}, the subdiagrams $\D'_1,\D'_2$ are called the \emph{minimal neighboring subdiagrams of $\D_1,\D_2$}.
\end{definition}
\subsubsection{Matching strata}
The shift operation and the matching are justified by the following lemma.
\begin{lemma}\label{lem:bdry_matching_geo}
    Let $\partial_{\zeta_i}S_\D$ be a matchable boundary. Assume that $\match(\D,\zeta_i)=(\D',\zeta'_j).$ Then $\partial_{\zeta_i}S_\D=\partial_{\zeta'_j}S_{\D'}.$
\end{lemma}
This lemma generalizes \cite[Lemma 7.17]{even2021amplituhedron}. It provides another evidence for how chord diagrams are suited for labeling BCFW cells, and how the combinatorics of the former govern the geometry of the latter.
\begin{proof}
    We will prove the lemma by direct calculation. One can do it either in the level of degenerate operations, where it is implied by several commutation relations, or in the level of weak domino forms. We will treat all cases except for the blue boundary and the case that a $\gamma_\star$-boundary is matched to a $\delta_\star$-boundary by showing commutation relations between the degenerate operators, and treat the remaining two cases by working in the domino form, since the calculations will be shorter there. There is nothing to prove in the case of the red boundary, since there this equality is the definition.
    \\\textbf{Commutations of degenerate operators.}
We will rely on the following lemma, whose proof is given below.
\begin{lemma}\label{lem:boundary_matching_operators}
\begin{enumerate}
\item\label{it:unobstracted start non short}
Let $b'$ be the element which succeeds $b.$
\[(\mbcfw^{\alpha})_{a,c}(L,[\zeta:\gamma:\delta:\varepsilon],\pre_b.R)=(\mbcfw^\beta)_{b,c}(\pre_{b'}.L,[\zeta:\gamma:\delta:\varepsilon],R).\]
\item\label{it:unobstracted end non short}
\[\pre_{n-1}(\mbcfw^{\gamma})_{a,c}(L,[\alpha:\beta:\zeta:\varepsilon],R)=(\mbcfw^\beta)_{a,d}(\pre_b.L,[\alpha:\beta:\zeta:\varepsilon],\pre_p.R).\]
\item\label{it:unobstracted short}
Assume now that $a\lessdot b\lessdot c\lessdot d.$ In this case the third input of $\mbcfw^\zeta$ is null.
\[\pre_{n-1}(\mbcfw^{\alpha})_{a,c}(L,[0:\beta:\gamma:\delta:\varepsilon])=(\mbcfw^\delta)_{b,d}(pre_c.L,[\beta:\gamma:\delta:\varepsilon]).\]
\end{enumerate}
In the cases below the passage between the left and right hand sides require an invertible rational transformation that takes positive data to positive data, and is explicated in the proof below. We denote the variables on the right hand, if they do not agree with those on the left, by a tilde notation. 
\begin{enumerate}[resume]
\item\label{it:obstructed end 2 step}
\begin{align*}&(\mbcfw)_{ce}\big((\mbcfw^\gamma)_{ac}(L,[\alpha_1:\beta_1:\delta_1:\varepsilon_1],M),~[\alpha_2:\beta_2:\gamma_2:\delta_2:\varepsilon_2],~R\big)=\\
&\quad\qquad=
(\mbcfw)_{ae}\big(L,~[{\talpha}_1:\tbeta_1:\tgamma_1:\tdelta_1:\tepsilon_1],~(\mbcfw^\gamma)_{ce}(\tilde{M},[\talpha_2:\tbeta_2:\tdelta_2:\tepsilon_2],R)\big)
\end{align*}\item\label{it:obstructed short 2 step}
Assume $a\lessdot b \lessdot c \lessdot d.$
\begin{align*}&(\mbcfw)_{ce}\big((\mbcfw^\alpha)_{ac}(L,[0:\beta_1:\gamma_1:\delta_1:\varepsilon_1]),~[\alpha_2:\beta_2:\gamma_2:\delta_2:\varepsilon_2],~R\big)=\\
&\quad\qquad=
(\mbcfw)_{be}\big(\pre_c.L,[\talpha_1:\tbeta_1:\tgamma_1:\tdelta_1:\tepsilon_1],(\mbcfw^\gamma)_{ec}([\talpha_2:\tbeta_2:\tdelta_2:\tepsilon_2],R)\big)
\end{align*}
\item\label{it:intermideate_short_case}Assume $a<b<c\lessdot d \lessdot e \lessdot f.$
\begin{align*}&(\mbcfw)_{ce}^\delta\big((\mbcfw)_{a,c}\big(L, [\alpha_1:\beta_1:\gamma_1:\delta_1:\varepsilon_1],R),
[\alpha_2:\beta_2:\gamma_2:0:\varepsilon_2]
\big)=
\\&=\pre_f.(\mbcfw)_{ad}\big(L,[\talpha_1:\tbeta_1:\tgamma_1:\tdelta_1:\tepsilon_1],((\mbcfw)_{ce}^\delta(\tilde{R},[\talpha_2:\tbeta_2:\tgamma_2:0:\tepsilon_2])_{\setminus f})
\big)
\end{align*}
where the subscript $_{\setminus f}$ means erasing the (zero) column $f.$
\item\label{it:intermediate_short_2}
Assume $a\lessdot b \lessdot c \lessdot d<e<f\lessdot n.$
\begin{align*}
&(\mbcfw)_{be}\big((\mbcfw^\delta)_{ac}(L,[\alpha_1:\beta_1:\gamma_1:0:\varepsilon_1])_{\setminus d},~[\alpha_2:\beta_2:\gamma_2:\delta_2:\varepsilon_2],~R\big)=\\
&\quad\qquad=
(\mbcfw)_{ae}\big(L,[\talpha_1:\tbeta_1:\tgamma_1:\tdelta_1:\tepsilon_1],(\mbcfw^\gamma)_{be}([\talpha_2:\tbeta_2:0:\tdelta_2:\tepsilon_2],R)\big)
\end{align*}
\item\label{it:chain_with_short_1}
Assume $a_h\lessdot b_h\leq a_{h-1}\lessdot b_{h-1}\ldots\leq a_1\lessdot b_1\leq c'\lessdot c\lessdot d\lessdot e\lessdot f<n,$ and let \[L_h\in\Mat_{k_h,[b_h]\cup\{n\};\ell_h},~L_{h-1}\in\Mat_{k_{h-1},[b_{h-1}]\cup\{n\};\ell_{h-1}},\ldots, L_1\in \Mat_{k_1,[b_1]\cup\{n\};\ell_1},~R\in\Mat_{k_0,\{b_1,\ldots,c,n\};\ell_0},\]
where every $\ell_i\in\{0,1\}$ and also $\sum\ell_i\in\{0,1\}.$
It holds that
\begin{align*}
(\mbcfw^\delta)_{ce}&\big\{(\mbcfw)_{a_hc}\big(
L_h,[\alpha_h:\beta_h:\gamma_h:\delta_h:\varepsilon_h],\\
&\quad(\mbcfw)_{a_{h-1}c}\big(
L_{h-1},[\alpha_{h-1}:\beta_{h-1}:\gamma_{h-1}:\delta_{h-1}:\varepsilon_{h-1}],
\cdots
,\pre_dR\big)\cdots
\big),
[\alpha_0:\beta_0:\gamma_0:0:\varepsilon_0]\\\
&=(\mbcfw)_{a_hd}\big(
\tilde{L}_h,[\talpha_h:\tbeta_h:\tgamma_h:\tdelta_h:\tepsilon_h],\\&\quad
(\mbcfw)_{a_{h-1}d}\big(
\tilde{L}_{h-1},[\talpha_{h-1}:\tbeta_{h-1}:\tgamma_{h-1}:\tdelta_{h-1}:\tepsilon_{h-1}],
\cdots
,(\mbcfw^\alpha)_{c'd}(\tilde{R},[0:\tbeta_0:\tgamma_0:\tdelta_0:\tepsilon_0])\big)\cdots
\big)
\end{align*}
\item\label{it:chain_with_short_2}
Assume $a_h\lessdot b_h\leq a_{h-1}\lessdot b_{h-1}\ldots\leq a_1\lessdot b_1\leq c'\lessdot c\lessdot d\lessdot e<f\lessdot g<n,$ and let \[L_h\in\Mat_{k_h,[b_h]\cup\{n\};\ell_h},~L_{h-1}\in\Mat_{k_{h-1},[b_{h-1}]\cup\{n\};\ell_{h-1}},\ldots, L_1\in \Mat_{k_1,[b_1]\cup\{n\};\ell_1},\]\[M\in\Mat_{k_0,\{b_1,\ldots,c,n\};\ell_0},~R\in\Mat_{k_{h+1},\{e,\ldots,f,g,n\};\ell_{h+1}}\]
where every $\ell_i\in\{0,1\}$ and also $\sum\ell_i\in\{0,1\}.$
It holds that
\begin{align*}
&(\mbcfw)_{cf}\big\{(\mbcfw)_{a_hc}\big(
L_h,[\alpha_h:\beta_h:\gamma_h:\delta_h:\varepsilon_h],\\
&(\mbcfw)_{a_{h-1}c}\big(
L_{h-1},[\alpha_{h-1}:\beta_{h-1}:\gamma_{h-1}:\delta_{h-1}:\varepsilon_{h-1}],
\cdots
,\pre_dM\big)\cdots
\big),\\&\qquad\qquad
[\alpha_0:\beta_0:\gamma_0:0:\varepsilon_0]
,(\mbcfw^\gamma)_{df}([\alpha_{h+1}:\beta_{h+1}:0:\delta_{h+1}:\varepsilon_{h+1}],R)
\big\}=\\
&=(\mbcfw)_{df}\big\{
(\mbcfw)_{a_hd}\big[
\tilde{L}_h,[\talpha_h:\tbeta_h:\tgamma_h:\tdelta_h:\tepsilon_h],
(\mbcfw)_{a_{h-1}d}\big(
\tilde{L}_{h-1},[\talpha_{h-1}:\tbeta_{h-1}:\tgamma_{h-1}:\tdelta_{h-1}:\tepsilon_{h-1}],
\cdots
\\&\qquad\qquad\qquad\qquad\qquad\qquad\qquad\qquad\cdots(\mbcfw^\alpha)_{c'd}(\tilde{M},[0:\tbeta_0:\tgamma_0:\tdelta_0:\tepsilon_0])\big)\cdots
\big)\big],
[\talpha_{h+1}:\cdots:\tepsilon_{h+1}],\tilde{R}
\big\}
\end{align*}
\end{enumerate}
\end{lemma}
We now use the previous lemma to prove the claim regarding uncolored matchable boundaries whose matched boundary is also uncolored. Denote by $\emptyset$ the empty diagram. Suppose $\check{\D}=\shift_{\zeta_i}\D$ and write  $\D',\check{\D}'$ for their minimal neighboring subdiagrams~\cref{obs:minimal_neighboring}. Then $\check{\D}'=\shift_{\zeta_i}\D'$, if we label chords according to their label in $\D.$ Since the steps of the generation sequences of $\D,\check\D$ which comes after generating $\D',\check\D'$ respectively, are the same, if the claim holds for ${\D}',\check{\D}'$ it will also hold for $\D,\check\D.$ Thus we may assume that in all cases $\D,\check\D$ are their own minimal neighboring subdiagrams. We will follow the cases of
\cref{obs:bdry_matching_comb}.
\begin{itemize}
\item[--] $\match(\D,\alpha_i)=(\shift_{\alpha_i}\D,\beta_i),$ where $c_i>b_i+1.$ $\partial_{\alpha_i}S_\D=S_{\D_L}\bcfw^\alpha_{a,c}S_{\D_R}$ where $\D_R$ has index set $\{b,b+1,\ldots,n\}$ and no chord starts at $b,$ and $\D_L$ has index set $[b]\cup\{n\}.$ By \cref{it:unobstracted start non short} this space equals $S_{\D_L}\bcfw^\beta_{b,c}S_{\D'_R}=\partial_{\beta_i}S_{\check{\D}},$ where the chord diagram $\D'_R$ is the same as $\D_R$ but starts at $b+1.$
\item[--]$\match(\D,\alpha_i)=(\shift_{\alpha_i}\D,\delta_i),$ $\D_i$ is short and no chord starts where it ends. In this case, since the shift is defined, there is no chord which descends from $\D_i.$ 
There are two possibilities, either no other chord ends where $\D_i$ ends. In this case we can write $\partial_{\alpha_i}S_\D=\pre_{n-1}S_{\D_L}\mbcfw^\alpha_{a,c}S_{\D_R},$ where $\D_R$ is the empty diagram on indices $b,c,d,n$ and $d=n-2.$ By \cref{it:unobstracted short} this space equals 
$\pre_c.S_{\D_L}\mbcfw^\delta_{b,d}S_{\D_R}=\partial_{\delta_i}S_{\check{\D}}.$

In the other case there is a chain of same-end ancestors of $\D_i$ of the same color.
We use the indices of \cref{it:chain_with_short_1}. We have
\[\partial_{\alpha_i}S_\D=
S_{\D_h}\bcfw_{a_hd}\left(S_{\D_{h-1}}\bcfw_{a_{h-1}d}(S_{\D_R}\bcfw_{c'd}^\alpha\emptyset)\cdots\right),\]where $\emptyset$ denotes the empty diagram.
By that item this space equals
\[\left(S_{\D_h}\bcfw_{a_hc}(S_{\D_{h-1}}\bcfw_{a_{h-1}c}(\cdots\bcfw_{a_1c}S_{\pre_d\D_R})\cdots)\right)\bcfw_{ce}^\delta\emptyset=
\partial_{\delta_j}S_{\shift_{\alpha_i}\D}.
\]
\item[--]$\match(\D,\alpha_i)=(\shift_{\alpha_i}\D,\gamma_j).$
In this case there is a chord which starts where $\D_i$ ends. Again we split into the case where no ancestor of $\D_i$ ends where $\D_i$ ends, and the case where there is a chain of same end ancestors of the same color.
For the first case, we use the indices of \cref{it:obstructed short 2 step}. 
\[\partial_{\alpha_i}S_\D=(S_{\D_L}\bcfw_{ac}^\alpha\emptyset)\bcfw_{ce}S_{\D_R}.
\]
By that item this space equals \[S_{\pre_c\D_L}\bcfw_{be}(\emptyset\bcfw^\gamma_{ce}S_{\D_R})=\partial_{\gamma_j}S_{\shift_{\alpha_i}\D}.\]

In the second case we use the indices of \cref{it:chain_with_short_2}.
\[\partial_{\alpha_i}S_\D=
\left\{S_{\D_h}\bcfw_{a_hd}\left(S_{\D_{h-1}}\bcfw_{a_{h-1}d}(\cdots(S_{\D_M}\bcfw_{c'd}^\alpha\emptyset)\cdots\right)\right\}\bcfw_{df}S_{\D_R}.\]
By that item this space equals
\[\left(S_{\D_h}\bcfw_{a_hc}(S_{\D_{h-1}}\bcfw_{a_{h-1}c}(\cdots S_{\pre_d\D_M})\cdots)\right)\bcfw_{cf}(\emptyset\bcfw_{df}^\gamma S_{\D_R})=
\partial_{\gamma_j}S_{\shift_{\alpha_i}\D}.
\]

\item[--]$\match(\D,\gamma_i)=(\shift_{\gamma_i}\D,\delta_i),~i\neq \star$ This case is similar to the case $\match(\D,\alpha_i)=(\shift_{\alpha_i},\beta_i),$ and follows by the same argument using \cref{it:unobstracted end non short}.
\item[--]$\match(\D,\gamma_i)=(\shift_{\gamma_i}\D,\gamma_j).$ This case is similar to the first case $\match(\D,\alpha_i)=(\shift_{\alpha_i},\delta_j),$ where $\D_i$ has no same-end, same-color ancestor, and follows by the same argument using \cref{it:obstructed end 2 step}.
\end{itemize}
We turn to the two remaining cases.
\\\textbf{The blue boundary and its matched partner}
Let $C\vdots D$ be a weak domino representative of $U\vdots V\in\partial_{\varepsilon_h}S_\D,$ where $h$ is the lowest blue chord. Assume first it is not very short.
Define $\check C\vdots \check D$ from $C\vdots D$ as follows. For every $j\in\{\star\}\Red_\D\cup\Blue_\D\setminus\{h\}$ we subtract from the corresponding row of $C$ or $D$ a multiple of $C_{h}$ which cancels its $d_{\star}$ entry. Since $\D_h$ is not very short this also puts the $c_\star$ entry to $0$. For $D_\star$ we also multiply by $-\sgn(\gamma_{h})=(-1)^{\below(\D_{h})+1},$ to fix signs. $\check D_0=(-1)^{\below(\D_{h})+1}D_0.$ Other rows of $\check C\vdots\check D$ are the same as the corresponding rows of $C\vdots D.$ The claim will follow from
\begin{lemma}\label{lem:blue_chane_of_coords_verification}
    $\check C\vdots \check D$ is a weak $\check\D=\shift_{\varepsilon_h}\D-$domino representative of $U\vdots V,$ in which $\theta^{\check\D}_\star=0,$ and other variables have the required sign, that is, an element of $\partial_{\gamma_\star}S_{\check\D}.$ Moreover, this procedure yields a diffeomorphism between $\partial_{\gamma_\star}S_{\check\D}$ and $\partial_{\varepsilon_h}S_\D.$ 
\end{lemma}
\begin{proof}
It is immediate from the construction that the resulting $1-$loop matrix is an almost $\check\D-$domino pair, with all domino variables finite and defined, and since $C\vdots D$ has full rank, also $\check C\vdots \check D$ has. $\theta^{\check\D}_\star=0$ is also immediate from the construction, since the $(c_\star^{\check\D},d_\star^{\check\D})$ entries of $\check D_\star$ are proportional to the starting domino of $\check C_{h_*}.$
The only part of the claim which requires verification is the domino signs. 
For this we first need to write the expressions for the new rows.
\[\check D_\star = (-1)^{\below(\D_{h})+1}(D_\star-\frac{1}{\gamma_{h}+\delta_{h}}C_{h}).\] $\gamma_\star^{\check\D}, \delta_\star^{\check\D},\varepsilon_\star^{\check\D}$ have the correct signs by construction. 
Since the ending domino of $\check \D_\star$ is positively proportional to the starting domino of $D_{h_*},$ if $\D_{h_*}$ starts where a sibling $\D_i$ ends, then this chord becomes a top purple chord. The weak sign of $\theta_i^{\D}$ induces the correct weak sign of $\eta_i^{\check\D}.$ Similarly, since the ending dominoes of $D_\star,\check{C}_{h}=C_{h}$ are proportional with proportionality constant of sign $(-1)^{\below(\D_{h})},$ if $\D_i$ is the top purple chord (if exists) in $\D,$ then $\check\D_i$ is a same end child of $\check\D_{h}.$ The weak sign of $\eta^{\D}_i$ induces the correct weak sign of $\eta^{\check\D}_i.$

A non top red row $C_i$ has the form \[\alpha_i\ee_{a_i}+\beta_i\ee_{b_i}+\gamma_iD_\star+\delta_iD_0+\varepsilon_i(\alpha_{\p(i)}\ee_{a_{\p(i)}}+\beta_{\p(i)}\ee_{b_{\p(i)}}),\]and hence
\begin{align*}\check C_i &= \alpha_i\ee_{a_i}+\beta_i\ee_{b_i}+\gamma_iD_\star+\delta_iD_0+\varepsilon_i([\alpha_{\p(i)}\ee_{a_{\p(i)}}+\beta_{\p(i)}\ee_{b_{\p(i)}})-\frac{\gamma_i}{\gamma_{h}+\delta_{h}}C_{h}=\\&=
\alpha_i\ee_{a_i}+\beta_i\ee_{b_i}+(-1)^{\below(\D_{h})+1}\gamma_i\check{ D}_\star+(-1)^{\below(\D_{h})+1}\delta_i\check{D}_0+\varepsilon_i([\alpha_{\p(i)}\ee_{a_{\p(i)}}+\beta_{\p(i)}\ee_{b_{\p(i)}}).
\end{align*}
The signs of $\gamma^{\check\D}_i,\delta^{\check\D}_i$ should indeed be $(-1)^{\below(\D_{h})+1}$ times those of $\gamma^{\D}_i,\delta^{\D}_i,$ by the domino signs and the definition of $\below(-).$ The signs of $\eta_i$ do not change by this new scaling.

Finally, a blue chord $C_i$ has the form
\[\alpha_i\ee_{a_i}+\beta_i\ee_{b_i}+\gamma_iu_\star+\delta_iD_\star+\varepsilon_i(\alpha_{\p(i)}\ee_{a_{\p(i)}}+\beta_{\p(i)}\ee_{b_{\p(i)}}),\]and hence $\check C_i$ equals
\begin{align*}&
\alpha_i\ee_{a_i}+\beta_i\ee_{b_i}+\gamma_iu_\star+\delta_iD_\star+\varepsilon_i([\alpha_{\p(i)}\ee_{a_{\p(i)}}+\beta_{\p(i)}\ee_{b_{\p(i)}})-\frac{\gamma_i+\delta_i}{\gamma_{h}+\delta_{h}}C_{h}=\\
&
\alpha_i\ee_{a_i}+\beta_i\ee_{b_i}+(-1)^{\below(\D_{h_*})+1}\gamma^{\check\D}_iu^{\check\D}_\star+\delta^{\check\D}_i\check{D}_\star+
\varepsilon_i([\alpha_{\p(i)}\ee_{a_{\p(i)}}+\beta_{\p(i)}\ee_{b_{\p(i)}})
,
\end{align*}
where \[u^{\check\D}_\star=\pr_{a_{h}b_{h}}\check{D}_\star = \frac{(-1)^{\below(\D_{h})}}{\gamma_{h}+\delta_{h}}\alpha_{h}\ee_{a_{h}}+\beta_{h}\ee_{b_{h}},\]and\[\delta^{\check\D}_i=(-1)^{\below(\D_{h})+1}\frac{\delta_i\gamma_{h}-\gamma_i\delta_{h}}{\gamma_{h}},\qquad\qquad\gamma^{\check\D}_i=(-1)^{\below(\D_{h})+1}\gamma_i(1+\frac{\delta_{h}}{\gamma_{h}}).\]
It is immediate from the domino sign rules that $\gamma_i^{\check\D}$ has the correct sign, and a simple calculation to see that $\delta_i^{\check\D}$ and $\eta_i^{\check\D},$ when $\p(i)\in\Blue_\D$, have the correct signs. Since no other signs can change, this finishes the verification.

This construction is clearly invertible, since we can reverse each step and the verification of signs is the same. The correspondence and its inverse are smooth.
\end{proof}
If $\D_h$ is very short, we act similarly, only that if there is a chord $\D_l$ which start at $(c_\star,d_\star)$ we also replace $C_h,C_l$ by $C'_{h},D'_h$ defined as
\[C'_h=C_h-\frac{\pr_{d_\star}C_h}{\pr_{d_\star}C_l}C_l,\qquad C'_l=C_l-\frac{\pr_{b_{\p(l)}}C_l}{\pr_{b_{\p(l)}}C'_h}C'_h.\]
These substitutions correspond to replacing $\D_h$ by a chord of support $(b_h,c_\star,c_l,d_l)$ and making it a parent to the chord that used to be labeled $\D_l.$ As usual $b_{\p(l)}=n$ if $\p(l)=\emptyset.$ 
Similarly, if there are black chords ending at $(c_\star,d_\star)$ we subtract a multiple of $C_h$ to set their $d_\star$ entry to $0.$ This operation corresponds to shifting their ends one marker to the left.

We then relabel the chords to agree with our labeling scheme for $\shift_{\varepsilon_h}\D.$
By construction the supports of chords agree with the expected supports for $\shift_{\varepsilon_h}\D.$ Most signs have been verified in the proof of \cref{lem:blue_chane_of_coords_verification} and the few remaining signs are easily verified.  
\\\textbf{The case $\match(\D,\gamma_\star)=(\shift_{\gamma_\star}\D,\delta_\star).$}
By our assumption on the minimality of the subdiagrams $\D_\star$ has no ancestors, no chord starts to its right and $d_\star\lessdot n-1$. Moreover, when the shift is defined the $(n-1)$th column is a zero column. In the level of matrices, the matrices of $\partial_{\gamma_\star}S_\D$ differ from those of $S_\D$ by setting the $c_\star$ entry of the yellow, red and blue chords to $0.$ The shifted diagram has no purple chords, and the yellow chord is $(d_\star,n-1).$ Matrices in $\partial_{\delta_\star}S_{\shift_{\gamma_\star}\D}$ differ from those of $S_{\shift_{\gamma_\star}\D}$ by setting the $n-1$th entry of the yellow, red and blue chords to $0,$ thus obtaining matrices of the exact same form, and satisfying the same sign rules, as those of $\partial_{\gamma_\star}S_\D.$

We have covered all cases, and the proof follows.
\end{proof}
\begin{proof}[Proof of \cref{lem:boundary_matching_operators}]
The first three items follow immediately from the definitions and direct substitution. 

The proofs of \cref{it:obstructed end 2 step},~\cref{it:obstructed short 2 step},~\cref{it:intermideate_short_case} and \cref{it:intermediate_short_2} are similar, so we write the explicit transformation for them, but prove only one case \cref{it:obstructed end 2 step}. All substitutions are easily seen to be invertible in the level of projective vectors with positive entries. 
The substitutions are:
\\\textbf{\cref{it:obstructed end 2 step}: }$M=\scale_n(\frac{\tdelta_2\tepsilon_1}{\tdelta_2\tepsilon_1+\tdelta_1\tepsilon_2}).\tilde{M}$ and 
\begin{align*}
    &\alpha_1=(\tdelta_2+\frac{\tdelta_1}{\tepsilon_1}\tepsilon_2)\talpha_1,~ \beta_1=(\tdelta_2+\frac{\tdelta_1}{\tepsilon_1}\tepsilon_2)\tbeta_1,\delta_1 =\tdelta_1\tbeta_2 
    ~\varepsilon_1=\tdelta_2\tepsilon_1
    \\&\alpha_2=\talpha_2,~\beta_2=\tbeta_2,~\gamma_2=\frac{\tgamma_1}{\tdelta_1}(\tdelta_2+\frac{\tdelta_1}{\tepsilon_1}\tepsilon_2),~\delta_2=\tdelta_2+\frac{\tdelta_1}{\tepsilon_1}\tepsilon_2,~\varepsilon_2=\tepsilon_2.
\end{align*}
\textbf{\cref{it:obstructed short 2 step}: }
 \begin{align*}
    &\beta_1=(\tdelta_2+\frac{\tdelta_1}{\tepsilon_1}\tepsilon_2)\talpha_1,~ \gamma_1=(\tdelta_2+\frac{\tdelta_1}{\tepsilon_1}\tepsilon_2)\tbeta_1,\delta_1 =\tdelta_1\tbeta_2 
    ~\varepsilon_1=\tdelta_2\tepsilon_1
    \\&\alpha_2=\talpha_2,~\beta_2=\tbeta_2,~\gamma_2=\frac{\tgamma_1}{\tdelta_1}(\tdelta_2+\frac{\tdelta_1}{\tepsilon_1}\tepsilon_2),~\delta_2=\tdelta_2+\frac{\tdelta_1}{\tepsilon_1}\tepsilon_2,~\varepsilon_2=\tepsilon_2.
\end{align*}
\textbf{\cref{it:intermideate_short_case}: }$R=\scale_n(\frac{\tdelta_2\tepsilon_1}{\tdelta_2\tepsilon_1+\tdelta_1\tepsilon_2}).\tilde{R}$ and
\begin{align*}
    &\alpha_1=(\tgamma_2+\frac{\tdelta_1}{\tepsilon_1}\tepsilon_2)\talpha_1,~ \beta_1=(\tgamma_2+\frac{\tdelta_1}{\tepsilon_1}\tepsilon_2)\tbeta_1,\delta_1 =\tdelta_1\tbeta_2 
    ~\varepsilon_1=\tgamma_2\tepsilon_1
    \\&\alpha_2=\talpha_2,~\beta_2=\tbeta_2+\frac{\tgamma_1\tgamma_2}{\tdelta_1}+\frac{\tgamma_1\tepsilon_2}{\tepsilon_1},~\gamma_2=\tgamma_2+\frac{\tdelta_1}{\tepsilon_1}\tepsilon_2,~\varepsilon_2=\tepsilon_2.
\end{align*}
\\\textbf{\cref{it:intermediate_short_2}: }
 \begin{align*}
    &\alpha_1=(\tdelta_2+\frac{\tdelta_1}{\tepsilon_1}\tepsilon_2)\talpha_1,~ \beta_1=(\tdelta_2+\frac{\tdelta_1}{\tepsilon_1}\tepsilon_2)\tbeta_1,\gamma_1 =\tdelta_1\tbeta_2 
    ~\varepsilon_1=\tdelta_2\tepsilon_1
    \\&\alpha_2=\talpha_2,~\beta_2=\tbeta_2,~\gamma_2=\frac{\tgamma_1}{\tdelta_1}(\tdelta_2+\frac{\tdelta_1}{\tepsilon_1}\tepsilon_2),~\delta_2=\tdelta_2+\frac{\tdelta_1}{\tepsilon_1}\tepsilon_2,~\varepsilon_2=\tepsilon_2.
    \end{align*}
\\\textbf{A proof of \cref{it:obstructed end 2 step}: }
We treat the tree part of \cref{it:obstructed end 2 step}. The loopy part is handled in the exact same way, and does not add difficulty. We write $[r,s]$ for the interal of integers $r\leq i\leq s.$
The tree part of the LHS of \cref{it:obstructed end 2 step} can be written as
$L'+v'+M'+u+R',$ with
\begin{align*}
    &L'=\pre_{[b+1,f]}.y_a(\frac{\alpha_1}{\beta_1}).L,\\
    &M'=\pre_{[1,a]\cup[d+1,f]}.y_c(\frac{\alpha_2}{\beta_2}).y_d(\frac{\delta_1}{\varepsilon_1}).M\\
    &R'=\pre_{[1,c]}.y_e(\frac{\gamma_2}{\delta_2}).y_f(\frac{\delta_2}{\varepsilon_2}).R,
    \\
    &v'=[\alpha_2 \ee_a+\beta_2 \ee_b+(-1)^{k_M} \frac{\alpha_2}{\beta_2}\delta_1 +(-1)^{k_M} \delta_1 \ee_d +(-1)^{k_M}\varepsilon_1 \ee_n],\\
    &u=\alpha_2 \ee_c+\beta_2 \ee_d+(-1)^{k_R}\gamma_2 \ee_e+(-1)^{k_R}\delta_2 \ee_f +(-1)^{k_R}\varepsilon_2 \ee_n
\end{align*}where as usual we will identify a matrix with its row span, the summation is in the level of vector spaces, $k_M,k_R$ are the number of rows of $M,R$ respectively.

Similarly, the right hand side can be written as
$\tilde{L}'+w+\tilde{M}'+\tilde{u}+\tilde{R}',$ with
\begin{align*}
    &\tilde{L}'=\pre_{[b+1,f]}.y_a(\frac{\talpha_1}{\tbeta_1})L,\\
    &\tilde{M}'=y_e(\frac{\tgamma_1}{\tdelta_1}).y_f(\frac{\tdelta_1}{\tepsilon_1}).\pre_{[1,a]\cup[d+1,f]}.y_c(\frac{\talpha_2}{\tbeta_2}).\tilde M\\
    &\tilde{R}'=\pre_{[1,c]}.y_e(\frac{\tgamma_1}{\tdelta_1}).y_f(\frac{\tdelta_1}{\tepsilon_1}+\frac{\tdelta_2}{\tepsilon_2}).\tilde R,
    \\
    &w=\talpha_1\ee_a+\tbeta_1 \ee_b+(-1)^{k_M+k_R+1}\tgamma_1 \ee_e+(-1)^{k_M+k_R+1}\tdelta_1 \ee_f +(-1)^{k_M+k_R+1} \tepsilon_1 \ee_n\\
    &\tilde{u}'=\talpha_2 \ee_c+\tbeta_2 \ee_d+(-1)^{k_R}\frac{\tgamma_1}{\tdelta_1}(\tdelta_2+\frac{\tdelta_1} {\tepsilon_1}\tepsilon_2)\ee_e +(-1)^{k_R} (\tdelta_2+\frac{\tdelta_1} {\tepsilon_1}\tepsilon_2)\ee_f
    +(-1)^{k_R}\tepsilon_2 \ee_n
\end{align*}
From the relation between the variables $\zeta_i,~\zeta\in\{\alpha,\ldots,\varepsilon\},~i\in\{1,2\}$ and $\tzeta_i,~\tzeta\in\{\talpha,\ldots,\tepsilon\},~i\in\{1,2\}$ we see that
\begin{itemize}
\item $u'=\tilde{u}'.$
\item $v'=(\tdelta_2+\frac{\tdelta_1}{\tepsilon_1}\tepsilon_2)w-\tdelta_1\tilde{u}',$ hence $u',v'$ and $w,\tilde{u}'$ have the same linear span.
\item $L'=\tilde{L}'.$
\item After subtracting from the matrix $\tilde{M}$ the matrix $\frac{\tdelta_1}{\tdelta_1\tepsilon_2+\tdelta_2\tepsilon_1}\tilde{M}_n\cdot\tilde{u}'$ we obtain $M'.$ Thus the linear spans of $\tilde{M}',\tilde{u}'$ and that of $M',u'$ are the same.

\item $R'=\tilde{R}'.$
\end{itemize}
Thus, the two vector spaces are the same.
\\\textbf{\cref{it:chain_with_short_1}: }Observe that 
\[(\mbcfw^\alpha)_{c'd}(\tilde{R},[0:\tbeta_0:\tgamma_0:\tdelta_0:\tepsilon_0]) = (\mbcfw)_{ce}^\delta(\pre_d\tilde{R},[\tbeta_0:\tgamma_0:\tdelta_0:0:\tepsilon_0])_{\setminus f}.\]
Set
\[\tilde{R}^{(0)}=\tilde R,~[\alpha^{(0)}_0:\beta^{(0)}_0:\gamma^{(0)}_0:0:\varepsilon^{(0)}_0]=[\tbeta_0:\tgamma_0:\tdelta_0:0:\tepsilon_0].
\]
For $i=1,\ldots, h$ suppose we have defined $R^{(i-1)},$ and $[\alpha^{(i-1)}_0:\beta^{(i-1)}_0:\gamma^{(i-1)}_0:0:\varepsilon^{(i-1)}_0].$
We use \cref{it:intermideate_short_case} to write an equality of the form
\begin{align*}&(\mbcfw)_{ce}^\delta\big((\mbcfw)_{a_i,c}\big(\tilde{L}_i, [\alpha^{(i)}_i:\beta^{(i)}_i:\gamma^{(i)}_i:\delta^{(i)}_i:\varepsilon^{(i)}_i],{R}^{'(i)}),
[\alpha_0^{(i)}:\beta_0^{(i)}:\gamma_0^{(i)}:0:\varepsilon_0^{(i)}]
\big)=
\\&=\pre_f.(\mbcfw)_{a_id}\big(\tilde{L}_i,[\talpha_i:\tbeta_i:\tgamma_i:\tdelta_i:\tepsilon_i],((\mbcfw)_{ce}^\delta(\tilde{R}^{(i-1)},[\talpha_0^{(i-1)}:\tbeta_0^{(i-1)}:\tgamma_0^{(i-1)}:0:\tepsilon_0^{(i-1)}])_{\setminus f})
\big).
\end{align*}
This defines $[\alpha_0^{(i)}:\beta_0^{(i)}:\gamma_0^{(i)}:0:\varepsilon_0^{(i)}],~[\alpha^{(i)}_i:\beta^{(i)}_i:\gamma^{(i)}_i:\delta^{(i)}_i:\varepsilon^{(i)}_i].$ Note that $R^{(i)},R^{'(i)}$ differ by an explicit scaling of the $n$th column, by \cref{it:intermideate_short_case}. Set 
\[L_i^{(i)}=\tilde{L}_i,\qquad\tilde{R}^{(i)}=(\mbcfw)_{a_i,c}\big(L_i^{(i)},[\alpha^{(i)}_i:\beta^{(i)}_i:\gamma^{(i)}_i:\delta^{(i)}_i:\varepsilon^{(i)}_i],{R}^{'(i)}).\]
Note that, recursively, $\tilde{R}^{(i)}$ can be written as
\begin{align*}
(\mbcfw)_{a_ic}\left(
L^{(i)}_h,[\alpha^{(i)}_i:\beta^{(i)}_i:\gamma^{(i)}_i:\delta^{(i)}_i:\varepsilon^{(i)}_i],(\mbcfw)_{a_{h-1}c}\left(
L^{(i)}_{i-1},[\alpha_{i-1}^{(i)}:\beta_{i-1}^{(i)}:\gamma^{(j)}_{i-1}:\delta_{i-1}^{(i)}:\varepsilon_{i-1}^{(i)}],
\cdots
,\pre_dR^{(i)}\right)\cdots
\right)
\end{align*}
where each $L_j^{(i)}$ equals the previously defined ${L}_j^{j}$ up to scaling of the $n$th row by an expression determined from the iterations of \cref{it:intermideate_short_case}, the same relation holds between ${R}^{(i)}$ and $\tilde{R}^{(0)},$ and by the projective vectors $[\alpha_{j}^{(i)}:\beta_{j}^{(i)}:\gamma_{j}^{(i)}:\delta_{j}^{(i)}:\varepsilon_{j}^{(i)}]$ and the corresponding vectors $[\alpha_{j}^{(j)}:\beta_{j}^{(j)}:\gamma_{j}^{(j)}:\delta_{j}^{(j)}:\varepsilon_{j}^{(j)}]$ so that only the $\varepsilon-$terms differ. 

Thus, for
\[[\alpha_{j}:\beta_{j}:\gamma_{j}:\delta_{j}:\varepsilon_{j}]=[\alpha_{j}^{(h)}:\beta_{j}^{(h)}:\gamma_{j}^{(h)}:\delta_{j}^{(h)}:\varepsilon_{j}^{(h)}],~L_j=L_j^{h},R=R^{(h)}\]the equality of \cref{it:chain_with_short_1} holds. The invertibility of the transformation of \cref{it:intermideate_short_case} implies the invertibility in this case as well. 
\\\textbf{\cref{it:chain_with_short_2}: } The proof is similar to the previous case, but requires an additional single application of \cref{it:intermediate_short_2}:
Use \cref{it:intermediate_short_2} to write\begin{align*}
&(\mbcfw)_{cf}\left((\mbcfw)_{a_hc}\left\{
L_h,[\alpha_h:\beta_h:\gamma_h:\delta_h:\varepsilon_h],\right.\right.\\
&\quad\qquad\left.\left.(\mbcfw)_{a_{h-1}c}\left(
L_{h-1},[\alpha_{h-1}:\beta_{h-1}:\gamma_{h-1}:\delta_{h-1}:\varepsilon_{h-1}],
\cdots
,\pre_dM\right)\cdots
\right\}\right.,\\&\quad\qquad\quad\qquad\quad\qquad\left.
[\alpha_0:\beta_0:\gamma_0:0:\varepsilon_0]
,(\mbcfw^\gamma)_{df}([\alpha_{h+1}:\beta_{h+1}:0:\delta_{h+1}:\varepsilon_{h+1}],R)
\right)=
\\&(\mbcfw^\delta)_{c,e}
(\big\{
(\mbcfw)_{a_hc}\left(
L_h,[\alpha_h:\beta_h:\gamma_h:\delta_h:\varepsilon_h],\right.\\
&\quad\qquad\left.(\mbcfw)_{a_{h-1}c}\left(
L_{h-1},[\alpha_{h-1}:\beta_{h-1}:\gamma_{h-1}:\delta_{h-1}:\varepsilon_{h-1}],
\cdots
,\pre_dM\right)\cdots
\right)|_{\setminus e+2}\big\},\\&\quad\qquad\quad\qquad\quad\qquad
[{\alpha'}_0:{\beta'_0}:{\gamma'}_0:0:{\varepsilon'}_0],
(\mbcfw)_{df}([{\alpha}'_{h+1}:{\beta}'_{h+1}:{\gamma}'_{h+1}:{\delta}'_{h+1}:{\varepsilon}'_{h+1}],\bar{R}).
\end{align*}
We now apply the previous item to the right hand side of the above equality, to obtain the right hand side of the statement of the item. Again the map is invertible, thanks to the invertibility of the transformations of \cref{it:intermideate_short_case} and \cref{it:intermediate_short_2}.
\end{proof}
\subsection{The spaces $\Sred_\D$ and $\Sblue_\D$}\label{subsec:redblue_bdries}
\begin{lemma}\label{cor:wwRedBlueRanks}Recall \cref{lem:uniqueness_or_SA}. Let $\D$ be a chord diagram, $U\vdots V\in\overline{S}_\D\setminus\SA$. Then
\begin{enumerate}
    \item The well defined lines $\ww_i,~i\in\Red_\D\cup\Blue_\D$ and $D_\star$ are linearly independent and span a space $W_{\Red\cup\Blue}$. Also the well defined vectors $C_h,~h\notin\Red_\D\cup\Blue_\D$ are linearly independent and moreover the space they span intersects $W_{\Red\cup\Blue}$ trivially. If $U^i\vdots V^i\in S_\D,~i=1,2,\ldots$ is a sequence with $\lim_{i\to\infty}U^i\vdots V^i=U\vdots V$ then \begin{align*}&\lim_{i\to\infty}\Span(D^i_\star,\ww^i_h,~h\in\Red_\D\cup\Blue_\D)=W_{\Red\cup\Blue},\\&\lim_{i\to\infty}\Span(C_h^i,~h\notin\Red_\D\cup\Blue_\D)=\Span(C_h,~h\notin\Red_\D\cup\Blue_\D).\end{align*}
    \item If $C\vdots D$ is a domino limit for $U\vdots V$ then the vectors $C_h,~h\notin\Red$ are linearly independent and their linear span intersects trivially the vectors space $W_{\Red}:=\Span(\ww_h,~h\in\Red_\D)+D_\star.$
    If $C^i\vdots D^i\in S_\D,~i=1,2,\ldots$ is a sequence of domino pairs of elements in $S_\D$ whose limit is $U\vdots V,$ then \[\lim_{i\to\infty}\Span(D^i_\star,\ww^i_h,~h\in\Red_\D)=W_{\Red},\qquad\lim_{i\to\infty}\Span(C_h^i,~h\notin\Red_\D)=\Span(C_h,~h\notin\Red_\D).\] 
    If $U\vdots V\notin \Sred_\D$ then also the red and yellow rows are linearly independent, and they span $W_\Red.$
    \item The three vectors $D_0,D_\star$ and $\ww_0$ are linearly dependent.\end{enumerate}
\end{lemma}
\begin{proof}
The first item and the claims on $W_\Red$ in the second item are immediate from \cref{lem:uniqueness_or_SA} and the proof of \cref{lem:ww}, which recursively constructs $\ww^i_h$ as a linear combination of $D^i_\star,D^i_0$ and $C^i_l$ for $\Red\cup\Blue\ni l\geq h.$

Note that the same rank argument used in \cref{lem:uniqueness_or_SA}, only without the complication caused by the yellow chord, shows that, relying on 
\cref{lem:Srem_cases}, the matrix obtained from $C$ by restricting to non red rows, is of full rank, since, by canceling the dominoes inherited to non top rows, except the top blue row, from their parents, the resulting matrix has an invertible triangular minor obtained by taking the columns corresponding to the first non zero entries of each row. The claim that this vector space is the limit of the corresponding subvectors spaces of $U^i$ also follows.

Next, assume towards contradiction that $W_{\Red^c}:=\Span(C_h:~h\notin\Red)$ intersects $W_\Red$ non trivially. We can similarly define the spaces $W_{\Red^c}^i.$ Note that  every blue row $C^i_h\in\Span(\ww^i_{\p(h)},\ww^i_h,D^i_\star)$ and the same holds for the limit vectors. Thus, by the first item, $W_{\Red^c}\cap W_\Red\neq 0$ only if $\Span(C_h:~h\in\Blue_\D)\cap W_\Red\neq 0.$ But again by 
\cref{lem:Srem_cases}, the coefficient of $\ww_h$ in the combination of it with $\ww_{\p(h)},D_\star$ which yields $C_h,$ is non zero. This implies that any non zero combination of $C_h$ for $h\in \Blue_\D$ must be a linear combination of $\ww_h,~h\in\Blue_\D, D_\star$ and $\ww_l$ for the lowest red chord $\D_l,$ such that at least one $\ww_h,~h\in\Blue_\D$ has non zero coefficient. By the definition of $W_\Red$ and the first item, no such combination exists.

Finally, the span of the yellow and red rows of $C\vdots D$ is always contained in $W_\Red.$ If in addition $U\vdots V\notin \Sred_\D$ then by \cref{lem:Srem_cases} the starting domino of every red row is non zero, and also $\varepsilon_r$ is defined and finite for every such row. Thus, it is straightforward to see that the alogrithm \cref{lem:ww} for producing the vectors $\ww_r$ from the red and yellow rows can still be applied, hence $W_\Red$ is contained in the span of the red and yellow rows.

For the third item, note that otherwise we can find a non zero linear combination of the three vectors supported only on $c_\star,d_\star,$ which implies $U\vdots V\in\SA_D.$
\end{proof}
We now show that if $U\vdots V\in\Sred$ or in $\Sblue$ then this can be seen from domino limits.
\begin{lemma}\label{lem:Srem_and_domino_limits}
Assume $U\vdots V\in\overline{S}_\D\setminus\SA.$ Let $C\vdots D$ be a domino limit pair for $U\vdots V.$ By \cref{lem:uniqueness_or_SA} the rows $D_0,D_\star$ of $D$ are independent of the limit pair, up to scaling. Then, 
  if $U\vdots V\in\Sblue_\D\setminus\Sred_\D,$ then there is at least one blue row which is a linear combination of $\ee_{a_h},\ee_{b_h},D_\star$ and $u_\star=\pr_{c_\star d_\star}D_\star.$ 
  Moreover, if $h\in\Blue_\D$ is the maximal index for which a nonzero vector in $U\cap\Span(\ee_{a_h},\ee_{b_h},D_\star,u_\star)$ exists, then also $C_h$ belongs to this intersection, and for every blue descendant $\D_j$ of $\D_h$ the intersection $U\cap\Span(\ee_{a_j},\ee_{b_j},u_\star,D_\star)$ contains a vector $u_j,$ with $\pr_{a_jb_j}u_j\neq 0,$ such that the $1-$loop matrix $\tilde C\vdots \tilde D$ obtained from $C\vdots D$ by replacing $C_j,$ for blue $\D_j$ with $j\leq h$ by $u_j$ is a weak domino representative for $U\vdots V.$
\end{lemma}
\begin{proof}We prove by arriving to contradiction.
Let $h$ be the index of the topmost red or blue chord for which $\Span(\ee_{a_h},\ee_{b_h},u_\star,D_\star)\cap U\neq 0,$ and let $v\neq0$ be a vector in this intersection. If $v$ is not unique up to scaling then we can find a non zero vector in $\Span(\ee_{a_h},\ee_{b_h},u_\star)\cap U$ which would imply $U\vdots V\in\SA.$ Also, since a linear combination of $v,D_\star$ lies in $\Span(\ee_{a_h},\ee_{b_h},u_\star),$ and cannot be $0$ or else $U\vdots V\in\SA_C,$ then this combination must be proportional to $\ww_h$ by \cref{lem:ww} and the second item of \cref{lem:uniqueness_or_SA}. In particular, $\pr_{a_hb_h}v$ is proportional to the starting domino the row labeled $h.$ Since $C\vdots D$ is a domino limit pair, 
\cref{lem:Srem_cases} shows that the starting domino of $C_h$ does not vanish, and the proportionality constant is not $0.$ 
Thus, $v\in\Span(\ww_h,D_\star),$ or else we can find a non zero linear combination of it with $\ww_h,D_\star$ supported on $c_\star,d_\star,$ which implies $U\vdots V\in\SA_D.$ Let $l$ be the index of the topmost red or blue chord with the property that $\Span(\ww_l,D_0,D_\star)\cap U\neq\{0\}.$ Then $\D_l$ is either $\D_{h}$ or an ancestor. Let $v'$ be a nonzero vector in the intersection.

Assume first $\D_h$ is red. Then also $\D_l$ is red.
If $l=0,$ then by the third item of \cref{cor:wwRedBlueRanks}, $v'$ is a linear combination of $D_0,D_\star,$ implying $U\vdots V\in\Sred_\D.$ If $l\neq 0,$ the starting domino of $C_l$ is nonzero, as well as the domino $C_l$ inherits from its parent, by \cref{lem:Srem_cases} and the assumption that $U\vdots V\notin\Sred_\D$. Again we may add an appropriate multiple of $v'$ to $C_l$ to obtain a non zero vector in $U\cap\Span(D_0,D_\star,u_\star,\ee_{a_\p(l)},\ee_{b_{\p(l)}})$. As above, by our assumptions on $U\vdots V,$ and \cref{lem:uniqueness_or_SA} this vector is a linear combination of $\ww_{\p(l)},D_0,D_\star,$ contradicting the choice of $l.$ 

Thus, $\D_h$ is blue. Assume towards contradiction $v$ is not proportional to $C_h.$ By \cref{lem:Srem_cases} the starting domino of $C_h$ does not vanish, and the proportionality constant between this domino and $\pr_{a_hb_h}v$ is nonzero. Thus, a linear combination of $C_h$ and $v$ may cancel it, thus obtaining a nonzero vector $v'\in \Span(\ee_{a_{\p(h)}},\ee_{b_{\p(h)}},u_\star,D_\star),$ contradicting the choice of $h.$ Thus $C_h$ is proportional to $v$ and belongs to $ U\cap\Span(\ee_{a_{h}},\ee_{b_h},D_\star,u_\star).$

For the last part, let $J\subset\Blue\cap[h]$ be the set of indices $j$ for which there exists a blue row $u_j\in U\cap\Span(\ee_{a_j},\ee_{b_j},u_\star,D_\star).$ Note that such $u_j$ must be in the linear span of $\pr_{a_jb_j}\ww_j, D_\star$ and $u_\star,$ hence also of $\ww_j,D_\star$ and $u_\star,$ since otherwise we can use a combination of $u_j,\ww_j,u_\star,D_\star$ to find a non-zero vector in $U+V$ with support in $a_j,b_j,$ which would imply $U\vdots V\in\SA_D.$ 
 \begin{claim}$J$ is a chain of child-to-parent blue chords which starts at the lowest blue and ends at $\D_h$.
 \end{claim}
 \begin{proof}
 Let $C^i\vdots D^i\to C\vdots D$ be a sequence of domino forms of $U^i\vdots V^i$ in $S_\D$, converging to $U\vdots V.$
Write $m=|\Blue_\D\cap[h]|.$ For $t\in[m],$ 
set $W_t^i=\Span(C^i_h,C^i_{h-1},\ldots,C^i_{h-t+1})\subset U^i,$ and let $W_t\subset U$ be the limiting vector space, whose dimension is $t.$ 

We claim that $W_t$ has a basis of vectors $u_h,\ldots,u_{h-t+1}$ such that $u_j\in\Span(\ww_j,u_\star,D_\star),$ and the coefficient of $\ww_j$ is positive. For $t=1$ this has been shown above, with $u_h=C_h.$

Assume, by induction, we have shown it for $t<m-1,$ let us now deduce it for $t+1.$ $W_{t+1}$ cannot contain a vector spanned by $u_\star,D_\star,$ or else $U\vdots V\in\SA_C.$ By induction $W_{t+1}$ contains the vectors $u_h,\ldots, u_{h-t+1}$ with the required properties. Any vector in $W_{t+1}$ that can complete them to a basis must be a linear combination of \[\pr_{a_{\p(h)}b_{\p(h)}}\ww_{\p(h)},\pr_{a_{h}b_{h}}\ww_{h},\pr_{a_{h-1}b_{h-1}}\ww_{h-1},\ldots,\pr_{a_{h-t+1}b_{h-t+1}}\ww_{h-t+1},\pr_{a_{h-t}b_{h-t}}\ww_{h-t},u_\star,D_\star.\]
If $\D_{h-t}$ is the lowest blue chord, and if it is very short, then we should replace $\pr_{a_{h-t}b_{h-t}}\ww_{h-t}$ by the combination of $\ww_{h-t}$ and $u_\star$ in which the $d_\star$ entry vanishes, but it does not affect the proof below.
We can use the vectors $u_{h},\ldots,u_{h-t+1}$ to find such a non zero vector which is of the form
\[u_{h-t}=x\pr_{a_{\p(h)}b_{\p(h)}}\ww_{\p(h)}+y\pr_{a_{h-t}b_{h-t}}\ww_{h-t}+r u_\star + sD_\star,~~y\geq 0.\]
Let $f_i,~i\in\Blue\cup\{\star\}$ be the first non zero entry of $\pr_{a_ib_{i}}\ww_{i}.$
By \cref{lem:Srem_cases} these $f_i$ are all different.
If $y=0$ then $u_{h-t}$ is a witness that $U\cap\Span(\pr_{a_{\p(h)}b_{\p(h)}}\ww_{\p(h)},u_\star,D_\star)\neq0$, contradicting the choice of $h.$ 
Thus we may assume $y>0.$ 
Note that the Plucker coordinates of $W^i_{t+1},$ after proper scaling, tend to those of $W_{t+1}$. 
If $x=0$ $u_{h-t}$ has the required properties.
Assume $x\neq 0.$
Now, it is easy to see, from the domino sign rules, that \[\sgn(\lr{W^i_{t+1}}_{f_{p(h)},\ldots,f_{h-t+1}}/\lr{W^i_{t+1}}_{f_{h,\ldots,f_{h-t}}})=\prod_{j=0}^t\sgn(\varepsilon_{h-j})=:\sigma.\]
In the limit $i\to\infty$ we may calculate the same ratio with the ordered basis of rows given by $u_h,\ldots,u_{h-t}$,
 and the sign is readily seen to be $(-1)^t\sgn(x)\sgn(y),$ thus $\sgn(x)=(-1)^t\sigma.$ Similarly, 
 \[\sgn(\lr{W^i_{t+1}}_{f_{p(h)},f_{h},\ldots,f_{h-t+2},f_\star}/\lr{W^i_{t+1}}_{f_{h},\ldots,f_{h-t+2},f_{h-t},f_\star})= -\sgn(\gamma_{h-t})\sgn(\gamma_{h-t+1})\prod_{j=0}^{t-1}\sgn(\varepsilon_{h-j)}.\]
 By the domino sign rules $-\sgn(\gamma_{h-t}\gamma_{h-t+1})=\sgn(\varepsilon_{h-t}),$ hence the above ratio is again $\sigma.$
 Calculating the same ratio in the limit, on the other hand, yields $(-1)^{t+1}\sgn(x)=-\sigma$. This contradiction shows that $x$ must be $0,$ and the induction follows. 
 \end{proof}
It follows from the construction that the  $1$-loop matrix $\tilde C\vdots D$ obtained by replacing the rows $C_j,~j\in J$ by the corresponding $u_j$ is an almost $\D-$domino pair.
From the definition of $h,$ the fact $U\vdots V\notin \SA\cup\Sred_{\D}$ and \cref{lem:Srem_cases} it also follows that $\varepsilon_j,\gamma_j$ is always defined and finite for all chords $\D_j,~j\notin J.$
Since $\tilde C\vdots  D$ is an almost domino pair, we can consider $\tilde\alpha_j,\tilde\beta_j,\tilde\gamma_j,\tilde\delta_j$ for $j\in J$ as the coefficients of $\ee_{a_j},\ee_{b_j},u_\star,D_\star$ respectively in the expression for $u_j.$ By \cref{lem:Srem_cases} and the above discussion they are finite and well defined. $\tilde C\vdots D$ is of full rank by \cref{cor:wwRedBlueRanks}, and the fact that the only rows of $\tilde C$ which differ from the corresponding rows of $C$ are blue, but from their construction the blue rows of $C,\tilde C$ span the same subspace of $U.$ 
We are left with proving that the domino signs hold in a weak sense.
$\tilde\alpha_j,\tilde\beta_j$ have the correct sign by construction. We can verify that $\tilde\gamma_j,\tilde\delta_j$ have the correct signs by considering the limits $i\to\infty$ of
\[\lr{W^i_t}_{f_h,\ldots,f_{h-t+2},c_\star}/\lr{W^i_t}_{f_h,\ldots,f_{h-t+1}},~\text{and }, \lr{W^i_t}_{f_h,\ldots,f_{h-t+2},n'}/\lr{W^i_t}_{f_h,\ldots,f_{h-t+1}}~\text{respectively},\]
where $j=h-t+1,$ and $n'=n$ if $\p(0)=\emptyset,$ and the first non zero entry of the starting domino of $C_{\p(0)}$ otherwise.
$\eta_h$ has the correct (weak) sign since $u_h=C_h.$ For blue $j=h-t+1$ we can verify that $\eta_j$ has the correct (weak) sign by considering the limit of $\lr{W^i_t}_{f_h,\ldots,f_{h-t+3},c_\star,n'}/\lr{W^i_t}_{f_h,\ldots,f_{h-t+1}},$ where $n'$ is as above. The proof follows.  
\end{proof}
\begin{lemma}\label{lem:from SREDSBLUE to smaller diags in diags order}
Define a partial order $<$ on chord diagrams in $\CD_{n,k}^1$ such that $\check{\D}<\D$ if the total number of red and blue chord in $\check{\D}$ is smaller than that of $\D.$

Let $\D$ be a chord diagram.
If $U\vdots V\in\Sred_\D\cup\Sblue_\D\setminus\SA$ then $U\vdots V\in\Srem_{\check\D}$ for some $\check\D<\D.$
\end{lemma}
\begin{proof}
We treat the cases of $U\vdots V\in\Sred_\D,~U\vdots V\in\Sblue_\D\setminus\Sred_\D$ differently.
While the proof is somehow lengthy, it can be explained easily. We will show that in these cases $U\vdots V$ is contained in $\overline{\partial_{\abe}S_\D}$, or $\overline{\partial_{\varepsilon_h}S_\D}$, where $\D_h$ is the lowest blue chord, respectively. We will then use the fact that these boundaries are also boundaries of other BCFW cells~\cref{lem:bdry_matching_geo}, with less red and blue chords. Iterating this will yield the claim. We split the proof into red and blue cases.
\\\textbf{Treating red chords:}
\\Assume $U\vdots V\in\Sred_\D$ has a domino limit $C\vdots D.$ By \cref{lem:uniqueness_or_SA},~\cref{rmk:Sred_without_red_chords} $|\Red_\D|>1,$ and by \cref{lem:uniqueness_or_SA},~\cref{rmk:Sred_without_red_chords} we can find a red row $C_{j}$ in the linear span of $D_0,D_\star.$ $\Span(D_0,D_\star)\cap U$ must be one dimensional since otherwise a linear combination of two non proportional vectors from the intersection would have witnessed that $U\vdots V\in\SA_C.$ Let $h_*$ be the maximal index of such a row. Write $\Red_\D=\{k'-r+1,\ldots,k'-1,k',k'+1=\tr(\D)\leftrightarrow0\},~\Blue_\D =\{k'-r,\ldots,k'-s\}$. Recall that by the second item of \cref{lem:uniqueness_or_SA} the vectors $\ww_j,~j\in\Red_\D\cup\Blue_\D\cup\{\star\}$ of \cref{nn:ww} are linearly independent, and moreover the matrix made of these rows and the rows $C_j,~j\notin(\Red_\D\cup\Blue_\D)$ is of full rank $k+2.$ Let $U^i\vdots V^i\in S_{\D},~i=1,2,\ldots$ be a sequence which tends to $U\vdots V,$ and let $C^i\vdots D^i$ be domino pairs for $U^i\vdots V^i$ with limit $C\vdots D.$ 
Every $C^i_j$, for $j\in\Red_\D\cup\Blue_\D$ is a linear combination of $D_0,D_\star,\pr_{a_jb_j}\ww_j$ and $\pr_{a_{\p(j)}b_{\p(j)}}\ww_{\p(j)}.$
Write 
\[u^i_{k'+1}=\sigma_{h_*}C^i_{h_*},~\text{and } \check{u}^i_{k'+1}=\sigma_{h_*}(\gamma_{h_*}D^i_\star+\delta_{h_*}D^i_0),~\text{where }\sigma_{h_*}=(-1)^{\below(\D^i_{h_*})+\below(\D^i_{0})+1}.\]For $j\in\Red_\D,~j<k'$ define $u^i_j$ from $C^i_{j}$ by subtracting a multiple of $C^i_{k'}$ so that the result is a linear combination of $\pr_{a_jb_j}\ww^i_j,\pr_{a_{\p(j)}b_{\p(j)}}\ww^i_{\p(j)},\pr_{a_{0}b_{0}}\ww^i_{0},\pr_{a_{k'}b_{k'}}\ww^i_{k'},$ and $u_\star^i=\pr_{c_\star d_\star} D^i_\star.$ 
These properties uniquely determine the vector, given the matrix $C^i\vdots D^i.$ 
For blue $\D_j$ define $u^i_j~(\check{u}^i_j)$ by subtracting from $C^i_j$ the unique multiple of $u^i_{k'+1}=C^i_{h_*}~(\check{u}^i_{k'+1})$ which sets the projection on the starting domino of $\D_{\p(0)}$, which are as usual considered to be just $\ee_n$ if $\p(0)=\emptyset$ to $0.$ Again this condition determines uniquely $u^i_j.$ We also define $\check D^i_0$ by adding to $\ww^i_{k'}$ the unique multiple of $\ww^i_0$ which sets the projection on the starting entries of $\p(0)$ to $0.$ Also write $\check D^i_\star=-\frac{\varepsilon^i_\star}{\varepsilon^i_0}\ww^i_0.$  

Let $\check\D=\shift_{\abe}\D.$ We recall that this is the chord diagram obtained from $\D$ by declaring the top red chord to be black, thus making its red child the new top red chord, that is $\tr(\check{\D})=k'.$ 
 Finally, define the $1-$loop matrices $\check C^i\vdots \check D^i,~\hat C^i\vdots \check{D}^i$ with $\check{C}^i,\hat{C}^i$ having rows labeled by $[k+1]\setminus{k'},$ where \[\check{C}^i_j=\begin{cases}\check{u}^i_j,&\text{if defined, else}\\
    u^i_j,&\text{if $u^i_j$ is defined}\\
    C^i_j,&\text{otherwise}.\\\end{cases}\qquad\qquad \hat{C}^i_j=\begin{cases}{u}^i_j,&\text{if defined}\\ C^i_j,&\text{otherwise}.\\\end{cases}\]
We write $\check{C},\hat{C}$ for the limiting matrices, properly scaled so that now row vanishes.

The red case of the lemma follows immediately from the next observation.
\begin{obs}\label{obs:for_red_limits}
    \begin{itemize}
    \item For large enough $i$ $\check C^i\vdots \check D^i$ is a weak $\check{\D}$-domino representative of an element in $\partial_{\gamma^{\check\D}_\star}S_{\check\D}.$ 
    \item Let $\check{U}^i\vdots \check{V}^i$ be $1$-loop vector space spanned by $\check{C}^i\vdots\check{D}^i.$ Then $\lim_{i\to\infty}\check{U}^i\vdots\check{V}^i\to U\vdots V,$ hence $U\vdots V\in\overline{S}_{\check\D}.$
    \end{itemize}
\end{obs}
\begin{proof}
Let us first write the formulas for the rows of $\check{C}\vdots\check{D}.$ \cref{obs:ww} implies
\begin{align*}
\check{D}^i_0 = (-1)^{\below(\D_{k'})+1}(\ww^i_{k'}+\varepsilon^i_{k'}\ww^i_0),\qquad\qquad\check{D}^i_\star=-\frac{\varepsilon^i_\star}{\varepsilon^i_0}\ww^i_0=u^i-\frac{\varepsilon^i_\star}{\varepsilon^i_0}(\alpha^i_0\ee_{a_0}+\beta^i_0\ee_{b_0}).
\end{align*}
For red $j<k',$
\begin{align*}
    \check{C}_j^i = &u^i_j=C_j^i-\frac{\varepsilon^i_0\delta^i_j+\varepsilon^i_\star\gamma^i_j}{\varepsilon^i_0\delta^i_{k'}+\varepsilon^i_\star\gamma^i_{k'}}C^i_{k'}=\\
    =&\alpha^i_{j}\ee_{a_j}+\beta^i_j\ee_{b_j}+\varepsilon^i_j(\alpha_{\p(j)}\ee_{a_{\p(j)}}+\beta_{\p(j)}\ee_{b_{\p(j)}})+\frac{\varepsilon^i_0(\gamma^i_j\delta^i_{k'}-\delta^i_j\gamma^i_{k'})}{\varepsilon^i_0\delta^i_{k'}+\varepsilon^i_\star\gamma^i_{k'}}\check{D}^i_\star+(-1)^{\below(\D_{k'})}\frac{\varepsilon^i_0\delta^i_j+\varepsilon^i_\star\gamma^i_j}{\varepsilon^i_0\delta^i_{k'}+\varepsilon^i_\star\gamma^i_{k'}}\check{D}^i_0.
\end{align*}
In addition,
\begin{align*}
    &\check{C}^i_{k'+1}=\check{u}^i_{k'+1}=\sigma_{h_*}(\gamma^i_{h_*}\D^i_\star+\delta^i_{h_*}D^i_0),\\
    &{u}^i_{k'+1}=\varepsilon^i_{h_*}(\alpha^i_{\p(h_*)}\ee_{a_{\p(h_*)}}+\beta^i_{\p(h_*)}\ee_{b_{\p(h_*)}}+\alpha^i_{h_*}\ee_{a_{h_*}}+\beta^i_{h_*}\ee_{b_{h_*}})+
    \sigma_{h_*}(\gamma^i_{h_*}D^i_\star+\delta^i_{h_*}D^i_0).
\end{align*}
Finally, for $j\in\Blue_\D$
\begin{align*}
    \check{C}^i_j&=\check{u}^i_j = C^i_j-\sigma_{h_*}\frac{\delta^i_j\varepsilon^i_\star}{\gamma^i_{h_*}\varepsilon^i_\star+\delta^i_{h_*}\varepsilon^i_0}\check{u}^i_{k'+1}=\\&=
    \alpha^i_{j}\ee_{a_j}+\beta^i_j\ee_{b_j}+\varepsilon^i_j(\alpha_{\p(j)}\ee_{a_{\p(j)}}+\beta_{\p(j)}\ee_{b_{\p(j)}})+\gamma^i_ju+\frac{\delta^i_j\delta^i_{h_*}\varepsilon^i_0}{\gamma^i_{h_*}\varepsilon^i_\star+\delta^i_{h_*}\varepsilon^i_0}\check{D}_\star
\end{align*}
and\[u^i_j=\check{u}^i_j-\frac{\delta^i_j\varepsilon^i_\star}{\gamma^i_{h_*}\varepsilon^i_\star+\delta^i_{h_*}\varepsilon^i_0}(\alpha^i_{\p(h_*)}\ee_{a_{\p(h_*)}}+\beta^i_{\p(h_*)}\ee_{b_{\p(h_*)}}+\alpha^i_{h_*}\ee_{a_{h_*}}+\beta^i_{h_*}\ee_{b_{h_*}}).\]
It is immediate from the construction that $\check{C}^i\vdots\check{D}^i$ is an almost $\check{\D}-$domino pair with all variables defined and finite. The signs of $\check{D}_0,\check{D}_\star,\check{C}^i_{k'+1}$ are correct by their construction, where for $\check{C}^i_{k'+1}$ we use for the signs of its $\gamma^{\check\D},\check\delta^{\check\D}$ coordinates that $1+\below(\D_0)=\below(\check\D_{k'+1})$, and for the sign of its $\varepsilon^{\check\D}$ coordinate we use \[1+\below(\D_0)+\after(\D_\star)=\after(\check\D_{k'+1}),\qquad\qquad1+\below(\D_0)+\Between(\D_\star)=\Between(\check\D_{k'+1}).\]
Moreover, because of these signs, and since the ending domino of $C^{i}_{k'+1}$ is proportional to that of $\check\D_\star,$ it follows that $\eta^{\check\D}_{k'+1},\theta^{\check\D}_{k'+1}$, if exist, have the correct sign as well. The $\eta_\star^{\check\D}$ variable is $0$ by construction.

For red $\check{C}^i_j$ the $\check\D-$domino entries have the correct signs, by \cref{def:L=1domino_signs} and \cref{obs:iterated_etas}. The last observation also shows that $\eta_j^{\check\D},$ for non top red chords, also have the correct signs.
All other $\check{\D}-$domino signs of $\check{C}^i\vdots\check{D}^i$ but that of $\eta_\star$ are correct, since they agree with the corresponding signs for $C^i\vdots D^i.$ 

The $\gamma^{\check\D}$ and $\delta^{\check\D}$ coordinates of blue chords are the same as the corresponding coordinates for $C^i,$ up to scaling by $1,\frac{\delta^i_{h_*}\varepsilon^i_0}{\gamma^i_{h_*}\varepsilon^i_\star+\delta^i_{h_*}\varepsilon^i_0}>0$, respectively. Since $\below(\D_j)=\below(\check\D_j)$ in these cases, their signs are correct, as well as of $\eta_j^{\check\D}$ variables for non top blue chords, by \cref{def:L=1domino_signs}.

For red and blue chords of $\check\D$ the $\varepsilon^{\check\D}$ coordinates are the same as for the corresponding chords in $\D$, up to  $(-1)^{\below(\D_{k'})+1}$ sign for the new top red chord, and the same holds for $\gamma^{\check\D}_h,~h\in\Blue_{\check\D}=\Blue_\D.$ For the yellow chord $\varepsilon^{\check\D}_\star=(-1)^{1+\below(\D_0)}\frac{\varepsilon^i_\star}{\varepsilon^i_0}.$ It follows from \cref{obs:ww} that \[\check{\ww}^i_{j}=\ww^i_j,~j\in\Red_\D\cup\Blue_\D\setminus\{k'+1\},\]where the vectors $\check\ww_j$ are as in \cref{nn:ww}, but for $\check\D.$
Moreover, it is easy to see that  
\[\lim_{i\to\infty}\Span(\check{C}^{i}_{k'+1},\check{D}^i_\star)=\Span(D_0,D_\star)=\Span(\ww_0,\D_\star).\]
We thus learn that 
\[\lim_{i\to\infty}\check{U}^i+\check{V}^i=\Span(C^j~|~j\notin\Red_\D\cup\Blue_\D\cup\{\star\})+\Span(\ww_j~|j~\in\Red_\D\cup\Blue_\D)+\Span(D_\star)=U+V,\]
by the second item of \cref{lem:uniqueness_or_SA}.
This also shows that for $i$ large enough $\check{U}^i\vdots\check{V}^i$ has full rank, hence it is an element of $\overline{S}_{\check\D},$ and even of $\partial_{\gamma^{\check\D}_\star}{S}_{\check\D},$ and $\check{C}^i\vdots\check{D}^i$ is $\check\D-$domino representative.

We need to show that $\lim_{i\to\infty}\check{U}^i=U.$
First, note that for every $h\notin\Red_{\check\D}$ $\check{C}_h=\hat{C}_h.$ Now, by \cref{cor:wwRedBlueRanks} the rows $C_l,~l\notin\Red_{\D}$ are linearly independent and intersect trivially 
$\Span(D_\star,\ww_h:h\in\Red_\D).$ By construction, the linear span of $\hat{C}^i_h,~h\notin\Red_{\check\D}$ equals $\Span(C^i_{h_*})+\Span(C^i_h,~h\notin\Red_\D).$ Thus, also the limiting vectors $\check{C}_h,~h\notin\Red_{\check\D}$ are linearly independent vectors in $U.$ Write $U_{\Red^c_{\check\D}}$ for their span.

In addition, for $h\in\Red_{\check\D},$ $\check{C}^i_h=\hat{C}^i_h,$ hence the limiting vector spaces spanned by these collections of vectors is the same subspace $U_{\Red_{\check\D}}$ of $U,$ whose dimension is $k-\dim(U_{\Red^c_{\check\D}}).$ To finish we must show that $U_{\Red^c_{\check\D}}\cap U_{\Red_{\check\D}}=\{0\},$ since it will imply $\lim_{i\to\infty}\check{U}^i=\check{U}.$ 

$U_{\Red^c_{\check\D}},$ by its construction is contained in \[\Span(D_0,D_\star)+\Span(C_h,~h\notin\Red_\D)=\Span(\ww_0,D_\star)+\Span(C_h,~h\notin\Red_\D).\] With the notations of \cref{cor:wwRedBlueRanks} (applied to $\D$ and $C\vdots D$) $U_{\Red_{\check\D}}$ is contained in $W_\Red$, hence, \cref{cor:wwRedBlueRanks} implies that if $U_{\Red^c_{\check\D}}\cap U_{\Red_{\check\D}}\neq\{0\}$ then it must be contained in $\Span(\ww_0,D_\star)$. Moreover, $\ww_0,D_\star\notin U,$ since otherwise $U\in\SA_C.$ Thus the intersection must be a linear combination with non zero coefficients of $\ww_0,D_\star.$ But then its projection on the starting domino of $\D_{\p(0)}.$ 
For every element of $ U_{\Red_{\check\D}},$ however, this projection is $0,$ yielding the required contradiction.
\end{proof}
\textbf{Treating blue chords:}
If $U\vdots V\in\Sblue_\D\setminus\Sred_\D,$ then by \cref{lem:Srem_and_domino_limits} we can find a weak domino representative $C\vdots D$ for this space. By \cref{lem:Srem_cases} and 
the construction of \cref{lem:Srem_and_domino_limits}, the starting domino of every row is non zero, and all $\varepsilon,\gamma$ variables are defined and finite. 
 For $C\vdots D$ all $\varepsilon_j\neq 0,$ except for a nonempty subset $J\subseteq\Blue_\D,$ which is an ascending chain of indices of blue chords, starting from $h_*,$ the lowest blue, to some $h^*\in\Blue_\D.$ 
It is easy to show, by changing the domino entries slightly, in a way that does not contradict the sign rules, that $C\vdots D$ is a $\D-$domino limit, and moreover, we can even write $C\vdots D=\lim _{i\to\infty}C^i\vdots D^i,$ where $C^i\vdots D^i\in\partial_{\varepsilon_{h_*}}S_\D.$ Let $\check\D=\shift_{\varepsilon_{h_*}}<\D.$ Then, by \cref{lem:blue_chane_of_coords_verification}
\[U\vdots V\in\overline{\partial_{\varepsilon_{h_*}}S_\D}=\overline{\partial_{\gamma^{\check\D}_\star}S_{\check\D}}.\]
\end{proof}
\begin{rmk}\label{rmk:iterative boundary}In the above analysis, if $h_*\neq k',$ $h_*\neq h^*$ or $U\vdots V\in \Sred_\D\cap\Sblue_\D$ the argument is easily elaborated to show that $U\vdots V$ belongs to a corner of $\overline{S}_{\D'},$ where $\D'$ is obtained by iterating several red and blue shift moves.
\end{rmk}
\subsection{Manifold structures of strata}\label{subse:mfld_str}
We now wrap up earlier results to study the manifold structure of strata.
We begin with two corollaries of \cref{lem:uniqueness_or_SA},~\cref{rmk:Sred_without_red_chords} and \cref{lem:Srem_and_domino_limits}.
\begin{cor}\label{cor:reg_domino_strata}
Let $\D$ be a chord diagram, and $A\subset\tVar_\D.$ 
Then if $\partial_AS_\D\cap \Srem_\D\neq\emptyset$ then $\partial_AS_\D\subseteq \Srem_\D.$ In particular, every element of $\partial_AS_\D$ have unique up to gauge weak $\D-$domino form, and for $A\neq A'$ with $\partial_AS_\D,~\partial_{A'}S_\D\subseteq \Srem_\D,$
$\partial_AS_\D\cap \partial_{A'}S_\D=\emptyset.$ Thus, $\Srem_\D$ is the disjoint union of subspaces $\partial_A S_\D.$ If the blue boundary of $S_\D$ is not contained in $\SA$ then it does not intersect it. 
\end{cor}
\begin{proof}
By the first item of \cref{lem:uniqueness_or_SA}, every point $U\vdots V$ in $\partial_AS_\D\cap \Srem_\D\neq\emptyset$ has a unique domino limit, which is a weak domino representative. Hence it must be the weak domino pair with the elements of $A$ being $0.$ Since it is a weak domino representative, it must be of full rank. Since it is not in $\SA,$ for every quadruple $\{i,i+1,j,j+1\}$where addition is taken cyclically modulo $n,$ there is a non zero minor of $U$ not using these columns, and for every $\{i,i+1(\mod n)\}$,there is a nonzero minor of $U+V$ avoiding these columns. 
By \cref{prop:pos_poly_plucker_rep} all Pl\"ucker coordinates of $U,U+V$ can be written as polynomials with nonnegative coefficients in the elements of $\tVar.$ Thus, if a minor is non zero for one point of $\partial_AS_\D,$ then the same minor will be non zero for all points of $\partial_A S_\D.$ Thus, $\partial_AS_\D\cap\SA=\emptyset.$ 

If $U\vdots V$ belongs to the blue boundary of $S_\D$ then it has a weak domino representative in which only $\varepsilon_h$ for the lowest blue chord vanish. 
If $U\vdots V\notin\SA,$ the argument of the last paragraph shows that the same holds for all points of that boundary.

In addition, since the weak domino form $C\vdots D$ of an arbitrary point of $\partial_A S_\D$ is of full rank, then in particular no non zero combination of $D_0,D_\star$ falls in $C,$ showing $\partial_AS_\D\cap\Sred_\D=\emptyset,$ by \cref{lem:uniqueness_or_SA},~\cref{rmk:Sred_without_red_chords}.

Finally, by \cref{lem:Srem_and_domino_limits}, if $C\vdots D$ is a weak domino representative of $U\vdots V\notin\SA\cup\Sred_\D$ then it must have a blue row $C_j$ with vanishing $\varepsilon_j.$ But this means that $\varepsilon_j\in A,$ hence all $\partial_A S_\D\subseteq\Sblue_\D.$ Since this containment is false it implies that $\partial_AS_\D\cap\Srem_\D=\emptyset.$ 
Thus, $\partial_AS_\D\subseteq\Srem_\D.$ The 'in particular' part now follows from \cref{lem:uniqueness_or_SA}, which guarantees a unique weak domino form for elements of $\Srem_\D.$ The final statement follows from the previous statements and the first item of \cref{lem:uniqueness_or_SA}.
\end{proof}
\begin{cor}\label{cor:regular_matchable_and_S_D_in_SREM}
    For every chord diagram $\D$, $S_\D$ and all its regular matchable boundaries are pairwise disjoint. The regular matchable boundaries are contained in $\Srem_\D$. The blue boundary, if not in $\SA,$ is contained in $\Sblue_\D\setminus(\Sred_\D\cup\SA),$ hence is also disjoint from them.
\end{cor}
\begin{proof}
$S_\D$ does not intersect $\Sred_\D,\Sblue_\D$ by \cref{obs:cell_not_red_blue}.
We will show that its regular matchable boundaries are in $\Srem_\D,$ since then by \cref{lem:uniqueness_or_SA} all their points have a unique up to gauge weak domino representative, and since different spaces have different vanishing domino coordinate, they must be disjoint. Note that we have not yet ruled out the case that $\Srem_\D$ is empty, but even if it were the case, our argument would have held in a trivial sense.

These boundary strata are not in $\SA$ by the definition of regular boundaries, see \cref{nn:reg_match}.

By 
definition every element of these spaces has a weak domino representative $C\vdots D$. In particular $D_0,D_\star,$ which are well defined by the second item of \cref{lem:uniqueness_or_SA} are linearly independent of $C,$ which by \cref{lem:uniqueness_or_SA},~\cref{rmk:Sred_without_red_chords} implies that these spaces avoid $\Sred_\D.$

Since every point of $S_\D$ and its regular matchable boundaries is a domino limit, obtained by setting the relevant domino coordinate to $0,$ and this coordinate is not $\varepsilon_h$ for $\D_h\in\Blue_\D,$ \cref{lem:Srem_and_domino_limits} shows that these spaces do not intersect $\Sblue_\D$ as well. 

If the blue boundary is regular then by \cref{cor:reg_domino_strata} it avoids $\SA$. Every point $U\vdots V$ of the blue boundary has a weak $\D-$domino representative $C\vdots D$ in which the lowest blue $\varepsilon$ variable vanishes. $C\vdots D$ is clearly a domino limit, which can be taken over a sequence of elements in $S_\D$ with the same domino variables as $C\vdots D,$ only that the lowest blue $\varepsilon$ is non zero. This weak domino limit is of full rank, since it is a domino representative. 
In particular, $D_0,D_\star,$ which are uniquely determined by the second item of this lemma, are linearly independent over $C.$ Thus, by \cref{lem:uniqueness_or_SA},~\cref{rmk:Sred_without_red_chords} $U\vdots V\notin\Sred_\D,$ as needed. 
\end{proof}
\cref{cor:reg_domino_strata} has the following corollary, which uses the same notations.
\begin{cor}\label{cor:strata_and_dim}
If $\partial_A S_\D\subseteq\Srem_\D$ for $A\neq\emptyset$ then it is a weak submanifold of $\Gr_{k,n;\ell}$ of dimension at most $4(k+\ell)-1.$ The dimension equals $4(k+\ell)-1$ if and only if it is one of the regular matchable boundaries. The blue and red matchable boundaries, if avoid $\SA$, are manifolds of dimension $4(k+\ell)-1.$ 
\end{cor}
\begin{proof}
By \cref{cor:reg_domino_strata} every point of $\partial_A S_\D$ has a unique up to gauge weak domino representative. Moreover, the proof of that corollary shows that every weak domino form in which only the elements of $A$ vanish, is a weak domino representative, as long as $\partial_A S_\D\subseteq\Srem_\D$ is non empty. Since $\partial_A S_\D\subseteq\Srem_\D$, by \cref{lem:Srem_cases} and the definition of $\Sblue_\D,\Sred_\D$ no $\varepsilon_i$ variable vanishes. We use the gauge freedom to fix them to $1.$ 
It is evident that if some domino variable $\zeta_i=0$ then the inequality of every $\eta$ or $\theta$ variable involving it becomes trivial, in the sense that it follows from the domino sign constraints \cref{def:L=1domino_signs} for the domino \emph{entries}.

For every $i$ where $\alpha_i,\beta_i\neq 0,$ write $\psi_i=\beta_i/\alpha_i$ and for every $i$ with $\gamma_i,\delta_i\neq 0$ write $\phi_i=\delta_i/\gamma_i.$ The inequalities of \cref{def:L=1domino_signs} which involve $\eta$ or $\theta$ variables, and do not become trivial by vanishings of variables, as explained in the previous paragraph, are equivalent, modulo the sign rules for domino entries, to inequalities of the form $\phi_i\leq\phi_j$ ($\D_j$ is a strict same end ancestor of $\D_i$), $\phi_i\leq\psi_j$ ($\D_j$ starts where $\D_i$ ends) 
The vanishings of $\theta$ or $\eta$ variables which is not implied from vanishings of single domino entries are thus equivalent to setting $\phi_i=\phi_j$ or $\psi_i=\phi_j$ for some of the pairs $i,j.$ These equalities divide the set of $\phi$ and $\psi$ variables into equivalence classes. The class of a $\phi_i$ or $\psi_i$ variable is a singleton if there is no equality involving that variable.

Thus, the set made of elements
\begin{itemize}
\item $\alpha_i$ if $\alpha_i\notin A,$ and otherwise $\beta_i,$ if it is not in $ A.$
\item $\gamma_i$ if $\gamma_i,\delta_i\notin A,$ and otherwise $\delta_i,$ if it is not in $ A.$
    \item a single not ill defined, nonzero  $\phi_i$ or $\psi_i$ variable from every equivalence class. This condition means that the numerator and denominator defining it are not in $A$.
\end{itemize}
Then these variables are independent, are subject to induced inequalities of the form $\sgn(\zeta_i)=\pm,$ $0\leq \phi_i\leq \phi_j$ or $0\leq\phi_i\leq\psi_j$ and yield a smooth parameterization of $\partial_A S_\D$. This smooth parameterization also shows that these manifolds are weak smooth submanifolds of $\Gr_{k,n;\ell}.$

For $S_\D$ we have precisely $4(k+\ell)$ independent variables. For every boundary $\partial_A S_\D\subseteq\Srem_\D$ we have at most $4(k+\ell)-1:$ if some domino entry vanishes then it will not appear in the parameterization, and neither the corresponding $\psi_i,\phi_i.$

If two domino entries vanish, or there is more than one equality between the $\phi_i,\psi_i$ variables, or there is at least one vanishing of an entry and at least one equality of $\phi_i,\psi_i$ variables, then it is evident that there are at most $4(k+\ell)-2$ independent variables. \cref{lem:not_codim_1} below shows that whenever $\D_j$ is a strict same end parent of $\D_i$ and $\{\gamma_i,\delta_j\}\cap A\neq\emptyset,$ or  $\D_i$ has a sibling $\D_j$ which starts where $\D_i$ ends and $\{\gamma_i,\beta_j\}\cap A\neq\emptyset,$ then also $\phi_j=\phi_j$ or $\phi_i=\psi_j$, respectively, and again there are at most $4(k+\ell)-2$ independent variables. 

Thus, the only strata of $\partial_A S_\D\subseteq\Srem_\D$ which are of codimension $1$ are obtained when $A$ consists of a single domino entry, not one of those treated in the above paragraph, or that there is a single equivalence classes of $\psi_i,\phi_i$ variable which is of size $2$ and the other are of size $1.$ If $\phi_i=\phi_j$ where $\D_j$ is a strict same end ancestor of $\D_i,$ then also all intermediate $\phi_h$ are equal, and the same is true if $\phi_i=\psi_j$ but $\D_i,\D_j$ are not siblings. Thus, the case of a single equivalence class which is of size greater than $1,$ and that size is $2,$ is precisely when some $\eta_i$ or $\theta_i$ vanishes.
So only the vanishing of a single $\zeta_i\in\Var_\D$ may give rise to a codimension $1$ boundary. If this boundary is contained in $\Srem_\D$ then going over the cases of \cref{lem:Srem_cases}, we see it must be a regular matchable boundary. For such a boundary indeed there are $4(k+\ell)-1$ independent coordinates, and by \cref{cor:reg_domino_strata} they yield a parameterization of that dimension.

If the red or blue boundaries are not contained in $\SA,$ then by \cref{obs:bdry_matching_comb} they are regular matchable boundaries of another BCFW cell, hence the conclusion of the previous paragraph holds for them.
\end{proof}
\begin{obs}\label{lem:not_codim_1}
Let $C\vdots D$ be a loopy matrix in a weak-$\D-$domino form.
By \eqref{eq:eta},~\eqref{eq:theta} and \cref{def:L=1domino_signs}\[\gamma_i\delta_j=\eta_i+\gamma_j\delta_i,\]
and
\[\gamma_i\beta_j=\theta_i+\delta_i\alpha_j\]
and the three summands in each formula have the same sign.
Thus, if either one of the domino variables in the left hand side of one of the equation vanishes, also the two terms on the right vanish. 
\end{obs}
\begin{prop}\label{prop:mfld_w_bdry}
The union of $S_\D$ and its matchable boundaries that avoid $\SA$ is a smooth manifold with boundary. It is moreover a topological and a weak smooth submanifold with boundary of $\Gr_{k,n;\ell}.$
\end{prop}
\begin{proof}
$S_\D$ is a smooth manifold by \cref{thm:dominoAndBCFWparams}.
We will show below that the red boundary, if exists, is contained in $\Sred_\D.$ This will imply, using \cref{cor:regular_matchable_and_S_D_in_SREM} and \cref{obs:cell_not_red_blue}, that $S_\D$ and all its matchable boundaries which avoid $\SA$ are distinct. Hence, it will be enough to show to each matchable boundary $\partial_{\zeta_i}S_\D$ which avoids $\SA$ separately that $S_\D\cup\partial_{\zeta_i}S_\D$ is a manifold with boundary. 

For every regular matchable boundary $\partial_{\zeta_i}S_\D$, we can gauge fix some of the domino parameters, e.g. every $\varepsilon_h,$ so that exactly $4k$ remain unfixed, including $\zeta_i$, which is never $\varepsilon_h$ in the regular matchable case. Let $T$ denote this set, excluding $\zeta_i.$ Then, after possibly fixing the signs, these coordinates yield a parameterization $\R_+^T\times\R_{\geq 0}\to S_\D\cup\partial_{\zeta_i}S_\D$ where the boundary component is the image of $\R_+^T\times\{0\}$.

The proof for the blue boundary is similar, where we use the uniqueness of the weak domino representation of this boundary, even though it is not in $\Srem_\D,$ obtained in the end of the previous proof, as long as the blue boundary avoids $\SA$. The only change is that we gauge fix to $1$ another domino entry in the lowest blue row $\D_h$ other than $\varepsilon_h$, to keep $\varepsilon_h$ unfixed, 

Turning to the red boundary, assuming the red boundary avoids $\SA,$ 
we will act similarly to the proof of the red case of \cref{lem:from SREDSBLUE to smaller diags in diags order}. We assume $|\Red_\D|>1,$ since otherwise there is no red boundary, and write $\check{\D}=\shift_{\abe}\D.$
For a $U\vdots V\in S_\D$ we define a \emph{nearly $\check\D$ form} $\hat{C}\vdots\check{D}$ from its domino form $C\vdots D$ as follows. Let $k'$ be the index of the second-to-top red chord, so that $\tr(\D)=k'+1.$ Note that  $\tr(\check{\D})=k'.$ We can write $\Red_\D=\{k'-r+1,\ldots,k'-1,k',k'+1\},~\Blue_\D =\{k'-s,\ldots,k'-r\}$. By the second item of \cref{lem:uniqueness_or_SA} the vectors $\ww_j,~j\in\Red_\D\cup\Blue_\D\cup\{\star\}$ of \cref{nn:ww} are linearly independent, and moreover the matrix made of these rows and the rows $C_j,~j\notin(\Red_\D\cup\Blue_\D)$ is of full rank $k+2.$ 
Every $C_j$, for $j\in\Red_\D\cup\Blue_\D$ is a linear combination of $D_0,D_\star,\pr_{a_jb_j}\ww_j$ and $\pr_{a_{\p(j)}b_{\p(j)}}\ww_{\p(j)}.$
Write
\[u_{k'+1}=(-1)^{\below(\D_{0})+1}\frac{C_{k'}}{\delta_{k'}},~\text{and } \check{u}_{k'+1}=(-1)^{\below(\D_{k'})+\below(\D_{0})+1}\frac{\gamma_{k'}D_\star+\delta_{k'}D_0}{\delta_{k'}}.\]For $j\in\Red_\D,~j<k'$ define $u_j$ from $C_{j}$ by subtracting a multiple of $C_{k'}$ so that the result is a linear combination of $\pr_{a_jb_j}\ww_j,\pr_{a_{\p(j)}b_{\p(j)}}\ww_{\p(j)},\pr_{a_{0}b_{0}}\ww_{0},\pr_{a_{k'}b_{k'}}\ww_{k'},$ and $u=\pr_{c_\star d_\star} D_\star.$ 
These properties uniquely determine the vector, given the domino pair $C\vdots D.$ 
For blue $\D_j$ define $u_j~(\check{u}_j)$ by subtracting from $C_j$ the unique multiple of $u_{k'+1}=C_{k'}~(\check{u}_{k'+1})$ which sets to $0$ the projection on the starting entries of $\p(0)$, which are as usual considered to be just $n$ if $\p(0)=\emptyset.$
In other words, 
\begin{equation}\label{eq:passage to u and check u}
u^j= C_j-\frac{\delta_{k'}\delta_j\varepsilon_\star}{\gamma_{k'}\varepsilon_\star+\delta_{k'}\varepsilon_0}u_{k'},\qquad \check{u}^j= C_j-\frac{\delta_{k'}\delta_j\varepsilon_\star}{\gamma_{k'}\varepsilon_\star+\delta_{k'}\varepsilon_0}\check{u}_{k'}
\end{equation}
Again this condition determines uniquely $u_j.$ We define $\check D_0$ by adding to $\ww_{k'}$ the unique multiple of $\ww_0$ which sets the projection on the starting entries of $\p(0)$ to $0.$ Finally, we set $\check D_\star=-\frac{\varepsilon_\star}{\varepsilon_0}\ww_0.$  

Define the $1-$loop matrices $\check C\vdots \check D,~\hat C\vdots \check{D}$ where $\check{D}$ has rows $\check{D}_0,\check{D}_\star$ and $\check{C},\hat{C}$ has rows labeled by $[k+1]\setminus{k'},$ where \[\check{C}_j=\begin{cases}\check{u}_j,&\text{if defined, else}\\
    u_j,&\text{if $u_j$ is defined}\\
    C_j,&\text{otherwise}.\\\end{cases}\qquad\qquad \hat{C}_j=\begin{cases}{u}_j,&\text{if defined}\\ C_j,&\text{otherwise}.\\\end{cases}\]
By \cref{obs:for_red_limits} $\check{C}\vdots\check{D}$ is a weak domino representative of an element in $\partial_{\gamma_\star}S_{\check\D},$ which is contained in $\Srem_{\check\D},$ hence has a unique-up-to-gauge weak $\check\D-$domino form by \cref{cor:regular_matchable_and_S_D_in_SREM} and \cref{lem:uniqueness_or_SA}. It is easily seen then every element of $\partial_{\gamma_\star}S_{\check\D},$ can be obtained this way from elements of $S_\D.$ It is also straight forward to see that the nearly $\check\D$ form is unique up to gauge, since the $\D-$domino form is unique up to gauge. 

Note that for $j\in\Blue_\D$, using \cref{eq:passage to u and check u} it holds that
\begin{equation}\label{eq:check u minus u 1}\check{u}_j=\alpha_j\ee_{a_j}+\beta_j\ee_{b_j}+\varepsilon_j(\alpha_{\p(j)}\ee_{a_{\p(j)}}+\beta_{\p(j)}\ee_{b_{\p(j)}})+\gamma^{\check\D}_j u^{\check\D}_\star+\delta_j^{\check\D}\check{D}_\star\qquad u_j=\check{u}_j-\delta_j^{\check\D}\frac{\check{D}_0}{\delta_{k'}}\end{equation}
where
\begin{equation*}\gamma^{\check\D}_j =\gamma_j,~u^{\check\D}_\star=u_\star,~\text{and }\delta_j^{\check\D}=\frac{\delta_j\delta_{k'}\varepsilon_0}{\gamma_{k'}\varepsilon_\star+\delta_{k'}\varepsilon_0},\end{equation*}and when we don't put superscript we mean the variables of $\D.$
Also for $j=k'+1$
\begin{equation}\label{eq:check u minus u 2}u_{k'+1}=\check{u}_{k'+1}+\frac{\check{D}_0}{\delta_{k'}}.\end{equation}
We will now use this to write a smooth parameterization of $S_\D\cup\partial_{\delta_0}S_\D:$ parameterize $\partial_{\delta_0}S_\D$ using the weak $\check{\D}-$parameterization obtained from gauge fixing and picking the correct number of domino coordinates. Let $T$ be this set of coordinates, whose signs are fixed to be taken in $\R_+$. Add a parameter $t,$ which is secretly $(-1)^{\below(k')}\frac{1}{\delta_{k'}}.$ And write a parametrization $F:\R_+^T\times\R_{\geq 0}\to S_\D\cup\partial_{\delta_0}S_\D$ which takes $(t,(\zeta_i)_{i\in A})$ to $(\hat{C}\vdots\check{D})(t,(\zeta_i)_{i\in A}),$
where $(\check{C}\vdots\check{D})((\zeta_i)_{i\in A})\in\partial_{\delta_0}S_\D=\partial_{\gamma_\star^{\check\D}}S_{\check\D}$ is the weak $\check{D}-$domino form associated to the parameters $(\zeta_i)_{i\in A},$ and $\hat{C}(t,(\zeta_i)_{i\in A})$ is obtained from $\check{C}(t,(\zeta_i)_{i\in A})$ by changing the rows $\{k'+1\}\cup\Blue_{\D}$ by the transformations \eqref{eq:check u minus u 1},~\eqref{eq:check u minus u 2} with $\frac{(-1)^{\below(k')}}{\delta_{k'+1}}:=t.$ 
Clearly the limit $t\to0$ yields the red boundary.

Note that when $t=0$ $u_{k'+1}=\check{u}_{k+1}\in\hat{C}=\check{C}.$ But $u_{k'+1}\in\Span(D_0,\D_\star),$ hence $\partial_{\delta_0}S_\D\subseteq\Sred_\D.$

For the 'moreover' part it is enough to show that the union of $S_\D$ and one such matchable boundary $\partial_{\zeta_i}S_\D$ is a weak smooth and a topological submanifold with boundary. The above parameterizations, which continue to the boundaries, show the weak smooth submanifold property. 
In order to show that the union is a topological submanifold, it is enough again to show this property for $S_\D\cup\partial_{\zeta_i}S_\D.$ This means showing that the inverse map from this space to $\R^{4(k+\ell)}$ is continuous. The proof is a minor modification of the explicit construction of the preimage done in the coarse of proof of \cref{thm:dominoAndBCFWparams}, using the row-by-row construction of preimage via \cref{lem:4biden}. We omit the details. 
\end{proof}
We write the next corollary of the above proof for later reference.
\begin{cor}\label{cor:red_bdry_in_sred}
    Let $\D$ be a chord diagram. If the red boundary exists and is not contained in $\SA$ then it is contained in $\Sred_\D.$
\end{cor}

\section{Twistors, functionaries and the $B$-amplituhedron}
\subsection{Twistor coordinates}
\begin{definition}\label{def:cones}
Recall \cref{def:pluckers_and_cone}. The amplituhedron map can be naturally lifted, by using the same formula, to a map between cones
\[\tZ:\CC\Gr_{k,n;\ell}\to\CC\Gr_{k,k+4;\ell}.\]We refer to the image of the map as the \emph{cone} $\CC\Ampl_{n,k,4}^\ell$ over $\Ampl_{n,k,4}^\ell.$ It is endowed with a natural projection $\Pi:\CC\Ampl_{n,k,4}^\ell\to \Ampl_{n,k,4}^\ell$ which is the restriction of the map $\Pi:\CC\Gr_{k,k+4;\ell}\to\Gr_{k,k+4;\ell}$ and
\[\forall Y\vdots L\in\CC\Ampl_{n,k,4}^1,~\tZ(\Pi(Y\vdots L))=\Pi(\tZ(Y\vdots L)),\qquad \forall Y\in\CC\Ampl_{n,k,4}^0,~\tZ(\Pi(Y))=\Pi(\tZ(Y)).\] 
\end{definition}
\begin{definition}[Twistor coordinates]
Fix
  $Z \in \Mat_{n,k+4}^{>}$
   with rows  $Z_1,\dots, Z_n \in \R^{k+4}$.
   Consider $Y \in \CC\Gr_{k,k+4}$ or $(Y\vdots L)\in \CC\Gr_{k,k+4;1},$ and bases $y_1, \dots,y_k$ for (a representative of) $Y,$ and $l_1,l_2$ for (a representative of) $L.$
   For $\{i_1,\dots, i_4\} \subset [n]$
   we define the \emph{(momentum) twistor coordinate}, denoted
	$$\llrr{Y Z_{i_1} Z_{i_2} Z_{i_3} Z_{i_4}} \text{ or }
	\llrr{{i_1} {i_2} i_3 {i_4}}$$
        to be  the determinant of the
        matrix with rows
        $y_1, \dots,y_k, Z_{i_1}, \dots, Z_{i_4}$.
Similarly, for $i_1,i_2\in [n]$ we write
$$\llrr{Y Z_AZ_B Z_{i_1} Z_{i_2} },~\llrr{YLZ_{i_1}Z_{i_2}},~\llrr{(Y+L)Z_{i_1}Z_{i_2}} \text{ or }
	\llrr{AB{i_1} {i_2}}$$
        to be the determinant of the
        matrix with rows
        $y_1, \dots,y_k, l_1,l_2, Z_{i_1}, Z_{i_2}$. For reasons that will be clarified later, we collectively refer to these two types of twistors as \emph{permissible}.  This notation will be justified in \cref{rmk:justification} and the analysis of the forward limit operation in the next subsections.
        We sometimes refer to twistors which involve the loop as \emph{loopy twistors}. We extend the notations to allow $A,B $ to be in various places in the twistors by requiring antisymmetry with respect to permuting the indices.        
         We say that an index $i$ \emph{participates} or \emph{appears} in a twistor coordinate $\llrr{YZ_{i_1}\ldots Z_{i_4}},$ if $i\in\{i_1,\ldots,i_4\}.$ In this case we will also say that the twistor contains $i.$ 
 
        These definitions extend naturally to more general index sets.
\end{definition}
We will often be slightly inaccurate and consider the twistors as functions on $\Gr_{k,k+4;\ell}.$ We will also abuse notations and refer to twistors of $\tZ(C\vdots D)$ as twistors of the point $C\vdots D$ in $\CC\Gr_{k,n;\ell}$ or $\Gr_{k,n;\ell}.$ 

Note that using Laplace expansion, we can express 
the twistor coordinates of $Y, ~Y+L$ as linear functions in the Pl\"ucker coordinates of $Y$ and $Y+L$ resp. whose coefficients are the $4 \times 4$ and $2\times 2$ minors of $Z$ resp. Thus, twistors of the form $\llrr{I},~I\in\binom{[n]}{4}$ are functions on $\CC\Gr_{k,k+4;1}$ which factor through the natural forgetful map $\CC\Gr_{k,k+4;1}\to\CC\Gr_{k,k+4}.$ 

We will use intensively the following Cauchy-Binet formula for twistors.
\begin{lemma}
	[The Cauchy-Binet expression for twistors]
	\label{lem:Cauchy-Binet} Let $(C\vdots D) \in \Gr_{k,n;\ell}^{\geq},~\ell\in\{0,1\}$ and set $Y:=\tZ(C)=CZ,$ and if $\ell=1$ also $L:=\tZ(D)=DZ$.
Then
 \begin{align}\label{eq:cauchy-binet}
	 \llrr{Y~Z_{i_1} \dots Z_{i_4}} &=
                \sum_{J=\{j_1<\dots<j_k\} \in {[n] \choose k}} \lr{C}_J \, \lr{Z}_{j_1, \dots, j_k, i_1,\dots, i_4}\\ 
		\notag& =\sum_{J \in {[n] \choose k}, \  J \cap I = \emptyset} \lr{C}_J \, 
		\lr{Z}_{j_1, \dots, j_k, i_1,\dots, i_4},
        \end{align}
	where the second equality follows from the first by noting that 
		$\lr{Z}_{j_1, \dots, j_k, i_1,\dots, i_4} = 0$
		if $\{j_1,\dots,j_k\}$
		and $\{i_1,\dots,i_4\}$ intersect.
        \begin{align}\label{eq:cauchy-binet_loop}
	 \llrr{Y~Z_{i_1} Z_{i_2}Z_AZ_B} :=\llrr{(Y+L)Z_{i_1}Z_{i_2}}=&
                \sum_{J=\{j_1<\dots<j_{k+2}\} \in {[n] \choose {k+2}}} \lr{C+D}_J \, \lr{Z}_{j_1, \dots, j_{k+2}, i_1, i_2}\\ 
		\notag& =\sum_{J \in {[n] \choose k+2}, \  J \cap I = \emptyset} \lr{C+D}_J \, 
		\lr{Z}_{j_1, \dots, j_{k+2}, i_1,i_2},
        \end{align}
	where again the second equality follows from the first by noting that 
		$\lr{Z}_{j_1, \dots, j_k, j_{k+2} i_1,i_2} = 0$
		if $\{j_1,\dots,j_{k+2}\}$
		and $\{i_1,i_2\}$ intersect.
\end{lemma}

\subsection{The $B-$amplituhedron}
 In \cite{karpwilliams} Karp and Williams define the $B-$amplituhedron, a space which is isomorphic to the amplituhedron, and that under this isomorphism twistor coordinates are mapped to Pl\"ucker coordinates. This construction can be performed also in the loopy case.
\begin{definition}\label{def:B_amp}
Fix non negative integers $k,m,n,\ell$ with $n\geq k+m.$
For $W\in \Gr_{k+m,n}^{>}$ we define the \emph{$B$-loop amplituhedron} $\mathcal{B}_{n,k,m}^\ell(W)$ to be
\[\{(\bz,\bw_1,\ldots,\bw_\ell)\in\Gr_{m,n}\times(\Gr_{2,n})^{\times\ell}|~\exists (U\vdots V_1,\ldots,V_\ell)\in\Gr_{k,n;\ell}^{\geq}~\text{s.t }\bz= U^\perp\cap W,~\bw_i=(U+V_i)^\perp\cap W,~i\in[\ell]\}.\]
$\mathcal{B}_{n,k,m}^\ell(W)$ is the image of $\Gr_{k,n;\ell}^\geq$ under the map $\tZ_B$ which maps
\[(U\vdots V_1,\ldots,V_\ell)\mapsto( U^\perp\cap W,(U+V_1)^\perp\cap W,\ldots, (U+V_\ell)^\perp\cap W),\]where $W$ is the column span of $Z.$
We denote elements of $\mathcal{B}_{n,k,m}^\ell(W)$ by $(\bz,\bw_1,\ldots,\bw_\ell),$ and think of $\bz$ as a matrix with \emph{columns} $z_1,\ldots,z_n\in\R^m,$ up to the left action of $\GL_m$, and of each $\bw_i$ as a matrix with columns $(\bw_i)_1,\ldots,(\bw_i)_n\in\R^2,$ up to the right action of $\GL_2.$ Note that each $\bw_i\subset\bz.$
The reason we denote the columns of $\bz$ by $z_i$ is that they can be thought of as projections of the rows $Z_i$ of $Z.$
\end{definition}
\begin{rmk}
If $Z^T$ is a positive matrix representation for $W$ then by \cite[Section 4]{karp2017sign} $CZ$ is of dimension $\dim(C)$ whenever $C$ is nonnegative of dimension smaller or equal $k+m.$ This implies that for such $C$ \[\dim(C^\perp\cap W)=k+m-\dim(C).\]
We now restrict to the case $m=4,\ell=1$. In this case $\mathcal{B}_{n,k,4}^1(W)$ is a subspace of $\FL_{2,4;n}$, the two step flag variety of  flags of vector spaces of dimensions $2$ and $4$ in $\R^n,$ and denote its elements by $(\bz,\bw).$ 
We define the \emph{cones} $\CC\FL_{2,4;n}$ and $\CC\mathcal{B}_{n,k,4}^1(W)$ over $\FL_{2,4;n}$ and $\mathcal{B}_{n,k,4}^1(W)$ respectively in complete analogy to the corresponding \cref{def:pluckers_and_cone}, and denote the elements of these cones by the same notation $(\bz,\bw).$

For elements $(\bz,\bw)\in\CC\FL_{2,4;n}$ we have two types of Pl\"ucker coordinates, which we will sometimes refer to as \emph{permissible Pl\"ucker coordinates}. 
For every distinct $a,b,c,d\in[n]$ we write 
\begin{align*}\lr{(\bz,\bw);ab}=\lr{\bw;ab}:&=\lr{ab}:=\lr{w_aw_b},\\\lr{\bz;abcd}:&=\lr{abcd}:=\lr{z_az_bz_cz_d}\end{align*}
for the Pl\"ucker coordinates of $\bw,\bz$ respectively. We shall sometimes consider them as functions over the partial flag variety, in which case each of the two sets of Pl\"ucker coordinates is only defined up to a common scale: 
Changing the representatives scales all coordinates for each type by a scalar.
The relations between these coordinates are generated by the usual Pl\"ucker relations for $4-$spaces, for $2-$spaces and in addition
\begin{align*}\forall a,b,c,d,e,f\in[n],~~~~&\lr{abcd}\lr{ef}-\lr{abce}\lr{df}+\lr{abcf}\lr{de}=0\\&\lr{abcd}\lr{ef}-\lr{abce}\lr{df}+\lr{abde}\lr{cf}-\lr{acde}\lr{bf}+\lr{bcde}\lr{af}=0.\end{align*}
These functions generate the ring of functions for the flag variety, see \cite[Proposition 3.1.6]{weyman2003cohomology} or \cite{fulton1997young}.

$\mathcal{B}^1_{n,k,4}(W)$ is embedded in $\FL_{2,4}(W)\subseteq\FL_{2,4;n}=\Gr_{2,n;1},$ the two step flag variety of $2$-planes-inside $4-$planes-inside $W,$ and it is straightforward to see that as we let $W$ vary in $\Gr_{k+4,n}^>$, $\mathcal{B}_{n,k,4}^1(W)$ covers an Hausdorff-open, hence Zariski dense, subspace of $\FL_{2,4;n}.$ Thus, $W-$independent relations satisfied by the Pl\"ucker coordinates are exactly the same relations satisfied by the Pl\"ucker coordinates for the $2-$step flag variety written above.
\end{rmk}
The following theorem is essentially a translation of \cite[Lemma 3.10, Proposition 3.12]{karpwilliams} to our setting.
\begin{thm}\label{lem:B_amp_A_amp_dual}For $Z\in\Mat_{(k+4)\times n}^{>}$ denote by $W$ its row span. 
Then the map $f_Z:\FL_{2,4}(W)\dashrightarrow\Gr_{k,k+4;1}$ defined by
\[(\bz,\bw)\to (\bz^\perp Z\vdots(\bw^\perp Z)/(\bz^\perp Z))\]is generically defined, well defined, and restricts to a diffeomorphism: \[f_Z(\B_{n,k,4}^1(W))\simeq\Ampl_{n,k,4}^1(Z).\]
The inverse map is given by
\[(Y\vdots L)\mapsto f_Z^{-1}(Y\vdots L):=\left(\tZ^{-1}(Y)^\perp,\tZ^{-1}(Y+L)^\perp\right).\]
Moreover, the following relations hold between the projective vectors of Pl\"ucker and twistor coordinates:
\[(\lr{\bz;abcd})_{a,b,c,d\in[n]} = (\llrr{Y Z_aZ_bZ_cZ_d})_{a,b,c,d\in[n]},~~(\lr{\bw;ab})_{a,b\in[n]} = (\llrr{(Y+L) Z_aZ_b})_{a,b\in[n]}.\]
We denote the induced map on function rings by $\psi_Z.$
\end{thm}
Most of the proof follows from the cases $m=4,m=2$ of the corresponding results of \cite{karpwilliams}, and we refer the reader there for details. The parts which are not shown in \cite{karpwilliams} are that $\bw^\perp Z$ contains $\bz^\perp Z,$ and that the quotient is two dimensional, but this is an easy exercise in linear algebra. Also \cite{karpwilliams} only claim a homeomorphism, but since their maps are smooth they are in fact diffeomorphisms.

As a corollary of the above remark and theorem we have

\begin{cor}\label{obs:relations_for_functionaries}
\begin{itemize}
\item The twistors change sign when two consecutive indices are swapped.
\item The twistors satisfy the following Pl\"ucker relations:
\begin{align*}
\forall i_1,\ldots,i_3,j_1\ldots,j_5\in[n],\sum_{s=1}^5(-1)^s\llrr{i_1i_2i_3j_s}\llrr{j_1\ldots\widehat{j_s}\ldots j_5}=0\end{align*}
\begin{align*}
\forall a,b,c,d\in[n],\llrr{abAB}\llrr{cdAB}-\llrr{acAB}\llrr{bdAB}+\llrr{adAB}\llrr{bcAB}=0\end{align*}
and
\begin{align*}
    \forall a,b,c,d,e,f\in[n],\llrr{abcd}\llrr{efAB}-\llrr{abce}\llrr{dfAB}+\llrr{abcf}\llrr{deAB}&=0,\\\llrr{abcd}\llrr{efAB}-\llrr{abce}\llrr{dfAB}+\llrr{abde}\llrr{cfAB}-\llrr{acde}\llrr{bfAB}+\llrr{bcde}\llrr{afAB}&=0
\end{align*}
\item The above relations generate all relations between twistor coordinates which are not $Z$-dependent, and are homogeneous with respect to the multigrading.
\end{itemize}
\end{cor}
\subsection{Embedding loopy Grassmannians into Grassmannians}
\cite{arkani-hamed_trnka} considered a certain map between Grassmannians and loopy Grassmannians, to which they had referred as \emph{hiding particles}. We now explain this construction, for the case of one loop.
\begin{defobs}\label{nn:varphi}
Let $N$ be an index set $p\lessdot A\lessdot B\lessdot n=\max{N}.$
Define \[\varphi_{AB}:\Mat_{k,N\setminus\{A,B\};1}\to\Mat_{k+2,N},\qquad\qquad (C\vdots D)\mapsto \begin{pmatrix}&D_{1,1}&\cdots&D_{1,p}&1&0&D_{1,n}\\
&D_{2,1}&\cdots&D_{2,p}&0&1&D_{2,n}\\
&C_{1,1}&\cdots&C_{1,p}&0&0&C_{1,n}\\
&\vdots&\vdots&\vdots&\vdots&\vdots&\vdots\\
&C_{k,1}&\cdots&C_{1,p}&0&0&C_{k,n}
\end{pmatrix},\]
that is, first writing the matrix $C+D$ where $D$ is in the upper two entries and $C$ in the lower $k$ entries, and then adding two new columns, positioned at $A,B$ which have the identity matrix in the upper two rows, and zero elsewhere.

Note that, for every $I\in\binom{N}{k+2}$ with $A,B\in I$
\begin{equation}\label{eq:phi_AB_1}
\lr{C}_{I\setminus\{A,B\}}=\lr{\varphi_{AB}(C\vdots D)}_I,
\end{equation}
while if $A,B\notin I$
\begin{equation}\label{eq:phi_AB_2}
\lr{C+D}_{I}=\lr{\varphi_{AB}(C\vdots D)}_I.
\end{equation}
If $C\vdots D$ is another representative of the same element of $\Gr_{k,N\setminus\{A,B\};1}$ then $\varphi_{AB}(C'\vdots D')$ may not equal $\varphi_{AB}(C\vdots D),$ but we have equalities of \emph{projective vectors}
\[(\lr{\varphi_{AB}(C'\vdots D'))}_{A,B\notin I}=(\lr{\varphi_{AB}(C\vdots D))}_{A,B\notin I},\qquad(\lr{\varphi_{AB}(C'\vdots D'))}_{A,B\in I}=(\lr{\varphi_{AB}(C\vdots D))}_{A,B\in I}.\]
Let $s:U\to\Mat_{k,N\setminus\{A,B\};1},$ where $U$ is an open set in $\Gr_{k,N\setminus\{A,B\};1}$ or $\CC\Gr_{k,N\setminus\{A,B\};1},$ be a local section of the projection map $\Pi.$ We write
\[\varphi_{AB}^s=\varphi_{AB}^{s,U} = \Pi\circ\varphi_{AB}\circ s : U\to\begin{cases}
    \Gr_{k+2,N},&\text{if }U\subset\Gr_{k,N\setminus\{A,B\};1}\\
    \CC\Gr_{k+2,N},&\text{if }U\subset\CC\Gr_{k,N\setminus\{A,B\};1}
\end{cases}\]We will sometimes be slightly inaccurate and omit the superscript, and just write $\varphi_{AB}$ for that map.

In addition, if $C\vdots D=\Gr_{k,N;1}^{\geq}$ is a \emph{nonnegative representative}, that is, a representative in which all maximal minors of $C,~C+D$ are either zero or positive, then by construction all maximal minors of $\varphi_{AB}(C\vdots D)$ which either include or exclude both indices $A,B$ are nonnegative, by \eqref{eq:phi_AB_1},~\eqref{eq:phi_AB_2}.

We also define the map
\[\tilde\varphi_{AB}:\Mat_{(N\setminus\{A,B\})\times(k+4)}\to\Mat_{N\times (k+6)},\qquad\qquad Z\mapsto\begin{pmatrix}
&0&0&Z_{1,1}&\ldots&Z_{1,k+4}\\
&\vdots&\vdots&\vdots&\vdots&\vdots\\
&0&0&Z_{p,1}&\ldots&Z_{p,k+4}\\
&1&0&0&0&0\\
&0&1&0&0&0\\
&0&0&Z_{n,1}&\ldots&Z_{n,k+4}
\end{pmatrix}.\]For every $I\in\binom{N}{k+6}$ with $A,B\in I$
\begin{equation}\label{eq:tilde_phi_AB_1}
\lr{Z}_{I\setminus\{A,B\}}=\lr{\tilde{\varphi}_{AB}(Z)}_I,
\end{equation}
while if $\{A,B\}\nsubseteq I$
\begin{equation}\label{eq:tilde_phi_AB_2}
\lr{Z}_{I}=0.
\end{equation}
Thus, if $Z$ is nonnegative also $\tilde\varphi(Z).$
\end{defobs}
\begin{definition}
    Write $\check\Gr_{k,n;1}$ for the subspace $\Gr_{k+2,[n]\cup\{A,B\}}^{(AB)}$ of $\Gr_{k+2,[n]\cup\{A,B\}},$ where $A,B$ are new markers taken to satisfy $n-1\lessdot A\lessdot B\lessdot n,$ in which the columns $A,B$ are linearly independent. We write $\CC\check\Gr_{k,n;1}$ for the corresponding cone, as in \cref{def:pluckers_and_cone}.
    We denote points of either of these two spaces by $E=(C+D; v^A; v^B),$ a notation that will be justified later.
    We have the map $\check\varrho:\CC\check\Gr_{k,n;1}\mapsto \CC\Gr_{k,n;1}$
    obtained by sending a point with Pl\"ucker coordinates $(P_I(E))_{I\in\binom{[n]\cup\{A,B\}}{k+2}}$ to the point $C\vdots D$ with \[P_I(C+D)=P_I(E),~\text{for }I\in\binom{[n]}{k+2},\qquad P_I(C)=P_{I\cup\{A,B\}}(E),~\text{for }I\in\binom{[n]}{k}.\]
    This map descends to the Grassmannian level, and we also denote it by $\check\varrho:\check\Gr_{k,n;1}\mapsto\Gr_{k,n;1}.$ We write
    \[\check\Gr_{k,n;1}^\geq=\check\varrho^{-1}(\Gr_{k,n;1}^\geq),~\check\Gr_{k,n;1}^>=\check\varrho^{-1}(\Gr_{k,n;1}^>),~\CC\check{\Gr}_{k,n;1}^\geq=\check\varrho^{-1}(\CC\Gr_{k,n;1}^\geq),~\CC\check\Gr_{k,n;1}^>=\check\varrho^{-1}(\CC\Gr_{k,n;1}^>)\]
    While these maps are not invertible, the maps $\varphi_{AB}^s$ of \cref{nn:varphi} provide local sections for $\check\varrho.$

The amplituhedron map lifts to the space $\check\Gr_{k,n;1}$ using the map $\tilde\varphi_{AB}$ from \cref{nn:varphi} by 
\[\check{Z}:\check\Gr_{k,n;1}\dashrightarrow\check\Gr_{k,k+4;1},\qquad (C+D; v^A; v^B)\mapsto (C+D; v^A;v^B){\tilde\varphi_{AB}(Z)}.\]
As in the case of the amplituhedron this map is well defined over
$\check\Gr_{k,n;1}^\geq,$ and the $\check Z$-image of this space is denoted $\check\Ampl_{n,k,4}^1(Z).$ We similarly define the cone $\CC\check\Ampl_{n,k,4}^1(Z)$ as the respective image of $\CC\check\Gr_{k,n;1}$ by the analogous map of right multiplication by $\tilde\varphi_{AB}Z,$ which we also write as $\check Z.$ 
On $\CC\check\Ampl_{n,k,4}^1(Z)$ one can define \emph{twistor coordinates} for $I\in\binom{[n]\cup\{A,B\}}{4}$ by letting 
\[\llrr{abcd}=\llrr{\tilde Y\tilde\varphi_{AB}(Z_{a})\tilde\varphi_{AB}(Z_b)\tilde\varphi_{AB}(Z_c)\tilde\varphi_{AB}(Z_d)}\] be the determinant of the $(k+4)\times(k+4)$ matrix obtained by taking a matrix representative of $\tilde Y$ and completing it to a square matrix by adding the $a,b,c,d$ rows of $\tilde\varphi_{AB}(Z)$ as the last four rows.
We will sometimes be slightly inaccurate and refer to the twistor coordinates as functions over $\check\Ampl_{n,k,4}^1,$ and sometimes as functions on $\Gr_{k,n;1}^\geq,~\CC\Gr_{k,n;1}^\geq$ via composition with the map $\check Z.$ A twistor coordinate is \emph{permissible} if $A,B$ either both appear or neither appear in the twistor's entries.

We similarly define the map
\[\check{Z}_B:\check\Gr_{k,n;1}\dashrightarrow\check\Gr_{2,n;1},\qquad (C+D; v^A; v^B)\mapsto (C+D; v^A; v^B)^\perp\cap\widetilde W,\]
where $\widetilde W$ is the column span of ${\tilde\varphi_{AB}(Z)}.$
The image of $\check\Gr_{k,n;1}^\geq$ under this map is denoted $\check{\mathcal{B}}_{n,k,4}^1( W).$ A Pl\"ucker coordinate of an element of $\check{\mathcal{B}}_{n,k,4}^1(W)$ is \emph{permissible} if it either includes or excludes both $A$ and $B.$ 
\end{definition}
\cref{lem:B_amp_A_amp_dual} has the following immediate corollary
\begin{cor}\label{obs:_check_B_amp_A_amp_dual}With the above notations,
$f_{\tilde\varphi_{AB}Z}:\check{\mathcal{B}}_{n,k,4}^1( W)\simeq\check\Ampl_{n,k,4}^1(Z)$. Under this isomorphism the projective vector of Pl\"ucker coordinates of an element of $\check{\mathcal{B}}_{n,k,4}^1( W)$ equals the projective vector of the twistor coordinates of its image in $\check\Ampl_{n,k,4}^1(Z).$
\end{cor}
\begin{rmk}\label{rmk:iso_between_different spaces}
For $n\geq k+2$ we also define the space $\hat\Gr_{k,n;1}$ as the space whose elements are $(U\vdots v_A\vdots v_B),$ where $U\in\Gr_{k,n}$ and $v_A,v_B$ are linearly independent vectors in $\R^n/U.$ This space is associated a map 
\[\hat\varrho:\hat\Gr_{k,n;1}\to\Gr_{k,n;1},\qquad (U\vdots v_A\vdots v_B)\mapsto (U\vdots \Span(v_A,v_B)).\]
There is a canonical diffeomorphism $\nu:\hat\Gr_{k,n;1}\simeq\check\Gr_{k,n;1}$ defined on representative by sending $(C\vdots D_A\vdots D_B)$ to the row span of $\varphi_{AB}(C\vdots D),$ where $D$ is a matrix whose rows are $D_A,D_B.$
This map is well defined since if $(C'\vdots D'_A\vdots D'_B)$ is another representative then $C' = g\cdot C,~g\in\GL_k$ and $D'_A-D_A,D'_B-D_B\in\Span(C),$ hence $(\varphi_{AB}(C\vdots D),\varphi_{AB}(C'\vdots D'))$ represent the same element of $\Gr_{k+2,[n]\cup\{A,B\}}.$ A similar argument shows that the map is invertible.

We can also similarly define the cone $\CC\hat\Gr_{k,n;1}$ and the analogues of the above statements hold between the cones, where we use the same notations for the maps.
We have that
\[P_I(\nu(U\vdots v_A\vdots v_B))=\begin{cases}P_I(C),&\text{if }A,B\in I\\
P_I(C+D_A+D_B),&\text{if }A,B\notin I\\(-1)^{k+1+|I\cap\{n\}|}P_I(C+D_B),&\text{if }A\in I\\
(-1)^{k+|I\cap\{n\}|}P_I(C+D_A),&\text{if } B\in I.
\end{cases}\]
We define the positive and nonnegative parts of $\hat\Gr_{k,n;1}$ as the $\hat\varrho-$preimages of $\Gr^>_{k,n;1},~\Gr_{k,n,;1}^\geq$ respectively, and similar definitions exist for the cones.
With this we can define the hat-version of the amplituhedron as follows.
Define the map \[\hat{Z}:\hat\Gr_{k,n;1}\to\hat\Gr_{k,k+4;1},\qquad (C\vdots D_A\vdots D_B)\mapsto (Y\vdots L_A\vdots L_B)=(C\vdots D_A\vdots D_B)Z.\]
The image of $\hat\Gr_{k,n;1}^\geq$ under this map is denoted $\hat\Ampl_{n,k,4}^1(Z).$
\end{rmk}
\begin{obs}\label{obs:commutative_diag}
The following diagram is commutative
\[
\begin{tikzcd}[row sep=3em, column sep=4em]
\hat\Gr_{k,n;1}^\geq \arrow[r, "\simeq"] \arrow[d] \arrow[rr, bend left=40] & \check\Gr_{k,n;1}^\geq \arrow[r] \arrow[d] \arrow[dd, bend right=30,, pos=0.4, yshift=4pt] & \Gr_{k,n;1}^\geq \arrow[d] \arrow[dd, bend right=30, pos=0.4, yshift=4pt] \\
\hat\Ampl_{n,k,4}^1(Z) \arrow[r, "\simeq"] \arrow[rr, bend left=40]\arrow[dr, "\simeq"', sloped] & \check\Ampl_{n,k,4}^1 (Z)\arrow[r] \arrow[d, "\simeq"', description] & \Ampl_{n,k,4}^1(Z) \arrow[d, "\simeq", description] \\
& \check{\mathcal{B}}_{n,k,4}^1(W) \arrow[r] & {\mathcal{B}}_{n,k,4}^1(W)
\end{tikzcd}
\]
where the top vertical maps are the amplituhedron maps, the vertical curvy maps are the $B-$amplituhedron maps, the bottom vertical maps are the isomorphisms of the amplituhedron and $B$-amplituhedron. The horizontal maps from the leftmost column to the middle one are $\nu,$ the horizontal maps from the middle column to the right one are $\check\varrho,$ and the curvy horizonal maps are $\hat\varrho,$ the map $\hat\Ampl_{n,k,4}^1(Z)\to\check{\mathcal{B}}^1_{n,k,4}(W)$ is $f_{\tilde\varphi_{AB}(Z)}^{-1}\circ\nu.$
The composed map $\hat\Gr_{k,n;1}^\geq\check\to{\mathcal{B}}_{n,k,4}(W)$ is given by
\[(C\vdots D_A\vdots D_B)\mapsto (\tZ_B(C);-\tZ_B(C)D_A^T;-\tZ_B(C)D_B^T).\]The same commutative diagram exists for the cones. 
\end{obs}
The proof is straightforward, where for the last statement we note that 
\[\begin{pmatrix}
    D_A&1&0\\
    D_B&0&1\\
    C&0&0
\end{pmatrix}^\perp = \begin{pmatrix}
    C^\perp; &-C^\perp D^T_A;&-C^\perp D^T_B
\end{pmatrix}\]
\begin{rmk}\label{rmk:justification}
A twistor coordinate on $\CC\hat\Ampl_{n,k,4}^1,\CC\check\Ampl_{n,k,4}^1$ is permissible if and only if it is pulled back from $\CC\Ampl_{n,k,4}^1$ by the maps of \cref{obs:commutative_diag}.

Explicitly, let $s:U\to \Mat_{k,n;1}$ be a local section of the projection map $\Mat_{k,n;1}\to\Gr_{k,n;1},$ where $U$ an open set. It induces local sections which we also denote by $s,$ ${s}:U\to\hat\Gr_{k,n;1},~{s}:U\to\check\Gr_{k,n;1}$ for $\hat\varrho,\check\varrho$ respectively.
If $(C\vdots D)\in U\subseteq\Gr_{k,n;\ell}^{\geq}$ then by plugging in the expansions \eqref{eq:cauchy-binet_loop},~\eqref{eq:cauchy-binet} the expressions \eqref{eq:phi_AB_1},~\eqref{eq:phi_AB_2},~\eqref{eq:tilde_phi_AB_1} and \eqref{eq:tilde_phi_AB_2}, we learn that independently of $s,$
\[\llrr{(Y+L)Z_{i_1}Z_{i_2}}=\llrr{\tilde Y\tilde{\varphi}_{AB}(Z)_{i_1}\tilde{\varphi}_{AB}(Z)_{i_2}\tilde{\varphi}_{AB}(Z)_{A}\tilde{\varphi}_{AB}(Z)_{B}},\]
while
\[\llrr{YZ_{i_1}Z_{i_2}Z_{i_3}Z_{i_4}}=\llrr{\tilde Y\tilde{\varphi}_{AB}(Z)_{i_1}\tilde{\varphi}_{AB}(Z)_{i_2}\tilde{\varphi}_{AB}(Z)_{i_3}\tilde{\varphi}_{AB}(Z)_{i_4}},\]
where $(Y\vdots L)=\tZ(C\vdots D)$ and $\tilde Y=\widetilde{\tilde{\varphi}_{AB}(Z)}(\varphi^s_{AB}(C\vdots D))=\varphi^s_{AB}(C\vdots D)\tilde{\varphi}_{AB}(Z)$.

A similar statement holds for $B-$amplituhedra and the corresponding permissible Pl\"ucker coordinates, using \cref{lem:B_amp_A_amp_dual} and \cref{obs:_check_B_amp_A_amp_dual}.
\end{rmk}
\begin{nn}\label{nn:extending notation l=0}
We extend the notations of the $\hat{}$ and $\check{}$ spaces to the case $\ell=0$ by defining $\check\Gr_{k,n;0}=\hat\Gr_{k,n;0}=\Gr_{k,n;0}=\Gr_{k,n},$ and make similar identifications in the level of nonnegative parts and the various amplituhedra. 
\end{nn}
\subsection{Functionaries}
\begin{definition}
\label{def:functionary}
Denote by $\mathcal{R}_{k,N}^{\ell}~(\mathcal{PR}_{k,N}^{\ell})$  for $\ell=0,1$ the polynomial ring generated by the (permissible) twistor coordinates of $\CC\check\Ampl_{N,k,4}^\ell$, and let $\mathcal{FR}_{k,N}^{\ell}~(\mathcal{FPR}_{k,N}^{\ell})$ be its field of fractions. In all cases, if $N=[n]$ we replace the subscript $N$ by $n$. By \cref{rmk:iso_between_different spaces} and \cref{obs:commutative_diag} these rings can be equivalently thought of as rings of functions over $\CC\hat\Ampl_{n,k,4}^\ell.$
These rings are graded by \emph{degree}, as usual for polynomial rings. A homogeneous element of $\mathcal{FR}_{k,N}^{\ell}~(\mathcal{FPR}_{k,N}^{\ell})$ is called a \emph{rational (permissible) functionary}. If it belongs to $\mathcal{R}_{k,N}^{\ell}~(\mathcal{PR}_{k,N}^{\ell})$ we will sometimes refer to it as a \emph{(permissible) polynomial functionary}. We will often omit the words rational or polynomial and just refer to the function as a functionary. By \cref{rmk:justification} a functionary is permissible if and only if it is descends to $\CC\Ampl_{n,k,4}^\ell$ via any local section $s$ from  $s:U\to \CC\check\Ampl_{n,k,4}^\ell$ for an open subset $U\subset \CC\Ampl_{n,k,4}^\ell,$ independently of the section. This is equivalent to requiring that it is invariant under the $\SL_2$ group action on $L$ (or $D$) acting by a change of basis. 

We can similarly define the rings $\mathcal{BR}_{k,N}^{\ell},\mathcal{BPR}_{k,N}^{\ell},\mathcal{BFR}_{k,N}^{\ell}\mathcal{BFPR}_{k,N}^{\ell}$ as the corresponding rings of rational functions in the Pl\"ucker coordinates of $\check{\mathcal{B}}_{n,k,4}^\ell$, and again permissibility is equivalent descending to $\CC\mathcal{B}_{n,k,4}^\ell$ via any local section from an open subspace of $\CC\mathcal{B}_{n,k,4}^\ell$ to $\CC\check{\mathcal{B}}_{n,k,4}^\ell,$ independently of the section. We sometimes refer to these functions as Pl\"ucker functionaries (which can be polynomial, rational permissible etc.). The $\mathcal{B}-$versions of the rings are canonically isomorphic to the $\mathcal{B}-$less versions by \cref{lem:B_amp_A_amp_dual} and \cref{obs:_check_B_amp_A_amp_dual}.

All these rings have also a finer multigrading defined as follows. A polynomial functionary $F\in \mathcal{R}_{k,N}^{\ell}$ is called \emph{pure} if for each index $i\in N\cup\{A,B\}$ there is a nonnegative integer $d_i$ such that for every monomial in $F$ $i$ appears in exactly $d_i$ of its factors. In this case we write $\deg_i(F)=d_i.$ This finer multidegree comes from the weights of the torus actions scaling the rows of $Z$ and the columns of $C\vdots D.$ Alternatively, this is the torus action scaling the columns of the elements of the associated $B-$amplituhedron. A rational functionary is pure if it can be presented as the quotient of two polynomial functionaries.

Throughout
this text we will only be considering functionaries which are pure. 
\end{definition}
Observe that for any pure rational functionary $F$
\[\sum_{i}\deg_i(F)=4\deg(F).\]
\begin{ex}
    The functionary \[F=\frac{\llrr{1234}\llrr{1256}-\llrr{1235}\llrr{1246}}{\llrr{1345}^2}\] is pure of degree $0,$ and has
    \[\deg_1(F)=0,~\deg_2(F)=2,\deg_3(F)=\deg_4(F)=\deg_5(F)=-1,\deg_6(F)=1.\]
    The functionary $\llrr{1234}\llrr{1256}-\llrr{1235}\llrr{1247}$ is not pure.
\end{ex}
Especially useful examples of functionaries are the following quadratics: 
\begin{nn}\label{nn:useful_quad}
Write $\llrr{abc|de|fgh}$ for $\llrr{abcd}\llrr{efgh}-\llrr{abce}\llrr{dfgh}.$
\end{nn}
\section{Promotions of functionaries and positivity}
\subsection{Promotions}
We now define certain maps between function fields that will be analogous to the three geometric maps used in the construction of a BCFW cell.  We extend the promotions defined in \cite{even2021amplituhedron,even2023cluster} to the loopy case, and include a promotion that was not discussed in these texts - the forward limit promotion.  
\begin{definition}\label{def:pre_promotion}
Let $N$ be an index set, and $i$ be a marker of $N.$
The map $\Psi^{\pre_i}:\mathcal{FR}_{k,N\setminus \{i\}}^{\ell}\to\mathcal{FR}_{k,N}^{\ell}$ is defined by sending every generator $\llrr{i_1,\ldots,i_4}$ of $\mathcal{FR}_{k,N\setminus \{i\}}^{\ell}$ to the generator $\llrr{i_1,\ldots,i_4}$ of $\mathcal{FR}_{k,N}^{\ell}.$ In the case $i=p,$ the penultimate marker of $N$ we write $\Psi^{\pre_i}$ simply as $\Psi^{\pre}.$

The map $\Psi^{\inc_i}:\mathcal{FR}_{k,N\setminus \{i\}}^{\ell}\to\mathcal{FR}_{k+1,N}^{\ell}$ is defined by sending every generator $\llrr{i_1,\ldots,i_4}$ of $\mathcal{FR}_{k,N\setminus \{i\}}^{\ell}$ to the generator $\llrr{i_1,\ldots,i_4}$ of $\mathcal{FR}_{k+1,N}^{\ell}.$ 
\end{definition}
\begin{definition}\label{def:bcfw_prom}
With the notations of \cref{nn:bcfwmap}, define the \emph{BCFW promotion}
\[\Psi^{\bcfw}_{ac}:\mathcal{FR}_{k_L,N_L}^{\ell_L}\times \mathcal{FR}_{k_R,N_R}^{\ell_R}\to \mathcal{FR}_{k,n}^{\ell}\]
defined by the following substitution:
\begin{align*}
&  \text{on $\mathcal{FR}_{k_L,N_L}^{\ell_L}$:} && b \;\mapsto\; b - \frac{\llrr{b\,c\,d\,n}}{\llrr{a\,c\,d\,n}}\,a 
\\
& \text{on $\mathcal{FR}_{k_R,N_R}^{\ell_R}$:} && n \;\mapsto\; \frac{\llrr{a\,c\,d\,n}}{\llrr{a\,b\,c\,d}}\,b - \frac{\llrr{b\,c\,d\,n}}{\llrr{a\,b\,c\,d}}\,a 
\\ 
& \text{on $\mathcal{FR}_{k_R,N_R}^{\ell_R}$:} && d \;\mapsto\; d - \frac{\llrr{a\,b\,d\,n}}{\llrr{a\,b\,c\,n}}\,c 
\end{align*} 
In the end case $a=1,$ the left factor becomes trivial, and the map restricts to the \emph{upper promotion}
\[\Psi^{\bcfw}_{ac}:\mathcal{FR}_{k_R,N_R}^{\ell_R}\to \mathcal{FR}_{k,n}^{\ell}\]defined using the two lower substitutions.
\end{definition}
The meaning of the above substitution is that whenever we see the index $b$ in a twistor coming from the left side, $\llrr{bi_1i_2i_3}$ we substitute $b - \frac{\llrr{b\,c\,d\,n}}{\lr{a\,c\,d\,n}}\,a$ instead of $b$ to obtain
\[\llrr{b,i_1,i_2,i_3}-\frac{\llrr{b\,c\,d\,n}}{\llrr{a\,c\,d\,n}}\llrr{a,i_1,i_2,i_3},\]and perform similar expansions with respect to $d,n$ for twistors coming from the right side. See \cite{even2023cluster} for more details.
\begin{rmk}\label{obs:BCFW_promotion_domino}
The following equality is an immediate consequence of the Pl\"ucker relations:
\[n - \frac{\llrr{a\,b\,c\,n}}{\llrr{a\,b\,c\,d}}\,d + \frac{\llrr{a\,b\,d\,n}}{\llrr{a\,b\,c\,d}}\,c  \;=\; \frac{\llrr{a\,c\,d\,n}}{\llrr{a\,b\,c\,d}}\,b - \frac{\llrr{b\,c\,d\,n}}{\llrr{a\,b\,c\,d}}\,a.\]
One can equivalently write the BCFW promotion \cref{def:bcfw_prom} using the same equations for $b,d$ but on $\mathcal{FR}_{k_R,N_R}^{\ell_R}$: make the following substitution for $n$\[n \;\mapsto\; =\;  n-\frac{\llrr{a\,b\,c\,n}}{\llrr{a\,b\,c\,d}}\,d + \frac{\llrr{a\,b\,d\,n}}{\llrr{a\,b\,c\,d}}\,c .\]
This definition was used in \cite{even2023cluster}, and the two definitions are of course equivalent. The convention we use here is more suited for domino forms, while the one of \cite{even2023cluster} is more natural when working with BCFW coordinates. 
\end{rmk}

\begin{definition}\label{def:FL_prom}
Let $N=\{\min N,\ldots,c,d,A,B,n\}$ be (an ordered) index set. Define the \emph{forward limit promotion}
\[\Psi^{\FL}:\mathcal{FR}_{k+1,N}^{0}\to \mathcal{FR}_{k,N
}^{1}\]
defined by the following substitution:
\begin{align*}
& A \;\mapsto\; A - \frac{\llrr{c\,d\,A\,n}}{\llrr{c\,d\,B\,n}}\,B 
\\
& d \;\mapsto\; d - \frac{\llrr{d\,A\,B\,n}}{\llrr{c\,A\,B\,n}}\,c .
\end{align*}
\end{definition}
\begin{rmk}\label{rmk:vector_promotion}
The above promotions are induced from geometric maps which we can think of as \emph{vectors promotions}, best seen in the level of $B$-amplituhedron using the isomorphism of \cref{def:functionary} between the ring of functionaries and the ring of Pl\"ucker functionaries.

The map $\Psi^{\pre_i}$, performed on Pl\"ucker functionaries as the canonical isomorphism of \cref{def:functionary} allows, is induced from the geometric map, which we denote with the same notation,
\[\Psi^{\pre_i}:\check\Gr_{4-2,N;\ell}\to\check\Gr_{4-2\ell,N\setminus\{i\};\ell},\qquad
\bz\mapsto\rem_i\bz,~\text{and if }\ell=1,~\text{also }\bw\mapsto\rem_i\bw.\]
The map descends to the (loopy) Grassmannians level.

$\Psi^{\bcfw}_{ac}$ is induced from the geometric map, also denoted
$\Psi^{\bcfw}_{ac}:\check\Gr_{4-2,N;\ell}\to\check\Gr_{4-2\ell_L,N_L;\ell_L}\times \check\Gr_{4-2\ell_R,N_R;\ell_R}$
\begin{align*}
&  \text{on the left component:} && z_b \;\mapsto\; z_b - \frac{\lr{b\,c\,d\,n}}{\lr{a\,c\,d\,n}}\,z_a 
\\
& \text{on the right component:} && z_n \;\mapsto\; \frac{\lr{a\,c\,d\,n}}{\lr{a\,b\,c\,d}}\,z_b - \frac{\lr{b\,c\,d\,n}}{\lr{a\,b\,c\,d}}\,z_a 
\\ 
&  && z_d \;\mapsto\; z_d - \frac{\lr{a\,b\,d\,n}}{\lr{a\,b\,c\,n}}\,z_c
\end{align*}
for all other indices in the left or right component $z_i\mapsto z_i.$ 
The map descends to the loopy Grassmannians level, where
\begin{align*}
&\text{if $\ell_L=1$:} &&
w_b \;\mapsto\; w_b - \frac{\lr{b\,c\,d\,n}}{\lr{a\,c\,d\,n}}\,w_a
\\&\text{if $\ell_R=1$:}
&&w_n \;\mapsto\; \frac{\lr{a\,c\,d\,n}}{\lr{a\,b\,c\,d}}\,w_b - \frac{\lr{b\,c\,d\,n}}{\lr{a\,b\,c\,d}}\,w_a 
\\ 
&  && w_d \;\mapsto\; w_d - \frac{\lr{a\,b\,d\,n}}{\lr{a\,b\,c\,n}}\,w_c
\end{align*}
and for every other $i,~w_i\mapsto w_i.$

The map $\Psi^{\FL}$ is induced from the geometric map
$\Psi^{\FL}:\check\Gr_{2,N\setminus\{A,B\};1}\to\Gr_{4,N}$ via
\begin{align*}
& z_A \;\mapsto\; z_A - \frac{\lr{c\,d\,A\,n}}{\lr{c\,d\,B\,n}}\,z_B 
\\
& z_d \;\mapsto\; z_d - \frac{\lr{d\,A\,B\,n}}{\lr{c\,A\,B\,n}}\,z_c,
\end{align*}
and again for all other indices in the target $z_i\mapsto z_i.$

In all these cases the induction of maps is by first moving to the rings of functions and then using the isomorphisms \cref{lem:B_amp_A_amp_dual},~\cref{obs:_check_B_amp_A_amp_dual}, and it holds that
\[\Psi^\bullet\llrr{{i_1}\ldots,{i_4}}=\Psi^\bullet\lr{{i_1}\ldots,{i_4}}=\lr{\Psi^\bullet(z_{i_1})\ldots\Psi^{\bullet}(z_{i_4})},~\bullet\in\{\pre,\bcfw,\FL\},\]
and moreover, for $\bullet\in\{\pre,\bcfw,\FL\},$ and functionary valued row $A$-vectors $v_1,\ldots,v_4$
\begin{equation}\label{eq:promotion_for_coords_prom}
    \Psi^\bullet\lr{{v_1}\cdot\bz,\ldots,{v_4}\cdot\bz}=\lr{\Psi^\bullet(v_1\cdot\bz),\ldots,\Psi^\bullet(v_4\cdot\bz)}
    ,
\end{equation}
where $v\cdot \bz$ denotes $\bz v^T,$ $v_i$ are assumed to be vectors for which this product makes sense, meaning $A=N\setminus\{i\}$ for $\bullet=\pre_i,~A=N_L$ or $N_R$ for $\bullet=\bcfw$, or $A=N$ for $\bullet=\FL.$ and both functionaries and columns $z_i$ are promoted by $\Psi^\bullet.$
\end{rmk}
The above remark implies
\begin{cor}\label{cor:purity_preserved_under_prom}
The promotions are well defined ring homomorphisms, which preserve purity of rational functionaries.
\end{cor} 
In practice will be interested in applying this map to rational functionaries that after the promotion become permissible.
\begin{definition}\label{def:pre_perm}
    A functionary $F\in\mathcal{FR}_{k+1,N}^{0}$ is \emph{pre-permissible of level $l$} if $F$ has a representation as a rational function $P/Q$ in twistors such that
\begin{itemize}
    \item[--] The is no twistor which includes $B$ but not $A$.
    \item[--] $\deg_A(P)-\deg_A(Q)=l.$
\end{itemize}
$F$ is \emph{pre-permissible} if it is pre-permissible of level $0.$
\end{definition}

\begin{obs}\label{obs:pre_permissible}
\begin{itemize}
\item If $F\in\mathcal{FR}^\ell_{k,N}$ is permissible then also $\Psi^\pre F$ is permissible.
\item Using \cref{nn:bcfwmap}, if $F$ is a permissible functionary on one of the index sets $N_L$ or $N_R$ then $\Psi^\bcfw_{ac}F$ is permissible.
\item
Let $F\in\mathcal{FR}_{k+1,N}^{0}$ be a pure rational functionary which is pre-permissible of level $l.$ 
Then the functionary obtained from $F$ by first substituting 
\[
 A \;\mapsto\; A - \frac{\llrr{c\,d\,A\,i}}{\llrr{c\,d\,B\,i}}\,B 
,~~~~~ d \;\mapsto\; d - \frac{\llrr{d\,A\,B\,i}}{\llrr{c\,A\,B\,i}}\,c .
\]where $i$ is an arbitrary index, and then multiplying by $\llrr{c\,d\,B\,i}^l$ is permissible. In particular, if $l=0$ then 
$\Psi^\FL(F)$ is permissible, and in particular descends to a function on $\CC\Ampl_{k,N\setminus\{A,B\},4}^1$.
\end{itemize}
\end{obs}
\begin{proof}
The first two items are not tautolofical only if $\ell=1,$ but in this case $A,B$ are not in the index sets, and hence they are evidently true.

For the third item we use the notations of \cref{def:pre_perm}. Note first that only twistors which use $A$ and not $B$ are affected by the promotion. Such a twistor $\llrr{efgA}$ promotes to \[\frac{\llrr{efgA}\llrr{cdBi}-\llrr{efgB}\llrr{cdAi}}{\llrr{cdBi}}=\frac{-\llrr{efgc}\llrr{dABi}+\llrr{efgd}\llrr{cABi}-\llrr{efgi}\llrr{cdAB}}{\llrr{cdBi}}.\] 
The numerator of this expression is permissible, while the denominator yields a non permissible factor of $\llrr{cdBi}.$ The constraint on the $A$-degrees of $P,Q$ implies that these factors cancel after multiplication by $\llrr{cdBi}^l$, implying  the result.
\end{proof}
\subsection{Signs of functionaries}
\begin{definition}\label{def:sign}
  For $S \subset \Gr_{k,n;\ell}$, we say a rational functionary $F$ has \emph{sign $s \in \{0, \pm 1\}$ on the image of $S$} if for all $Z\in \Mat_{n,k+4}^{>}$ and for all $Y$ or $(Y\vdots L)\in \tZ(S)$, $F(Y)$ has sign $s$.  If $F$ has sign 0, we also say it \emph{vanishes} on the image of $S$. If $F\geq 0$ on the intersection of the image of $S$ and the domain where $F$ is defined, or $F\leq 0$ on that set, we say that $F$ is \emph{weakly positive or negative} respectively, or that $F$ has a \emph{weak $\pm1$ sign} there.
	In these cases, we say that $F$ has \emph{fixed (weak) sign} on 
	$\tZ(S)$.

A rational functionary $F$ has \emph{strong sign} $+1$ or is \emph{strongly positive} on the image of $S$, if for all $(C\vdots D)\in S$, $F((C\vdots D)Z)$ can be written as a ratio of polynomials in the Pl\"ucker coordinates of $C,D$ and maximal minors of $Z$ with all coefficients of sign $+1$. We similarly define \emph{strongly negative} functionaries and strong sign $-1.$ A functionary which is either strongly positive, strongly negative, or vanishes identically, on the image of $S$ is said to have a strong sign on the image of $S$. 
\end{definition}
The importance of this definition comes from the following simple observation:
\begin{obs}\label{obs:strong_sign_and_strata}
If a permissible functionary $F$ has a strong nonzero sign on the BCFW cell $S_\D$ then 
\[V(F)\cap\overline{S}_D=\{U\vdots V\in \overline{S}_\D|~F(\tZ(U\vdots V),Z)=0\}\]is a union of (loopy) positroid strata, independent of the positive matrix $Z.$
\end{obs}
The next lemmas provide examples of permissible twistors which have strong signs.
\begin{lemma}\label{lem:sign_bdry_twistors}
If $i\lessdot j$, and $i'\lessdot j'$ satisfy $|\{i,j,i',j'\}|=4,$ then $\llrr{ijAB},~\llrr{iji'j'}$ are nonnegative on $\Ampl^{\ell=1}_{n,k,4}$. If $i=\max N,j=\min N$ and $i'\lessdot j'$ satisfy $|\{i,j,i',j'\}|=4,$ then $(-1)^k\llrr{ijAB},~(-1)^k\llrr{iji'j'}$ have strong sign $+1$ on the image of $\Gr_{k,n;\ell}^>$ and are nonnegative on $\Ampl^{\ell=1}_{n,k,4}$.
\end{lemma}
\begin{proof}
By \cref{lem:Cauchy-Binet} it is easy to see that all terms $\lr{Z}_{j_1, \dots, j_{k+2}, i,j}, \lr{Z}_{j_1, \dots, j_k, i,j,i',j'},$ have the same sign, which is exactly the one described in the statement of the lemma.
\end{proof}
\begin{lemma}\label{lem:sign_of_chord_twistors}
\begin{enumerate}
\item\label{it:chord_BCFW}
Let $S_L\subseteq \Gr_{k_L,N_L;\ell_L}^{\geq},~S_R\subseteq \Gr_{k_R,N_R;\ell_R}^{\geq},$ where $k_L,\ldots,\ell_R$ are as in \cref{nn:bcfwmap} be subspaces which satisfy \cref{condition:BCFW}.
Let $S=S_L\bcfw S_R$ 
Then the twistor coordinates 
	\[(-1)^{k_R}\llrr{bcdn}, \quad (-1)^{k_R+1}\llrr{acdn}, \quad \llrr{abdn},\quad -\llrr{abcn}, \quad \llrr{abcd}\]
	have strong sign $+1$ on the image of $S$.
 \item\label{it:chord_FL}
Let $S=\FL(S')$ where $S'\subseteq\Gr_{k+1,[n]\cup\{A,B\}}^{\geq}$ is a positroid cell which satisfies the assumption of \cref{condition:FL} then
the twistor coordinates
 \[\llrr{d_\star nAB}, \quad -\llrr{c_\star nAB}, \quad \llrr{c_\star d_\star AB}\]are strongly positive on the image of $S.$
 If $\{a_0,b_0\}$ is coindependent for $\PosD(S),$ where $a_0,b_0$ are consecutive indices, then $\llrr{a_0b_0AB}$ is strongly positive on the image of $S.$
 \end{enumerate}
\end{lemma}
\begin{proof}
The proofs for the general cases of three types of twistor coordinates from the two items are similar, and are identical to the corresponding proofs in \cite{even2021amplituhedron,even2023cluster}, see, for example, \cite[Lemma 12.4]{even2023cluster}, and rely on \cref{lem:bcfw_after_fl},~\cref{it:bcfw_after_fl_chords},~\cref{it:bcfw_after_fl_yellow},~\cref{it:bcfw_after_fl_reds}. We will therefore only prove the general case of second item for the yellow chord, leaving the details of the other cases to the reader.

Expand each twistor $\llrr{i_\star j_\star AB}$ as in \eqref{eq:cauchy-binet_loop}. We obtain a sum of products of Pl\"ucker coordinate of the matrix $C+D$ where $(C\vdots D) \in S_\D$, and the matrix $Z$. The sum is over all Pl\"ucker coordinates of $C+D$ which are non zero and intersect the set $\{c_\star,d_\star,n\}$ precisely in the singleton $\{c_\star,d_\star,n\}\setminus\{i_\star,j_\star\}.$ \cref{lem:bcfw_after_fl},~\cref{it:bcfw_after_fl_yellow} guarantees that this summation is non empty. 
\end{proof}
\begin{nn}\label{nn:bdry_twistors}We refer to the twistors appearing \cref{lem:sign_bdry_twistors} as \emph{boundary twistors}. This name is motivated by \cref{prop:CD bdries}. The boundary twistors which are defined using four indices are called \emph{tree boundary twistors}, and those defined using two indices are the \emph{loop boundary twistors}.

We refer to the twistors appearing in the first item of \cref{lem:sign_of_chord_twistors} as \emph{(BCFW) chord's twistors}. We refer to the $\star-$labeled  twistors appearing in the second item as the \emph{yellow chord's twistors}. The motivation for these names is explained in \cref{cor:5_3}.
\end{nn}

The twistors of \cref{lem:sign_of_chord_twistors} are closely related to the domino variables, as we will now see.
\begin{lemma}\label{lem:5_3}
\begin{enumerate}
    \item\label{it:5_3_BCFW}
    	Let $k \geq 1$ and let $Y \in \Mat_{k \times (k+4)}$ and $Z \in \Mat_{5 \times (k+4)}$,
	with row vectors $Y_1,\dots,Y_k$, and $Z_1,\dots,Z_5$, respectively. Define 
	$$u:= \llrr{2345} Z_1 - \llrr{1345} Z_2 + \llrr{1245} Z_3 - \llrr{1235} Z_4 + \llrr{1234} Z_5.$$
	Suppose at least one of the $5$ twistor coordinates
	$\llrr{2345}, \llrr{1345}, \llrr{1245}, \llrr{1235}, \llrr{1234} $ is nonzero. Then $\Span(Y_1,\ldots, Y_k) \cap \Span(Z_1,\ldots, Z_{5}) = \Span(u)$, and in particular is the trivial vector space if and only if $u=0$. \item\label{it:5_3_FL}
    Let $k \geq 1$ and let $(Y\vdots L) \in \Mat_{k \times (k+4);1}$ and $Z \in \Mat_{3 \times (k+4)}$,
	with row vectors $Y_1,\dots,Y_k,L_1,L_2$, and $Z_1,Z_2,Z_3$, respectively. Define 
	$$u:= \llrr{23AB} Z_1 - \llrr{13AB} Z_2 + \llrr{12AB} Z_3.$$
	Suppose at least one of the $3$ twistor coordinates
	$\llrr{23AB}, \llrr{13AB}, \llrr{12AB}$ is nonzero. Then $\Span(Y_1,\ldots, Y_k,L_1,L_2) \cap \Span(Z_1,Z_2, Z_{3}) = \Span(u)$, and in particular is the trivial vector space if and only if $u=0$. 
\end{enumerate}
\end{lemma}
\begin{proof}
The first item is \cite[Lemma 12.13]{even2023cluster}.  
The second is proven in an identical manner, keeping in mind that we interpret $\llrr{ijAB}$ as the determinant of the square matrix made whose rows are $Y_1,\ldots, Y_k,L_1,L_2,Z_i,Z_j.$
\end{proof}
\begin{cor}\label{cor:5_3}
\begin{enumerate}
    \item\label{it:cor5_3_BCFW}Let $S=S_L\bcfw S_R$ be as in \cref{lem:sign_of_chord_twistors},~\cref{it:chord_BCFW}. Then for every \[U\vdots V=\mbcfw(U_L\vdots V_L,[\alpha:\ldots:\varepsilon],U_R\vdots V_R)\in S\] one can recover the BCFW parameters $\alpha,\ldots,\varepsilon$ from $\tZ (U\vdots V)$ as
	\[(-1)^{k_R}\llrr{bcdn}, \quad (-1)^{k_R+1}\llrr{acdn}, \quad \llrr{abdn},\quad -\llrr{abcn}, \quad \llrr{abcd},\]
 respectively, uniquely up to a multiplication by a positive scalar.
    \item\label{it:cor5_3_FL}
    Let $S=\FL(S')$ as in \cref{lem:sign_of_chord_twistors},~\cref{it:chord_FL}.
    Then for every $U\vdots V\in S$ one can recover the forward limit parameters $\gamma_\star,\delta_\star,\varepsilon_\star$  as
    equal
 \[\llrr{d_\star nAB}, \quad -\llrr{c_\star nAB}, \quad \llrr{c_\star d_\star AB},\]
  respectively, up to a multiplication by a positive scalar. 
\end{enumerate}
\end{cor}
\begin{proof}
The two items follow immediately from the corresponding items in \cref{lem:sign_of_chord_twistors},~\cref{lem:5_3}. In the case of the BCFW product, let $v\in U$ be as in the notations of \cref{def:bcfw-map}. Then \[vZ=\alpha Z_{a}+\ldots+\varepsilon Z_{n},\] 
is a non zero vector, since the amplituhedron map does not lose rank .On the other hand, it belongs to both $\Span(Z_a,Z_b,Z_c,Z_d,Z_n)$ and to $UZ.$ Using \cref{lem:5_3},~\cref{it:5_3_BCFW}, whose conditions are verified in  \cref{lem:sign_of_chord_twistors},~\cref{it:chord_BCFW}, we have equality of projective vectors
\[[\alpha:\beta:\gamma:\delta:\varepsilon]=[(-1)^{k_R}\llrr{bcdn}:(-1)^{k_R+1}\llrr{acdn}:\llrr{abdn}:-\llrr{abcn}:\llrr{abcd}].\]
The proof of the case of the forward limit is identical.
\end{proof}
\begin{lemma}\label{cor:epsilon_signs}
Let $\D\in\CD_{n,k}^1$ be a chord diagram. Then the following twistors are strongly positive on the image of $S_\D.$
\begin{enumerate}
\item\label{it:cor_eps_of_most_twistors} 
 $\llrr{a_ib_ic_id_i}$, where $\D_i$ be a black or purple chord.
 \item\label{it:cor_top_yellow} $\llrr{c_\star d_\star AB}$.

\item\label{it:cor_eps_red} $\llrr{a_ib_iAB}$ for a red chord $\D_i.$
\end{enumerate}
\end{lemma}
\begin{proof}
For the proof for the first item, by \cref{lem:sign_bdry_twistors} it is enough to show that the expansion  \eqref{eq:cauchy-binet} contains at least one term. This follows from \cref{lem:bcfw_after_fl},~\cref{it:bcfw_after_fl_chords}.

For the second and third 
items, we need to show that the expansion \eqref{eq:cauchy-binet_loop} has contains at least one term. This follows from
\cref{lem:bcfw_after_fl},~\cref{it:bcfw_after_fl_yellow} and \cref{it:bcfw_after_fl_reds}, respectively.
\end{proof}
\subsection{Signs under promotion}
The next three propositions show that signs of functionaries behave well under promotions. Their proofs are postponed to the next subsection.
\begin{prop}\label{prop:sign_under_pre_promotion}
Assume $F$ is a functionary which has a fixed sign (strong sign, weak sign) $s$ on the image of $S\subseteq\Gr_{k,N\setminus\{\p\};\ell}^{\geq},$ where $p$ is the penultimate index in $N$. Then $\Psi^{\pre}(F)$ has a fixed (strong, weak) sign $s$ on $\pre_pS.$ 
\end{prop}
The proof of this claim is simple, and is identical to the proof of the '$\pre$' case in \cite[Theorem 12.6]{even2023cluster}. We include, of course, the case that the sign is $0.$

\begin{prop}\label{prop:sign_under_FL_promotion}
Recall \cref{def:pre_perm}.
Let $N=\{1,2,\ldots,n-1,A,B,n\},$ and let $S\subseteq\Gr_{k+1,N}^{\geq}$ be a positroid cell satisfying the assumptions of \cref{condition:FL}.
Let $F=\frac{f}{g},$ be a pure pre-permissible functionary which has a fixed (weak) sign $s\in\{-1,0,1\}$ on the image of $S$, where $f,g$ are polynomial functionaries. Then either $\Psi^{\FL}(F)$ has a fixed (weak) sign $(-1)^{\deg_A(F)+\deg_B(F)+\deg_n(F)}s$ on $\FL(S),$ or at least one of $\Psi^\FL(f),~\Psi^\FL(g)$ is identically $0.$
\end{prop}
\begin{prop}\label{prop:sign_under_BCFW_promotion}
Let $S_L\subseteq\Gr_{k_L,N_L;\ell_L}^{\geq},~S_L\subseteq\Gr_{k_R,N_R;\ell_R}^{\geq}$ be subspaces satisfying \cref{condition:BCFW}, where $k_L,k_R,N_L,N_R,\ell_L,\ell_R$ are as in  \cref{nn:bcfwmap}.	Let $F$ be a pure rational functionary with indices contained in $N_L$ (resp. $N_R$). 
	\begin{enumerate}
		\item\label{it:vanishing} If $F$ vanishes on the image of $S_L$ (resp. $S_L$), then $\Psi_{a c}^{\bcfw} (F)$ vanishes on the image of $S_L\bcfw S_R$.
\item\label{it:strong_sign_BCFW} If 
$F$ has weak sign $s \in \{\pm 1\}$ on the image of $S_{L}$ and $\ell_L=1$ (resp. $S_R$,~$\ell_R=1$) then $\Psi_{ac}^{\bcfw} F$ has weak sign $(-1)^{(k_R+1) \deg_n F}s$ (resp. $s$) on the image of $S_L\bcfw S_R$. 

If $F$ has a strong sign $s\in\{\pm1\}$ on the image of $S_{L}$ and either $\ell_L=0$ (resp. $S_R$,~$\ell_R=0$) or $F$ does not use the loopy twistors $\llrr{ijAB}$, then $\Psi_{ac}^{\bcfw} F$ has strong sign $(-1)^{(k_R+1) \deg_n F}s$ (resp. $s$) on the image of $S_L\bcfw S_R$. 
\end{enumerate}
\end{prop}

\subsubsection{Proofs of \cref{prop:sign_under_FL_promotion} and \cref{prop:sign_under_BCFW_promotion}}
\begin{proof}[Proof of \cref{prop:sign_under_FL_promotion}]
Write $c=n-2,d=n-1.$
We will follow the sequence of operations which generates $\FL(S)$ from $S,$ and understand its effect on twistors.
\\\textbf{Step 1 - $\scale_l\circ\inc_l:$}
In this step we prove that if $F$ has a (weak) sign $s\in\{0,-1,+1\}$ on the image of $S$ then $\tilde{\Psi}^{\inc_l}(F)$ has a (weak) sign $(-1)^{\deg_A(F)+\deg_B(F)+\deg_n(F)}s$ on the image of $S_0:=\scale_l(\R_+)(\inc_l(S)).$

We start by showing that $\tilde{\Psi}^{\inc_l}(F)$ has this sign on the image of $\inc_l(S).$ This follows from \cite[Lemma 4.21]{even2021amplituhedron}. This lemma shows that for any positive $Z'$ and $C':=\inc_{i}(C)$, there is a $Z$ such that for each twistor $\llrr{I}$ appearing in $F$, $\llrr{C'Z'~Z'_I}= (-1)^{|I\cap\{A,B,n\}|}\llrr{CZ~Z_I}$. We have $\lr{C}_I= \lr{C'}_{I \cup \{i\}}$ and by the proof of \cite[Lemma 4.21]{even2021amplituhedron}, the same statement holds for $Z, Z'$. So adding $i$ to each Pl\"ucker coordinate and minor appearing in $f$ and $g$ gives a formula for $\Psi^{\inc_i}F$ on the image of $\inc_{i} (S)$.
Now, since
\[\tZ(\scale_l(\lambda).C)=\widetilde{Z'}(C), \] where $Z'$ is the scaling of $Z_l$ by $\lambda^{-1}$, and since the index $l$ does not appear in any twistor in the representation of $\Psi^{\inc_l}(F)$, it is immediate to verify that the scaling does not affect the value, hence also not the sign of the functionary.
\\\textbf{Step 2 - $x_*,y_*$ operations.}
In this step we prove that if $\Psi^{\inc_l}F$ has a (weak) sign $s'\in\{0,-1,+1\}$ on the image of $S_0$ then $\widetilde{\Psi}^{\FL}(F)$ has a (weak) sign $s'$ on the image of \[S_1:=y_c(\R_+).y_d(\R_+).x_A(\R_+).x_l(\R_+).S_0,\]
where
\[\widetilde{\Psi}^{\FL}:
\Psi^{\inc_l}(\mathcal{FR}_{k+1,N}^{0})\to\mathcal{FR}_{k+2,N\cup\{l\}}^{0}.
\]
defined on generators by the following substitution:
\begin{align*}
& A \;\mapsto\; A - \frac{\llrr{c\,d\,A\,l}}{\llrr{c\,d\,B\,l}}\,B 
\\
& d \;\mapsto\; d - \frac{\llrr{d\,A\,B\,l}}{\llrr{c\,A\,B\,l}}\,c .
\end{align*} 
Note that replacing $l$ by $n$ in the above substitution yields the substitution which defines $\Psi^{\FL}.$ 

	We rely on \cite[Section 4]{even2021amplituhedron}, and on the following lemma whose proof is given below.
 \begin{lemma}\label{lem:4biden_intermediate}
 On the image of $S_1$ every twistor $\llrr{i_1i_2i_3i_4}$ for $\{i_1,\ldots,i_4\}\in\binom{c,d,A,B,n}{4}$ has a fixed nonzero sign.
 \end{lemma}
 The proof of \cite[Lemma 4.28]{even2021amplituhedron} together with \cite[Lemma 4.22]{even2021amplituhedron} show that for every functionary $F$ on index set $N$ there is a functionary $G,$ on index set $N\cup\{l\},$ with the following property: For every $C \in S_1$ and a positive matrix $Z$ there are explicit elements $C' \in S$ and positive matrix $Z',$
 such that $F(C'Z', Z')= G(CZ,Z)$. \cite[Lemma 4.35]{even2021amplituhedron} implies that $G = \widetilde{\Psi}^{\FL}(F),$ up to scaling by some twistors coordinates appearing in \cite[Lemma 4.35]{even2021amplituhedron}. These twistor coordinates are precisely those of \cref{lem:4biden_intermediate}. By \cref{lem:4biden_intermediate}, these twistor coordinates are all nonzero on the image of $S_1$, so one may divide $G$ by them to obtain exactly the formulas defining $\tilde{\Psi}^{\FL}$.
 \\\textbf{Step 3 - $\rem_{AB}\circ\addL_{AB}:$}
 Note that by \cref{obs:pre_permissible} $\widetilde{\Psi}^{\FL}(F)$ has a permissible representation.  Note also that since $S$ is $B-$independent, for every element of the positroid cell $S_1$ the columns $A,B$ are linearly independent. 
 
 In this step we show that if $G(C',Z')\in \mathcal{FR}_{k+2,N\cup\{l\}}^{0}$ has a permissible representation and has (weak) sign $s'$ on the image of $S_1$, then another functionary $G'(C\vdots D,Z)\in\mathcal{FR}_{k,[n]\cup\{l\}}^{1}$ has either (weak) sign $s'$ or is identically $0$ on the image of $S_2=\rem_{AB}(\addL_{AB}(S_1))\subseteq\Gr_{k,[n]\cup\{l\};1}^{\geq}.$ 
 $G'$ is defined by taking any permissible representation of $G,$ and interpreting the twistor coordinates appearing there as in \eqref{eq:cauchy-binet},~\eqref{eq:cauchy-binet_loop}. This just means that we interpret twistors not including $A,B$ as usual, and twistors $\llrr{\tZ(C)Z_{i_1}Z_{i_2}Z_AZ_B}$ are interpreted as $\llrr{\tZ(C+D)Z_{i_1}Z_{i_2}}.$ It is easy to see that $G'$ is independent of the choice of the representation.
Thanks to \cref{rmk:justification} we have
\begin{equation}\label{eq:G2G'}G'(C\vdots D,Z)=G(\varphi_{AB}^\sigma(C\vdots D),\tilde\varphi_{AB}(Z)),\end{equation}
where $\sigma$ is an arbitrary local section of the projection $\Mat_{k,[n]\cup\{l\}\setminus\{A,B\};1}\to\Gr_{k,[n]\cup\{l\};1},$  and $\varphi_{AB}^\sigma$ is defined in \cref{nn:varphi}.
 Now, if $C\vdots D\in S_2$ then $\varphi^\sigma_{AB}(C\vdots D)\in S_1,$ for some section $\sigma.$ Indeed, $\varphi_{AB}^\sigma$ is a local inverse of $\rem_{AB}\circ\addL_{AB},$ and for every $V\in S_1$ we can find $\sigma=\sigma_V$ such that $V=\varphi_{AB}^\sigma(\rem_{AB}(\addL_{AB}(V)))$. Now, every $C\vdots D\in S_2$ equals $(\rem_{AB}(\addL_{AB}(V))$ for some $V\in S_1,$ thus, $\varphi^{\sigma_V}_{AB}(C\vdots D)\in S_1.$
 
 In addition, if $Z$ has positive minors, $\tilde{\varphi}_{AB}(Z)$ is a matrix with nonnegative maximal minors, hence is in the topological closure of the space of positive matrices.
 Thus, if $G$ has (weak) sign $s'$ on $S_1\times\Mat_{(N\cup\{l\})\times (k+4)}^{>},$ then by \eqref{eq:G2G'} and the above discussion, if $G'$ is defined on $S_2\times\Mat_{([n]\cup\{l\}))\times (k+4)}^{>}$, meaning that its denominator in a reduced representation is not identically $0,$ then it has a weak sign $s'$. If the numerator of $G'$ in that representation is not identically $0,$ then also $G'$ will not be identically $0$ on $S_2\times\Mat_{([n]\cup\{l\}))\times (k+4)}^{>}$. Moreover, by construction $G'$ is permissible.
\\\textbf{Step 4: adding column $l$ to column $n$ and erasing the former:}
Let $G'$ be a permissible functionary which has a weak sign $s'$ on the image of $S_2\times\Mat_{([n]\cup\{l\})\times (k+4)}^{>}.$
In this step we show that if the denominator of $G'$ does not vanish under the substitution $l\to n,$ which replaces each index $l$ in the twistor representation by the index $n$, then the functionary $G''$ obtained by this assignment has a weak sign $s'$ on the image of $\FL(S),$ or is identically $0.$

First, note that, by taking the limit $Z_l\to Z_n$ inside $\Mat_{([n]\cup\{l\})\times (k+4)}^{>}$ we obtain that $G'$ either has weak sign $s',$ or is identically zero/undefined, on the image of $S_2\times\Mat'_{([n]\cup\{l\})\times (k+4)}$, where $\Mat'_{([n]\cup\{l\})\times (k+4)}$ is the collection of nonnegative matrices with rows labeled $[n]\cup\{l\}$ and $k+4$ rows, where $n-1\lessdot l\lessdot n$, whose restriction to the $[n]\times(k+4)$ submatrix is positive, and the $l$th row equals the $n$th row. 
Note that for $Z'\in \Mat'_{([n]\cup\{l\})\times (k+4)}$ with $Z=\rem_l Z',$ 
\[\lr{Z'}_I=0~\text{if }l,n\in I,\]
and if exactly one of $l,n,$ say $l$ is in $I$ then
\[\lr{Z'}_I=\lr{Z}_{I\cup\{n\}\setminus\{l\}}.\]
Let $\psi:S_2\to \FL(S)$ be the map $(C\vdots D)\mapsto\rem_lx_l(1)(C\vdots D).$
We have, for $M=C,C+D$
\[\lr{\rem_lx_l(1)(M)}_I=\begin{cases}\lr{M}_I&\text{if }n\notin I\\
\lr{M}_I+\lr{M}_{I\cup\{l\}\setminus\{n\}}&\text{if }n\in I
\end{cases}\]
Let
$\tilde\psi:\Mat^{>}_{[n]\times(k+4)}\to\Mat'_{([n]\cup\{l\})\times (k+4)}$ be the map $Z
\mapsto y_l(1)\pre_lZ,$ where as usual $n-1\lessdot l\lessdot n$.
Note that $\psi(C\vdots D
)Z=(C\vdots D)\tilde\psi(Z
),$ hence
\begin{obs}\label{obs:twistors_psi}
For every permissible twistor given by indices $i_1,\ldots,i_4\neq l$ (which may include $A,B$) it holds that
\[\llrr{\tZ(\psi(C\vdots D
))Z_{i_1}\ldots Z_{i_4}}=\llrr{\widetilde{Z'}(C\vdots D)Z'_{i_1}\ldots Z'_{i_4}},\]
where $Z'=\tilde\psi(Z
).$
\end{obs}
Every $C\vdots D\in \FL(S)$ equals $\psi(C'\vdots D')$ for $C'\vdots D'\in S_2.$ Thus, for $Z\in\Mat^>_{n\times(k+4)},$
\[G''(\tZ(C\vdots D),Z) = G'(\widetilde{Z'}(C'\vdots D'),\widetilde{Z'})\]with $Z'=\tilde\psi(Z
)\in\Mat'_{([n]\cup\{l\})\times(k+4)}.$
Hence, if $G'$ has weak $s'$ on $S_2\times\Mat_{([n]\cup\{l\})\times (k+4)}'$ then $G''$ has weak sign $s'$ on the image of $\FL(S).$ Following the steps of the argument, we see that the resulting function is identically $0$ or undefined precisely if the $\FL-$promotion of $f$ or $g$ is identically $0.$ 

\end{proof}
\begin{proof}[Proof of \cref{lem:4biden_intermediate}]
As in the proof of \cref{lem:sign_of_chord_twistors} it is easy to see that all summands in the Cauchy-Binet expansion \eqref{eq:cauchy-binet} of $\llrr{i_1\ldots i_4}$ have the same sign, and we just need to show that there is at least one nonzero summand. Using \cref{condition:FL}, we pick a basis $I\in\binom{N}{4}$ with $I\cap\{c,d,A,B,n\}=\emptyset.$ It is easy to see that $I\cup\{h\},$ for the unique element $\{h\}=\{c,d,l,A,B\}\setminus\{i_1,\ldots,i_4\}$ is a basis for the elements of $S_1.$ 
\end{proof}


\begin{proof}[Proof of \cref{prop:sign_under_BCFW_promotion}]
The proof is similar to the proof of \cite[Theorem 11.3]{even2023cluster}, which shows that in the tree case a functionary with a strong sign on the left or right component gets promoted under $\Psi^\bcfw$ to a functionary of strong sign. 
We recall the argument used there, and only list the changes required for our purposes.
In what follows we only treat the case that $F$ comes from the right component, to avoid heavy notation or repetitions. The other case is similar.

The argument in \cite{even2023cluster} was: Write $S=S_L\bcfw S_R.$ $S$ is sandwiched between two cells, 
\begin{equation}\label{eq:sandwich}S_\bulR\subset \overline{S}\subset\overline{S^\bulR}\end{equation}where
$S^{\bulR}=\Gr^{>}_{k_L,N_L} \bcfw S_{R}$, and $S_{\bulR}=S'\bcfw S_{R},$ where $S'$ is a positroid cell whose only basis is some $I\in\binom{N_L}{k_L}$. Note that such one can find such $S'$ in $\overline{S}$ thanks to \cref{lem:4coind_tree}. Then, iterated applications of \cite[Lemma 11.10]{even2023cluster} are used to show that $\Psi^{\bcfw}_{ac}(F)$ has the prescribed strong sign on $S^{\bulR},S_{\bulR},$ which implies it also has the prescribed sign on $S.$

We act similarly. Again for convenience we will only sketch the case that $F$ comes from the right side. We define $S^{\bulR},S_{\bulR}$ as above.
If $\ell_R=1$ then it is easy to see that again \eqref{eq:sandwich} holds, and we proceed as above, like in \cite[Theorem 11.3]{even2023cluster}. The only change is that the argument used in \cite[Theorem 11.10]{even2023cluster} to show that the sign is not only weak, but also strong may fail, if $F$ uses the loop twistors $\llrr{ijAB}$: in the tree case it was shown that if $V\in S_G$ is obtained from $V'\in S_{G'}$ by performing $x_i(t)$ or $y_i(t)$, where $S_G,S_{G'}$ are $r,r-1$ dimensional positroid cells respectively, then the Pl\"ucker coordinates of $V'$ are quotients of polynomials with positive coefficients in the Pl\"ucker coordinates of $V.$ This argument breaks in the presence of loops.
Thus, $\Psi^\bcfw_{ac}F$ will always have the prescribed weak sign, and if $F$ has strong sign and does not use the loopy twistors, then also the sign of $\Psi^\bcfw_{ac}F$ will be strong.

If $\ell_R=0$ and $F$ has strong sign $s$ then neither $F$ nor $\Psi^{\bcfw}_{ac}(F)$ involve the loop twistors $\llrr{i_1i_2AB}$. Thus, depends only on the tree part. In this case, even though \eqref{eq:sandwich} does not hold, its tree level analog holds
\[S'\bcfw \PosC(S_R)\subset \overline{\PosC(S)}\subset\overline{\Gr_{k_L,N_L}^>\bcfw \PosC(S_R)},\]where $S'$ is the above $0$ dimensional positroid cell. Observe that $S_R\subseteq\PosC(S_R)$ and $\{U|~U\vdots V\in S\}\subseteq\PosC(S),$ by \cref{prop:fixed_positroids}. In addition, by \cref{obs:strong_sign_and_strata}, since $F$ has strong sign on $S_R$ it also has a strong sign on $\PosC(S_R).$ 
The rest of the proof is  identical to that of \cite[Theorem 11.3]{even2023cluster}, including the part involving the strong sign. Thus, $\Psi^\bcfw_{ac}F$ has strong sign $F$ on $\PosC(S)$ and hence on $S$.
\end{proof}
\section{Boundaries of the $1-$loop amplituhedron}
\begin{prop}\label{prop:CD bdries}
If $C\vdots D\in\Gr_{k,n;1}^\geq$ is a preimage of a point from the zero locus of a loop boundary twistor $\llrr{ijAB},$ then $D$ contains a vector supported on positions $i,j,$ and vice versa. Thus, if $C\vdots D\in\Gr_{k,n;1}^\geq$ is a preimage of a point from the zero locus of a loop boundary twistor then $C\vdots D\in\SA_D,$ and vice versa.

If $C\vdots D\in\Gr_{k,n;1}^\geq$ is a preimage of a point from the zero locus of a tree boundary twistor $\llrr{iji'j'}$ then $C$ contains a vector supported on positions $i,j,i',j'$ and vice versa. Thus, if $C\vdots D\in\Gr_{k,n;1}^\geq$ is a preimage of a point from the zero locus of a tree boundary twistor then $C\vdots D\in\SA_C,$ and vice versa.

Moreover, in both cases the zero loci of these twistor coordinates in the amplituhedron are contained in its topological boundary.
\end{prop}
\begin{proof}
The statements which describe the zero loci are easy: From \cref{eq:cauchy-binet_loop}, and the fact that all summands have the same signs, we see that $\llrr{ijAB}$ is $0$ if and only if all minors of $C+D,\langle C+D\rangle_J$ with $i,j\notin J$ is $0.$ This is equivalent to the existence of a basis of the vector space $C+D$ with a row supported only on the positions $i,j.$ Similar argument works for $\llrr{iji'j'}.$ 

We will prove the last part for of the proposition for loop boundary twistors $\llrr{ijAB}.$ The proof for the tree boundary twistors is similar, and in fact, the case of a twisor $\llrr{iji'j'}$ is also identical to the proof of \cite[Proposition 8.2]{even2021amplituhedron}. We will assume $j=i+1,$ the case that $i=n$ is proven in a similar way. 

Let $p\in\Gr_{k,k+4;1}^{\geq}$ be a point in the zero locus of $\llrr{i,i+1,A,B}=0.$ We need to show that every open neighborhood of $p$ in $\Gr_{k,k+4;1}$ intersects $\Gr_{k,k+4;1}\setminus\Ampl^1_{n,k,4}(Z).$
It is enough to find points in $\Gr_{k,k+4;1}$ arbitrarily close to $p$ for which $\llrr{i,i+1,A,B}$ is negative, since, by the previous parts of the proposition, such points cannot belong to the amplituhedron. Moreover, by  
\eqref{eq:cauchy-binet_loop}, the dependence of $\llrr{i,i+1,A,B}$ on $C,D$ is through $C+D$ solely, and it is a continuous function of $C+D,$ as \cref{lem:Cauchy-Binet} shows. In addition, by \cref{lem:submersion2Grass} the map $(C\vdots D)\to C+D$ is open. Thus, it is enough to find, for an arbitrary preimage $(U\vdots V)\in\Gr_{k,n;1}^{\geq}$ of $p,$ a matrix $M$ close enough to a matrix representation of $U+V$ on which the twistor is negative. Indeed, such $M$ will have full rank, if it is close enough to $U+V$, and one can find a $k$-dimensional subvector space $U'$ of the row span $W$ of $M$ such that $U'$ is close to $U$ and $W/U'$ is close to $V$ in the natural topologies.

Let $M\in\Mat_{(k+2)\times n}$ be a matrix representing $C+D$ for a preimage $(C\vdots D)\in\Gr_{k,n;1}^{\geq}$ of $p.$ For $R=(R_{i}^j)_{i\in [k+2]}^{j\in [n]}\in\Mat_{(k+2)\times n},$ write $M(R)=M+R.$ The function $G(R)=\llrr{\tZ(C(R)) Z_{i}Z_{i+1}}$ is a multilinear function of the variables $R_{i}^{j},~i\in [k+2],~j\in [n].$  Order the pairs $(a,b)_{a\in [k],b\in [n]}$ arbitrarily. Let $(a',b')$ be the minimal element in that order such that 
\[G(R)
|_{R_{a}^{b}=0~\forall (a,b)>(a',b')}\not\equiv 0,\]
where $'>'$ is the aforementioned order on pairs. 
{In other words, $(a',b')$ is defined by the property that if we restrict $G(R)$ to the domain where $R_a^b=0,$ for all $(a,b)>(a',b')$, then it is not identically $0,$ but if we restrict it to the domain where $R_a^b=0,$ for all $(a,b)\geq (a',b')$ then it is. 
Such $(a',b')$ exists since the function $G(R)$ vanishes when all $R_{i}^{j}=0,$ but it is not identically zero, since any element of $\Mat_{k+2,n}$ can be represented by $M(R)$ for some $R,$ and $\llrr{i,i+1,A,B}$ is not identically zero.
We can write 
\begin{equation}\label{eq:multilinear_arg}G(R)|_{R_{a}^{b}=0~ \forall (a,b)>(a',b')}=H(R)+R_{a'}^{b'}F(R),\end{equation}
where $F(R),H(R)$ are multilinear functions in $(R_{a}^b)$ for $(a,b)<(a',b').$
By the choice of $(a',b')$ it follows that $H\equiv 0$ while $F\not\equiv0.$  We will find a sequence of $(k+2)\times n$ matrices $\{R(l)\}_{l=1,2,\ldots}$ with $\lim_{l\to\infty}R(l)= 0,$ or equivalently $M(R(l))\to M,$ for which \[\sgn(\llrr{\tZ(M(R(l))) Z_{i}Z_{i+1}})=-,\]thus establishing the claim.} To this end we write \[N_{(a',b')}:=\{\overline{R}\in \Mat_{(k+2)\times n}|\overline R_a^b=0~\forall (a,b)\geq (a',b')\}\simeq \R^{\{(a,b)|~(a,b)<(a',b')\}}.\]
Since $F$ is a multilinear function which is not identically $0,$ it is non zero on a dense subset of $N_{(a',b')}.$ Let $\overline{R}(l)\in N_{(a',b')}$ be a sequence of matrices on which $F(\overline{R}(l))\neq 0,$ and $\lim_{l\to\infty}\overline{R}(l)\to 0$. Define $R(l)\in \Mat_{(k+2)\times n}$ by \begin{equation}\label{eq:R_a^b}R(l)_a^b=\begin{cases}
 \overline{R}(l)_a^b,~~(a,b)<(a',b'),\\ 
 -\sgn(F(\overline{R}(l))/l,~~(a,b)=(a',b')\\
 0,~~(a,b)>(a',b')
\end{cases}.\end{equation}
Then $M(R(l))\to M,$ and, by \eqref{eq:multilinear_arg} and \eqref{eq:R_a^b}, \[\sgn(\langle \tZ(M(R(l))) Z_{i,i+1}\rangle)=\sgn\left(\frac{-\sgn(F(\overline{R}(l))}{l} F(\overline{R}(l))\right)=-.\]
\end{proof}
\section{Injectivity}
The main goal of this section is to prove
\begin{thm}\label{thm:inj}
Let $\D \in \CD_{n,k}^1$, be a chord diagram, $Z\in \Mat^{>}_{n \times(k+4)}$, and $Y \in \tZ(S_\D)$. Then, $Y$ has a unique preimage in $S_\D$ under $\tZ$. Thus, the restriction of $\tZ$ to $S_\D$ is injective.
\end{thm}We will actually prove much more - BCFW cells and their boundaries which do not intersect $\SA$ map injectively under the amplituheron map.
We will also show that the natural coordinates, domino and BCFW, are presented by functionaries, that these functionaries have fixed signs on the cell, and that moreover they evolve under the $\pre,\FL,\bcfw$ steps in the construction of BCFW cells by the corresponding promotions $\Psi^\pre,\Psi^\FL,\Psi^\bcfw.$ These results extend analogous results of \cite{even2021amplituhedron,even2023cluster}.
\subsection{A general framework for showing injectivity}
Our general method for showing injectivity is by finding an explicit preimage.
More precisely, given a chord diagram $\D,$ for a weak domino form $C\vdots D$ of $U\vdots V,$ with $\tZ(U\vdots V)=Y\vdots L,$ we first solve the rows corresponding to top chords by using \cref{lem:5_3}, and then solve descendant rows, again using \cref{lem:5_3}, but applied to vectors calculated in earlier steps of the recursion. In order to apply \cref{lem:5_3} we will need to verify that at least one of the five determinants which appear in the lemma are non zero. This will introduce some technical difficulty, that will be solved in later sections by restricting to $\Srem_\D.$

\begin{lemma}\label{lem:inj_when_eps_neq0}
Let $\D$ be a chord diagram.
Suppose that $U\vdots V\in\overline{S}_\D$ has a unique weak domino form weakly outside $\D_h=(a_h,b_h,c_h,d_h),$ and this form, with the exception of $\D_h$, can be calculated uniquely from $Y\vdots L=\tZ(U\vdots V).$ 
Define a functionary $F^\D_{\varepsilon_h}$ by
\begin{equation}\label{eq:def_F_eps}F^\D_{\varepsilon_h}=\begin{cases}
    \llrr{a_hb_hc_hd_h},&\D_h~\text{is black or purple}\\
    \llrr{a_0b_0AB},&h=0\\
    \llrr{c_\star d_\star AB},&h=\star\\
    \llrr{a_h,b_h,D_\star Z,D_0Z},&h\in\Red_\D\setminus\{0\}\\
    \llrr{a_h,b_h,u_\star Z,D_\star Z},&h\in\Blue_\D
\end{cases},\end{equation}
where for blue and non top red chords we use that we can uniquely read the chords above $\D_h.$ 
Then if $F^\D_{\varepsilon_h}\neq 0$ then also the row corresponding to $\D_h$ is uniquely determined: For $h\neq 0,\star,$ it equals $\alpha_hv_1+\beta_hv_2+\gamma_hv_3+\delta_hv_4+\varepsilon_h v_5,$ where
\[[\alpha_h:\ldots:\varepsilon_h]=[\llrr{v_2v_3v_4v_5}:-\llrr{v_1v_3v_4v_5}:\llrr{v_1v_2v_4v_5}:-\llrr{v_1v_2v_3v_5}:\llrr{v_1v_2v_3v_4}],\]
where \[v_1=\ee_{a_h}\qquad v_2=\ee_{b_h}\qquad v_5=\begin{cases}\ee_n,&\p(h)=\emptyset\\\alpha_{\p(h)}\ee_{a_{\p(h)}}+\beta_{\p(h)}\ee_{b_{\p(h)}},&\p(h)\neq\emptyset\end{cases}\]
and
\[v_3=\begin{cases}\ee_{c_h},&\D_h~\text{is black or purple}\\
D_\star,&h\in\Red_\D\\
\gamma_\star\ee_{c_\star}+\delta_\star \ee_{d_\star},&h\in\Blue_\D\end{cases} \qquad\qquad 
v_4=\begin{cases}
\ee_{d_h},&\D_h~\text{is black or purple}\\D_0,&h\in\Red_\D\\
D_\star,&h\in\Blue_\D\\
\end{cases}
\]
For $h=0,~\star,$ the row corresponding to $\D_h$ equals $\alpha_0v_1+\beta_0v_2+\varepsilon_0 v_3,$ $\gamma_\star v_1+\delta_\star v_2+\varepsilon_\star v_3,$ respectively, where the three coefficients equal 
\[[\llrr{v_2v_3}:-\llrr{v_1v_3}:\llrr{v_1v_2}],\]with 
\[v_1=\begin{cases}\ee_{a_0},&h=0\\\ee_{c_\star},&h=\star\end{cases}\qquad v_2=\begin{cases}\ee_{b_0},&h=0\\ \ee_{d_\star},&h=\star\end{cases}\qquad v_3=\begin{cases}\ee_n,&\p(h)=\emptyset\\\alpha_{\p(h)}\ee_{a_{\p(h)}}+\beta_{\p(h)}\ee_{b_{\p(h)}},&\p(h)\neq\emptyset\end{cases}\]
\end{lemma}
\begin{rmk}\label{rmk:EPSES}Note that while \eqref{eq:def_F_eps} depends on the unique domino form \emph{strictly} outside $\D_h,$ and on $Y\vdots L,Z$ for $h\notin\Red\cup\Blue\setminus\{0\},$ the expression for the \eqref{eq:def_F_eps} can be calculated directly from the knowledge of $Y\vdots L, Z,$ without the additional assumption on uniquely recovering other rows. For non top red or blue chords this expression involves only $Y\vdots L,Z$ and the top red and yellow rows, which are assumed to be calculated. 
In addition, for a top chord, if $F^\D_{\varepsilon_h}\neq 0$, the lemma implies that we can calculate the corresponding row in the preimage directly, without the need to recover any other row. 
\end{rmk}
\cref{lem:inj_when_eps_neq0} has the following corollary.
\begin{cor}\label{cor:inj_when_all_epses_neq0}
Assume $U\vdots V\in\overline{S}_\D$ has a unique weak $\D$-domino representative $C\vdots D,$ and that $\tZ(U\vdots V)=Y\vdots L.$ Assume that all expressions \eqref{eq:def_F_eps} for $h\in[k+1]$ are non zero.
    Then $C\vdots D$ is recovered, uniquely up to scaling the rows, via iterated application of \cref{lem:inj_when_eps_neq0}, performed in parent-to-child order. Moreover, $U\vdots V$ is the unique preimage of $Y\vdots L$ among the elements of $\overline{S}_\D$ which have a unique weak domino form and the expressions \eqref{eq:def_F_eps} do not vanish for them.  
\end{cor}
Using the second item of \cref{lem:uniqueness_or_SA} and \cref{prop:CD bdries} we also have the following corollary
\begin{cor}\label{cor:inj_when_epses_above_yellow_red_neq0}
If $U\vdots V\in\overline{S}_\D\setminus\SA$ or $U\vdots V\in S_\D,$ then the rows $D_0,D_\star$ and $C_h,$ for every $\D_h$ which does not descend from $\D_0$ are uniquely determined from $Y\vdots L,Z.$ Moreover, when $\ell=1$, for all these chords, but $\D_0,\D_\star$, the indices $A,B$ do not appear in the functionaries which describe the solutions.
\end{cor}
\begin{proof}
In the first case, by \cref{lem:uniqueness_or_SA} there is a unique weak domino form weakly outside the top red chord. Since $U\vdots V\notin\SA$, by \cref{prop:CD bdries} the boundary twistors $\llrr{a_hb_hc_hd_h},~h\notin\Red\cup\Blue,$ $\llrr{a_0b_0AB},\llrr{c_\star d_\star AB}$ are non zero. We can therefore apply \cref{lem:inj_when_eps_neq0} iteratively to uniquely recover all rows which do not descend from $\D_0$. In the second case we act in a similar manner, but use \cref{thm:dominoAndBCFWparams} for the uniqueness of the domino form, and \cref{cor:epsilon_signs} for the non-vanishing of the twistors. Since $A,B$ are not in the index set, looking at the expressions of \cref{lem:inj_when_eps_neq0}, we see that these indices are only introduced for the yellow and top red chord, and their descendants.
\end{proof}
\begin{proof}[Proof of \cref{lem:inj_when_eps_neq0}]
For $h\neq\{0,\star\}$ the domino row $C_h$ can be written as $\alpha_hv_1+\beta_hv_2+\gamma_hv_3+\delta_hv_4+\varepsilon_h v_5,$ where $\alpha_h,\ldots,\varepsilon_h$ are the domino variables and the vectors $v_1,\ldots,v_5$ are as in the statement of the lemma. Note that all of them have been calculated already. $C_hZ$ lies in the intersection of $Y$ and $\Span(v_1Z,\ldots,v_5Z).$
Note that $F^\D_{\varepsilon_h}=\llrr{v_1Z,v_2Z,v_3Z,v_4Z},$ and is non zero by assumption.

We now apply \cref{lem:5_3},~\cref{it:5_3_BCFW} with this $Y,$ and $v_1Z,\ldots,v_5Z$ playing the role of $Z_1,\ldots,Z_5,$ respectively. Since $F^\D_{\varepsilon_h}\neq 0,$ the coefficients $\alpha_h,\ldots,\varepsilon_h$ are uniquely recovered as in the statement of the lemma.

For $D_0,D_\star$ we act similarly, only that we apply \cref{lem:5_3},~\cref{it:5_3_FL} this time, with $Y+L$ playing the role of $Y$ in the lemma, and $v_1Z,v_2Z,v_3Z$ playing the role of $Z_1,Z_2,Z_3.$ Since $F^\D_{\varepsilon_h}=\llrr{v_1Z,v_2Z,A,B}\neq 0,$ we obtain the result.
\end{proof}
\subsection{Injectivity for BCFW cells}
By \cref{cor:inj_when_epses_above_yellow_red_neq0}, for every $U\vdots V\in\overline{S}_\D\setminus\SA$ we can define the vectors \[Z_\star=D_\star Z,~Z_0=D_0Z,~\Zu=u_\star Z=(\pr_{c_\star d_\star}D_\star)Z,\]uniquely up to scalings, where $D_0,D_\star$ are the rows of $D$ corresponding to the top red and yellow chord, for an arbitrary domino limit $C\vdots D.$
\begin{obs}\label{obs:well_def_red_blue_outside_SA}
Consider $\D\in\CD_{n,k}^1$. Then on $S_\D\cup(\overline{S}_\D\setminus\SA)$ the \emph{vector functionaries} describing $Z_\star, \Zu,Z_0$ are well defined up to scaling, and thus also the functionaries
\[\Fdelr^\D := \llrr{Z_{a_0}Z_{b_0}Z_\star Z_0},\qquad\qquad\Fepsh^\D=\llrr{Z_{a_h}Z_{b_h}\Zu Z_\star},~\D_h~\text{is blue}.\]
\end{obs}
\begin{obs}\label{lem:quad_lemma}On $S_\D\cup(\overline{S}_\D\setminus\SA)$ $\Fdelr^\D$ is permissible and
\[\llrr{\cdot,\cdot,Z_\star,Z_0}=\llrr{\cdot,\cdot,A,B}\llrr{\hat{n}a_0b_0|AB|c_\star d_\star \hat{n}}=
\frac{\llrr{\cdot,\cdot,A,B}}{\llrr{a_0b_0AB}}\Fdelr^\D,\]
where $Z_{\hat{n}}$ denoted the starting domino of $\D_{\p(0)}$ times $Z,$ that is, $Z_n,$ if $\p(0)=\emptyset,$ and otherwise it is $\alpha_{\p(0)}Z_{a_{\p(0)}}-\beta_{\p(0)}Z_{b_{\p(0)}}$.
\end{obs}
\begin{proof}By \cref{cor:inj_when_epses_above_yellow_red_neq0} and \cref{lem:inj_when_eps_neq0} we have that
\begin{align*}&Z_\star=
\llrr{d_\star \hat{n}AB}Z_{c_\star}-\llrr{c_\star \hat{n}AB}Z_{d_\star}+\llrr{c_\star d_\star AB}Z_{\hat{n}},\\
&Z_0=
\llrr{b_0\hat{n}AB}Z_{a_0}-\llrr{a_0\hat{n}AB}Z_{b_0}+\llrr{a_0b_0AB}Z_{\hat{n}}.\end{align*}
Substituting, and using the Pl\"ucker relations, we get\begin{align*}\llrr{Z_iZ_jZ_\star Z_0}=
\llrr{ijAB}\llrr{\hat{n}a_0b_0|AB|c_\star d_\star\hat{n}},~\text{in particular }\Fdelr^\D=\llrr{a_0b_0AB}\llrr{\hat{n}a_0b_0|AB|c_\star d_\star\hat{n}}.\end{align*}
The rightmost equality follows for combining the two equations. 

By \cref{cor:inj_when_all_epses_neq0} $Z_{\hat{n}}$ does not include the indices $A,B.$ Using the Pl\"ucker relations we can write $\llrr{ijAB}\llrr{\hat{n}a_0b_0|AB|c_\star d_\star\hat{n}}=-\llrr{\hat{n}a_0b_0|c_\star d_\star|AB\hat{n}}$, which shows the permissibility.\end{proof}
\begin{lemma}\label{lem:eps_red_vanishes_on_red}
Let $\D$ be a chord diagram. Then $\Fdelr^\D$ does not vanish on $S_\D,$ and for $U\vdots V\in \overline{S}_\D\setminus\SA$ 
\[\Fdelr^\D=0\Leftrightarrow U\vdots V\in\Sred_\D.\]
\end{lemma}
\begin{proof}
    If $U\vdots V\in\overline{S}_\D\setminus \SA$ then by \cref{lem:uniqueness_or_SA} the rows $D_0,D_\star$ of a domino limit $C\vdots D$ of $U\vdots V$ are well defined. 
    Using \cref{obs:well_def_red_blue_outside_SA} 
\[\Fdelr=\llrr{Z_{a_0}Z_{b_0}Z_\star Z_0}=\det(GZ) \]where $G$ is a matrix whose first $k$ rows are a basis of $C,$ the next two rows are $D_\star,D_0$ and the last two rows are $\ee_{a_0},\ee_{b_0}.$

By definition on $\Sred\setminus\SA$ we have a linear combination of $D_0,D_\star$ which belongs to $U,$ hence $G$ is not of full rank and thus $\Fdelr^\D=0.$

Conversely, if $\Fdelr^\D$ then $\det(GZ)=0.$ Write $G'$ for the first $k+2$ rows of $G.$ If $G'$ is not of full rank then $\Span(D_0,D_\star)\cap U\neq \{0\}.$ 
which implies $U\vdots V\in\Sred,$ by definition. 
Otherwise $G'$ is of full rank, and is a non negative matrix, since it is a representative of $U+V.$ $\det(GZ)$ is precisely the twistor $\llrr{Y'Z_{a_0}Z_{b_0}}$ where $Y'=G'Z,$ which by \cref{prop:CD bdries} can only vanish if $U+V$ contains a vector with support $a_0,b_0,$ but this implies $U\vdots V\in\SA_D,$ which is a contradiction.

If $U\vdots V\in S_\D$ then $D_0,D_\star$ are again well defined, by \cref{thm:dominoAndBCFWparams}, and $G'$ can be taken to be the domino matrix whose rows are the rows of $C,D,$ hence it is of full rank. The argument proceeds as above, only that we use \cref{cor:epsilon_signs} to show that the twistor $\llrr{(C+D)Z_{a_0}Z_{b_0}}\neq 0.$\end{proof}
An immediate corollary of the above proof is
\begin{cor}\label{cor:sign_del0_on_FL}$\llrr{i,i+1,D_\star,D_0}\geq 0$ for every $\D\in\CD^1_{n,k},i<n-1,$ in particular, $\Fdelr^\D\geq 0$ for every $\D\in\CD^1_{n,k}.$
\end{cor}
\begin{lemma}\label{lem:eps_blue_vanishes_on_blue}
 Let $\D$ be a chord diagram. Then for $U\vdots V\in\overline{S}_\D\setminus\SA$ 
\[\Fepsh^\D=0~\text{for some blue }\D_h\Leftrightarrow U\vdots V\in\Sblue_\D.\]No $\Fepsh^\D$ for $\D_h\in\Blue_\D$ vanishes on $S_\D.$
\end{lemma}
\begin{proof}
If $U\vdots V\in\overline{S}_\D\setminus \SA$ then by \cref{lem:uniqueness_or_SA} the vectors $D_\star,u$ of a domino limit $C\vdots D$ of $U\vdots V$ are well defined. They are also well defined for $U\vdots V\in S_\D,$ by the argument used in the previous lemma.

    Using \cref{obs:well_def_red_blue_outside_SA} 
\[\Fepsh=\llrr{Z_{a_0}Z_{b_0}\Zu Z_\star}=\pm\det(GZ) \]where $G$ is a matrix whose first $k$ rows are a basis of $C,$ the next row is $D_\star,$ then $u_\star$ and the last two rows are $\ee_{a_h},\ee_{b_h}.$

By definition on $\Sblue\setminus\SA$ we have a linear combination of $\ee_{a_0}, \ee_{b_0},u,D_\star$ which belongs to $U,$ hence $G$ is not of full rank and thus $\Fepsh=0.$

Conversely, if $\Fepsh=0$ then $\det(GZ)=0.$ Write $G'$ for the first $k+2$ rows of $G.$ If $G'$ is not of full rank then $\Span(u,D_\star)\cap U\neq \{0\},$ but the support of $u$ is $c_\star,d_\star$ and the support of $D_\star-u$ is either $n$ or $a_{\p(0)},b_{\p(0)}.$ Thus, the existence of a non zero vector $v$ in the intersection is impossible under our assumptions: for $U\vdots V\in S_\D$ it cannot happen by \cref{lem:bcfw_after_fl},~\cref{it:bcfw_after_fl_for_Sblue_argument}, and for $U\vdots V\notin\SA$ it cannot happen since the existence of $v$ is witnesses that $U\vdots V\in\SA.$

So we may assume $G'$ is of full rank.
Note that the row span of $G'$ is non negative:
$D_\star-u$ is proportional to $\ee_n$ or to $\pr_{a_\p(0)b_{\p(0)}}C_{\p(0)},$ which is a well defined vector by \cref{lem:uniqueness_or_SA}. The row span of $G'$ can be taken as the limit $\lim_{t\to0}G'_t$ where $G'_t$ is a matrix whose first $k+1$ rows are those of $G',$ and its $(k+2)$th row is $D_0^t$ obtained from $D_0$ by scaling its $a_0,b_0$ entries by $t.$ Thus, $G'_1$ spans $U\vdots V,$ which is a nonnegative space, and it is easy to see that replacing $D_0$ by $D_0^t$ for $t>0$ does not change this property. Since $G'=G_0$ it must also be nonnegative. 

 $\det(GZ)$ is precisely the twistor $\llrr{Y'Z_{a_h}Z_{b_h}}$ where $Y'=G'Z,$ which by \cref{prop:CD bdries} can only vanish if the row span of $G'$ contains a vector with support $a_h,b_h,$ but this implies that\[\left(U+\Span(u,D_\star)\right)\cap\Span(\ee_{a_h},\ee_{b_h})\neq 0\Leftrightarrow U\cap\Span(\ee_{a_h},\ee_{b_h},u,D_\star)\neq 0.\]Hence, by definition, $U\vdots V\in\Sblue,$ and by \cref{obs:cell_not_red_blue} $U\vdots V$ cannot belong to $S_\D.$
\end{proof}
\begin{lemma}\label{lem:Sinj_strata}
For every positive matrix $Z,$ on $S_\D\cup\Srem_\D$ all functions $F^\D_{\varepsilon_h}$ of \eqref{eq:def_F_eps} are well defined and non zero.
\end{lemma}
\begin{proof}
By definition, 
\[\Srem_\D=\overline{S}_\D\setminus\left(\SA\cup\Sred_\D\cup\Sblue_\D\right).\]
By \cref{prop:CD bdries} a boundary twistor vanishes on an element of $\overline{S}_\D$ if and only if it belongs to $\SA.$ By \cref{lem:eps_red_vanishes_on_red} $U\vdots V$ belongs to  $\Sred_\D\setminus\SA$ if and only if $\Fdelr^\D=\llrr{Z_{a_0}Z_{b_0}Z_\star Z_0}=0,$ which combined with \cref{lem:quad_lemma} is equivalent to $F^\D_{\varepsilon_h}$ for some red $\D_h.$
By \cref{lem:eps_blue_vanishes_on_blue}, $U\vdots V\in \Sblue_\D\setminus\SA$ if and only if $\llrr{Z_{a_h}Z_{b_h}\Zu Z_\star}=0$ for some blue $\D_h.$  
Thus, no $F^\D_{\varepsilon_h}$ vanishes on $\Srem_\D$ for $h\in[k+1]\cup\{\star\},$ and they are all well defined by \cref{cor:inj_when_epses_above_yellow_red_neq0}.  

Similar argument works for $S_\D,$ only that in this case we cannot argue immediately about all boundary twistors, but using \cref{cor:epsilon_signs} we can argue for the boundary twistors appearing in the expressions for $F^\D_{\varepsilon_h}$ for $h\notin\Blue_\D\cup\Red_\D.$ For the remaining $F^\D_{\varepsilon_h}$ we act as above. We use the non vanishing of $\llrr{a_hb_hAB},$ for $\D_h\in\Red_\D$, obtained from \cref{cor:epsilon_signs},~\cref{it:cor_eps_red} to deduce from \cref{lem:eps_red_vanishes_on_red} and \cref{lem:quad_lemma} the non vanishing of $F^\D_{\varepsilon_h}$ for $\D_h\in\Red_\D.$ We use \cref{lem:eps_blue_vanishes_on_blue} to show the non vanishing of $F^\D_{\varepsilon_h}$ for $\D_h\in\Blue_\D.$
\end{proof}
\begin{prop}\label{prop:inj_bdries}
The map $\tZ|_{S_\D\cup\Srem_\D}
$ is injective, and the inverse map is smooth. \end{prop}
\begin{proof}
By \cref{cor:reg_domino_strata} every element of $\Srem_\D$ has a unique weak $\D-$domino representative. The same holds for elements of $S_\D$ by \cref{thm:dominoAndBCFWparams}, and the spaces are disjoint, since for elements of $\Srem_\D\setminus S_\D$ at least one element of $\tVar_\D$ vanishes. By \cref{lem:Sinj_strata} on $S_\D\cup\Srem_\D$ no $F^\D_{\varepsilon_h}$ vanishes. Hence, by \cref{cor:inj_when_all_epses_neq0} $\tZ$ is injective on this space. The iterative construction of preimages via \cref{lem:inj_when_eps_neq0} is provides a functionary description of the domino coordinates, as long as the functions $F^\D_{\varepsilon_h}$ are non zero, thus the formula for the preimage of a point is smooth.


\end{proof}
\cref{prop:inj_bdries} has the following corollaries.
\begin{proof}[Proof of \cref{thm:inj}]
It follows trivially from \cref{prop:inj_bdries}. 
\end{proof}
\begin{cor}\label{cor:open_map}
    The map $\tZ:S_\D\to\Gr_{k,k+4;1}$ has a smooth inverse, hence is open.
\end{cor}
\begin{cor}\label{cor:cell_is_not_SA}
$S_\D\subseteq \Srem_\D,$ and in particular $S_\D\cap\SA=\emptyset.$ Thus, all boundary twistors have strong nonzero sign on the image of $S_\D.$
\end{cor}
\begin{proof}
Since the amplituhedron map restricted to a BCFW cell is a open map between manifolds of the same dimension, the image is open. $S_\D$ is contained in a single loopy positroid cell, by \cref{prop:fixed_positroids}, hence by \cref{obs:SA_strata} it may intersect $\SA$ only if it is contained in $\SA.$ In this case by \cref{prop:CD bdries} it must be contained in the zero locus of some twistor, but such zero locus it easily seen not to have interior points.

Thus, $S_\D\cap\SA=\emptyset.$ Then \cref{lem:eps_red_vanishes_on_red},~\cref{lem:eps_blue_vanishes_on_blue} and \cref{lem:Sinj_strata} imply that $S_\D\subseteq \Srem_\D.$ Now \cref{lem:sign_bdry_twistors} and \cref{prop:CD bdries} show that all boundary twistors have a strong non zero sign on the image of the cell. 
\end{proof}
\begin{thm}\label{prop:inj_codim_1}
    For every matchable boundary of a BCFW cell which is not contained in $\SA$, $\partial_{\zeta_i}S_\D$ the map $\tZ|_{\partial_{\zeta_i}S_\D}$ is invertible with a smooth inverse, hence injective.
\end{thm}
\begin{proof}By \cref{cor:regular_matchable_and_S_D_in_SREM} every such boundary of $S_\D$ contained in $\Srem_{\D'}$ where $\D'$ is either $\D$ or a chord diagram smaller than $\D$ in the partial order on chord diagrams defined in \cref{lem:from SREDSBLUE to smaller diags in diags order}. Thus, the result follows from \cref{prop:inj_bdries}.
\end{proof}
\begin{rmk}\label{rmk:inv_prob_bcfw}
The inverse problem procedure described in this section, leaning on \cref{lem:inj_when_eps_neq0}, allows a calculation of the domino form. Using \cref{cor:BCFW2dominoExplicit} one can also deduce the BCFW form. However, one can also calculate the BCFW form directly, via a BCFW analog of \cref{lem:inj_when_eps_neq0}: Recall the vectors $v^{\D_i}_\balp,\ldots, v^{\D_i}_\beps$ of \cref{subsec:BCFW_form}, defined by the requirement that the row corresponding to $\D_i$ in the BCFW form of $U\vdots V\in S_\D$ can be written as
\[\balp_i v^{\D_i}_\balp+\ldots\beps_i v^{\D_i}_\beps,\]where for $i=0$ we omit $v^{\D_0}_{\bgam},v^{\D_0}_{\bdel}$ from the summation, and for $i=\star$ we omit $v^{\D_\star}_{\balp},v^{\D_\star}_{\balp}.$
By \cref{obs:bcfw_form} $v^{\D_i}_\bzet,~\bzet\in\{\balp,\ldots,\beps\}$ depends only on reduced BCFW coordinates of chords $\D_h$ for $h\in\ahead(i).$
Then as in \cref{lem:inj_when_eps_neq0} we have for $i\neq 0,\star$
\[[\balp_i:\bbet_i:\cdots:\beps_i]=[\llrr{v^{\D_i}_\bbet Z,v^{\D_i}_\bgam Z,v^{\D_i}_\bdel Z,v^{\D_i}_\beps Z}:-\llrr{v^{\D_i}_\balp Z,v^{\D_i}_\bgam Z,v^{\D_i}_\bdel Z,v^{\D_i}_\beps Z}:\cdots:\llrr{v^{\D_i}_\balp Z,v^{\D_i}_\bbet Z,v^{\D_i}_\bgam Z,v^{\D_i}_\bdel Z}],\]as long as one of them, say, $\llrr{v^{\D_i}_\balp Z,v^{\D_i}_\bbet Z,v^{\D_i}_\bgam Z,v^{\D_i}_\bdel Z}\neq 0,$ and an analogous statement can be made for $i=0,\star.$ \cref{rmk:v_becomes_explicit} shows that this twistor agrees with $F^\D_{\varepsilon_i}$ of \eqref{eq:def_F_eps}, which is indeed non zero on $S_\D$ by \cref{thm:inj}. 

Since the vectors $[\balp_i:\bbet_i:\cdots:\beps_i]$ are well defined functions of $Y\vdots L,$ they have a non projective representative with all entries permissible. Since \eqref{eq:def_F_eps} is a permissible expression, this non projective representative can be taken to be the vector appearing in the right hand side of the above equation.
\end{rmk}
\subsection{The functionary representation of BCFW coordinates}
We now turn to study the functionary representation of the BCFW 
coordinates, and their behavior under promotions.
Recall \cref{rmk:inclusion_of_Var}, which defines a correspondence between the reduced BCFW variables of ${\D'}$ and ${\D}$ whenever the former is a subchord diagram of the latter. 
\begin{prop}\label{prop:bcfw_coords_promote_2_bcfw_coords}
Let $\D\in\CD_{n,k}^{\ell}$ be a chord diagram. Every element of $\BBCFWV^\D
$ has a representation by a rational pure permissible functionary of fixed sign on the image of the cell. Moreover, 
\begin{itemize}
\item if $\D=\pre_{n-1}\D'$ and $F'$ is a representing functionary for an element of $\BBCFWV^{\D'}
,$ then the representing functionary for the corresponding element of $\BBCFWV^\D
$ is $\Psi^\pre F'.$ In addition, $\Fdelr^\D=\Psi^\pre\Fdelr^{\D'},$ when they are defined, where we recall that if $\ell=1$ $\Fdelr^\D$ is defined in \cref{obs:well_def_red_blue_outside_SA}.
\item If $\D=\D_L\bcfw\D_R$ then $\alpha_k=\pm\halp_k,\ldots,\varepsilon_k=\pm\heps_k$ are realized by the twistors of \cref{cor:5_3}, up to signs which can be read from \cref{def:L=1domino_signs}. Let $F'$ be a representing functionary for an element of $\BBCFWV^{\D_L}$ or $\BBCFWV^{\D_R}$ then the representing functionary for the corresponding element of $\BBCFWV^\D$ equals $\Psi^\bcfw_{ac} F'$. In addition if either $\D'=\D_L$ or $\D'=\D_R$ is a $1-$loop diagram, then $\Fdelr^\D=\Psi^\bcfw_{ac}\Fdelr^{\D'}.$ 
\item If $\D=\FL(\D')$ then $\gamma_\star=\pm\hgam_\star,\delta_\star=\pm\hdel_\star,\varepsilon_\star=\pm\heps_\star$ are given by the twistors of \cref{cor:5_3}, up to signs which can are fixed by \cref{def:L=1domino_signs}. $\alpha_0=\pm\halp_0,\beta_0=\pm\hbet_0,\varepsilon_0=\pm\heps_0$ are given by the twistors $\llrr{b_0n_0AB},-\llrr{a_0n_0AB},\llrr{a_0b_0AB},$ up to signs from \cref{def:L=1domino_signs} again. If $F'$ is the representative functionary for an element of $\zeta^{\D'}_i\in
\BBCFWV^{\D'},$ including $\delta_{\tr(\D)},$ then the functionary for the corresponding variable in $
\BBCFWV^\D$ is $(\Psi^\FL F')/c_{\zeta,i}$, where \[c_{\zeta,i}=\begin{cases}
F^\D_{\alpha_\star}:=\llrr{c_\star d_\star Bn},&\zeta=
\bgam,~\D_i\in\Red_\D
\\\frac{1}{F^\D_{\alpha_\star}}=\frac{1}{\llrr{c_\star d_\star Bn}},&\zeta\in\{\balp,\bbet,\bgam,\beps\},~\D_i\in\Blue_\D
\\
\Psi^\FL F^{\D'}_{\delta_0}=\frac{\llrr{na_0b_0|AB|c_\star d_\star n}}{\llrr{c_\star d_\star Bn}},&\zeta=
\bdel,~\D_i\in\Red_\D
\\1,&\text{otherwise}\end{cases}\]
\end{itemize}
\end{prop}
Note that the scaling in the case of the forward limit is compatible with the scaling that the passage from BCFW to reduced BCFW coordinates involves, see \cref{def:BCFW coords}.
In view of the previous proposition we introduce the following notation.
\begin{nn}\label{nn:representing_functionaries}
For a chord diagram $\D$ and $\zeta\in\BBCFWV^\D
$ we write $F^\D_\zeta = F^\D_\zeta(Y\vdots L;Z)$ for the (permissible, rational, pure) functionary representing it.
By slight abuse of notations we will sometimes think of $F^\D_\zeta$ as a function of $C\vdots D$ and $Z,$ by composing with the amplituhedron map. 
\end{nn}
\begin{proof}
We use the notations of \cref{rmk:inv_prob_bcfw}. 
The proof walks in the path of the proof of \cite[Proposition 11.14]{even2023cluster}. We will often consider inner products of row vectors. Whenever we write $u \cdot v,$ with $u\in \R^{N_1},~v\in\R^{N_1\cup N_2}$ are row vectors, we implicitly identify $u$ as the element $\pre_{N_2} u\in\R^{N_1\cup N_2}.$ By  \cref{rmk:inv_prob_bcfw} we can write, for $i\in[k+1]_\D$, 
\begin{align*}[\balp^\D_i(Y\vdots L;Z)&:\cdots:\beps^\D_i(Y\vdots L;Z)]=[F_{\balp_i}^\D(Y\vdots L;Z):\cdots:F_{\beps_i}^\D(Y\vdots L;Z)]
\\\notag&=[\llrr{v^{\D_i}_\bbet Z, v^{\D_i}_\bgam Z, v^{\D_i}_\bdel Z, v^{\D_i}_\beps Z}:\llrr{v^{\D_i}_\balp Z, v^{\D_i}_\bgam Z,v^{\D_i}_\bdel Z, v^{\D_i}_\beps Z}:-\llrr{v^{\D_i}_\balp Z,v^{\D_i}_\bbet Z,v^{\D_i}_\bdel Z, v^{\D_i}_\beps Z}:\\&\notag\quad\qquad\qquad\qquad\qquad\qquad\qquad:-\llrr{v^{\D_i}_\balp Z, v^{\D_i}_\bbet Z, v^{\D_i}_\bgam Z, v^{\D_i}_\beps Z}:\llrr{v^{\D_i}_\balp Z, v^{\D_i}_\bbet Z, v^{\D_i}_\bgam  Z, v^{\D_i}_\bdel Z}]\end{align*}
and a similar expression for the BCFW coordinates of the top red and yellow chord, if exist. We also have a similar expression for $[\balp^{\D'}_i:\cdots:\beps^{\D'}_i].$ This determines the BCFW coordinates up to scalings of the rows, where the signs are easily determined from the description of the BCFW form in \cref{subsec:BCFW_form} and the signs of boundary twistors \cref{lem:sign_bdry_twistors}, but we will not need them for this claim.
For this proof it will be more convenient to work in the $B$-amplituhedron, and use the vector promotions of \cref{rmk:vector_promotion}.
Using the correspondence between functionaries on $\Ampl_{n,k,4}^\ell$ and Pl\"ucker functionaries on $\mathcal{B}_{n,k,4}^\ell$ given in \cref{lem:B_amp_A_amp_dual} and \cref{obs:_check_B_amp_A_amp_dual}, we have $i\in[k+1]_\D$

\begin{align}\label{eq:domino_funcs_D}&[\balp^\D_i(Y\vdots L;Z):\cdots:\beps^\D_i(Y\vdots L;Z)]=[F_{\balp_i}^\D(\bz):\cdots:F_{\beps_i}^\D(\bz)]
\\\notag&=[\lr{v^{\D_i}_\bbet\cdot \bz, v^{\D_i}_\bgam\cdot \bz, v^{\D_i}_\bdel\cdot \bz, v^{\D_i}_\beps \cdot \bz}:-\lr{v^{\D_i}_\balp \cdot \bz, v^{\D_i}_\bgam \cdot \bz,v^{\D_i}_\bdel \cdot \bz, v^{\D_i}_\beps \cdot \bz}:\lr{v^{\D_i}_\balp\cdot \bz,v^{\D_i}_\bbet \cdot \bz,v^{\D_i}_\bdel \cdot \bz, v^{\D_i}_\beps \cdot \bz}:\\\notag&\qquad\qquad\qquad\qquad\qquad\qquad\qquad\qquad:-\lr{v^{\D_i}_\balp \cdot \bz, v^{\D_i}_\bbet \cdot \bz, v^{\D_i}_\bgam \cdot \bz, v^{\D_i}_\beps \cdot \bz}:\lr{v^{\D_i}_\balp \cdot \bz, v^{\D_i}_\bbet \cdot \bz, v^{\D_i}_\bgam  \cdot \bz, v^{\D_i}_\bdel \cdot \bz}]\end{align}
where $\bz\in\CC\Gr_{2,n;1}$ equals $f_Z^{-1}(Y\vdots L),$ where $f_Z$ is the map of \cref{lem:B_amp_A_amp_dual}, $\cdot$ is the standard inner product performed row-by-row, that is $v\cdot \bz=\bz v^T$, and we slightly abused notations and referred to $F^\D_\bzet$ as a function of $\bz,$ which is fine by the same theorem. 
Let $\check{\bz}$ be the $4-$plane obtained from $\bz$ by adding the columns $z_A,z_B,$ via the choice of a local lift $s:U\to\CC\check{\Gr}_{2,n;1}$ where $U$ is an open neighborhood of $\bz\in\CC\Gr_{2,n;1}.$
Then, independently of the local lift, for $i=0,$ 
\begin{align*}\notag[\balp^\D_0(Y\vdots L;Z):&\bbet^\D_0(Y\vdots L;Z):\beps^\D_0(Y\vdots L;Z)]
=[F_{\balp_0}^\D(\bz):F_{\bbet_0}^\D(\bz):F_{\beps_0}^\D(\bz)]
\\&=[\lr{v^{\D_0}_\bbet \cdot \bz, v^{\D_0}_\beps \cdot \bz,z_A,z_B}:-\lr{v^{\D_0}_\balp \cdot \bz, v^{\D_0}_\beps \cdot \bz,z_A, z_B}:\lr{v^{\D_0}_\balp \cdot \bz,v^{\D_0}_\bbet \cdot \bz,z_A, z_B}]\end{align*}and for $i=\star$ 
\begin{align*}\notag[\bgam^\D_\star(Y\vdots L;Z):&\bdel^\D_\star(Y\vdots L;Z):\beps^\D_\star(Y\vdots L;Z)]
=[F_{\bgam_\star}^\D(\bz):F_{\bdel_\star}^\D(\bz):F_{\beps_\star}^\D(\bz)]
\\&=[\lr{v^{\D_\star}_\bdel \cdot \bz, v^{\D_\star}_\beps \cdot \bz,z_A,z_B}:-\lr{v^{\D_\star}_\bgam \cdot \bz, v^{\D_\star}_\beps \cdot \bz,z_A, z_B}:\lr{v^{\D_\star}_\bgam \cdot \bz,v^{\D_\star}_\bdel \cdot \bz,z_A, z_B}]\end{align*}
\\\textbf{The case $\D=\pre_{n-1}(\D').$}
\\The effect of $\Psi^\pre$ on a functionary of $\D$ is trivial. To show that the same holds for the BCFW coordinates, note that the effect of $\Psi^{\pre_{n-1}}$ on $\bz$ is the removal of the $n-1$th column of $\bz$, by \cref{rmk:vector_promotion}. The lift $s$ can be chosen so that $z_A,z_B$ are unaffected. The effect on $v^{\D_i}_\bzet$ is adding a new zero entry at position $n-1.$ Thus, the comparing right hand side of \eqref{eq:domino_funcs_D} applied to a point in $S_\D,$ and to the corresponding point of $S_{\D'},$ we see that they are the same. The same argument works for $\Fdelr^\D$ using its expression given in \cref{obs:well_def_red_blue_outside_SA}.
\\\textbf{The case $\D=\D_L\bcfw\D_R.$}
\\Moving to the BCFW step case, we first observe that the row corresponding to the rightmost top chord $\D_k$ has a representation by the twistors given in \cref{cor:5_3}:
\[[F^\D_{\balp_k}(\bz):F^\D_{\bbet_k}(\bz):\cdots:F^\D_{\beps_k}(\bz)]=[\llrr{bcdn}:-\llrr{acdn}:\cdots:\llrr{abcd}]
\]
These twistors have strong non zero signs by \cref{lem:sign_of_chord_twistors}. The above projective vector equals $[\lr{z_{b}z_{c}z_{d}z_n}:-\lr{z_{a}z_{c}z_{d}z_n}:\cdots:\lr{z_{a}z_{b}z_{c}z_{d}}]$ by \cref{lem:B_amp_A_amp_dual}.
Let $\D_i$ be a chord coming from the right component.
By the construction of the vectors $v^{\D_i}_\bzet,$ performed in \cref{subsec:BCFW_form}, we easily see that
\[v^{\D_i}_\bzet=y_{c}(\frac{\bgam_k}{\bdel_k}) y_{d}(\frac{\bdel_k}{\beps_k})v^{(\D_R)_i}_\bzet.\]Since the action of $y_j$ are realized by right matrix multiplication, and its transpose is $x_{j}$ we have
\begin{equation}\label{eq:relate_vecs_before_after_bcfw}
v^{\D_i}_\bzet \cdot \bz = v^{(\D_R)_i}_\bzet (x_{d}(\frac{\bdel_k}{\beps_k})x_{c}(\frac{\bgam_k}{\bdel_k})  \bz).
\end{equation}
We now induct on chords, in descending order of endpoints, that the BCFW variables of $\D_R$ are promoted under $\Psi^\bcfw$ to the corresponding BCFW coordinates of $\D.$
First observe that after substituting in \eqref{eq:relate_vecs_before_after_bcfw} the twistor representations for $\bgam_k,\bdel_k,\beps_k$ the expression in the bracket is precisely the BCFW promotion of $\bz$, \cref{rmk:vector_promotion}, in the form of \cref{obs:BCFW_promotion_domino}. If $\D_i$ corresponds to a top chord of $\D_R$ then by \cref{obs:bcfw_form} $v^{(\D_R)_i}_\bzet$ is a constant vector, the right hand side of \eqref{eq:relate_vecs_before_after_bcfw} is the BCFW promotion of the corresponding vector for $\D_R.$ By \eqref{eq:domino_funcs_D}, and \cref{rmk:vector_promotion} which shows that the promotion on functionaries is induced from the promotion of vectors, the BCFW coordinates of $\D_i$ are the promotions of the BCFW coordinates of $(\D_R)_i.$

If $\D_i$ is not a top chord, then by \cref{obs:bcfw_form} $v^{\D_i}_\bzet$ is a Laurent polynomial in the BCFW coordinates of chords $\D_h$ for $h\in\ahead(i),$ and the same holds for $v^{(\D_R)_i}_\bzet.$By induction, the functionaries $F^{\D}_{\bzet_h}$ describing the BCFW coordinates of $\D_h$ for $h\in\ahead(i)$ are either the twistors of the $k$th chord, or are promotions of the corresponding functionaries $F^{\D_R}_{\bzet_h}$ for $(\D_R)_h.$
Thus,\begin{align*}
v^{\D_i}_\bzet((F^\D_{\bsigm_h})&_{h\in\ahead(i),\bsigm\in\{\balp,\ldots,\beps\}})\cdot \bz= \\&=v^{{\D_R}_i}_\bzet((F^\D_{\bsigm_h})_{h\in\ahead(i)\setminus\{k\},\bsigm\in\{\balp,\ldots,\beps\}})(x_d(-\frac{\lr{abdn}}{\lr{abcd}})x_c(-\frac{\lr{abdn}}{\lr{abdn}})\cdot \bz) =
\\&=v^{{\D_R}_i}_\bzet((\Psi^\bcfw_{ac}F^{\D_R}_{\bsigm_h})_{h\in\ahead(i)\setminus\{k\},\bsigm\in\{\balp,\ldots,\beps\}})\cdot(\Psi^\bcfw_{ac}\bz) =\\&= \Psi^\bcfw_{ac}([v^{{\D_R}_i}_\bzet((F^{\D_R}_{\bsigm_h})_{h\in\ahead_{\D_R}(i),\bsigm\in\{\balp,\ldots,\beps\}})]\cdot \bz),
\end{align*}
where $\ahead_{\D_R}$ is the function $\ahead$ calculated with respect to the diagram $\D_R,$ the first equality is  \eqref{eq:relate_vecs_before_after_bcfw}, the second is induction, and the third is by definition. Using \eqref{eq:promotion_for_coords_prom} 
and \eqref{eq:domino_funcs_D} we see that also the BCFW describing functionaries of the coordinates of $\D_i$ are promotions of the corresponding coordinates for $(\D_R)_i.$

Chords coming from the left component are treated similarly. Also $\Fdelr^\D$ is treated similarly, using its expression given in \cref{obs:well_def_red_blue_outside_SA}.
\\\textbf{The case $\D=\FL(\D').$}
\\
We can obtain the representing functionaries and vectors $v_{\bzet}^{\D_i}$ in using \cref{rmk:inv_prob_bcfw} again.
The BCFW functionaries  $F^\D_{\hgam_\star},F^\D_{\hdel_\star},F^\D_{\heps_\star}$ are given by the twistors of \cref{cor:5_3}, and have a strong non zero sign by \cref{lem:sign_of_chord_twistors}. 
In addition
\[[F^\D_{\halp_0}:F^\D_{\hbet_0}:F^\D_{\heps_0}]=[\llrr{b_0ABn},-\llrr{a_0ABn},\llrr{a_0b_0AB}].\]This follows from \cref{cor:5_3} applied to $\D',$ where $\D_0=\D'_{k+1}$ is the rightmost top chord, and the definition of the forward limit which just chops the $A,B$ entries of $\D'_{k+1}$ to obtain $\D_0.$ Note that these three twistors are invariant under the forward limit promotion, and they cannot be identically zero, by the proof \cref{thm:inj} which shows that they are coordinates on $\tZ(S_\D).$ Hence by \cref{prop:sign_under_FL_promotion} they have fixed weak non-zero sign on $S_\D.$

Let $s,\check\bz,z_A,z_B$ be as above. For every such choice of $z_A,z_B$ there is a unique choice of $F^{\D_\star}_\halp,F^{\D_\star}_\hbet$ such that \begin{equation}\label{eq:v_star_perp_z}\check{v}^{\D_\star}:=\left(\pre_{AB}(\lr{d_\star ABn}\ee_{c_\star}-\lr{c_\star ABn}\ee_{d_\star}+\lr{c_\star d_\star AB}\ee_{n})+F^{\D_\star}_\halp \ee_A+F^{\D_\star}_\hbet \ee_B\right) \perp \check{\bz},\end{equation}given by \[F^{\D_\star}_\halp=\lr{c_\star d_\star Bn},\qquad F^{\D_\star}_\hbet=-\lr{c_\star d_\star An},\]
by the $B$-amplituhedron analog of \cref{lem:5_3}, whose proof is identical. Note that the Pl\"ucker coordinates are calculated with respect to $\check{z},$ but this affects only the non permissible twistors.
Similarly, there is a unique choice of $F^{\D_0}_\hgam,F^{\D_0}_\hdel$ such that 
\begin{equation}\label{eq:v_0_perp_z}
\check{v}^{\D_0}:=\left(\pre_{AB}(\lr{b_0 ABn}\ee_{a_0}-\lr{a_0 ABn}\ee_{b_0}+\lr{a_0b_0AB}\ee_{n})+F^{\D_0}_\hgam \ee_A+(F^{\D_0}_\hdel +\frac{\hbet_\star}{\halp_\star}F^{\D_0}_\hgam)\ee_B\right) \perp \check{\bz}.
\end{equation}
given by
\[F^{\D_0}_\hgam=\lr{a_0b_0Bn},\qquad F^{\D_0}_\hdel=-\frac{\lr{na_0b_0|AB|c_\star d_\star n}}{\lr{c_\star d_\star Bn}}.\]

Recall \cref{obs:parameterizing_S_check_D}. Let $S_{\check\D}$ be constructed from $S_{\D'}$ with the same operations as in the definition of the forward limit, only without chopping the columns $A,B$ of $C\vdots D.$ Then the argument of \cref{cor:5_3} would show that above twistors and functionaries 
would indeed represent the additional BCFW coordinates $\halp_\star,\hbet_\star,\hgam_0,\hdel_0.$ Therefore we can think of the lift to $\check\bz$ as a lift to $S_{\check\D}.$ 
Observe that $F^{\D_0}_\hzet=\Psi^\FL F^{\D'_0}_\hzet,~\hzet\in\{\halp,\ldots,\heps\}.$


As in the BCFW case, we will induct on chords, in descending order of endpoints, that their \emph{non reduced} BCFW coordinates are the promotions of the corresponding BCFW coordinates for $\D'.$ We have shown it for $\D_0.$ It will be more convenient to work with the vectors $v_{\hzet}^{\D_i}$ obtained from $v_{\bzet}^{\D_i}$ by scaling in factors of $1,\halp_\star^{\pm1},\hdel_0$ according to the different cases of $i,\hzet,$ see \eqref{eq:passage_from_v_hat_to_v_bar}. This scaling has the effect that \eqref{eq:domino_funcs_D} holds with every $\bzet$ being replaced by $\hzet.$

For $i\in[k+1]_\D$,
by \cref{rmk:bcfw_vecs_fl} we can write
\begin{align}\label{eq:FL_v_promotion}
\notag v^{\D_i}_{\hzet}\cdot \bz&=\rem_{AB}\left(y_{c_\star}(\frac{\gamma_\star}{\delta_\star})x_{A}(\frac{\hbet_\star}{\halp_\star})\check{v}_{\hzet}^{\D'_i}-s\check{v}^{\D_\star}-r\check{v}^{\D_0}\right)\cdot \bz\\\notag&=\left(y_{c_\star}(\frac{\gamma_\star}{\delta_\star})x_{A}(\frac{\hbet_\star}{\halp_\star})\check{v}_{\hzet}^{\D'_i}-s\check{v}^{\D_\star}-r\check{v}^{\D_0}\right)\cdot\check{\bz}\\&=y_{c_\star}(\frac{\gamma_\star}{\delta_\star})x_{A}(\frac{\hbet_\star}{\halp_\star})\check{v}_{\hzet}^{\D'_i}\cdot\check{\bz}-s\check{v}^{\D_\star}\cdot\check{\bz}-r\check{v}^{\D_0}\cdot\check{\bz}=y_{c_\star}(\frac{\gamma_\star}{\delta_\star})x_{A}(\frac{\hbet_\star}{\halp_\star})\check{v}_{\hzet}^{\D'_i}\cdot\check{\bz},
\end{align}
where the vectors $\check{v}_{\hzet}^{\D'_i}$ are the vectors denoted $\check{v}_{\hzet,i}^{t-1}$ in \cref{rmk:bcfw_vecs_fl}, the second equality is because the $A,B$ entries of the vector in the parentheses are $0,$ and the last equality follows from \eqref{eq:v_0_perp_z},~\eqref{eq:v_star_perp_z}. As in the case of the BCFW step, by \cref{obs:bcfw_form}, $\check{v}_{\hzet}^{\D'_i}$ depends only on BCFW variables of $\D_h$ for $h\in\ahead(i)\setminus\{\star\}.$ By induction the  functionaries $F^{\D}_{\bzet_i}$ representing these BCFW variables in $\D$ for are the forward limit promotions of the  functionaries $F^{\D'}_{\bzet_i}$ representing the corresponding variables on $\D'.$ Thus, as in the case of BCFW step
\begin{align*}
v^{\D_i}_\hzet((F^\D_{\hsigm_h})&_{h\in\ahead(i),\hsigm\in\{\halp,\ldots,\heps\}})\cdot \bz= \\&=v^{{\D'}_i}_\hzet((F^\D_{\hsigm_h})_{h\in\ahead(i)\setminus\{\star\},\hsigm\in\{\halp,\ldots,\heps\}})(x_{c_\star}(-\frac{\lr{d_\star ABn}}{\lr{c_\star ABn}})y_A(-\frac{\lr{c_\star d_\star An}}{\lr{c_\star d_\star B n}})\cdot\check{\bz}) =
\\&=v^{{\D'}_i}_\hzet((\Psi^\FL F^{\D'}_{\hsigm_h})_{h\in\ahead(i)\setminus\{\star\},\hsigm\in\{\halp,\ldots,\heps\}})\cdot(\Psi^\FL \check{\bz}) =\\&= \Psi^\FL\left([v^{{\D'}_i}_\hzet((F^{\D'}_{\hsigm_h})_{h\in\ahead_{\D'}(i),\hsigm\in\{\halp,\ldots,\heps\}})]\cdot\check{\bz}\right),
\end{align*}
where $\ahead_{\D'}$ is the function $\ahead$ calculated with respect to the diagram $\D'$,  the first equality is \eqref{eq:FL_v_promotion}, the second is induction, and the third is from the definition of vector promotions in \cref{rmk:vector_promotion}.
Using the $\bullet=\FL$ case of \eqref{eq:promotion_for_coords_prom} in which we substitute $\check\bz$ for $\bz,$ and \eqref{eq:domino_funcs_D}, we see that also the BCFW describing functionaries of the coordinates of $\D_i$ are promotions of the corresponding coordinates for $\D'_i.$ Scaling by the vectors by $c_{\zeta,i},$ which are the functionaries representing the scalings by $1, \halp_\star^{\pm1},\hdel_0$ in the different cases yield the reduced vectors $v_{\bzet}^{\D_i}$, by \eqref{eq:passage_from_v_hat_to_v_bar}. Observe that by \cref{obs:bcfw_form} the vectors $v^{\D_i}_{\bzet}$ indeed depend only on reduced BCFW coordinates. 
\\\textbf{Signs, permissibility and purity.}
\\The functionaries representing BCFW variables cannot be identically zero on a BCFW cell, since they form coordinates by the proof of \cref{thm:inj}. Moreover, since they represent coordinates which are non-zero on the cell, the functionaries themselves cannot vanish on the cell. 
$\Fdelr^\D$ is non-zero on the cell by \cref{lem:Sinj_strata} and \cref{cor:regular_matchable_and_S_D_in_SREM}.

That the reduced BCFW coordinates are pure follows from simple induction on chord diagrams, using that for $\ell=0,k=1$ diagrams the twistor solution of a single (top) row is pure, and purity is preserved under promotions by \cref{cor:purity_preserved_under_prom}. For $\Fdelr^\D$ it follows from \cref{obs:well_def_red_blue_outside_SA} applied in the case $\D=\FL(\D'),$ and iterated applications of \cref{cor:purity_preserved_under_prom}.

$\Fdelr^\D$ is permissible by \cref{lem:quad_lemma}. 
Regarding the other reduced BCFW variables, note that if $\ell=0$ then all functionaries are permissible by definition. Also, it follows from induction and \cref{obs:pre_permissible} that a permissible functionary gets promoted under $\Psi^\pre,\Psi^\bcfw$ to a permissible functionary.
Thus, we only need to show that if $\D=\FL(\D')$ then all reduced BCFW variables of $\D$ are permissible.
By \cref{obs:pre_permissible} and the first part of the proposition, it is enough to show that the functionaries which get promoted to the functionaries representing reduced BCFW variables are pre-pemissible.
We will prove that by induction on the generation sequence of $\D'.$ Observe first that before performing the BCFW steps introducing the blue and red chords of $\D$ no representing functionary of a BCFW coordinate involves the indices $A,B.$

Let $\D^h,\D^{h+1},\ldots,\D^{k+1}$ be the subdiagrams whose rightmost top chord is the chord which in $\D$ is the lowest blue, second lowest blue (if they exist),...,top red chord.
Let $\D^{j}_L$ be the left subdiagram of $\D_j$ and $\D^{j}_R$ be the right subdiagram of $\D_j,$ which is $\D^{j-1}$ if $j>h.$ Then $\D_{j+1}=\D^j_L\bcfw\D^j_R.$ We will show below that
\begin{obs}\label{obs:pre_perm_BCFW}
\begin{itemize}
\item All functionaries representing BCFW variables have a representation in which each twistor either includes or excludes both $A,B$, except than those representing $\halp_j,\hbet_j,\hgam_j,\heps_j$ $j\in\Blue_\D$, $\hgam_j,\hdel_j$ for $j\in\Red_\D.$
\item The functionaries representing $\halp_j,\hbet_j,\hgam_j,\heps_j$ $j\in\Blue_\D$, are pre-permissible of level $1.$
\item The functionaries $F^{(\D^l)_j}_\hgam$ representing $\hgam_j,~j\in\Red_{\D_l}$ in $\D_l$ for $l>j$ are pre-permissible of level $-1;$  
 $F^{(\D^l)_l}_\hgam=\llrr{a_lb_lBn}.$
\item The functionaries representing $\hdel_j,~j\in\Red_{\D_l}$ in $\D_l$ are pre-permissible of level $1,$ hence the functionary representing $\hdel_j/\hdel_{k+1}$ is pre-permissible.
\end{itemize}
\end{obs}
Given the observation, the permissibility of the reduced variables of $\D$ follows from \cref{obs:pre_permissible} using $F^{\D_\star}_\halp=\llrr{c_\star d_\star Bn}$.
\end{proof}
\begin{proof}[Proof of \cref{obs:pre_perm_BCFW}]
The observation is also proven by induction. All variables of $\D^j_L,\D^h_R$ are permissible, as explained above, have permissible functionaries, which by \cref{obs:pre_permissible} remain permissible throughout the process.

When a blue chord $\D_j$ is added, the functionary $F^{(\D^j)_j}_\hdel$ representing $\hdel_j$ is $\llrr{a_jb_jd_\star n}$. The other BCFW variables of this chord involve $A,$ but not yet $B.$ The other steps of the construction, obtained by performing BCFW promotions via blue and red chords and a single $\pre_B$ operation, are easily seen to take $\delta_j$ to a pre-permissible functionary, and the other variables to pre-permissible functionaries of level $1.$

When a red chord $\D_j$ is added its $\halp_j,\hbet_j,\heps_j$ variables are pre-permissible, and this property is preserved under the BCFW promotions via red chords. Similarly, when $\D_j$ is added, $F^{(\D^j)_j}_\hdel=\llrr{a_jb_jAn}$ is pre-permissible of level $1,$ and it is easy to see that this property is preserved under the BCFW promotions via red chords.
Regarding $F^{(\D^j)_j}_\hgam,$ it equals $\llrr{a_jb_jBn}$ when $\D_j$ is introduced. If $j\neq k+1,$ then after the next BCFW step which adds $\D_{\p(j)}$, by the first part of the proof, the functionary representing $F^{(\D^{\p(j)})_j}_{\hgam}$ gets promoted to the pre-permissible functionary of level $-1$\[\Psi_{a_{\p(j)}A}^\bcfw(\llrr{a_jb_jBn})=-\frac{\llrr{na_jb_j|a_{\p(j)}b_{\p(j)}|ABn}}{\llrr{a_{\p(j)}b_{\p(j)}An}}.\]This property persists throughout the BCFW promotions via the remaining chords of $\Red_\D.$
\end{proof}
\begin{cor}\label{cor:strong_pos_coords}
For every $\D\in\CD_{n,k}^\ell,~\ell\in\{0,1\}$ the boundary twistors, the functionaries describing the reduced BCFW coordinates, the functionaries describing the domino coordinates and $\Fdelr^\D$ (which is not included in these sets when $\ell=1$) are not identically zero on $\Gr_{k,k+4;1}$. 
\end{cor}
\begin{proof}
Boundary twistors are clearly not identically $0.$ Functionaries describing the reduced BCFW or domino coordinates are non zero at least on the image of $S_\D,$ by the preimage finding procedure of \cref{thm:inj} and \cref{rmk:inv_prob_bcfw}, hence they cannot be identically $0$ as well.
$\Fdelr^\D$ is not identically $0$ by \cref{lem:eps_red_vanishes_on_red} and \cref{cor:regular_matchable_and_S_D_in_SREM}.
\end{proof}
\section{Tiling the amplituhedron}

The section is devoted to proving that the images of BCFW cells tile the amplituhedron. We will also describe the boundary of the amplituhedron. 
For this we need first to define what is a tiling. We follow \cite{BaoHe}.
\begin{definition}\label{def:tiling}
Let $f:X\to Y$ be a map of manifolds, and let $X^\geq\subseteq X$ compact. Let $S_1,\ldots, S_N\hookrightarrow X^\geq$ be manifolds of dimension $\dim(Y).$ 
We say that the images of $\{S_i\}_{i=1}^N$ \emph{tile} $f(X^\geq)$ if
\begin{itemize}
\item \emph{Injectivity}: $f:S_i \to f(S_i)$ is injective for every $i=1,\ldots,N.$
\item \emph{Separation}: $f(S_i)\cap f(S_j)=\emptyset$ for every $i\neq j.$
\item \emph{Surjectivity}: $\bigcup_{i=1}^N \overline{f(S_i)}=f(X^\geq).$
\end{itemize}
\end{definition}
We would like to prove that 
\begin{theorem*}for $\ell=0,1$ $k+4\leq n,$ and every positive $(k+4)\times n$ matrix $Z,$ the images of $\{S_\D\}_{\D\in\CD^\ell_{n,k}}$ tile $\Ampl_{n,k,4}^\ell(Z).$
\end{theorem*}
Since we have shown the injectivity of BCFW cells in \cref{thm:inj}, it remains to prove the following two theorems.
\begin{thm}\label{thm:sep}
Fix $k+4\leq n,$ and $\ell\in\{0,1\}.$ Let $\D\neq\D' \in \CD_{n,k}^\ell$, be chord diagrams, and $Z\in \Mat^{>}_{n \times(k+4)}$. Then, $\tZ(S_{\D})\cap\tZ(S_{\D'})=\emptyset,$ that is, the images of $S_\D,S_{\D'}$ are separated.
\end{thm}
\begin{thm}\label{thm:surj}
For every $k,n$ with $n\geq k+4$, $\ell\in\{0,1\}$, and $Z\in \Mat^{>}_{n \times(k+4)}$
\[\Ampl_{n,k,4}^\ell=\bigcup_{\D\in\CD_{n,k}^1}\overline{\tZ(S_{\D})}.\]Moreover, every boundary point of $\Ampl_{n,k,4}^\ell$ lies in the zero locus of a boundary twistor of the form $\llrr{i,i+1,j,j+1}$ or $\llrr{i,i+1,A,B},$ where
$+1$ is taken cyclically.
\end{thm}
Tilings of amplituhedra were studied in many works: \cite{BaoHe,PSW} studied tilings of the $m=2$ tree amplituhedron, \cite{akhmedova2023tropical} studied the tropical $m=2$ tree case; \cite{even2021amplituhedron,even2023cluster} studied the $m=4$ tree case; \cite{galashin2024amplituhedra} studied the tree momentum amplituhedron case; \cite{perlstein2025bcfw} studied the ABJM amplituhedron case.

\subsection{Local separation}
This section is devoted to proving the following result.
\begin{prop}\label{prop:local_sep}
Let $\D_1,\D_2\in\CD_{n,k}^\ell$ be two neighboring chord diagrams. Then for every positive $Z\in\Mat_{n,k+4}^>,$ $\tZ(S_{\D_1})\cap\tZ(S_{\D_2})=\emptyset.$ 
\end{prop}
\begin{proof}
Our strategy will be to find a \emph{separating functionary}, that is a functionary which is not identically $0$ on $\tZ(S_{\D_1}),\tZ(S_{\D_2}),$ and has opposite weak signs on them. We will then apply \cref{obs:weak_sign_separator}, which shows that in this case the spaces do not intersect. In order to guarantee that the separating functionaries we construct are not identically $0$ we will pick permissible functionaries which are Lauren monomials in 
reduced BCFW coordinates and $\Fdelr^\D$ for some $S_\D,~\D\in\CD_{n,k}^\ell.$

Recall \cref{def:minimal_neighboring}. Let $\bar{\D}_1,\bar{\D}_2$ be the minimal neighboring subdiagrams of $\D_1,\D_2.$ 
We prove the result by induction on the lexicographically ordered $(\ell-\bar\ell,k-\bar{k},n-\bar{n}),$ where $\D_1\in\CD_{n,k}^\ell$ and $\D'_1\in\CD_{\bar n,\bar k}^{\bar\ell}.$ 
If $(\ell-\bar\ell,k-\bar k,n-\bar n)=(0,0,0)$ then $\D_1,\D_2$ are their own minimal neighboring subdiagrams.
Recall that in this case, by \cref{obs:minimal_neighboring}, $\D_1,\D_2$ do not have the same rightmost top chord, or equivalently, the last operation performed in their generation sequences are different. There are three cases to consider.
\medskip
\noindent
\\\textbf{(I) $\D_1=\pre_{n-1}\D'_1,$ and $\D_2$ is not of this form.}
In this case the rightmost top chord of $\D_2$ ends at $(n-1,n-2).$ If this chord is black, then by \cref{lem:sign_of_chord_twistors} the functionary $F=\Fdel_{k}^{\D_2}=\llrr{a_k,b_k,n-2,n}$ is strongly negative on $S_{\D_2},$ and is a reduced BCFW functionary for $\D_2.$ 
If the rightmost top chord is red, then by \cref{lem:sign_of_chord_twistors}
$F=\Fdel_{\star}^{\D_2}=\llrr{A,B,n-2,n}$ is strongly negative on $S_{\D_2},$ and is a reduced BCFW functionary for $\D_2.$ 
In both cases $F$ is a boundary functionary of $\Ampl_{[n]\setminus\{n-1\},k,4}^\ell$, hence by \cref{cor:cell_is_not_SA} it is strongly positive on $S_{\D'_1}$. By \cref{prop:sign_under_pre_promotion} $\Psi^\pre F$ has the same strong sign on $S_{\D_1}.$
The same argument works also if $\D_2=\pre_{n-1}\D'_2$ and $\D_1$ is not of that form.

\medskip
\noindent
\textbf{(II) $\D_2=\FL(\D'_2),$ and $\D_1$ is not of this form. }
The case $\D_1=\pre(\D'_1)$ has been treated above. It is enough to consider $\D_1=\D_L\bcfw_{ac} \D_R,$ and $(c,d)=(n-2,n-1).$ If $\D_R\in\CD_{\{b,b+1,\ldots,c,d,n\},k'}^1,$ consider the boundary functionary $\llrr{ABdn}$ on $\D_R.$ It is a strongly positive functionary, by \cref{cor:cell_is_not_SA}. Thus, by \cref{prop:sign_under_BCFW_promotion}
\[\Psi^\bcfw_{ac}\llrr{ABdn}=-\frac{\llrr{nab|cd|ABn}}{\llrr{abcn}}\]has a weak sign $+$ on $S_{\D_1}$. Using \cref{lem:sign_of_chord_twistors} this implies that $\llrr{nab|cd|ABn}\geq0$ on $S_{\D_1}$.  

If $\D_L\in\CD_{\{1,2,\ldots,a,b,n\},k'}^1$ we act similarly, only that this time we employ the twistor $\llrr{ABbn}$, which is a strongly positive boundary twistor for $S_{\D_L}$. 
\[\Psi^\bcfw_{ac}\llrr{ABbn}=-\frac{\llrr{ncd|ab|ABn}}{\llrr{acdn}}=-\frac{\llrr{nab|AB|cdn}}{\llrr{acdn}},\]
where the second equality follows from the Pl\"ucker relations. 
By \cref{prop:sign_under_BCFW_promotion} this functionary
has weak sign $(-1)^{k_R+1}$ on $S_{\D_1}$. Using \cref{lem:sign_of_chord_twistors} this shows that $\llrr{nab|AB|cdn}\leq 0$ on $S_{\D_1}.$ 

On 
$S_{\D_2}$ the same functionary equals $\frac{\llrr{a,b,Z_\star,Z_0}}{\llrr{abAB}},$ by \cref{lem:quad_lemma}, which is nonnegative by \cref{cor:sign_del0_on_FL} and \cref{lem:sign_bdry_twistors}.

Note that if $\D\in\CD_{n,k}^1$ is a chord diagram of the form $\D_L\bcfw_{ac}\D_R$ where $\D_R=\FL(\D'_R),$ then $\llrr{nab|AB|cdn}$ is $\frac{\Fdelr^\D}{F_{\heps_0}^\D}$, which is Laurent monomial in the reduced BCFW coefficients and $\Fdelr^\D.$

Of course the same argument works if $\D_1=\FL\D'_1$ and $\D_2$ is not of that form. 

\medskip
\noindent
\textbf{(III) $\D_1=\D_{L_1}\bcfw_{a,c}\D_{R_1},~\D_2=\D_{L_2}\bcfw_{a',c'}\D_{R_2}$.}
This is the last case to consider. From the previous cases we may assume $c=c'=n-2,~a\neq a'.$ Without loss of generality $a<a'.$ Write $b=a+1,b'=a'+1,d=n-1.$

First, if $k=1$ then $\llrr{a'cdn}$ is strongly positive on $S_{\D_1}$ and strongly negative on $S_{\D_2}$ by \cref{lem:Cauchy-Binet}, since all terms in \eqref{eq:cauchy-binet} are easily seen to have the same plus sign on $S_{\D_1}$ and minus sign on$S_{\D_2}$. This functionary is a reduced BCFW functionary for $S_{\D_2},$ hence satisfies our requirements.

Assume $k>1.$  
We will show that $F=\llrr{na'b'|cd|abn}$ is a separating functionary. 
$\llrr{a'b'dn}$ is a strongly positive boundary twistor of $S_{\D_{R_1}},$ by \cref{cor:cell_is_not_SA}. It is not a loopy twistor, hence by \cref{prop:sign_under_BCFW_promotion} 
\[\Psi^\bcfw_{ac}(\llrr{a'b'dn})=\frac{\llrr{na'b'|cd|abn}}{\llrr{abcn}}>0\]on $S_{\D_1}.$ Combining with \cref{lem:sign_of_chord_twistors}, this shows $F=\llrr{na'b'|cd|abn}<0$ on $S_{\D_1}.$

Similarly, $\llrr{abb'n}$ is a strongly positive boundary functionary on $S_{\D_{L_2}},$ by \cref{cor:cell_is_not_SA}. Again it is not a loopy twistor. Hence, by \cref{prop:sign_under_BCFW_promotion}
\[\Psi^\bcfw_{a'c}(\llrr{abb'n})=-\frac{\llrr{ncd|a'b'|abn}}{\llrr{a'cdn}}=\frac{\llrr{na'b'|cd|abn}}{\llrr{a'cdn}}\]has strong sign $(-1)^{k_R+1}$ on $S_{\D_2}.$ The second equality follows from the Pl\"ucker relations. Thus, using \cref{lem:sign_of_chord_twistors} $F=\llrr{na'b'|cd|abn}>0$ on $S_{\D_2}.$

Finally, for $\D = \D_L\bcfw\D_R\in\CD_{n,k}^\ell$ with $\D_R$ having rightmost top chord $(a',b',c,d)$ the twistor $\llrr{a'b'dn}$ is a reduced BCFW functionary, hence by \cref{prop:bcfw_coords_promote_2_bcfw_coords} $F$ is indeed a Laurent monomial in reduced BCFW variables.

\medskip
\noindent
We now assume that $\D_1,\D_2$ are not their minimal neighboring subdiagrams, but we have shown the claim for lexicographically smaller $(\ell-\bar\ell,k-\bar k,n-\bar n).$
Then by \cref{obs:minimal_neighboring} the last operation in the generation sequence of $\D_1,\D_2$ is the same. We split into cases according to what this operation is.
\\\textbf{(i) $\D_1=\pre_{n-1}\D'_1,~\D_2=\pre_{n-1}\D'_2.$ }

\noindent
If both $D_1$ and $D_2$ have no chord endings at $(n-1,n-2)$ then we can write
\[\D_1=\pre_{n-1}\D'_1,~\D_2=\pre_{n-1}\D'_2,~\D'_1,\D'_2\in\CD_{[n]\setminus\{n-1\},k}^{\ell}\]which by \cref{obs:minimal_neighboring} are also neighboring, but have smaller value of $(\ell-\bar\ell,k-\bar k,n-\bar n).$ By induction we can find a separating functionary $F$ for $S_{\D'_1},S_{\D'_2}$ which is a Laurent polynomial in reduced BCFW functionaries of a $S_{\D'}$ for some $\D'\in\CD_{[n]\setminus\{n-1\},k}^\ell.$ $\Psi^\pre F$ is a Laurent monomial in the reduced BCFW functionaries of $S_{\pre_{n-1}\D'}$ by \cref{prop:bcfw_coords_promote_2_bcfw_coords}, and $\pre_{n-1}\D'\in\CD_{n,k}^\ell$. 
By \cref{prop:sign_under_pre_promotion} it is a separating functionary between $S_{\D_1},S_{\D_2}.$ If the sign of $F$ is strong on the cells, so it the sign of $\Psi^\pre F,$ by \cref{prop:sign_under_pre_promotion}.

\medskip
\noindent
\textbf{(ii) $\D_1=\D_{L_1}\bcfw_{a,c}\D_{R_1},~\D_2=\D_{L_2}\bcfw_{a,c}\D_{R_2}.$}

\noindent
We use the notations $a,b,c,d$ as in \cref{nn:bcfwmap}. In this case, by \cref{obs:minimal_neighboring}, either $\D_{L_1}, \D_{L_2}$ are neighboring and $\D_{R_1}=\D_{R_2}$ or $\D_{R_1}, \D_{R_2}$ are neighboring and $\D_{L_1}=D_{L_2}$. Denote the neighboring subdiagrams $\D'_1,\D'_2\in\CD_{N,k'}^{\ell'}.$ In both cases $(\ell'-\bar\ell,k'-\bar k,|N|-\bar n)$ is smaller that $(\ell-\bar\ell,k-\bar k,n-\bar n)$ By induction we can find a separating functionary $F$ for the neighboring subdiagrams. This functionary is a Laurent monomial in the reduced variables for a BCFW cell $S_{\D'}$,~$\D'\in\CD_{N,k'}^{\ell'}.$
By \cref{prop:bcfw_coords_promote_2_bcfw_coords} $\Psi^\bcfw_{ac} F$ is again a Laurent monomial in reduced BCFW variables for a cell in $\D''\bcfw_{ac}\D'\in\CD_{n,k}^\ell.$ 
By \cref{prop:sign_under_BCFW_promotion} it is a separating functionary between $\D_1,\D_2$. Moreover, if $\D_1,\D_2$ are tree BCFW cells then by induction and \cref{prop:sign_under_BCFW_promotion} the separator functionary has opposite strong signs on the images of the two cells. 
\medskip
\noindent
\\\textbf{(iii) $\D_1=\FL(\D'_1),~\D_2=\FL(\D'_{2}).$}

\noindent In this case $\ell=1,$ the rightmost top chord of $\D'_1,\D'_2$ is not short and ends at $(c_\star,d_\star).$
By \cref{obs:minimal_neighboring} $\D'_1,\D'_2$ are neighboring tree diagrams.
Let $F$ be a separating functionary between $S_{\D'_1},~S_{\D'_2}$ which can be written as a Laurent monomial in the reduced BCFW variables of $S_{\D'}$ for some $\D'\in\CD_{[n]\cup\{A,B\},k+1}^0$. The existence of $F$ follows from the induction. By \cref{prop:sign_under_FL_promotion} $\Psi^\bcfw_{ac} F$ has opposite weak signs on $S_{\D_1},S_{\D_2},$ and by \cref{prop:bcfw_coords_promote_2_bcfw_coords} it is a Laurent monomial in BCFW variables of $S_\D:=S_{\FL(\D')}.$
Following closely the steps $I,III,i,ii$ of the construction of the separating functionary for $S_{\D'_1},S_{\D'_2}$ we see that $\D$ also has a rightmost top chord $(a_0,b_0,c_\star,d_\star)$, which is not short. Moreover $F_{\hgam_{\tr(\D)}}^{\D'}$ is not used in the separating Lauren monomial $F$. Using \cref{obs:pre_perm_BCFW} and \cref{obs:pre_permissible} we learn that $G=F\llrr{c_\star d_\star Bn}^l$ is pre-permissible, for some $l,$ and is still a separating functionary since $\llrr{c_\star d_\star Bn}>0$ on $\D'_1,\D'_2$ by \cref{cor:cell_is_not_SA}.
Thus $\Psi^\FL G$ is a permissible separating functionary. Note that $\Psi^\FL\llrr{c_\star d_\star Bn}=\llrr{c_\star d_\star Bn}.$ By the forward limit case of \cref{prop:bcfw_coords_promote_2_bcfw_coords} $\Psi^\FL G$ differs from a Laurent monomial in reduced BCFW variables of $\D$ by $\llrr{na_0b_0|AB|c_\star d_\star n}^m,$ for some $m.$ Since $\llrr{na_0b_0|AB|c_\star d_\star n} = \Fdelr^{\D}/F_{\heps_{\tr(\D)}}^{\D}$, which follows from \cref{lem:quad_lemma}, and since the numerator and the boundary twistor in the denominator are  positive on $S_{\D_1},S_{\D_2}$, by \cref{cor:sign_del0_on_FL} and \cref{cor:cell_is_not_SA} respectively, we deduce that $H=\Psi^\FL(G)\llrr{na_0b_0|AB|c_\star d_\star n}^m$ is a separating functionary which is a Laurent monomial in reduced BCFW variabels of $S_\D.$


To summarize, we have shown that for any two neighboring BCFW cells $S_{\D_1},S_{\D_2}\in\CD_{n,k}^\ell$ there is a permissible functionary which has opposite weak signs on their images under the amplituhedron map. 
Moreover, this functionary is a Laurent polynomial in reduced BCFW coordinates of a cell $S_\D$ for $\D\in\CD_{n,k}^\ell.$ By \cref{cor:strong_pos_coords} this functionary is not identically $0$ on the amplituhedron. Since it is algebraic and $\Gr_{k,k+4;1}$ is reduced, this means that the complement of its zero locus is Zariski open and dense, and also Hasudorff open and dense. We know that for every positive $Z$ $\tZ(S_{\D_1}),\tZ(S_{\D_2})$ is open, by \cref{cor:open_map}. Thus, the proposition immediately follows from 
\begin{obs}\label{obs:weak_sign_separator}
Let $X$ be a manifold and $F:X\to\R$ a function. Denote by $Y$ the zero locus of $F,$ and assume $Y$ has an empty interior. Let $U,V\subseteq X$ be open sets for which $F|_U\geq 0,~F|_V\leq 0.$ Then $U\cap V=\emptyset.$
\end{obs}
\begin{proof}
Assume towards contradiction the existence of $p\in U\cap V.$ Then $p$ has an open neighborhood $W$ also contained in this intersection. $W\setminus Y\neq\emptyset.$ But on $U\setminus Y,~V\setminus Y$ $F$ has opposite non zero signs, which is a contradiction.
\end{proof}
\end{proof}
Even though we do not need it, we point out that the separation argument can be slightly adjusted to show the following claim:
if $S_\D,S_{\D'}$ are neighboring BCFW cells, the separating functionary $F$ can actually be taken to be a functionary vanishing on the $\tZ-$image of their common boundary.
\begin{prop}\label{prop:cell_separated_from_all}
Let $\D$ be the unique chord diagram in $\CD_{n,k}^\ell$ in which all chords' supports are contained in $1,\ldots,k+3.$ Then for all positive $Z,$ and every other chord diagram $\D'\in\CD_{n,k}^\ell$ 
\[\tZ(S_\D)\cap\tZ(S_{\D'})=\emptyset.\]
\end{prop}
\begin{proof}
$\D$ has the property that the generation sequence of $S_\D$ ends by $n-k-4$ $\pre_p$ operations. For any other BCFW cell of the same $\ell,k,n$ there are less $\pre_p$ operations in the suffix of the generation sequence.

Let $\D'\in\CD_{n,k}^\ell$ be another chord diagram. As in \cref{prop:local_sep} we will find a permissible separating functionary between $\D,\D'$ which is a Laurent monomial in reduced BCFW coordinates of  $S_{\D'}.$

We will induct on $n.$ If $n=k+4$ there is no other chord diagram except $\D$, and there is nothing to prove. Otherwise, if $\D'=\pre_{n-1}\D_2,$ we can write $\D=\pre_{n-1}\D_1.$ By induction we can find a separating functionary $F$ between $\D_1,\D_2$ which is Laurent monomial in the reduced BCFW coordinates of $S_{\D_2}.$ As in case $(i)$ in the proof of \cref{prop:local_sep}, $\pre_{n-1} F$ is a separating functionary for $\D,\D'$ which is Laurent monomial in the BCFW coordinates of $S_{\D'}.$

The other possibility is that $\D'$ has a chord ending at $(n-2,n-1).$ In this case we proceed as in case (I) of the proof of \cref{prop:local_sep}, to find a functionary $F$ representing a reduced BCFW coordinate of $S_{\D'}$ which serves as a separating functionary.
The induction follows.

Using $F$, the proof follows from \cref{obs:weak_sign_separator}, just like the end of the proof of \cref{prop:local_sep}.
\end{proof} \subsection{Surjectivity and separation}
In order to prove \cref{thm:surj} and \cref{thm:sep} we need to equip ourselves with several topological and geometric lemmas.
\begin{prop}\label{prop:criterion_for_surj}
Let $f:B\to N$ be a smooth submersion 
between two smooth manifolds (\emph{without boundary}) $B,N,$ where $\dim(N)=n$. Let $L$ be an open connected subspace of $B$ with a compact closure $\overline{L}\subset B.$ Write $K$ for $f(\overline{L})$. Let $M\subset N$ be a subspace with empty interior, such that $M\cup\partial K$ is closed and for every open connected $U\subseteq N,$ $U\setminus M$ is connected. 
Let $\{S_a\}_{a\in A}$ be a finite collection of $n$-dimensional topological and weak smooth submanifolds \emph{without boundary} of $B,$ which are contained in $\overline{L}$ and satisfy the following properties:
\begin{enumerate}
    \item $\overline{S_a}$ is compact. 
    $\partial \overline{S_a}=S'_a\cup\bigcup_{i=1}^{k_a} S_{a;i},$ 
    where each $S_{a;i}$ is an $n-1$-dimensional manifold (\emph{without boundary}), $S_a\cup\bigcup_{i\in[k_a]}S_{a;i}$ is a topological and a weak smooth submanifold with boundary of $B$, and \[f(S'_a)\subseteq\partial K\cup M.\] 
    \item 
    Every $S_{a;i}$ equals $S_{b;j}$ for some other $(b;j).$ 
    Boundaries of the this type are called \emph{shared boundaries}, and we refer to $S_a,S_b$ in this case as neighboring.
    \item For every $a\in A,$ $f$ is injective on the union $S_a$ and its shared boundaries.
    \item $f(S_a)\cap f(S_b)=\emptyset,$ for neighboring $a, b.$
\end{enumerate}
Then 
\begin{enumerate}
\item \[f(\bigcup_{a\in A}\overline{S_a})=K.\]
\item $f(S_{a;i})\cap \partial K=\emptyset$, for every shared boundary $S_{a;i},$ and 
\[\partial K=\bigsqcup_a\overline{f(\partial\overline{S_a})\cap\partial K\setminus M)}.\]
\item Write $R=f^{-1}\left(\bigcup_{a\in A} (S_a\bigcup_{i\in[k_a]}S_{a;i})\setminus M\right).$ Then $R$ is open and dense in $\bigcup_{a\in A} (S_a\bigcup_{i\in[k_a]}S_{a;i}),$ and $f|_R:R\to f(R)$ is a covering map between manifolds. Thus, for every $x\in f(R)$ has a constant number of preimages.
\end{enumerate}
\end{prop}
This proposition is similar to the analogous propositions in \cite{even2021amplituhedron} and \cite{perlstein2025bcfw}. There are differences since here we require weaker control over the different strata $\overline{S}_a,$ since in the application to the loop amplituhedron below we have weaker control over the corresponding strata.
\begin{proof}
It follows from the statement that $int(K)$ is connected (see first sentences of the analogous proof of \cite[Proposition 8.5]{even2021amplituhedron}), hence by the assumption on $M$ also $int(K)\setminus M$ is open and connected, and $\overline{int(K)\setminus M}=K$.

We now show that $T=\bigcup_a f(S_a\cup\bigcup_{i\in[k_a]} S_{a;i})$ is a manifold: By the first condition $S_a\cup S_{a;i}$ is a smooth manifold with boundary for every $i\in [k_a].$ If $S_{a;i}=S_{b;j}$ We can glue the manifold with boundaries $S_a\cup S_{a;i},~S_b\cup S_{b;j}$ along their common boundary to obtain a \emph{topological manifold} (without boundary). Using the third and fourth condition this manifold maps injectively to $N,$ hence by Brouwer's Invariance of Domain Theorem \cite[Proposition 7.4, Ch. IV]{Dold} $f$ restricted to the union is a homeomorphism on the image, and the image is a (topological) $n-$manifold. Taking the union over $a\in A,i\in[k_a]$ we see that $T$ is a (topological) $n-$manifold. Note that it also follows that $f$ is an open map, when restricted to the gluing of ${\bigcup_{a\in A}S_a\cup\bigcup_{i\in[k_a]}S_{a;i}}$ along the common boundaries. 

Since $\partial K$ is compact with an empty interior, and $T\subseteq K,$ it follows that $T\cap\partial K=\emptyset,$ or $T\subseteq int(K).$
Using the assumptions on $M$ and since $\partial K$, we  have that $$T\setminus M=\bigcup_a f(S_a\cup\bigcup_{i\in[k_a]} S_{a;i})\setminus M\subseteq int(K)\setminus M$$is a non empty topological $n-$manifold, hence it is open in $int(K)\setminus M.$ But from the first condition it is also relative closed in $int(K)\setminus M.$ Hence, since $int(K)\setminus M$ is connected, we have $T\setminus M=int(K)\setminus M$. By the above, and the compactness of each $\overline{S}_a$ it follows that $f(\bigcup_a\overline{S_a})=K.$ 

For the second item, let 
\[S' = \bigcup_a S'_a=\bigcup_a (\partial\overline{S}_a\setminus(\bigcup_{i} S_{a;i})),\qquad\qquad S = S'\setminus f^{-1}(M).\]
It follows from above that $\partial K\subseteq f(S').$ Note that $f(S)\subseteq\partial K,$ by the first condition, hence also $\overline{f(S)} \subseteq\partial K,$ and in addition $f(S')\setminus f(S)\subseteq M.$
We need to show that $\partial K=\overline{f(S)}.$ Assume there is a point $p\in\partial K\setminus \overline{f(S)}$. Then $p$ has a connected neighborhood $U$ with $U\cap\overline{f(S)}=\emptyset,$ hence \begin{equation}\label{eq:K,M bdry}U\cap\partial K\subseteq U\cap f(S')\setminus f(S)\subseteq M.\end{equation} Since $p$ is a boundary point $U$ must contain points of $int(K)$ and of $K^c.$ Now $U\setminus M$ is connected, and $int(K)\cap U\setminus M$ is open and nonempty. But it is also relatively closed in $U\setminus M,$ by \eqref{eq:K,M bdry}. Thus, the two set must be equal, hence $(U\setminus M)\cap K^c=\emptyset.$ But since $M$ has empty interior, this contradicts $U\cap K^c$ being open and non empty. And the proof follows. 

For the last item, we will show that $f|_R:R\to f(R)=T\setminus M=int(K)\setminus M$ is a covering map. Since $f(R)=int(K)\setminus M$ is connected, it will also imply the claim about the cardinalities of preimages of points in $int(K)\setminus M$.
It follows from above that $f$ is a local homeomorphism from $\bigcup_{a\in A} (S_a\bigcup_{i\in[k_a]}S_{a;i})$ onto its image. Since $f(S_a),f(S_{a;i})$ avoid $\partial K$ for every $i,a$, $M$ has no interior, and $f$ restricted to $\bigcup_{a\in A} (S_a\bigcup_{i\in[k_a]}S_{a;i})$ is open, $R$ is open and dense in $\bigcup_{a\in A} (S_a\bigcup_{i\in[k_a]}S_{a;i}),$ and maps locally homeomorphically onto $int(K)\setminus M.$ 

In order to show that $f|_R$ is a covering map, we just need to show it is proper (see e.g. \cite[Lemma 2]{covering_proper}). To this end, let $Q$ be a compact subset of $int(K)\setminus M.$ We must show that $f|_R^{-1}(Q)$ is compact. Since we work with manifolds, it is enough to prove sequential compactness. Let $y_1,y_2,\ldots\in f|_R^{-1}(Q)$ be an arbitrary sequence. We will find a converging subsequence. Let $x_i=f(y_i)\in Q.$ Since $Q$ is compact we can find a converging subsequence, which, without loss of generality, may be taken to be $x_1,x_2,\ldots.$ Thus, $\lim_{i\to\infty }x_i=x\in Q.$ By the third item of the assumptions, and the first item of the conclusions, each $x_i$ has finitely many preimages in $R,$ thus, there is at least one $a\in A$ such that infinitely many of the $y_i$'s belong to it. By passing to a subsequence and relabeling the $x_i$'s we may assume that every $y_i\in \overline{S}_a.$ Since $\overline{S}_a$ is compact there exists $y\in \overline{S}_a$ with 
\[\lim_{i\to\infty} y_i = y.\]
If $y\in R$ we are done. Otherwise, by the first item, $y$ must map to $M\cup\partial K.$ But $f(y)=x\notin M\cup\partial K,$ by assumption and the second item of the conclusion. Thus, $y\in R$ showing that $f|_R$ is proper, hence a covering map.
\end{proof}
\begin{rmk}\label{rmk:why_weak_smooth_etc}
    We used that $S_a\cup\bigcup_{i\in[k_a]}S_{a;i}$ is a topological and smooth submanifold of $B$ for the gluing argument: 
    The topological embeddings guarantee that open sets in $S_{a;i}$ as a boundary of $S_a$ are the same as open sets as a boundary of $S_b,$ where $S_a,S_b$ are neighboring at $S_{a;i}=S_{b;j}.$ Being a smooth, or even continuous, injective image of a smooth manifold with boundary guarantees the existence of collar neighborhood which we can glue in order to construct the manifold topology of the glued space. 
\end{rmk}
The next result is analogous to \cite[Lemma 8.1]{even2021amplituhedron}, and it proven via the exact same argument
{\begin{lemma}\label{lem:smooth_extension}The amplituhedron map can be extended, via the same formula to a smooth submersion
\[\tZ:B\to \Gr_{k,k+m},~~V\mapsto \tZ(V)=V\cdot \bz,\]
where $B\subset\Gr_{k,n;1}$ is an open neighborhood of $\Gr^{\geq}_{k,n;1}.$
\end{lemma}}

Our next lemma is a simple variant of a nice argument from \cite{4168411,4169287}. For a proof see \cite[Lemma 10.4]{perlstein2025bcfw}
\begin{lemma}\label{lem:connectedness_in_complement_stratification}
Let $X$ be smooth connected manifold of dimension $n$, and $D_1,\ldots,D_h$ smooth manifolds of dimensions at most $n-2.$ Let $f_i:D_i\to X$ be smooth injections, for $i=1,\ldots, h.$ Then $X\setminus \bigsqcup f_i(D_i)$ is connected.
\end{lemma}
\begin{proof}[Proof of \cref{thm:surj}]
Write
\begin{itemize}
    \item $f=\tZ,$ the amplituhedron map;
    \item $N=\Gr_{k,k+4;1}.$ The set $M$ is the union of images of boundary strata of BCFW cells which are neither internal mathcable, nor included in $\SA$. 
    \item 
    $L=\Gr^{>}_{k,n;1}$ and $B\subset\Gr_{k,n;1}$ is an open set containing $\overline{L}=\Gr_{k,n;1}^{\geq}$.
    \item Finally $\{S_a\}=\{S_\D\}_{\D\in\CD_{n,k}^1}$ are the BCFW cells, and the boundary strata $S_{a;i}$ are the internal matchable boundaries of $S_\D$, that is, the matchable strata which are regular with respect to at least one of the two cells containing them.
\end{itemize}
We will verify that the assumptions of~\cref{prop:criterion_for_surj} hold, with these $f,L,B$ and the collection $\{S_a\}$, for all positive $Z.$ The first conclusion of \cref{prop:criterion_for_surj} will then imply that for all positive $Z$ the images of BCFW cells cover the amplituhedron.
\\\underline{Assumptions on $L,B,N$:} 
$B,N$ are clearly manifolds. $L$ is connected by \cref{lem:connected}. \\\underline{Assumptions on $f$:} $f$ is a smooth submersion by \cref{lem:smooth_extension}. 
\\\underline{Assumptions on $M$:} 
\\\underline{Empty interior.} $M$ is the union of images of subspaces $X_1,\ldots,X_N$, which are contained in $\Srem_\D$ for some $\D,$ by definition and \cref{lem:from SREDSBLUE to smaller diags in diags order}. They are the smooth injective images of smooth manifolds of dimensions at most $\dim N-2,$ by \cref{cor:strata_and_dim}. 
$f$ is injective on every $X_i$ by \cref{prop:inj_bdries}. Hence by the classical Morse-Sard's theorem \cite[Chapter 3]{Hirsch} $int(M)=\emptyset.$
\\\underline{$M\cup\partial K$ is closed.} 
By \cref{prop:CD bdries} $f(\SA)\subseteq\partial K,$ and both are easily seen to be closed, hence compact sets.

The spaces $X_i$ are of the form $\partial_AS_\D\subseteq\Srem_\D,$ for some $A\subseteq\tVar_\D,$ by \cref{lem:from SREDSBLUE to smaller diags in diags order} and \cref{cor:reg_domino_strata}. Every point in the closure of $\overline{X}_i$ is therefore either in $\SA,$ or in $\partial_{A'}S_{\D'}$ where $\D'\leq\D$ in the partial order of \cref{lem:from SREDSBLUE to smaller diags in diags order}, and it is easy to see that unless $A'=A,\D'=\D$ then $|A'|>1,$ in which case the point belongs to some other $X_j.$
Thus \begin{equation}\label{eq:M_partialK}\overline{
\bigcup_i X_i}\subseteq \bigcup_i X_i\cup\SA,\end{equation}and the right hand side is closed, hence compact. Since $K$ is compact and contains $f(\SA)$, also $M\cup\partial K$ is compact: if $x_1,x_2,\ldots $is a sequence of points in $M$, we can find for them preimages $y_1,y_2,\ldots$ in $\bigcup_i X_i.$ Since $\overline{\bigcup_i X_i}$ is compact, we can extract a converging subsequence, $y_{i_1},y_{i_2},\ldots$. Its limit $y$ belongs to $\SA\cup\bigcup_i X_i$, hence \[\lim_{j\to\infty}x_{i_j}=\lim_{j\to\infty}f(y_{i_j})=f(y)\in M\cup \partial K.\]
\\\underline{$U\setminus M$ is connected for every open $U.$} 
Since for all positive $Z,$ $M$ is the union of injective images of finitely many manifolds of codimensions $2$ or more, the connectedness of $U\setminus M,$ for an open and connected $U\subseteq N,$ 
follows from applying \cref{lem:connectedness_in_complement_stratification}, with $X=U$ and each $D_i$ is a smooth manifold which maps smoothly and injectively to the weak smooth submanifold $X_i.$
\\\underline{Assumptions on $\{S_a\}_{a\in A}:$}
We turn to verify the assumptions on $S_{a;i}.$ The first assumption follows from 
 \cref{cor:strata_and_dim}, \cref{lem:from SREDSBLUE to smaller diags in diags order} and \cref{prop:mfld_w_bdry}. The second is \cref{lem:bdry_matching_geo}. The fourth assumption is \cref{prop:local_sep}.

We turn to the third condition.
\cref{prop:inj_bdries} shows that the map is injective when restricted to the union of $S_\D$ and its regular matchable boundaries. 
Its image avoids the images of $\Sred_\D,\Sblue_\D$ by \cref{lem:Sinj_strata},~\cref{lem:eps_blue_vanishes_on_blue},~\cref{lem:eps_red_vanishes_on_red}.
The blue boundary of $S_\D,$ if it is internal matchable, is contained in $\Sblue_\D$ and avoids $\Sred_\D$ by \cref{cor:regular_matchable_and_S_D_in_SREM}. By \cref{lem:eps_red_vanishes_on_red} its image does not intersect the image of $\Sred_\D.$ The red boundary, if is an internal mathcable boundary, is contained in $\Sred_\D$ by \cref{cor:red_bdry_in_sred}. Thus, the union of $S_\D$ and its internal matchable boundaries maps injectively. It only remains to show that the red and blue boundaries, if they are matchable, also map injectively. But in this case they are also regular matchable boundaries of other cells, hence map injectively, by \cref{prop:inj_codim_1}. 

We will show a slightly strong claim. Let 
Let $T_{\zeta_i}$ be the gluing of $S_\D\sqcup\partial_{\zeta_i}S_\D$ and $S_{\D'}\sqcup\partial_{\zeta_i}S_\D$, along their common boundary, where 
$\D'=\shift_{\zeta_i}\D.$  By \cref{prop:mfld_w_bdry}, $T_{\zeta_i}$ is a manifold without boundary of dimension $4k+4,$ obtained as the gluing of two smooth manifolds with boundary along their common boundary. We will show that $\tZ$ is injective on $T_{\zeta_i}$ with a continous inverse, which will imply the claim. It is injective on $S_\D,S_{\D'}$ by \cref{thm:inj}. It is injective on $\partial_{\zeta_i}S_\D$ by \cref{prop:inj_codim_1}. The images of $S_\D,S_{\D'}$ are disjoint by \cref{prop:local_sep}. They are also disjoint from the image of $\partial_{\zeta_i}S_\D,$ as shown in the previous paragraph.  Thus, the amplituhedron map is injective when restricted to $T_{\zeta_i}.$  
Hence, by Brouwer's Invariance of Domain Theorem \cite[Proposition 7.4, Ch. IV]{Dold} the continuous map $\tZ$ maps $T_{\zeta_i}$ homeomorphically onto its image in $\Gr_{k,k+4;1},$ which must be an open submanifold of $\Gr_{k,k+4;1},$ and the inverse map is also continuous.

Thus, \cref{prop:criterion_for_surj} shows that for all positive $Z$
\[\Ampl_{n,k,4}^1(Z)=\bigsqcup_{\D\in\CD_{n,k}^1} \tZ(\overline{S}_\D).\]
\\\textbf{The moreover part: } The second conclusion of \cref{prop:criterion_for_surj} shows that
for all positive $Z,$
\[\partial \Ampl_{n,k,4}^1(Z) = \bigsqcup_{\D\in\CD_{n,k}^1} \overline{\left(\tZ(\overline{S}_\D)\cap\partial \Ampl_{n,k,4}^1(Z)\setminus M\right)}.\]
The above discussion shows that every point of a BCFW cell or an internal matchable boundary maps to an internal point of the amplituhedron. Thus, from the definition of $M$
\[\tZ(\overline{S}_\D)\cap\partial \Ampl_{n,k,4}^1(Z)\setminus M\subseteq \tZ(\overline{S}_\D\cap \SA),\]and since the right hand side is closed, it also contains the closure of the left hand side.
Thus, for all positive $Z$ 
\[\partial \Ampl_{n,k,4}^1(Z)\subseteq
\bigsqcup_{\D\in\CD_{n,k}^1} \tZ(\overline{S}_\D\cap \SA).\] Combining with \cref{prop:CD bdries} we deduce that 
\[\partial\Ampl_{n,k,4}^1(Z)=\Ampl_{n,k,4}^1(Z)\cap\left(\bigsqcup_{i,i+1,j,j+1\in[n]~\text{are all distinct}}V(\llrr{i,i+1,j,j+1})\sqcup\bigsqcup_{i\in[n]}V(\llrr{i,i+1,A,B})\right),\]
where as usual $+1$ is taken cyclically.
The theorem follows.
\end{proof}
\begin{proof}[Proof of \cref{thm:sep}]
The proof of \cref{thm:surj}, shows that with $f=\tZ,N=\Gr_{k,k+4},\{S_a\}_{a\in A}$ being the BCFW cells, and $M$ is the subspace described in that proof, the assumptions of \cref{prop:criterion_for_surj} hold. Thus, by the third conclusion of \cref{prop:criterion_for_surj}, for every point in an open dense subset of the amplituhedron, the number of preimages inside the union of the BCFW cells is constant. But this constant must be $1,$ since by \cref{prop:cell_separated_from_all} there is, for every $k,n$, a BCFW cell whose image is separated from the images of all other BCFW cells. This image is open by \cref{cor:open_map}, hence intersects the above open dense set.

This implies that the images of every two BCFW cells are disjoint. Indeed,  had two such images intersected, by \cref{cor:open_map} again, they would have intersected in an open set. Every point in this open set would have had at least two preimages, which is a contradiction.
\end{proof}

\bibliographystyle{alpha}
\bibliography{1Lbib}
\end{document}